\documentclass[12pt]{article}
\usepackage[utf8]{inputenc}
\usepackage[english]{babel}
\usepackage[T1]{fontenc}
\usepackage{hyperref}
\usepackage{amsmath}
\usepackage{amsthm}
\usepackage{amsfonts}
\usepackage{amssymb}
\usepackage{graphicx}
\numberwithin{equation}{section}
\usepackage{graphicx}
\usepackage{tikz}
\usetikzlibrary{calc}
\usepackage{framed}
\usepackage{caption}
\usepackage{float}
\usepackage{gensymb}
\usetikzlibrary{arrows}
\tikzset{>=latex}
\usetikzlibrary{patterns}
\usepackage{authblk}
\usepackage[top=2.5cm, bottom=2.5cm, left=2cm, right=2cm]{geometry}
\usepackage{fancyhdr}
\usepackage{cuted}
\usepackage{appendix}
\usepackage{chngcntr}

\newtheorem{theo}{Theorem}[section]
\newtheorem{prop}[theo]{Proposition}
\newtheorem{proper}[theo]{Property}
\newtheorem{coro}[theo]{Corollary}
\newtheorem{lemm}[theo]{Lemma}

\newtheorem{defi}[theo]{Definition}
\theoremstyle{definition}

\newtheorem{rmk}[theo]{Remark}
\newtheorem{conv}[theo]{Convention}

\newcommand{\N}{\mathbb{N}}
\newcommand{\A}{\mathcal{A}}

\newcommand{\R}{\mathbb{R}}
\newcommand{\C}{\mathcal{C}}
\newcommand{\Z}{\mathbb{Z}}

\newcommand{\Prob}{\mathbb{P}}
\newcommand{\E}{\mathbb{E}}
\newcommand{\e}{\mathcal{E}}
\newcommand{\lcal}{\mathcal{L}}

\newcommand{\T}{\mathbb{T}}

\renewcommand{\epsilon}{\varepsilon}

\newcommand\xqed[1]{%
  \leavevmode\unskip\penalty9999 \hbox{}\nobreak\hfill
  \quad\hbox{#1}}
\newcommand\demo{\xqed{$\blacksquare$}}

\usepackage[]{hyperref}
\hypersetup{
    colorlinks=false,
}

\makeatletter
\newcommand{\subjclass}[2][2020]{%
  \let\@oldtitle\@title%
  \gdef\@title{\@oldtitle\footnotetext{#1 \emph{Mathematics subject classification.} #2}}%
}
\newcommand{\keywords}[1]{%
  \let\@@oldtitle\@title%
  \gdef\@title{\@@oldtitle\footnotetext{\emph{Key words:} #1.}}%
}
\makeatother

\title{Motion by curvature and large deviations for an interface dynamics on $\Z^2$}
\author{Benoit Dagallier}
\affil{Courant Institute of Mathematical Sciences, New York University.\\ E-mail: {\tt bd2543@nyu.edu}}
\date{}

\subjclass{60F10, 82C22, 82C24}
\keywords{Large deviations, interface dynamics, motion by curvature, Ising model}

\begin{document}
\maketitle

\noindent\textbf{Abstract:}
We study large deviations for a Markov process on curves in $\Z^2$ mimicking the motion of an interface. Our dynamics can be tuned with a parameter $\beta$, which plays the role of an inverse temperature and coincides at $\beta=\infty$ with the zero temperature Ising model Glauber dynamics, where curves correspond to the boundaries of droplets of one phase immersed in a sea of the other one. The diffusion coefficient and mobility of the model are identified and correspond to those predicted in the literature. We prove that contours typically follow a motion by curvature with an influence of the parameter $\beta$ and establish large deviation bounds at all large enough $\beta<\infty$. 
\section{Introduction}
A basic paradigm in non-equilibrium statistical mechanics is the following. Consider a system with two coexisting pure phases separated by an interface, undergoing a first-order phase transition with non-conserved order parameter. 
Then, macroscopically, the interface should evolve in time to reduce its surface tension, according to a motion by curvature. For microscopic models on a lattice, some trace of the lattice symmetries should remain at the macroscopic scale and the resulting motion by curvature should be anisotropic. 
The following general behaviour, known as the Lifshitz law, is expected: if a droplet of linear size $N\gg 1$ of one phase is immersed in a sea of the other phase, then it should disappear in a time of order $N^2$. 
(Anisotropic) motion by curvature should correspond to the limiting dynamics, when $N$ is large, under diffusive rescaling of space and time. 
Phenomenological arguments in favour of this picture go back to Lifshitz \cite{Lifshitz1962} and can be summarised as follows. 
Consider a model with surface tension $\mathfrak{t} = \mathfrak{t}({\bf N})$, which depends on the local inwards normal ${\bf N}$ to an interface. 
We work in two dimensions to keep things simple. The surface energy associated with a curve $\gamma$ separating two phases reads:
\begin{align}
F(\gamma) = \int_{\gamma} \mathfrak{t}({\bf N}(s))ds,
\end{align}
where $s$ is the arclength coordinate on $\gamma$. 
The postulate is then that the local inwards normal speed $v$ to the interface reads:
\begin{equation}
v = \mu \frac{\delta F}{\delta \gamma}.\label{eq_speed_is_variation_free_energy}
\end{equation}
Above, $\delta F/\delta \gamma$ is the variational derivative of $F$, 
defined informally below. The quantity $\mu = \mu({\bf N})$ is the mobility of the model, computed by Spohn in \cite{Spohn1993} using linear response arguments. 
Let us relate~\eqref{eq_speed_is_variation_free_energy} and motion by curvature. 
The change $\delta F = \delta F({\bf N})$ in energy induced by the motion of a length $ds$ in the normal direction ${\bf N}$ is equal to $(\mathfrak{t}({\bf N})/R({\bf N}))ds$, which can be written $\delta F/\delta \gamma = \mathfrak{t}({\bf N})/R({\bf N})$, with $R({\bf N})$ the radius of curvature at ${\bf N}$. As such:
\begin{equation}
v = \mu \mathfrak{t}k =: ak,\quad \text{with:}\quad a({\bf N}) = \mu({\bf N}) \mathfrak{t}({\bf N}) \text{ the anisotropy and }k=1/R\text{ the curvature}.\label{MCM_phenom}
\end{equation}
A closed curve satisfying~\eqref{MCM_phenom} is said to evolve according to anisotropic motion by curvature. 
A bounded set with boundary following this equation is known to shrink to a point in finite time for a wide range of anisotropies $a$, see e.g.~\cite{Lacoin2014a} and references therein.\\

The paper~\cite{Spohn1993}, already cited, 
 is a landmark in the rigorous study of interface motion starting from microscopic models. 
 Ideally, one would like to start from a microscopic model with short-range interactions, 
with at least two different phases initially segregated on a macroscopic scale, 
and derive motion by curvature~\eqref{MCM_phenom} of the boundaries between the phases in the diffusive scaling. 
To this day however, 
results on microscopic models are scarce. 
 A major difficulty is to understand how to decouple, 
 from the comparatively slower motion of the interface, 
 the fast relaxation inside the bulk of each phase. 
 Indeed, in a diffusive time scale and at least for models with local interactions, 
 one expects the bulk to behave as if at equilibrium. 
Let us provide a (non-exhaustive) account of works on the subject.

In models where the interface is the graph of a function of a one-dimensional parameter, 
motion by curvature has been proven for a number of interacting particle systems. 
Motion by curvature usually boils down to the heat equation in this case, 
and the Lifshitz law is related to freezing/melting problems, see~\cite{Spohn1993,Chayes1996,Chayes2008,Chayes2012}, 
as well as~\cite{Lacoin2014b} and the monograph~\cite{Carinci2016}. 

For one-dimensional interfaces in two dimensions, a landmark is the proof of anisotropic motion by curvature for the Glauber dynamics of the zero temperature Ising model (henceforth zero-temperature Glauber dynamics). 
The drift of the interface at time $0$ was computed in~\cite{Cerf2007} for several types of initial conditions, before the full motion by curvature~\eqref{MCM_phenom} was proven in~\cite{Lacoin2014}--\cite{Lacoin2014a}. 
Their proof crucially relies on monotonicity of the Glauber dynamics.

More is known on another type of microscopic models for which some sort of mean-field mesoscopic description can be achieved. 
This comprises the so-called Glauber+Kawasaki process \cite{DeMasi1986} (see also \cite{Bertini2018} for an account of works on the model), 
which has local evolution rules, 
and models with long range interactions such as the Ising model with Kac potentials \cite{Comets1987,DeMasi1993,DeMasi1994,Katsoulakis1994}. 
For these models, studied in any dimension, 
the derivation takes place in two steps: first deriving a mean-field description of the dynamics, 
then rescaling space-time to derive motion by curvature. 
As a result, lattice symmetries are blurred and the resulting motion by curvature is isotropic. 
A series of recent works~\cite{Funaki2019,Funaki2022ConstantspeedIF,MR4528356} consider a Glauber+Kawasaki dynamics (see also~\cite{ElKettani2020MeanCI} for Glauber+Zero-range). 
In these works, the existence of an interface between regions at high- and low-density is established, 
and motion by curvature for this interface is obtained directly from the microscopic model, 
suitably scaling the Glauber part of the dynamics. 

A last category of models comprises the so-called effective interface models. 
In these models, an interface between phases is represented by the graph of a function, 
with which an "interfacial" cost is associated. 
Only the interface is relevant, 
and the phases it separates are not described. 
Effective interface models comprise the Ginzburg-Landau model in any dimension, see \cite{Funaki1997}, 
and more recently Lozenge-tiling dynamics in dimension three \cite{Laslier2018}.\\

To better understand the structure of interface dynamics, another related line of investigation concerns large deviations of the motion of an interface around motion by curvature. 
Assuming Gaussian-like fluctuations around the mean behaviour~\eqref{MCM_phenom}, the rate function describing the cost of observing an abnormal trajectory $\gamma_\cdot=(\gamma_t)_{t\leq T}$ should read:
\begin{equation}
I(\gamma_\cdot) 
= 
\int_0^T dt \int_{\gamma_t} \frac{(v-ak)^2}{2\mu}\, ds
,
\label{eq_conjectured_form_rate_function}
\end{equation}
with $s$ the arclength coordinate on $\gamma_t$. 
In the assumption of Gaussian fluctuations leading to~\eqref{eq_conjectured_form_rate_function}, 
one of the difficulties is that it is not even clear how the noise should be incorporated into the deterministic equations describing the interface motion. 
Extensive work to address this question and obtain rate functions of the form~\eqref{eq_conjectured_form_rate_function} has been carried out in the last decades on several models. 
We discuss them below and refer to the recent works \cite{Bertini2017}--\cite{Bertini2018}--\ref{MR3769818} and references therein for a more complete picture. 

One particularly studied class of model involves the stochastic Allen-Cahn equation, see~\cite[Chapter 4]{MR3587372} for an introduction.   
It is known that, in the diffusive (or sharp interface) limit, solutions to the (deterministic) Allen-Cahn equation satisfy motion by mean curvature in some sense, see \cite{Ilmanen1993,Evans1992,Barles1993}. 
One can then study large deviations around motion by curvature by starting from the Allen-Cahn equation perturbed by a small noise, then taking suitable limits in all parameters involved: regularisation of the noise if any, temperature and sharpnes of the interface. 
The addition of a noise term in the equation leads to very different solution theories depending on the regularisation and strength of the noise, 
with the equation for instance becoming ill-posed in dimension $d>1$ for space-time white noise. 
At the level of the large deviations however, this ill-posedness does not affect the rate functions at least in a suitable region of the regularisation, noise strength and interface sharpness as mentioned in~\cite{MR2284215} and discussed in detail in~\cite{Hairer2014LargeDF}.

Rate functions of the form~\eqref{eq_conjectured_form_rate_function} have been obtained from the stochastic Allen-Cahn equation taking various orders in the limits involved, see the seminal work~\cite{MR2284215} where the order of the limits ensures well-posedness of the stochastic Allen-Cahn equation and the more recent~\cite{Bertini2017} where a joint limit in all parameters is considered, with corresponding regularity estimates on solutions established in~\cite{Bertini2017a}. 
The rate function obtained as upper bounds in these papers is a more general version of~\eqref{eq_conjectured_form_rate_function} that coincides in simple cases, 
such as for a droplet trajectory with smooth boundary. 
More general trajectories that may feature nucleation events are also treated.  
A different way of adding noise to the Allen-Cahn equation and associated large deviations are considered in~\cite{MR3769818}.

Microscopic models have also been considered. In \cite{Bertini2018}, upper bound large deviations are proven for the aforementioned Glauber+Kawasaki process and Ising model with Kac potentials. 
The rate function~\eqref{eq_conjectured_form_rate_function} is proven to be the correct one for smooth trajectories and extensions to more general paths are discussed.\\

To the best of our knowledge however, no results on large deviations from motion by curvature for microscopic interface dynamics with local interactions have yet been published. In particular the question of large deviations for the zero temperature Ising Glauber dynamics is still open.

In this work, we present a family of interface dynamics that we call the contour dynamics. 
In the scaling limit, this dynamics typically evolves according to motion by curvature and we characterise the large deviations. 
The contour dynamics is closely related to the zero temperature Glauber dynamics for the Ising model: it has the same updates, except that additional moves depending on a parameter $\beta>0$ are allowed. 
This parameter $\beta$ plays the role of an inverse temperature acting on local portions of the contours. 
The model at each $\beta>0$ has reversible dynamics and, contrary to the Glauber dynamics for the zero temperature Ising model, the dynamics is not monotonous. When $\beta=\infty$, the update rules of the contour dynamics are exactly the same as the Ising ones. 
Large deviations for the contour dynamics are studied using the method initiated by Kipnis, Olla and Varadhan in \cite{Kipnis1989} (see also Chapter 10 in \cite{Kipnis1999}). There are substantial difficulties as we are dealing with curves, i.e. one-dimensional objects, evolving in two-dimensional space. One of the advantages of the method is that we no longer rely on monotonicity of the dynamics as in \cite{Lacoin2014a}. Monotonicity appears difficult to use for large deviations in any case, as atypical events, such as closeness to some atypical trajectory, are in general not monotonous. At each large enough $\beta>0$, we prove that the dynamics approaches anisotropic motion by curvature in the large size limit, with a dependence on the parameter $\beta$. At the formal level, the $\beta=\infty$ case indeed corresponds to anisotropic motion by curvature in the sense of \cite{Lacoin2014}. We then obtain large deviations for the model, with a rate function that agrees with~\eqref{eq_conjectured_form_rate_function} for sufficiently nice trajectories.\\

The rest of this article is structured as follows. In Section~\ref{sec_notations_results}, we introduce the microscopic model and fix notations. The dynamics is introduced in details using the zero temperature Glauber dynamics as comparison, while useful topological facts are collected in Appendix~\ref{app_prop_e_r}. The main results of the paper are listed in Section~\ref{sec_notations_results}, 
with Section~\ref{sec_heuristics} presenting the structure of the proof as well as a connection of the contour dynamics with the exclusion process, a guideline of the paper.

In Section~\ref{sec_relevant_martingales}, 
following the large deviation approach of \cite{Kipnis1989}, 
we compute Radon-Nikodym derivatives for a large class of tilted dynamics. 
Under the assumption that trajectories live in a nice enough space, 
we show how motion by curvature emerges from the microscopic computations as well as the influence of the parameter $\beta$. 
The computations of the Radon-Nikodym derivatives are then used to prove large deviations, with the upper bound in Section~\ref{sec_large_dev_upper_bound} and the lower bound in Section~\ref{sec_large_dev_lower_bound}. 
A number of technical results and sub-exponential estimates are postponed to Section~\ref{app_behaviour_pole} and Appendices~\ref{sec_replacement_lemma}--\ref{app_prop_e_r}. In particular, Section~\ref{app_behaviour_pole} is a collection of estimates that are genuinely particular to our model, 
concerning the dynamical behaviour of the poles, 
i.e. the sections of the contours on which the parameter $\beta$ affects the dynamics. 
\section{Model and results}\label{sec_notations_results}
\subsection{Zero temperature Glauber dynamics for the Ising model}
The contour dynamics studied in this paper is closely related to the Glauber dynamics of the zero temperature, two-dimensional Ising model on $(\Z^*)^2$ (henceforth zero temperature Glauber dynamics), 
with $\Z^*:= 1/2 + \Z$ the dual graph of $\Z$. 
Looking at $(\Z^*)^2$ rather than $\Z^2$ is meant to ensure that contours are lattice paths on $\Z^2$, see below. 
Let us first define this zero temperature Glauber dynamics.

On the space $\Sigma := \{-1,1\}^{(\Z^*)^2}$ of all spin configurations $\sigma = (\sigma(i))_{i\in(\Z^*)^2}\in\Sigma$, define the dynamics as follows: each site $i\in (\Z^*)^2$ is updated independently at rate $1$. The spin $\sigma(i)$ at site $i$ takes the same value as the majority of its neighbours, where spins $\sigma(j),\sigma(k)$ are neighbours for $j,k\in(\Z^*)^2$ if $\|j-k\|_1=1$. 
If spin $\sigma(i)$ has exactly two neighbours of each sign, then: with probability $1/2$, 
$\sigma(i)$ remains unchanged; with probability $1/2$, it is \emph{flipped}, 
i.e. changed to $-\sigma(i)$. 
A spin with three or more neighbours of the same sign is not changed, while a spin with three or more neighbours of opposite sign is flipped instantaneously, 
and the process is repeated until no such spin remains. 
This is summarised in the following jump rates (see also Figure~\ref{fig_updates_Ising}): for each configuration $\sigma$ and each $i\in(\Z^*)^2$,
\begin{equation}
c(\sigma,\sigma^i) = \begin{cases}
0\quad &\text{if }\sigma(i)\text{ and at least three neighbours have the same sign},\\
1/2\quad &\text{if }\sigma(i)\text{ has two neighbours of each sign},\\
+\infty \quad &\text{if }\sigma(i) \text{ has at least three neighbours with opposite sign}.
\end{cases}\label{eq_jump_rates_Ising}
\end{equation}
Above, the configuration $\sigma^i$ is the same as $\sigma$, except that the spin at $i$ has been flipped:
\begin{equation}
\forall j \in (\Z^*)^2\setminus\{i\},\qquad \sigma^i(j) = \sigma(j),\qquad \sigma^i(i) = -\sigma(i).\label{eq_def_sigma_flipped}
\end{equation}
\begin{figure}
\begin{center}
\includegraphics[width = 10cm]{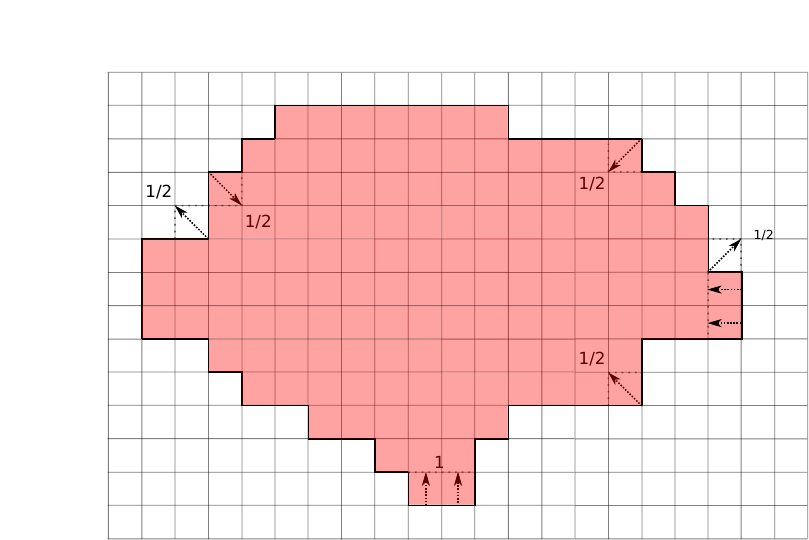} 
\caption{Some possible updates for the zero temperature Glauber dynamics. 
Red squares represent $-$ spins and white squares $+$ spins, assimilating a square with its centre, a point of $(\Z^*)^2$. 
If either of the lowest two red squares disappear 
(at rate 1/2, since both have two neighbours of each colour), 
then the remaining square has three neighbours of opposite colour and is flipped instantaneously. 
Both these squares thus become white at rate $1$. 
After that move, the Glauber rules preclude any square of the line from becoming red: the droplet shrinks. 
\label{fig_updates_Ising}}
\end{center}
\end{figure}
Rather than spins, the zero temperature Glauber dynamics can alternatively be defined in terms of \emph{blocks}: a block is a subset of $\R^2$ of the form $i+[-1/2,1/2]^2$, with $i\in (\Z^*)^2$ the centre of the block. Flipping a spin amounts to changing the colour of the corresponding block. The colour of a block (red or white in Figure~\ref{fig_updates_Ising}) is determined by the sign of the spin at its centre. This alternative terminology will be used preferentially throughout the article. In fact, we will consider configurations of the type depicted in Figure~\ref{fig_updates_Ising}, where all red blocks form a bounded connected region (that we call a droplet, see next paragraph) surrounded by white blocks. 
We will then not even focus on colours and instead say that a block is added/deleted to mean that the new droplet contains one more/one less block.\\

An important property of the zero temperature Glauber dynamics is that a large region of $-$ spins surrounded by $+$ spin \emph{shrinks}: 
with probability going to $1$ as the diameter $D$ of a region diverges, 
all spins will have become $+$ spins after a time of order $D^2$. 

In \cite{Lacoin2014}--\cite{Lacoin2014a}, the precise evolution of such a droplet was obtained in the scaling limit (for a slightly different choice of jump rates, but the result applies to the present case~\eqref{eq_jump_rates_Ising}). 
Let us describe their result. 
Let $\gamma^0\subset \R^2$ be a Jordan curve, i.e. a closed, simple curve. 
Let $\Gamma^0$ be the \emph{droplet associated with} $\gamma^0$, meaning the compact subset of $\R^2$ with boundary $\partial\Gamma^0 = \gamma^0$. Assume for simplicity that $\Gamma^0$ is convex and $\gamma^0$ is $C^\infty$ (the non-convex case is treated in \cite{Lacoin2014a}). 
Fix a scaling parameter $N\in\N_{\geq 1}$ and let $\sigma^{N,0}\in\Sigma$ be the spin configuration obtained by setting $\sigma^{N,0}(i) =-1$ if $i\in N\Gamma^0$, $\sigma^{N,0}(i) = 1$ if $i\notin N\Gamma^0$. 
Up to adding a finite number of $-$ spins, we may assume that each $-$ spin in $\sigma^{N,0}$ has at least two $-$ neighbours as in Figure~\ref{fig_updates_Ising}. 
The zero temperature Glauber dynamics~\eqref{eq_jump_rates_Ising} starting from $\sigma^{N,0}$ is then well defined for all time.\\

In \cite{Lacoin2014}, the authors prove that, rescaling space by $1/N$ and time by $N^2$, the rescaled droplet converges uniformly in time, 
in Hausdorff distance, 
to the unique solution of an anisotropic motion by curvature starting from $\Gamma^0$. 
To state a precise result, we need some notation. A solution $(\Gamma_t)_{t\geq 0}$ of motion by curvature~\eqref{MCM_phenom} with initial condition $\Gamma^0$ is a flow of droplets starting at $\Gamma^0$ and satisfying the following: there is a time $T_f>0$ such that, for $t<T_f$, the boundaries $(\gamma_t)_{t<T_f}$ of $(\Gamma_t)_{t<T_f}$, parametrised on the unit torus $\T$, solve~\eqref{MCM_phenom}:
\begin{equation}
\forall u\in\T,\forall t<T_f,\qquad  \partial_t \gamma_t(u) = a(\theta(t,u))\partial^2_s\gamma(t,u) = a(\theta(t,u))k(t,u) {\bf N}(t,u).\label{eq_def_MCM}
\end{equation}
Moreover, after time $T_f$, 
each droplet $\Gamma_t$ is reduced to a point. 
In~\eqref{eq_def_MCM}, the letter $s$ denotes the arclength coordinate on the curve $\gamma_t$ for $t<T_f$, 
while $k(t,u)$ is the curvature and $\theta(t,u)$ is the angle between the tangent vector at point $\gamma_t(u)$ and the first basis vector ${\bf b}_0 := (1,0)$. 
The vector ${\bf N}(t,u)$ is the unit inwards normal at $\gamma_t(u)$. 
The $\pi/2$-periodic anisotropy factor $a$ is a quantity with symmetries reflecting those of the square lattice. 
It reads:
\begin{equation}
a(\theta):=\frac{1}{2(|\sin(\theta)| + |\cos(\theta)|)^2},\qquad \theta\in[0,2\pi].\label{eq_def_a}
\end{equation}
Existence and uniqueness of a flow of sets solving~\eqref{eq_def_MCM} is part of the results of \cite{Lacoin2014}--\cite{Lacoin2014a}. \\
For a set $\Gamma\subset\R^2$ and $\epsilon>0$, 
let $\Gamma^{(-\epsilon)}$ (resp.: $\Gamma^{(\epsilon)}$) denote its $\epsilon$-shrinking (resp.: $\epsilon$-fattening):
\begin{equation}
\Gamma^{(\epsilon)} = \bigcup_{x\in\Gamma} B_1(x,\epsilon),\qquad \Gamma^{(-\epsilon)} = \Big[\bigcup_{x\notin\Gamma}B_1(x,\epsilon)\Big]^c,
\end{equation}
where $B_1(x,\epsilon)$ is the ball of centre $x$ and radius $\epsilon$ in $1$-norm. For future reference, recall:
\begin{equation}
\forall (u,v)\in\R^2,\qquad \|(u,v)\|_1 := |u|+|v|,\quad \|(u,v)\|_2 = \sqrt{u^2+v^2},\quad \|(u,v)\|_{\infty} = \max\{|u|,|v|\}.\label{eq_def_normes_1_2}
\end{equation}
The main result of \cite{Lacoin2014} is then the following. 
Denote as before by $(\Gamma_t)_{t\geq 0}$ the flow of droplets satisfying~\eqref{eq_def_MCM} with initial condition $\Gamma^0$. 
Let $\Prob$ denote the probability associated with the zero temperature Glauber dynamics starting from the configuration $\sigma^0$, 
and let $\Gamma^N$ be the notation for a microscopic droplet of $-$ spins. 
Then the rescaled droplet trajectory evolves in diffusive time and satisfies~\eqref{eq_def_MCM}, in the sense that:
\begin{equation}
\forall\epsilon>0,\qquad \lim_{N\rightarrow\infty}\Prob\Big(\forall t\geq 0,\quad \Gamma_t^{(-\epsilon)}\subset N^{-1}\Gamma^N_{tN^2} \subset \Gamma_t^{(\epsilon)}\Big)=1
\label{eq_Lacoin_1}
\end{equation}
and:
\begin{equation}
\forall\epsilon>0,\qquad \lim_{N\rightarrow\infty}\Prob\Big(\forall t\geq T_f+\epsilon,\quad\Gamma^N(tN^2) = \emptyset\Big)=1.\label{eq_Lacoin_2}
\end{equation}
For future reference, 
note that~\eqref{eq_Lacoin_1} is a statement on the Hausdorff distance $d_{\mathcal H}$ of $N^{-1}\Gamma^N_{tN^2}$ and $\Gamma_t$ at each time $t\geq 0$. 
The Hausdorff distance between two non-empty, compact sets $A,B\subset \R^2$ reads:
\begin{equation}
d_{\mathcal H}(A,B) = \inf\big\{\epsilon>0 : A\subset B^{(\epsilon)}\text{ and }B\subset A^{(\epsilon)}\big\}.\label{eq_def_distance_Hausroff}
\end{equation}
The proof of~\eqref{eq_Lacoin_1}--\eqref{eq_Lacoin_2} relies strongly on two ingredients. 
The first ingredient is the fact that the zero temperature Glauber dynamics has the monotonicity property (see e.g. Section 3.3 in \cite{Martinelli1999}): 
for two spin configurations $\sigma,\eta$, write $\sigma \leq \eta$ when $\sigma(i)\leq \eta(i)$ for each $i\in(\Z^*)^2$. 
There is then a coupling such that, 
with probability $1$, $\sigma_t\leq \eta_t$ for all $t\geq 0$. 

The second ingredient is the observation that local portions of the interface can be mapped to one-dimensional interacting particle processes, in particular to the symmetric simple exclusion process (SSEP), which is well known. 
This mapping will also be used in the present paper and is detailed in Section~\ref{sec_heuristics}.
\subsection{The contour dynamics}
In this article, we consider a microscopic interface dynamics that we call the contour dynamics. 
It is closely related to the zero temperature Glauber dynamics, 
with a number of interesting contrasting features. 
In this section, we first describe the state space, 
then define the dynamics. 
A connection of the contour dynamics with the simple exclusion process is also presented 
(and further discussed in Section~\ref{sec_heuristics}). 
Throughout the article, this connection will be used as a guideline for the study of the contour dynamics at both microscopic and macroscopic level. 
\subsubsection{The state space}
We first define the state space of the contour dynamics using results of~\cite{Lacoin2014}--\cite{Lacoin2014a} for the Ising model as a guide to formulate conditions 
on the shape of the contours we consider.

The arguably simplest yet interesting case where the scaling limit of the zero temperature Glauber dynamics is known is when the starting droplet of $-$ spins is (the discretisation of) a convex droplet with smooth boundary. 
At the microscopic level convexity is not preserved and 
we instead consider curves that can be split into four parts as defined next. 
Let ${\bf b}_{0}$, 
${\bf b}_{\pi/2}$ and ${\bf b}_\theta$ 
for $\theta\in[0,2\pi]$ be the vectors:
\begin{equation}
{\bf b}_{0}  := (1,0),
\quad {\bf b}_{\pi/2} := (0,1)
,\qquad
{\bf b}_\theta 
:= 
\cos(\theta) {\bf b}_0 + \sin(\theta){\bf b}_{\pi/2}
.
\label{eq_def_b_theta}
\end{equation}
\begin{defi}[The set $\Omega$]
Define $\Omega$ as the set of Lipschitz closed curves $\gamma\subset\R^2$ surrounding $0$ (including $\{0\}$) and satisfying the following condition when $\gamma\neq\{0\}$. 

The curve $\gamma$ can be split into four (intersecting) connected regions of maximal length. 
In region $k$ with $1\leq k\leq 4$, 
the tangent vector ${\bf T}$ to the interface, defined with clockwise orientation,  
satisfies ${\bf T}\cdot {\bf b}_{-(k-1)\pi/2}\geq 0$ and ${\bf T}\cdot {\bf b}_{-k\pi/2}\geq 0$ (see Figure~\ref{fig_quadrants_convex}).\label{property_state_space}
\end{defi}
In particular all Lipschitz closed curves that are the boundary of a convex droplets belong to $\Omega$. 
The advantage of $\Omega$, however, is that it also contains approximations of convex curves by closed lattice paths, 
as well as much more general curves. 
\begin{conv}
Interfaces in this article will always be oriented clockwise.
\end{conv}
\begin{figure}[H]
\begin{center}
\includegraphics[width=12cm]{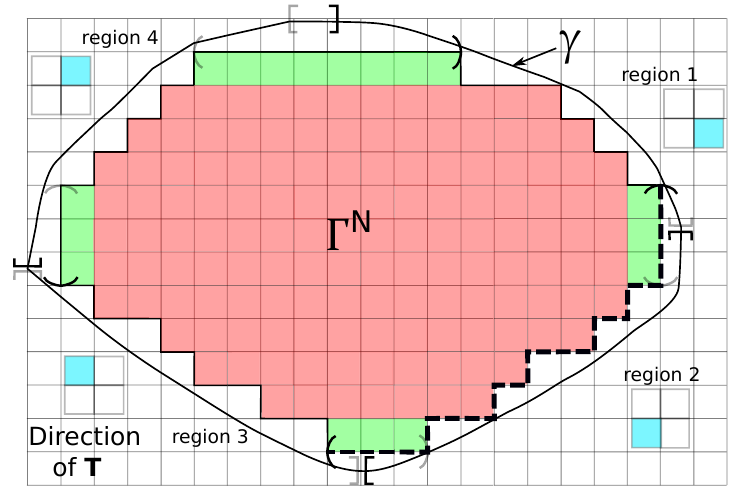}
\caption{The convex interface $\gamma$ and the droplet $\Gamma^N$ associated with the discretisation $\gamma^N := \partial\Gamma^N$ of $\gamma$. 
The four regions of both $\gamma$ and $\gamma^N$ are materialised by opening and closing brackets (for $\gamma$) or parentheses (for $\gamma^N$), 
with each region starting with an opening symbol and ending with a closing one. 
Note that regions overlap at the poles. This is the case for the north pole of $\gamma$ and all poles of $\gamma^N$ as the latter must contain at least two edges. 
To avoid confusion the delimiters of regions 1,3 are in grey and the other two in black. 
As an example, region 2 of $\gamma^N$ corresponds to the thick dashed lines. 
In each region, the quarterplane to which the tangent vector ${\bf T}$ belongs is indicated by a cyan square. 
The spins (i.e. the blocks) which have an edge belonging to a pole are coloured in green.\label{fig_quadrants_convex}}
\end{center}
\end{figure}

Note that, for a discrete interface, the tangent vector can only be one of the four vectors ${\bf b}_{k\pi/2}$ ($1\leq k\leq 4$), 
corresponding to $\pm {\bf b}_0,\pm {\bf b}_{\pi/2}$.

For an interface in $\Omega$, 
the intervals corresponding to the intersection of two consecutive regions will play a special role. 
These intervals (see Figure~\ref{fig_quadrants_convex}) are called \emph{poles}, 
and can equivalently be defined as points of the interface with extremal abscissa or ordinate. 
Pole $k$ ($1\leq k \leq 4$) is defined as the intersection of regions $k-1$ and $k$, 
where by convention $k-1 := 4$ if $k=1$. 
We shall also refer to poles in terms of cardinal directions: pole $1$ is the north pole, 
corresponding to the interval of points with maximal ordinate. 
Pole $2$ is the east pole, made of points with maximal abscissa, etc. \\

We can now define the state space of the contour dynamics. 
In view of the hydrodynamic behaviour~\eqref{eq_Lacoin_1}, 
we directly work with \emph{rescaled} microscopic curves, 
i.e. lattice paths on $(N^{-1}\Z)^2$ with $N\in\N_{\geq 1}$  the scaling parameter. 
In the following, an edge of the graph $(N^{-1}\Z)^2$ is identified with a segment of length $1/N$ between two neighbouring vertices. 
\begin{defi}[State space $\Omega^N_{\text{mic}}$]\label{def_state_space}
For $N\in\N_{\geq 1}$, 
the microscopic state space $\Omega^N_{\text{mic}}$ is the subset of $\Omega$ such that:
\begin{itemize}
	\item Each curve $\gamma^N\in \Omega^N_{\text{mic}}$ is a simple, closed lattice path on $(N^{-1}\Z)^2$.
	\item Each of the four poles of $\gamma^N$ contains at least two edges, 
	i.e. it is a segment of length at least~$2/N$.
	%
%
\end{itemize}
\end{defi}
\begin{rmk}
We are only interested in the shape of interfaces, not their position in $\R^2$. 
The condition that elements of $\Omega\supset\Omega^N_{\text{mic}}$ surround $0$ should therefore be understood as a way to lift degeneracies: it ensures that the state space $\Omega^N_{\text{mic}}$ does not contain infinitely many curves that are translations one of the other (and this is its only use). 
\demo
\end{rmk}
\noindent\textbf{Notation:} to avoid confusion between microscopic and macroscopic interfaces in the following, whenever both microscopic and macroscopic interfaces are considered, microscopic interfaces are denoted with a superscript $N$: $\gamma^N\in \Omega^N_{\text{mic}}$, with associated droplet $\Gamma^N$. 
In that case, the letters $\gamma,\Gamma$ without the $N$ superscript are used for macroscopic objects.

\subsubsection{The dynamics}
\label{sec_def_contour_dynamics}
We now define the contour dynamics on $\Omega^N_{\text{mic}}$ (see Figure~\ref{fig_image_possible_jumps_beta_dynamics}), 
starting with some notations.

The segments that we call poles (see Figure~\ref{fig_quadrants_convex}) are going to play an important role in the dynamics.
Take a curve $\gamma\in \Omega$ and let $P_k = P_k(\gamma)$ denote its pole $k$ ($1\leq k \leq 4$). 
Write $P_k= [L_k,R_k]$, where the points $L_k,R_k = L_k(\gamma),R_k(\gamma)$ of $\R^2$ are respectively the left and right extremities of $P_k$ when $\gamma$ is oriented clockwise as always, 
see Figure~\ref{fig_image_possible_jumps_beta_dynamics} below. 
The length of pole $k$ is denoted by $|P_k|$.

In analogy with the Ising case, 
the block $C_i$ with centre $i\in (N^{-1}\Z^*)^2$ is defined as:
\begin{equation}
C_i := i + \frac{1}{2N}[-1,1]^2.\label{eq_def_block}
\end{equation}
Consider now a microscopic curve $\gamma^N\in\Omega^N_{\text{mic}}$.  
We say that a block in the droplet $\Gamma^N$ delimited by $\gamma^N$ is \emph{in pole }$k$ 
($1\leq k\leq 4$) 
if one of the edges in the boundary of the block is included in pole $k$. 
Blocks in a pole are in green on Figure~\ref{fig_quadrants_convex}. 
Let $p_k = p_k(\gamma^N)$ denote the number of blocks in pole $k$. 
It is related to the length $|P_k|$ of the pole by:
\begin{equation}
p_k 
:= 
N|P_k|,\qquad 
1\leq k \leq 4
.
\label{eq_def_pk}
\end{equation}
With these notations, we proceed to define dynamical moves.

Let $\gamma^N\in\Omega^N_{\text{mic}}$ and $\Gamma^N$ denote the associated droplet as usual. 
If $i\in(N^{-1}\Z^*)^2$, 
adding to or deleting from $\Gamma^N$ the block $C_i := i +\frac{1}{2N}[-1,1]^2$ is the operation $\Gamma^N\rightarrow(\Gamma^N)^i$, 
with:
\begin{equation}
(\Gamma^N)^i := \begin{cases}
\Gamma^N\cup C_i\quad &\text{if }i\notin \Gamma^N,\\
\Gamma^N \setminus C_i\quad &\text{if }i\in\Gamma^N.
\end{cases}\label{eq_def_single_flips}
\end{equation}
Define then $(\gamma^N)^i := \partial\big((\Gamma^N)^i\big)$.\\ 

Consider now moves affecting the poles. 
For $k$ with $1\leq k\leq 4$, 
assume that pole $k$ of $\gamma^N$ contains exactly two blocks. 
Define then $(\gamma^N)^{-,k}$ as the boundary of $(\Gamma^N)^{-,k}$, 
where $(\Gamma^N)^{-,k}$ is obtained from $\Gamma^N$ by deleting the two blocks in pole $k$:
\begin{equation}
(\Gamma^N)^{-,k} := \Gamma^N \setminus \bigcup_{i\in (N^{-1}\Z^*)^2:C_i\text{ is in Pole } k}C_i
,
\qquad (\gamma^N)^{-,k} = \partial\big((\Gamma^N)^{-,k}\big) 
.
\label{eq_def_gamma_-}
\end{equation}
Define conversely a transformation that makes a droplet grow at the pole as follows. 
Let $x\in (N^{-1}\Z)^2\cap P_k$ ($1\leq k \leq 4$) be any point of $P_k$ different from $R_k,L_k$. 
Define then $(\gamma^N)^{+,x}$ as the boundary of $(\Gamma^N)^{+,x}$, with:
\begin{align}
(\Gamma^N)^{+,x} := \Gamma^N \cup \bigcup_{i\in (N^{-1}\Z^*)^2\setminus \Gamma^N:x\in C_i}C_i,\qquad (\gamma^N)^{+,x} = \partial\big((\Gamma^N)^{+,x}\big).
\end{align}
In words, $(\Gamma^N)^{+,x}$ is the droplet $\Gamma^N$ to which the two blocks with boundary that contains $x$ 
(e.g. for the north pole, the block for which $x$ is at the lower right corner and the block for which it is at the lower left) 
and that are not in pole $k$ have been added. 
We can now define the contour dynamics, 
illustrated on Figure~\ref{fig_image_possible_jumps_beta_dynamics}.  
\begin{defi}\label{def_contour_dynamics}
The contour dynamics on $\Omega^N_{\text{mic}}$ at inverse temperature $\beta\geq 0$ is defined through the jump rates $c(\gamma^N,\tilde\gamma^N)$ 
for curves $\gamma^N,\tilde\gamma^N\in\Omega^N_{\text{mic}}$:
\begin{itemize}
	\item $c\big(\gamma^N,(\gamma^N)^i\big) = (1/2){\bf 1}_{(\gamma^N)^i\in \Omega^N_{\text{mic}}}$ for $i\in(N^{-1}\Z^*)^2$;
	\item $c\big(\gamma^N,(\gamma^N)^{-,k}\big) = {\bf 1}_{p_k(\gamma^N)=2}{\bf 1}_{(\gamma^N)^{-,k}\in \Omega^N_{\text{mic}}}$, 
	with $p_k$ defined in~\eqref{eq_def_pk} ($1\leq k\leq 4$);
	\item (Growth at the poles) $c\big(\gamma^N,(\gamma^N)^{+,x}\big)= e^{-2\beta}$ for each $x\in (N^{-1}\Z)^2\cap P_k(\gamma^N)$ 
	with $x\notin\{R_k(\gamma^N),L_k(\gamma^N)\}$, $1\leq k \leq 4$;
	\item $c(\gamma^N,\tilde\gamma^N)=0$ for any other $\gamma^N,\tilde\gamma^N$.
\end{itemize}
\end{defi}
\begin{figure}[H]
\begin{center}
\includegraphics[width=12.1cm]{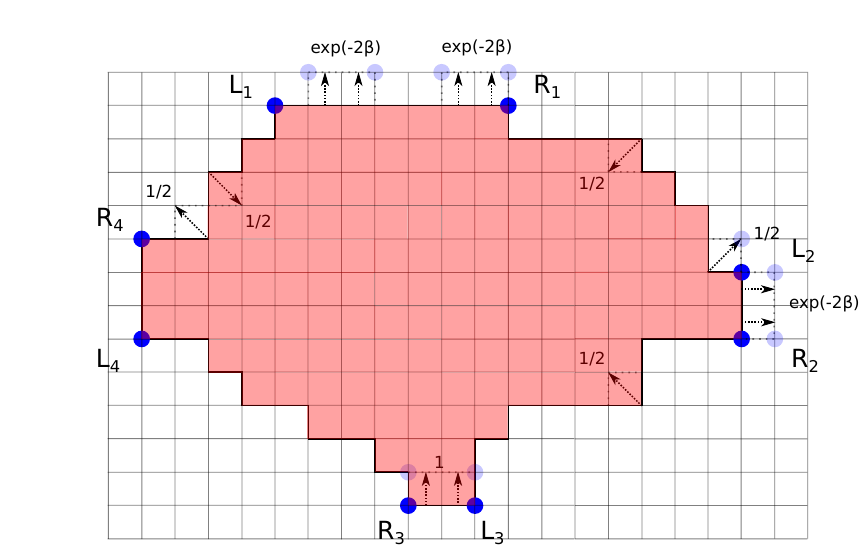} 
\caption{Some moves and associated jump rates for the contour dynamics acting on an element of $\Omega^N_{\text{mic}}$. 
Positions of the extremities $L_k,R_k$, $1\leq k\leq4$ of the poles are represented by dark dots. 
Possible positions of $L_k,R_k$ after a jump are represented by light dots. 
Dynamical moves amount to adding or deleting squares of side-length $1/N$ ("blocks"). 
The pole $P_3$ contains two blocks, i.e. $p_3=2$, 
thus both of its blocks can be simultaneously deleted by an update.\label{fig_image_possible_jumps_beta_dynamics}}
\end{center}
\end{figure}
\begin{rmk}\label{rmk_def_contour_dyn}
Let us comment on Definition~\ref{def_contour_dynamics}.
\begin{itemize}
	\item If $\beta = +\infty$, then the contour dynamics and the zero temperature Glauber dynamics~\eqref{eq_jump_rates_Ising} act on a contour $\gamma^N\in\Omega^N_{\text{mic}}$ in the same way, 
	provided the resulting contour is in $\Omega^N_{\text{mic}}$. 
	In particular, the contour dynamics at $\beta=+\infty$ is monotonous.
	\item At each $\beta\in[0,+\infty)$, 
	the contour dynamics is not monotonous (see Figure~\ref{fig_dyn_not_monotone_non_local}). 
	However, it is built to be reversible with respect to the measure $\nu^N_\beta$, 
	defined by:
\begin{equation}
\nu^N_\beta(\gamma) := \frac{1}{\mathcal Z^N_\beta}e^{-\beta N|\gamma^N|},\qquad \gamma^N\in\Omega^N_{\text{mic}},\quad \mathcal{Z}^N_\beta\text{ a normalisation factor.}\label{eq_def_nu_beta}
\end{equation}	
Recall that elements of $\Omega^N_{\text{mic}}$ must surround the point $0$ by Definition~\ref{def_state_space}. 
This breaks translation invariance so that $\nu^N_\beta$ is well-defined as soon as $\beta>\log 2$ 
(the number of curves of length $n/N$ in $\Omega^N_{\text{mic}}$ is bounded by $cn^42^n$ for some $c>0$ and each $n\in\N_{\geq 1}$).

\item The regrowth moves are an important difference from the zero temperature Glauber dynamics~\eqref{eq_jump_rates_Ising}, 
where regrowth of the droplet was not possible. 
This regrowth is designed to make the dynamics reversible with respect to the measure $\nu^N_\beta$. 
Compared to the Dirac at the configuration with only $-$ spins that is invariant for the zero temperature Glauber dynamics, 
the measure $\nu^N_\beta$ is much more convenient to work with: 
it has full support, 
and can be used for explicit computations.  
These properties make avalaible the entropy method of Guo, Papanicolaou and Varadhan \cite{Guo1988}.

The downside of the regrowth term is that one has to carefully control the motion of the poles, which turns out to be the main difficulty to study the contour dynamics.

\item The contour dynamics is non-local: 
one cannot find $\rho\in\N_{\geq 1}$ independent of $N$ such that, 
for any $\gamma^N\in\Omega^N_{\text{mic}}$ and any $x\in \gamma^N\cap(N^{-1}\Z)^2$, 
deciding whether $c\big(\gamma^N,(\gamma^N)^x\big)>0$ require only the knowledge of all points of the curve at $1$-distance at most $\rho/N$ from $x$.

The non-locality is due to the fact that regions $1$ and $3$, 
or $2$ and $4$ of a curve in $\Omega^N_{\text{mic}}$ may be very close to each other as subsets of $\R^2$, 
so that deleting a single block would create self-intersections in the interface, 
which is forbidden. 
This point is illustrated on Figure~\ref{fig_dyn_not_monotone_non_local}.
\end{itemize}
\end{rmk}
\begin{figure}[H]
\begin{center}
\includegraphics[width=12cm]{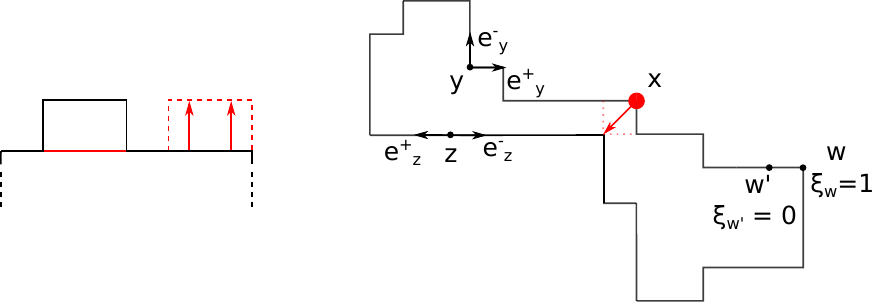} 
\caption{\emph{Left figure}: two microscopic curves equal everywhere except at the north pole: 
the configuration represented by the black line has a pole containing $2$ blocks, 
the one with the red line a pole containing $6$ blocks. 
Initially, the droplet delimited by the black line contains the one delimited in red. 
A possible update after which the inclusion does not hold is represented in dashed red lines: 
the contour dynamics at $\beta<\infty$ is not monotonous.\\
\emph{Right-figure}: only looking at points of the interface in a neighbourhood of $x$, 
the update indicated by an arrow should be allowed, 
as the corresponding block has two neighbours in and two neighbours out of the droplet. 
However, this update is forbidden as the resulting curve would not belong to $\Omega^N_{\text{mic}}$ (it would not be simple). 
The contour dynamics is therefore non local. 
The vectors ${\bf e}^{\pm}_y,{\bf e}^{\pm}_z$ are indicated for two points $y,z$ of the interface. 
The edges $[y+{\bf e}^-_y,y]$ and $[y,y+{\bf e}^+_y]$ are perpendicular: 
a block can be added or removed at $y$ (in this example, added). 
The same situation occurs at site $w$: the edge $[w,w+{\bf e}^+_w]$ is vertical, corresponding to $\xi_w = 1$, 
while the edge $[w',w] = [w+{\bf e}^-_{w},w]$ is horizontal, i.e. $\xi_{w'} = 0$.\label{fig_dyn_not_monotone_non_local}}
\end{center}
\end{figure}
\noindent\textbf{Link with simple exclusion.} 
The jump rates $c\big(\gamma^N,(\gamma^N)^i\big)$, 
$\gamma^N\in\Omega^N_{\text{mic}}$, $i\in(N^{-1}\Z^*)^2$ appear to involve the entire space as $i\in (N^{-1}\Z^*)^2$. 
They however vanish when $i$ is not at distance $1/(2N)$ from $\gamma^N$ and it is in fact possible to express these jump rates only in terms of points of the interface.  
This connects the contour dynamics to the Symmetric Simple Exclusion Process (SSEP for short) as we now explain.

For $\gamma^N\in\Omega^N_{\text{mic}}$, 
let $V(\gamma^N)$ denote the set of vertices of $\gamma^N$:
\begin{equation}
V(\gamma^N) := (N^{-1}\Z)^2 \cap \gamma^N.
\label{eq_def_vertices}
\end{equation}
If $x\in V(\gamma^N)$, let ${\bf e}^+_x = n(x,\gamma^N)-x$ 
and ${\bf e}^-_x = p(x,\gamma^N)-x$ where $n(x,\gamma^N)$, $p(x,\gamma^N)$ are respectively the next and previous points relative to $x$ when going through $V(\gamma^N)$ clockwise (see Figure~\ref{fig_dyn_not_monotone_non_local}).

Let $\xi_x = \xi_x(\gamma^N)$ be the state of the edge $[x,x+{\bf e}^+_x]$, 
defined by:
\begin{equation}
\xi_x = 1\quad \text{if }[x,x+{\bf e}^+_x]\text{ is vertical},\qquad \xi_x = 0\quad \text{if }[x,x+{\bf e}^+_x]\text{ is horizontal}.\label{eq_def_xi}
\end{equation}
A block can be added/deleted to a droplet provided it has at least two neighbours of opposite colours, 
see Figure~\ref{fig_image_possible_jumps_beta_dynamics}. 
This means that the interface has a corner at this block, 
i.e. there is a point $x\in V(\gamma^N)$ (corresponding to the corner of the block) 
such that the two edges $[x+{\bf e}^-_x,x]$ and $[x,x+{\bf e}^+_x]$ are perpendicular. 
Orthogonality of the two edges can equivalently be stated as follows:
\begin{equation}
{\bf 1}_{{\bf e}^+_x\cdot\,  {\bf e}^-_x=0}
=
1
\quad \Leftrightarrow \quad
\xi_{x+{\bf e}^-_x}(1-\xi_x) + \xi_x(1-\xi_{x+{\bf e}^-_x}) 
= 
1
.
\label{eq_link_ssep_intro}
\end{equation}
Associating a site with each edge of $\gamma^N$ and viewing $\xi_\cdot$ as a configuration of particles (with e.g.~$\xi_x=1$ if there is a particle at $x$; the precise mapping is given later in Figure~\ref{fig_Ising_ssep_no_bords}), 
the right-hand side of~\eqref{eq_link_ssep_intro} says that a jump may occur in the contour dynamics when there is a particle at site $x+{\bf e}^-_x$ but not at site $x$ (or vice versa). 
This is precisely the exclusion condition of the SSEP.  
If $i\in (N^{-1}\Z)^2$ is the centre of the block with corner $x$, 
we can now
define the curve $(\gamma^N)^x$ (involving the point $x$ of the interface) as the curve $(\gamma^N)^i$ of~\eqref{eq_def_single_flips}, 
and set:
\begin{equation}
c\big(\gamma^N,(\gamma^N)^x\big) = {\bf 1}_{(\gamma^N)^x\in \Omega^N_{\text{mic}}}c_x(\gamma^N),\qquad c_x(\gamma^N) := \frac{1}{2}\big[\xi_{x+{\bf e}^-_x}(1-\xi_x) + \xi_x(1-\xi_{x+{\bf e}^-_x})\big].\label{eq_def_c_x_gamma}
\end{equation}
The indicator function above is related to the non-locality of the dynamics, 
see the right figure of Figure~\ref{fig_dyn_not_monotone_non_local} and the second point of Remark~\ref{rmk_def_contour_dyn}. 
We will also say that "$x$ is flipped" to mean that the block with centre $i$ is added or deleted. 
The connection with the SSEP is further discussed in Section~\ref{sec_heuristics}. 

Recalling the jump rates at the poles in Definition~\ref{def_contour_dynamics}, 
the generator $\lcal_\beta$ of the contour dynamics at $\beta>0$ then acts on functions $f:\Omega^N_{\text{mic}}\rightarrow\R$ according to:
\begin{align}
&\forall \gamma^N\in\Omega^N_{\text{mic}},\qquad N^2\lcal_\beta f(\gamma^N) = N^2\sum_{x\in V(\gamma^N)}{\bf 1}_{(\gamma^N)^x\in \Omega^N_{\text{mic}}}c_x(\gamma^N)\big[f\big((\gamma^N)^x\big)-f(\gamma^N)]\label{eq_def_generateur_H_is_0}\\
&+N^2\sum_{k=1}^4\bigg[{\bf 1}_{p_k(\gamma^N)=2}{\bf 1}_{(\gamma^N)^{-,k}\in\Omega^N_{\text{mic}}}\big[f\big((\gamma^N)^{-,k})-f(\gamma^N)\big] 
+ \sum_{\substack{x\in V(\gamma^N)\cap P_k(\gamma^N) \\
x+{\bf e}^{\pm}_x\in P_k(\gamma^N)}}e^{-2\beta}\big[f\big((\gamma^N)^{+,x}\big)-f(\gamma^N)\big]\bigg].\nonumber
\end{align}
In~\eqref{eq_def_generateur_H_is_0}, 
the first line corresponds to the SSEP-like updates, 
and the second line to the poles, 
with the last term corresponding to regrowth moves. 
Note the $N^2$ factor in the generator corresponding to a diffusive rescaling of time, 
which already appeared in the hydrodynamics~\eqref{eq_Lacoin_1}.
\subsubsection{Initial condition of the dynamics, topology and effective state space}
\label{section_IC}
In this section, we explain how a suitable choice of initial condition, 
not restrictive at the macroscopic level, makes the contour dynamics local and prevents various pathologies.

The contour dynamics and state space $\Omega^N_{\text{mic}}$ have to be defined with great generality, which means that the following degenerate microscopic curves have to be considered:
\begin{itemize}
	\item[(i)] curves with two non-consecutive regions at microscopic distance as in Figure~\ref{fig_dyn_not_monotone_non_local}, 
	corresponding to self-intersecting curves in the limit;
	\item[(ii)] curves with two poles at microscopic distance from one-another, corresponding to curves which have one or more regions reduced to a point in the limit;
	\item[(iii)] curves at microscopic distance from the origin $0$;
		\item[(iv)] curves with one or more poles such that the volume of all blocks at macroscopic distance beneath that pole vanishes in the $N\to\infty$ limit. 
		This corresponds in the limit to droplets with poles at strictly positive (possibly infinite) distance from the interior of the droplet, see the right-hand side of Figure~\ref{fig_non_simple_curve}.
\end{itemize}
Let us discuss how to avoid these degenerate cases.

Starting from a curve that does not feature the pathologies of items (i)-(iii), 
we will see in Lemma~\ref{lemm_local_jump_rates} below that a change in the volume of the droplet (i.e. order $N^2$ blocks) is required for any of these pathologies to arise.  
One expects a change in volume to require a diffusive amount of time in analogy with the zero temperature Glauber dynamics (this is indeed the case, see Proposition~\ref {prop_short_time_existence_of_weak_solution}). 
Thus to avoid (i)-(iii) it will be enough to start from a nice enough initial condition as defined below.

The last item is different and responsible for much of the difficulty in the study of the dynamics. 
Indeed, growth at the poles can occur independently of the shape of the curves and could therefore occur on time scales much shorter than diffusive.  
This is independent of the initial condition and will have to be ruled out by a separate dynamical estimate (see Proposition~\ref{prop_short_time_existence_of_weak_solution}).

\paragraph{Initial condition.} 
We now explain how to choose the initial condition to avoid situations in items (i)-(iii). 
For $\gamma\in\Omega$ (see Definition~\ref{def_state_space})),  
recall that $\Gamma$ denotes the associated droplet. 
Introduce $\Gamma'\subset\Gamma$ as the largest droplet with simple boundary such that $d_{L^1}(\Gamma,\Gamma')=0$ (see Figure~\ref{fig_gamma_prime_z_k_w_k} in Appendix~\ref{app_prop_e_r}). 
When $\gamma$ is simple one has simply $\Gamma'=\Gamma$.  
The volume distance $d_{L^1}$ between bounded sets $A,B\subset\R^2$ is defined by:
\begin{equation}
d_{L^1}(A,B) 
= 
\int_{\R^2}\big|{\bf 1}_{A}-{\bf 1}_{B}\big|\, du\, dv
.
\label{eq_def_distance_L1}
\end{equation}
\noindent\textbf{Notation:} in the rest of the article, to avoid constantly alternating between interfaces and their associated droplets, we chose as much as possible to state results in terms of interfaces $\gamma,\tilde \gamma\in \Omega$ exclusively. In particular, we will use the convention:
\begin{equation}
d_{L^1}(\gamma,\tilde\gamma) 
:= 
d_{L^1}(\Gamma,\tilde\Gamma)
.
\label{eq_dist_volume_of_interf_as_dist_droplets}
\end{equation}
Define now $q(\gamma)$ as the distance between $0$ and $\gamma$ and let $r(\gamma)$ denote the smallest vertical or horizontal distance between consecutive poles of $\partial\Gamma'$. 
Define also $r'(\gamma)$ as the minimum of the distances between regions $1,3$ and between regions $2,4$ of $\partial\Gamma'$ 
(see~\eqref{eq_def_q_appB}--\eqref{eq_condition_sur_CI}--\eqref{eq_def_r'_appB} for precise definitions). 
\begin{lemm}\label{lemm_local_jump_rates}
Let $\gamma\in\Omega$ be such that $q(\gamma)>0$, $r(\gamma)>0$ and $r'(\gamma)>0$.   
There is then $r_0=r_0(\gamma)>0$ such that:
\begin{itemize}
	\item all droplets $\tilde \Gamma$ associated with a curve $\tilde\gamma\in\Omega$ and such that $d_{L^1}(\gamma,\tilde \gamma)\leq r_0^2$ satisfy $q(\tilde\gamma)>q(\gamma)/2$, 
$r(\tilde\gamma)>r(\gamma)/2$ and $r'(\tilde\gamma)>r'(\gamma)/2$. 
	\item The jump rates $c(\tilde \gamma^N,\cdot)$ of any $\tilde \gamma^N\in\Omega^N_{\text{mic}}$ with $d_{L^1}(\gamma,\tilde\gamma^N)\leq r_0^2$ are local:
	there is $N(\gamma)\in\N_{\geq 1}$ such that, for any $N\geq N(\gamma)$ and any $x\in V(\tilde\gamma^N)$, 
	the value of $c(\tilde\gamma^N,(\tilde\gamma^N)^x)$ can be determined through the knowledge of points at $1$-distance at most $3$ from $x$. 
\end{itemize} 
\end{lemm}
Note that $r'(\gamma)>0$ is true as soon as $\gamma$ is a simple curve. 
Note also that $r'(\gamma)>0$ ensures the locality of the jump rates for all microscopic curves close enough to $\gamma$ since it forbids the situation in the right figure of Figure~\ref{fig_dyn_not_monotone_non_local}. 
Lemma~\ref{lemm_local_jump_rates} is proven in Appendix~\ref{app_prop_e_r}, 
see Lemma~\ref{lemm_stability_initial_condition_appendix}.

How to choose an initial condition to avoid pathologies and to have a local dynamics is now clear and stated next. 
\begin{defi}[Initial condition of the dynamics]\label{def_CI}
Let $\gamma^{\mathrm{ref}}\in\Omega$ be such that $q(\gamma^{\mathrm{ref}})>0$ and $r(\gamma^{\mathrm{ref}})>0$. 
Assume also that $\gamma^{\mathrm{ref}}$ is a simple curve (thus $r'(\gamma^{\mathrm{ref}})>0$). 

The dynamics is started from the microscopic curve $\gamma^{\mathrm{ref},N}\in\Omega^N_{\text{mic}}$ defined as follows. 
Let $\Gamma^{\mathrm{ref},N}$ be the droplet obtained by discretising $\Gamma^{\mathrm{ref}}$ according to:
\begin{align}
\Gamma^{\mathrm{ref},N} = \bigcup_{i\in (N^{-1}\Z^*)^2 : C_i\subset \Gamma^{\mathrm{ref}}}C_i,\qquad C_i := i + \frac{1}{2N}[-1,1]^2.
\end{align}
Then $\partial\Gamma^{\mathrm{ref},N}$ is a simple curve for large enough $N\in\N$. 
Further, if $N$ is large enough, up to adding blocks at the poles to ensure each pole of $\partial\Gamma^{\mathrm{ref},N}$ contains at least two blocks, 
we may assume $\partial\Gamma^{\mathrm{ref},N} \in\Omega^N_{\text{mic}}$. 
We then set $\gamma^{\mathrm{ref},N} := \partial\Gamma^{\mathrm{ref},N}$.
\end{defi}
Starting from $\gamma^{\mathrm{ref},N}$, 
proving that in the large $N$ limit no pathology arise after a positive time is one of the difficulties, treated in Proposition~\ref{prop_short_time_existence_of_weak_solution}.\\

\paragraph{Effective state space.}
Starting from $\gamma^{\mathrm{ref}}$, we will study the dynamics on a subset of $\Omega$ composed of curves that, like $\gamma^{\mathrm{ref}}$, 
do not fall under the degenerate situations of items (i)--(iii). 
Such curves satisfy the following property. 
\begin{proper}\label{prop_IC}
A curve $\gamma\in\Omega$ satisfies Property~\ref{prop_IC} if $q(\gamma),r(\gamma),r'(\gamma)>0$, 
where these quantities are defined informally above~\eqref{eq_def_distance_L1} and rigorously above Lemma~\ref{lemm_stability_initial_condition_appendix}. 
\end{proper}
We \emph{could} then carry out the study of the contour dynamics focussing on elements of $\Omega$ satisfying Property~\ref{prop_IC}. Results in that direction are stated in Theorem~\ref{theo_large_dev_general}. 

However, this choice creates difficulties at the level of the topology. 
To simplify the exposition and focus on the probabilistic aspects of the droplet evolution, 
we choose to study the contour dynamics acting on curves in a small volume neighbourhood of $\gamma^{\mathrm{ref}}$, 
which in view of Lemma~\ref{lemm_local_jump_rates} is enough to ensure Property~\ref{prop_IC} still holds.  
In this sense, all results stated in Section~\ref{sec_results} apart from Theorem~\ref{theo_large_dev_general} can be understood as short-time results, as we only consider interfaces close to the initial condition of the dynamics. 
\begin{defi}[Effective state space $\e$]\label{def_effective_state_space}
The effective state space $\e$ is the subset of $\Omega$ made of curves $\gamma$ satisfying $d_{L^1}(\gamma^{\mathrm{ref}},\gamma)\leq r_0^2$, 
with $r_0 = r_0(\gamma^{\mathrm{ref}})$ the quantity in Lemma~\ref{lemm_local_jump_rates}. 
In particular, all curves in $\e$ satisfy Property~\ref{property_state_space} by item 1 of Lemma~\ref{lemm_local_jump_rates}.

At the microscopic level, 
we will consider elements of $\Omega^N_{\text{mic}}\cap \e$. 
The jump rates of the contour dynamics for each such curve are local by item 2 of Lemma~\ref{lemm_local_jump_rates}.
\end{defi}
\subsubsection{Test functions and tilted dynamics}
In the breakthrough paper~\cite{Kipnis1989}, a powerful method was introduced to study large deviations for interacting particle systems. 
As a basic ingredient, 
it relies on the introduction of suitable tilted dynamics. In our case, these dynamics are defined as follows. Consider the following set $\C$ of test functions:
\begin{equation} 
\C 
:= 
C^2_c\big(\R_+\times\R^2\big)
.
\label{eq_def_ensemble_test_functions}
\end{equation}
In~\eqref{eq_def_ensemble_test_functions}, the subscript $c$ means compactly supported. We will frequently write $G_t$ for the function $x\in\R^2\mapsto G(t,x)$, for $t\geq 0$. For $H\in\C$, define another (time-inhomogeneous) Markov chain with generator $N^2\lcal_{\beta,H}$ by modifying the jump rates as follows. If $\gamma\in \Omega$, recall that $\Gamma$ stands for the droplet associated with $\gamma$ and let:
\begin{align}
\big<\Gamma,H_t\big> 
:= 
\int_{\Gamma}H_t(u,v) \, du\, dv,
\qquad t\geq 0
.
\label{eq_def_crochet}
\end{align}
Then, for each $\gamma^N,\tilde\gamma^N\in \Omega^N_{\text{mic}}\cap \e$ and associated droplets $\Gamma^N,\tilde \Gamma^N$, the tilted jump rates are:
\begin{equation}
\forall t\geq 0,\qquad 
c^{H_t}(\gamma^N,\tilde\gamma^N) 
:= 
c(\gamma^N,\tilde\gamma^N)\exp\Big[ N\big<\tilde\Gamma^N,H_t\big> - N\big<\Gamma^N,H_t\big>
\Big]
.  
\label{eq_def_jump_rates_H}
\end{equation}
As each site is associated with a block of area $1/N^2$, 
the tilt by $N\big<\Gamma,H_t\big>$ can be thought of as adding a small magnetic field of size $O(N^{-1})$ at each site. 
This magnetic field is inhomogeneous but regular in space and time. 
The probability measure associated with the speeded-up generator $N^2\lcal_{\beta,H}$ will be denoted by $\Prob^N_{\beta,H}$, or simply $\Prob^N_{\beta}$ when $H\equiv 0$ (recall that the diffusive, $N^2$ scaling is the correct one for motion by curvature). The corresponding expectations are denoted by $\E^N_{\beta,H},\E^N_{\beta}$ respectively.
\subsection{Results}\label{sec_results}
Our first result is a stability estimate. 
It states that, in the large $N$ limit, 
trajectories starting from the discretisation $\gamma^{\mathrm{ref},N}$ of the curve $\gamma^{\mathrm{ref}}$ of Definition~\ref{def_CI} typically have length bounded independently of $N$ and stay close to $\gamma^{\mathrm{ref}}$ in volume for short time, 
thereby avoiding the pathologies of items (i)--(iv) of Section~\ref{section_IC}.
\begin{prop}\label{prop_short_time_existence_of_weak_solution}
Let $\beta>\log 2$ and $H\in\C$. 
Recall that the initial condition of the dynamics $\gamma^{\mathrm{ref},N}$ for $N\in\N_{\geq 1}$ is given in Definition~\ref{def_CI}. 
Then:
\begin{enumerate}
	\item The length of an interface is of order $1$ in the following sense. 
	For each time $T>0$, there are constants $C(\beta,H,T),C(H)>0$ independent of $\gamma^{\mathrm{ref}}$ such that:
	\begin{align}
\forall A>0,\qquad 
\limsup_{N\rightarrow\infty}\frac{1}{N}\log \Prob^N_{\beta,H}\Big(\sup_{t\leq T}|\gamma^N_t|\geq A\Big) 
\leq 
-C(\beta,H,T)A + |\gamma^{\mathrm{ref}}|\beta + C(H)
.	
	\end{align}
	\item There is a time $t_0(\beta,H,r_0,|\gamma^{\mathrm{ref}}|)>0$, with $r_0$ given in Definition~\ref{def_effective_state_space} of $\e$, 
	such that:
	\begin{equation}
\lim_{N\rightarrow\infty}
\Prob^N_{\beta,H}\Big(\forall t\leq t_0(\beta,H,\gamma^{\mathrm{ref}}), \gamma^N_t\in \e\Big)=1
.
\label{eq_estimate_tau_in_exp_moins_N}
	\end{equation}
\end{enumerate}
\end{prop}
%
We use item 2) in Proposition~\ref{prop_short_time_existence_of_weak_solution} to only work with trajectories taking values in the effective state space $\e$ at each time (except Theorem~\ref{theo_large_dev_general}, where general trajectories are treated).\\

The second result, Proposition~\ref{prop_value_slope_at_poles}, concerns the role of the parameter $\beta$. 
This result is perhaps the most striking feature of the contour dynamics. 
To state it, we need some notations. 
For $\gamma^N\in \Omega^N_{\text{mic}}\cap \e, \ell\in \N_{\geq 1}$ and a vertex $x\in V(\gamma^N)=(N^{-1}\Z)^2\cap \gamma^N$, 
recall the definition~\eqref{eq_def_xi} of $\xi_x$ and denote by $\xi_x^{+,\ell}$ the local average (see Figure~\ref{fig_penteopole}):
\begin{align}
\xi_x^{+,\ell} 
= 
\frac{1}{\ell+1}\sum_{\substack{y\in V(\gamma^N)\cap B_1(x,\ell/N) \\ y\geq x   }}\xi_y
,
\end{align}
where $B_1(x,a)$ is the ball of centre $x$ and radius $a>0$ in 1-norm~\eqref{eq_def_normes_1_2}. 
By $y\geq x$ we mean that $y$ is encountered after $x$ when travelling on $\gamma^N$ clockwise. 
The parameter $\ell$ will always be chosen much smaller than the number of points in $V(\gamma^N)$, so that no vertex is counted twice in $\xi^{+,\ell}_x$. 

We shall informally refer to $\xi_x^{+,\ell}$ as the slope (on the right-side of $x$). 
The slope $\xi_x^{-,\ell}$ on the left of $x$ is defined similarly by averaging over points $y\leq x$. 
\begin{prop}\label{prop_value_slope_at_poles}
Take $\beta>\log 2$ and a time $T>0$. Then, for any bias $H\in\C$, any test function $G\in\C$ and any $\delta>0$, if $k\in\{1,3\}$:
\begin{align}
&\lim_{\epsilon\rightarrow 0}
\limsup_{N\rightarrow\infty}\frac{1}{N}\log\Prob^{N}_{\beta,H}\bigg(\forall t\in[0,T], \gamma^N_t \in\e;
\nonumber\\
&\hspace{5cm}\bigg| \int_0^{T}G(t,L_k(\gamma^N_t))\big(\xi_{L_k(\gamma^N_t)}^{\pm,\epsilon N} - e^{-\beta}\big)\, dt\bigg|\geq \delta\bigg) 
= 
- \infty
.
\end{align}
If on the other hand $k\in\{2,4\}$:
\begin{align}
&\lim_{\epsilon\rightarrow 0}
\limsup_{N\rightarrow\infty}\frac{1}{N}\log\Prob^{N}_{\beta,H}\bigg(\forall t\in[0,T], \gamma^N_t \in\e;\nonumber\\
&\hspace{5cm}\bigg| \int_0^{T}G(t,L_k(\gamma^N_t))\big(1-\xi_{L_k(\gamma^N_t)}^{\pm,\epsilon N} - e^{-\beta}\big)\, dt\bigg|\geq \delta\bigg) 
= 
- \infty
.
\end{align}
\end{prop}
\begin{figure}
\begin{center}
\includegraphics[width=10cm]{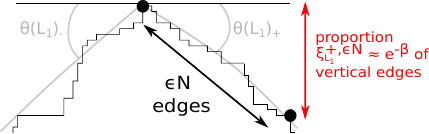} 
\caption{North pole of a curve with the proportion $\xi^{+,\epsilon N}_{L_1}$ of vertical edges to the right of the pole. For drawing convenience, $\xi^{+,\epsilon N}_{L_1}$ is assumed to be close to $e^{-\beta}$ (this is only shown to be true for its time average in Proposition~\ref{prop_value_slope_at_poles}).  
The corresponding angles $\theta(L_1)_\pm$ are also drawn. \label{fig_penteopole}}
\end{center}
\end{figure}
Proposition~\ref{prop_value_slope_at_poles} shows that, as long as trajectories remain in the effective state space $\e$, the time average of the slopes on either side of the poles are fixed in terms of $\beta$. 
As we explain in Section~\ref{sec_heuristics}, Proposition~\ref{prop_value_slope_at_poles} can be understood as a statement that the pole dynamics has an effect similar to reservoirs at density $e^{-\beta}$ or $1-e^{-\beta}$ in the simple exclusion process.\\ 

In the following, it will be useful to rephrase the condition on the slope described in Proposition~\ref{prop_value_slope_at_poles} in terms of a condition that makes sense for general curves in $\e$. 
We shall say that a curve $\gamma\in\e$ has slope $e^{-\beta}$ at pole $k$ with $1\leq k\leq 4$ (see Figure~\ref{fig_penteopole}) 
if the angle $\theta(L_k(\gamma)_\pm)$ between the tangent vector ${\bf T}$ approaching $L_k(\gamma)$ from the left ($-$) or the right ($+$) and the vector ${\bf b}_0 = (1,0)$ satisfies:
\begin{equation}
\tan\Big(\theta(L_k(\gamma)_-) + \frac{(k-1)\pi}{2}\Big) = \frac{e^{-\beta}}{1-e^{-\beta}} 
= 
-\tan\Big(\theta(L_k(\gamma)_+) + \frac{(k-1)\pi}{2}\Big)
.
\label{eq_tangente_angle_at_the_pole}
\end{equation}
\noindent\textbf{Hydrodynamic limit}\\
Next, we investigate the typical evolution, in the large $N$ limit, 
of interfaces following the contour dynamics with a bias $H\in\C$. 
We prove that they evolve according to an anisotropic motion by curvature as in~\eqref{eq_def_MCM}, 
but with an influence of the parameter $\beta$. 
To prove such a result, a suitable topology on trajectories is required. 
In the proof of the hydrodynamic limit for the zero temperature stochastic Ising model in \cite{Lacoin2014}--\cite{Lacoin2014a}, the authors prove uniform convergence in time for the topology associated with the Hausdorff distance~\eqref{eq_def_distance_Hausroff}. The Hausdorff distance between sets appears as a natural distance to put on the state space. Indeed, away from each pole, portions of the interface can be mapped to a SSEP (see Section~\ref{sec_heuristics}). Hausdorff convergence of the interface can then be shown to be equivalent to weak convergence of the empirical measure in the associated SSEP, a topology in which hydrodynamics are known for this model.
\begin{figure}
\begin{center}
\includegraphics[width=12cm]{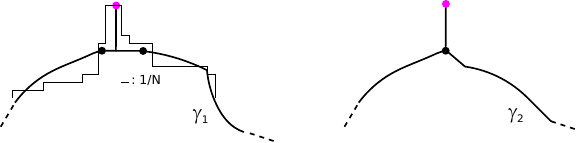} 
\caption{Portions of two curves $\gamma_1,\gamma_2\in\Omega$ and of one element of a sequence of microscopic curves in $\Omega_{\text{mic}}^N,N\in\N_{\geq 1}$ that converges to $\gamma_1$ in Hausdorff distance. 
The north poles of $\gamma_1,\gamma_2$ are point-like and represented by the magenta dots. 
The black dots delimit the flat portion on which the vertical line beneath the north poles of $\gamma_1,\gamma_2$ rest. This portion may ($\gamma_2$) or may not ($\gamma_1$) be reduced to a point. This possibility causes difficulties at the level of the topology, in Appendix~\ref{app_prop_e_r}.\label{fig_non_simple_curve}}
\end{center}
\end{figure}
In the case of the contour dynamics, the Skorokhod topology associated with the Hausdorff distance seems like a suitable choice. 
However, this topology turns out to be too strong. 
Indeed, although it is possible to control the volume of droplets pointwise in time (in terms of the volume distance $d_{L^1}$ defined in~\eqref{eq_def_distance_L1}), 
a pointwise control seems out of reach for the poles. 
This is because poles can grow regardless of the shape of the interface, making it very difficult to prove their macroscopic diffusive motion.

Still, some information on the poles needs to be present in the topology as we need to characterise curves with degenerate poles as in Figure~\ref{fig_non_simple_curve}.  
It turns out to be possible and sufficient to control the trajectory of the poles in a time-integrated way. 
For each $T>0$, define the set:
\begin{align}
E([0,T],\Omega) &:= D_{L^1}([0,T],\Omega)\cap\Big\{(\gamma_t)_{t\leq T} : \int_0^T |\gamma_t|\, dt<\infty\Big\},\nonumber\\
E([0,T],\e) &:= E([0,T],\Omega)\cap \Big\{ (\gamma_t)_{t\leq T}: \gamma_t\in \e\text{ for }t\in[0,T]\Big\}.\label{eq_def_E_0_T_Omega}
\end{align}
\textbf{Notation: } we often use the subscript $\cdot$ (as in $\gamma_\cdot$ in~\eqref{eq_def_d_E}) to denote a trajectory, provided the time interval on which it is defined is clear from the context. \\

Let $d_{L^1}^S$ denote the Skorokhod distance associated with the volume distance $d_{L^1}$ between curves (recall the convention~\eqref{eq_dist_volume_of_interf_as_dist_droplets} that the $L^1$ distance between two curves is by definition the $L^1$ distance between the droplets they delimit), 
see e.g.~\cite[Chapter 4 Section 1]{Kipnis1999} for the classical construction of the Skorokhod distance. 
Recall also that $d_{\mathcal{H}}$ is the Hausdorff distance~\eqref{eq_def_distance_Hausroff} on $\Omega$. 
The set $E([0,T],\e)$ is then equipped with the distance:
\begin{align}
\forall \gamma_\cdot,\tilde \gamma_\cdot \in E([0,T],\e),
\qquad 
d_E(\gamma_\cdot,\tilde \gamma_\cdot) 
:= 
d^S_{L^1}(\gamma_\cdot,\tilde \gamma_\cdot) + \int_0^T d_{\mathcal H}(\gamma_t,\tilde \gamma_t)\, dt
.
\label{eq_def_d_E}
\end{align}
Note that, since trajectories in $E([0,T],\e)$ have almost always finite length, the associated droplets are almost always bounded subsets of $\R^2$, thus the Hausdorff distance in~\eqref{eq_def_d_E} is almost always well-defined. 
Properties of $(E([0,T],\e),d_E)$ are established in Appendix~\ref{sec_the_set_E(0,T0)}. 
\\

Informally stated, our hydrodynamic limit result is the following. 
Introduce the mobility $\mu$ of the model:
\begin{equation}
\mu(\theta) 
:= 
\frac{|\sin(2\theta)|}{2(|\sin(\theta)|+|\cos(\theta)|)},\qquad \theta\in[0,2\pi]
.
\label{eq_def_mobility_first}
\end{equation}
The sequence $(\mathbb{P}^N_{\beta,H})_{N\in \N_{\geq 1}}$ of laws of the trajectory of interfaces converges weakly to a probability measure concentrated on trajectories in $E([0,T],\e)$ that are weak solutions, in the sense defined below in~\eqref{eq_formulation_faible_avec_der_en_temps}, of:
\begin{equation}
\begin{cases}
\partial_t \gamma\cdot {\bf N}= a\partial^2_s \gamma\cdot {\bf N} - \mu H = ak- \mu H \qquad\text{away from the poles},\\
\gamma_t\text{ satisfies }\eqref{eq_tangente_angle_at_the_pole} \text{ at almost every }t\in[0,T].
\end{cases}\label{eq_MCM_with_beta_and_H}
\end{equation}
Above, ${\bf N}$ is the inwards normal vector, $s$ the arclength coordinate, 
$k$ the curvature.  
The anisotropy $a$ (recall~\eqref{eq_def_a}) and mobility $\mu$ are functions of the angle $\theta(x)\in[0,2\pi]$ between the tangent ${\bf T}(x)$ at a point $x\in\gamma_t$ ($t\in[0,T]$) and the first basis vector ${\bf b}_0$ (recall~\eqref{eq_def_b_theta}). 

The precise hydrodynamic limit result is given next, for sufficiently short time only (see however Remark~\ref{rmk_longer_time}). 
To state it, let $\alpha$ denote a primitive of the anisotropy $a$ defined in~\eqref{eq_def_a}, 
in the sense that $\alpha'(\theta) = -a(\theta)$ for each $\theta \in[0,2\pi]\setminus(\pi/2)\Z$:
\begin{equation}
\alpha(\theta) = \frac{a(\theta)}{2}\frac{\sin(2\theta)\cos(2\theta)}{|\sin(2\theta)|} = \frac{{\bf T}(\theta)\cdot {\bf b}_0 {\bf T}(\theta)\cdot {\bf b}_{\pi/2}}{4\|{\bf T}(\theta)\|_1}\Big[\frac{1}{|{\bf T}(\theta)\cdot {\bf b}_{\pi/2}|}-\frac{1}{|{\bf T}(\theta)\cdot {\bf b}_0|}\Big],\label{eq_def_alpha}
\end{equation}
with ${\bf T}(\theta) = \cos(\theta){\bf b}_0+\sin(\theta){\bf b}_{\pi/2}$. 
\begin{prop}\label{prop_limite_hydro}
Recall Definition~\ref{def_CI} of the initial condition $\gamma^{\mathrm{ref},N}$ for $N\in\N_{\geq 1}$. 
Let $\beta>\log 2$, $H\in\C$ and $t_0 = t_0(\beta,H,r_0,|\gamma^{\mathrm{ref}}|)$ be the time of Proposition~\ref{prop_short_time_existence_of_weak_solution}. Then $(\mathbb{P}^N_{\beta,H})_N$ converges, in the weak topology associated with $d_E$, to a measure concentrated on trajectories in $E([0,t_0],\e)$ that have almost always point-like poles, i.e. for a.e. $t\in[0,t_0]$ and each $1\leq k\leq 4$, 
the pole $P_k(\gamma_t) := [L_k(\gamma_t),R_k(\gamma_t)]$ is reduced to the point $L_k(\gamma_t)=R_k(\gamma_t)$. 
Moreover, these trajectories are weak solutions of anisotropic motion by curvature with drift on $[0,t_0]$ in the following sense: for any $t\leq t_0$ and any test function $G$ in the set $\C$ defined in~\eqref{eq_def_ensemble_test_functions},
\begin{align}
\big<\Gamma_t,G_t\big> - \big<\Gamma^{\mathrm{ref}},G_0\big> - \int_0^t \big<\Gamma_{t'},\partial_{t'}G_{t'}\big>\, dt'
&= 
\int_0^{t}\int_{\gamma_{t'}\setminus \cup_kP_k(\gamma_{t'})} \alpha(\theta(s)) \partial_{s} G({t'},\gamma_{t'}(s))\, ds \, d{t'}
\nonumber\\
&\quad - \sum_{k=1}^4\int_0^{t}\Big(\frac{1}{2}-e^{-\beta}\Big)G({t'},L_k(\gamma_{t'}))\, d{t'}
\nonumber\\
&\quad
+ \int_0^{t}\int_{\gamma_{t'}}\mu(\theta(s)) (HG)({t'},\gamma_{t'}(s))\, ds\, d{t'}
.
\label{eq_formulation_faible_avec_der_en_temps}
\end{align}
Above, $\Gamma_t$ is the droplet associated with $\gamma_t$, $\big<\Gamma_t,G_t\big>$ is the integral of $G_t$ on $\Gamma_t$ as in~\eqref{eq_def_crochet} and $s$ is the arclength coordinate on $\gamma_t$ at time $t\in[0,t_0]$. 
\end{prop}
\begin{rmk}\label{rmk_longer_time}
The time $t_0= t_0(\beta,H,r_0,|\gamma^{\mathrm{ref}}|)$ until which Proposition~\ref{prop_short_time_hydro} is proven does not make use of the structure of solutions to~\eqref{eq_formulation_faible_avec_der_en_temps}, 
and one can in fact improve the result as follows (this improvement is carried out in Section~\ref{sec_up_to_time_T0}). 
Take $H\in\C$, $T\geq t_0(\beta,H,|\gamma^{\mathrm{ref}}|)$ and make the following assumptions:
\begin{enumerate}
	\item Equation~\eqref{eq_formulation_faible_avec_der_en_temps} admits only one solution on $[0,T]$, call it $\gamma^H_{\cdot} = (\gamma^H_t)_{t\leq T}$. 
	\item $\gamma^H_\cdot$ remains in $\e$ until time $T$, 
	in the sense:
\begin{equation}
	\exists \zeta_H>0,\qquad 
	B_{d_E}\big(\gamma^H_\cdot,\zeta_H\big)\subset E([0,T],\e).
	\end{equation}
\end{enumerate}
Then $(\Prob^N_{\beta,H})_N$, as a sequence of measures on $E([0,T],\Omega)$, converges weakly to the measure $\delta_{\gamma^H_\cdot}$.\demo 
\end{rmk}
\begin{rmk}
The term on the second line of~\eqref{eq_formulation_faible_avec_der_en_temps} fixes the value of the slope at the pole of curves to the one prescribed by Proposition~\ref{prop_value_slope_at_poles}. 
To see why, assume that the curvature $(k_t)_{t\leq t_0}$ on a solution $(\gamma_{t})_{t\leq t_0}$ of~\eqref{eq_formulation_faible_avec_der_en_temps} is, 
say, continuous and bounded on $\gamma_t\setminus \cup_k P_k(\gamma_t)$ at each time $t\leq t_0$. 
By definition, the tangent angle $s\mapsto\theta(s) = \theta(\gamma_t(s))$ then satisfies $\partial_{s}\theta(s) = -k_t(s)$ 
for each arclength coordinate $s$ corresponding to a point in $\gamma_t\setminus \cup_k P_k(\gamma_t)$, 
with the $-$ sign due to the clockwise parametrisation of $\gamma_t$. 
Let $G\in\C$. 
Integrating $\alpha\partial_{s} G(t,\cdot)$ by parts on each region in~\eqref{eq_formulation_faible_avec_der_en_temps} for a fixed $t\in[0,t_0]$, one finds by definition~\eqref{eq_def_alpha} of $\alpha$:
\begin{align}
\int_{\gamma_t\setminus \cup_k P_k(\gamma_t)} \alpha(\theta(s)) &\partial_{s} G(t,\gamma_t(s))\, ds
\nonumber\\
&\quad 
= \sum_{k=1}^4 \big[\alpha\big(\theta(L_{k+1}(\gamma_t))_-\big)G(t, L_{k+1}(\gamma_t))-\alpha\big(\theta(R_k(\gamma_t))_+\big)G(t, R_k(\gamma_t))\big] 
\nonumber\\
&\qquad 
- \int_{\gamma_t\setminus \cup_k P_k(\gamma_t)} a(\theta(s))k(\gamma_t(s))G(t,\gamma_t(s)) \, ds
.
\label{eq_IPP_curvature}
\end{align}
Since each pole is almost always reduced to a point, $L_k(\gamma_t) = R_k(\gamma_t)$ for each $k$ and almost every $t\in[0,t_0]$ and the sum in~\eqref{eq_IPP_curvature} compensates the second line of~\eqref{eq_formulation_faible_avec_der_en_temps} provided $ \alpha\big(\theta(L_{k+1}(\gamma_t))_-\big)= 1/4-e^{-\beta}/2 = -\alpha\big(\theta(L_k(\gamma_t))_+\big)$. 
This can be shown to hold when the tangent angle on either side of each pole satisfies~\eqref{eq_tangente_angle_at_the_pole}.\demo
\end{rmk}
\noindent\textbf{Large deviations}\\
We obtain upper-bound large deviations for the contour dynamics at each $\beta>\log 2$. 
Assuming solutions of~\eqref{eq_formulation_faible_avec_der_en_temps} to be unique, lower-bound large deviations can also be derived. 
Upper and lower bounds match for suitably regular trajectories. Specific to our model is, again, the control of the poles of the curves.\\
Let $T>0$ and $\beta>\log 2$. 
Given $\gamma^{\mathrm{ref}}$ as in Definition~\ref{def_CI} a trajectory $\gamma_\cdot\in  E([0,T],\e)$ with associated droplets $(\Gamma_t)_{t\leq T}$, define, recalling that $L_k,R_k$ are the extremities of the pole $P_k$:
\begin{align}
\ell^\beta_H(\gamma_\cdot) &= \big<\Gamma_{T},H_{T}\big> - \big<\Gamma^{\mathrm{ref}},H_0\big> - \int_0^{T} \big<\Gamma_t,\partial_t H_t\big>\, dt  -\int_0^{T}\, dt\int_{\gamma_t\setminus \cup_k P_k(\gamma_t)}\alpha(\theta(s))\partial_{s}H(t,\gamma_t(s))\, ds
\nonumber
\\
&\hspace{3.5cm}+\Big(\frac{1}{4}-\frac{e^{-\beta}}{2}\Big)\int_0^{T}\sum_{k=1}^4 \big[H(t,L_k(\gamma_t))+H(t,R_k(\gamma_t))\big]\, dt
.
\label{eq_def_ell_H}
\end{align}
Define also:
\begin{align}
J^\beta_H(\gamma_\cdot) = 
\ell^\beta_H(\gamma_\cdot) - \frac{1}{2}\int_0^{T}\int_{\gamma_t} \mu(\theta(s))H^2(t,\gamma_t(s))\, ds\,  dt
, 
\qquad\gamma_\cdot\in E([0,T],\e)
,
\label{eq_def_J_H}
\end{align}
where the mobility $\mu$ is defined in~\eqref{eq_def_mobility_first}. \\
To build the rate function, we will have to restrict the state space to control the behaviour of the poles. Introduce thus the subset $E_{pp}([0,T],\e)\subset E([0,T],\e)$ of trajectories with almost always point-like poles:  
\begin{equation}
E_{pp}([0,T],\e) = \bigg\{\gamma_\cdot\in E([0,T],\e) : \sum_{k=1}^4\int_0^{T}\|L_k(\gamma_t)-R_k(\gamma_t)\|_\infty dt = 0\bigg\}.\label{eq_def_E_pp}
\end{equation}
Recall that $R_k$ $(L_k)$ is the right (left) extremity of pole $k\in\{1,...,4\}$. Let us now define the rate function $I_\beta(\gamma_\cdot|\gamma^{\mathrm{ref}})$ for trajectories $\gamma_\cdot\in E([0,T],\e)$:
\begin{equation}
I_\beta(\gamma_\cdot|\gamma^{\mathrm{ref}}) = \begin{cases}\sup_{H\in\C} J_H^ \beta(\gamma_\cdot)\quad &\text{if }\gamma_\cdot\in E_{pp}([0,T],\e), \\
+\infty &\text{otherwise.}
\end{cases}\label{eq_def_rate_functions}
\end{equation}
\begin{rmk}
\begin{itemize}
	\item It is possible by Proposition~\ref{prop_value_slope_at_poles} to enforce that only trajectories with slope $e^{-\beta}$ at the poles at almost every time have finite rate function. One would expect this condition to already be present in~\eqref{eq_def_rate_functions}, but the very weak topology at the poles makes it more complicated to see than e.g. for a SSEP with reservoirs, as done in \cite{Bertini2009}.
	
	\item If $\beta=\infty$ and $\gamma_{\cdot}$ is a sufficiently regular trajectory in $C([0,T],\e)$ starting from $\gamma^{\mathrm{ref}}$ (say, with well-defined, continuous and bounded normal speed and curvature at each time $t\in(0,T]$), then setting $\beta=\infty$ in~\eqref{eq_def_rate_functions} one formally obtains:
\begin{equation}
I_{\infty}(\gamma_\cdot|\gamma^{\mathrm{ref}}) 
=
\frac{1}{2}\int_0^ {T}\int_{\gamma_t} \frac{\Big(v\big(\gamma_t(s)\big)-a\big(\theta(s)\big)k\big(\gamma_t(s)\big)\Big)^ 2}{\mu\big(\theta(s)\big)}\, ds\, dt
.
\label{eq_rate_function_beta_infty}
\end{equation}
As conjectured in~\eqref{eq_conjectured_form_rate_function}, the rate function $I_{\infty}(\cdot|\gamma^{\mathrm{ref}})$ thus measures the quadratic cost of deviations from anisotropic motion by curvature. At $\beta<\infty$, $I_\beta(\cdot|\gamma^{\mathrm{ref}})$ can also be written in the form~\eqref{eq_rate_function_beta_infty}, but only for trajectories that are not smooth: they must have kinks at the poles, in the sense that they satisfy the condition~\eqref{eq_tangente_angle_at_the_pole} at almost every time.\demo
\end{itemize}
\end{rmk}
Define the set $\mathcal A_{\beta,T}\subset E_{pp}([0,T],\e)$ of trajectories that can be obtained as a solution of the anisotropic motion by curvature with a smooth drift $H\in \C$~\eqref{eq_formulation_faible_avec_der_en_temps}:
\begin{align}
\mathcal A_{\beta,T}= \Big\{\gamma_\cdot \in E_{pp}([0,T],\e) : &\text{ there is a bias }H\in\C\text{ such that }\eqref{eq_formulation_faible_avec_der_en_temps}\text{ has a}\nonumber\\
&\text{ unique solution in }E([0,T],\e)\text{,  
and this solution is }\gamma_\cdot\Big\}
.
\label{eq_def_A_T_0_r_beta}
\end{align}
\begin{theo}\label{theo_large_dev}
Let $T>0$ and $\beta>\log 2$. For any closed set $C\subset E([0,T],\e)$:
\begin{equation}
\limsup_{N\rightarrow\infty} \frac{1}{N}\log\mathbb{P}^N_{\beta}\big(\gamma^N_\cdot \in C\big) 
\leq 
-\inf_{C}I_\beta(\cdot|\gamma^{\mathrm{ref}})
.
\end{equation}
Moreover, for any open set $O$ with $O\subset E([0,T],\e)$:
\begin{equation}
\liminf_{N\rightarrow\infty} \frac{1}{N}\log\mathbb{P}^N_{\beta}\big(\gamma^N_\cdot \in O\big) \geq -\inf_{O\cap \mathcal A_{\beta,T}}I_\beta(\cdot|\gamma^{\mathrm{ref}}).
\end{equation}
\end{theo}
\begin{rmk}
\begin{itemize}
	\item 
	The set $\A_{\beta,T}$ is expected to contain a large class of trajectories. In the $\beta=\infty$ case, it would for instance contain all classical solutions of the equation $v=ak-\mu H$, $H\in\C$, which can be studied by the method of \cite{Lacoin2014}\cite{Lacoin2014a}. When $\beta<\infty$ however, 
	even classical solutions of~\eqref{eq_MCM_with_beta_and_H} are extremely difficult to study due to the poles. A fortiori, the study of uniqueness and regularity of solutions of the weak formulation~\eqref{eq_formulation_faible_avec_der_en_temps} is difficult.
	\item A possible application of Theorem~\ref{theo_large_dev} is the analysis of metastability. For instance, applying a small, uniform field of the form $h/N$ ($h>0$), 
	one can use Theorem~\ref{theo_large_dev} to study the optimal trajectory for a nucleated droplet to cover the whole space. \\
	One can also ask about the typical speed at which such a droplet grows. 
	This speed is conjectured to be proportional to the size of the applied field \cite{Schonmann1998}, i.e. of order $1/N$. For the contour dynamics, curves move diffusively, which readily confirms the conjecture. The interested reader will find much more on metastability and its relation to large deviations in the book \cite{Olivieri2005}.
 \demo
\end{itemize}
\end{rmk}
We conclude this section by rephrasing Theorem~\ref{theo_large_dev} in a more general context. Elements of $\e$ are, by assumption (see Definition~\ref{def_effective_state_space}), in a small neighbourhood of the initial condition $\gamma^{\mathrm{ref}}$ for the volume distance. 
Working with curves in $\e$ is useful to avoid topology-related problems and obtain a large deviation bound valid for general sets. 
As claimed above Definition~\ref{def_effective_state_space} of $\e$, however, 
it is not important that microscopic curves be close to $\gamma^{\mathrm{ref}}$ (i.e. in $\e$), 
only that they satisfy Property~\ref{prop_IC}. 
The next theorem therefore improves on Theorem~\ref{theo_large_dev} for events corresponding to small balls around a given trajectory possibly far from $\gamma^{\mathrm{ref}}$, 
but satisfying the same Property~\ref{prop_IC} as $\gamma^{\mathrm{ref}}$ at each time. 

To state it, assume that $J^\beta_H$ is defined on the entire space $E([0,T],\Omega)$ with the same expression~\eqref{eq_def_J_H} (rather than on $E([0,T],\e)$). 
The rate function $I_\beta(\cdot|\gamma^{\mathrm{ref}})$ is correspondingly given for $\gamma_\cdot\in E([0,T],\Omega)$ by:
\begin{equation}
I_\beta(\gamma_\cdot|\gamma^{\mathrm{ref}}) = \begin{cases}\sup_{H\in\C} J_H^ \beta(\gamma_\cdot)\quad &\text{if }\gamma_\cdot\text{ has almost always point-like poles,} \\
+\infty &\text{otherwise.}
\end{cases}
\end{equation}
Similarly, $\mathcal A_{\beta,T}$ is now assumed to contain trajectories in $E([0,T],\Omega)$ with almost always point-like poles and that satisfy Property~\ref{prop_IC} at each time, rather than trajectories in $E_{pp}([0,T],\e)$.
\begin{theo}\label{theo_large_dev_general}
Let $\beta>\log 2$ and let $\bar\gamma_\cdot\in E([0,T],\Omega)$ be such that $\bar\gamma_t$ satisfies Property~\ref{prop_IC} at each time $t\leq T$. Then:
\begin{equation}
\limsup_{\zeta\rightarrow0}\limsup_{N\rightarrow\infty} \frac{1}{N}\log\mathbb{P}^N_{\beta}\big(\gamma^N_\cdot \in B_{d_E}\big(\bar\gamma_\cdot,\zeta\big)\big) 
\leq 
-I_\beta(\bar\gamma_\cdot|\gamma^{\mathrm{ref}})
.
\end{equation}
Moreover, if $\bar\gamma_\cdot$ is in $\mathcal A_{\beta,T}$, 
then:
\begin{equation}
\liminf_{\zeta\rightarrow 0}\liminf_{N\rightarrow\infty} \frac{1}{N}\log\mathbb{P}^N_{\beta}\big(\gamma^N_\cdot \in B_{d_E}\big(\bar\gamma_\cdot,\zeta\big)\big) 
\geq 
-I_\beta(\bar\gamma_\cdot|\gamma^{\mathrm{ref}})
.
\end{equation}
\end{theo}
\subsection{Heuristics on large deviations: link with the SSEP}\label{sec_heuristics}
In this section, we highlight the relationship between the contour dynamics away from the poles and the SSEP. 
This relationship serves as a guideline in the proof of large deviations (the structure of the proof is detailed in Section~\ref{sec_outline_proof}). 
A heuristic derivation of the rate function of Theorem~\ref{theo_large_dev} is proposed below using the link with the SSEP.\\

Take a curve $\gamma^N\in\Omega^N_{\text{mic}}$ (see Definition~\ref{def_state_space}) as in Figure~\ref{fig_quadrants_convex}. 
By Definition~\ref{property_state_space} of $\Omega$, 
$\gamma^N$ can be split into four regions. Let $1\leq k \leq 4$. 
Rotating the canonical reference frame $({\bf b}_0,{\bf b}_{\pi/2})$ by $\pi/4 + (k-1)\pi/2$, region $k$ of the boundary is turned into the graph of a $1$-Lipschitz function $f^{N,k}$ with derivative $\pm 1$. The $k=1$ case is illustrated on Figure~\ref{fig_Ising_ssep_no_bords}.

There is a well-known bijection between the graph of $f^{N,k}$ and a particle configuration which goes as follows. 
With each edge in region $k$, associate a site. 
Put a particle in the site if the corresponding edge corresponds to an interval on which $f^{N,k}$ has derivative $-1$ and no particle if $f^{N,k}$ has derivative $1$. 
Updates of the contour dynamics away from the poles then correspond to SSEP updates, as remarked in~\eqref{eq_link_ssep_intro}.\\

\begin{figure}[H]
\begin{center}
\includegraphics[width=12cm]{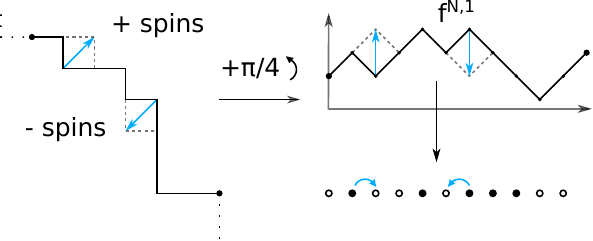} 
\caption{The mapping from a portion of region $1$ (left figure), delimited by the two big black dots, to the graph of a Lipschitz function $f^{N,1}$ (upper right figure). This Lipschitz function is then mapped to a particle configuration (lower right figure). Each edge of the initial interface corresponds to a possible site for particles. In the portion of the original interface delimited by the black dots, dynamical updates (blue arrows) do not change the length of the portion. These updates correspond to SSEP moves.}\label{fig_Ising_ssep_no_bords}
\end{center}
\end{figure}
If $\eta^{N,k} = (\eta^{N,k}(i))_i$ denotes the particle configuration obtained in region $k$ by the above mapping, with the index $i\in\N$ labelling particles sites with the convention that the site associated with the first edge of region $k$ has label $0$, 
then $\eta^{N,k}$ and $f^{N,k}$ are related through:
\begin{align}
\eta^{N,k}(i) 
:= 
\frac{1}{2}+\frac{(-1)^k\sqrt{2}N}{2}\Big[f^{N,k}\Big(\frac{i+1}{\sqrt{2}N}\Big)-f^{N,k}\Big(\frac{i}{\sqrt{2}N}\Big)\Big]
.
\label{eq_discrete_eta_tof}
\end{align}
Above, the $(-1)^k$ comes from the fact that the $f^k$ are defined in different referent frames ($f^1$ in $({\bf b}_{-\pi/4},{\bf b}_{\pi/4})$, $f^2$ in $({\bf b}_{-3\pi/4},{\bf b}_{-\pi/4})$, etc.).

Through this correspondence, the contour dynamics inside each region of an interface can be viewed as a SSEP. The dynamics at the poles (deletion or growth of two blocks at respective rates 1 and $e^{-2\beta}$) can then be viewed as a boundary dynamics that couples these SSEP. 
In Proposition~\ref{prop_value_slope_at_poles}, we saw that the dynamics at each pole acts like a moving reservoir, fixing the density of particles at the extremity of each SSEP in terms of $\beta$. As a first, informal approximation, it thus makes sense to view the contour dynamics as four SSEP coupled with reservoirs. \\

This approximation is further vindicated by the fact that using large deviation results for the SSEP on each region and rewriting them as quantities depending only on the associated curve yields the rate function~\eqref{eq_conjectured_form_rate_function} as we now explain.

From~\cite{Bertini2003}, 
the rate function for a SSEP with reservoirs, 
evaluated at a trajectory $(\rho^k_t)_{t\leq T}$ defined at time $t$ on $[a^k_t,b^k_t]$ (with all these objects smoothly varying in space and/or time) 
and starting from a profile $(\rho^0)^k$,
is known explicitly and should be finite only if the value of $\rho^k_t$ is fixed at $a^k_t,b^k_t$ to the value of the density of the reservoir 
(the argument in~\cite{Bertini2003} applies only when the interval $[a^k_t,b^k_t]$ does not depend on time, 
but we assume their result extends to the present case for the purpose of this discussion).

In view of Proposition~\ref{prop_value_slope_at_poles}, 
for the SSEP associated with the contour dynamics the condition at the extremity of the interval should be:
\begin{equation}
\forall t\in(0,T],\qquad \rho^k_t\big(a^k(\gamma_t)\big) = e^{-\beta} = 1- \rho^k_t\big(b^k(\gamma_t)\big),\quad 1\leq k \leq 4.\label{eq_bc_heuristics}
\end{equation} 
For such a SSEP trajectory, the rate function of~\cite{Bertini2003} reads:
\begin{equation}
I_{SSEP,k}\big((\rho^{k}_t)_{t\leq T}|(\rho^0)^k\big) 
=
 \int_0^T \int_{a^k_t}^{b^k_t} \frac{\big(\partial_t f^k - (1/4)\Delta f^k\big)^2}{2\rho^k(1-\rho^k)}\sqrt{2}\, du\, dt
.
\label{eq_rate_function_SSEP}
\end{equation}
The $\sqrt{2}$ factor in~\eqref{eq_rate_function_SSEP} is again related to the reference frame being tilted, 
and $f^k_t$ a $1$-Lipschitz function built from $\rho^k_t$ in analogy with~\eqref{eq_discrete_eta_tof} at each time $t\in[0,T]$ according to:
\begin{equation}
\rho^k_t(u) 
=
\frac{1+(-1)^k\partial_u f^k(u/\sqrt{2})}{2},
\qquad u\in [a^k_t,b^k_t]
. 
\label{eq_lien_rho_f}
\end{equation}
Note that Equation~\eqref{eq_lien_rho_f} only defines $f^k_t$ up to a constant. 
For the graph of $f^k_t$ in an appropriate tilted reference frame to coincide with region $k$ of an element $\gamma_t\in\Omega$ as we will use below, 
this constant must be chosen as ($k-1 = 4$ if $k=1$):
\begin{equation}
f^k_t(a^k_t) 
= 
a^{k-1}_t
.
\end{equation}
Let us now use~\eqref{eq_rate_function_SSEP} to connect with the contour dynamics rate function obtained in Theorem~\ref{theo_large_dev}. 
As in the microscopic case, a macroscopic curve $\gamma\in\e$ can be associated with "particle densities'' on each region. 
Indeed, for $1\leq k\leq 4$, 
region $k$ is by definition of $\Omega\supset\e$ the graph of a $1$-Lipschitz function $f^k$ on an interval $[a^k(\gamma),b^k(\gamma)]$, where ($R_{k+1} := R_1$ if $k=4$):
\begin{equation}
a^k (\gamma)
:= 
L_k\cdot {\bf b}_{\pi/4 - k\pi/2},\qquad 
b_k(\gamma) = R_{k+1}\cdot {\bf b}_{\pi/4 - k\pi/2}
. 
\end{equation}
The function $f^k$ is then associated with a function $\rho^k$ through~\eqref{eq_lien_rho_f}.

Let $\gamma_\cdot\in E([0,T],\e)$ start from $\gamma^{\mathrm{ref}}$ and consider the associated $(\rho^k_\cdot,f^k_\cdot)_{1\leq k \leq 4}$. 
We focus on $k=1$. 
The tangent vector at a point $x = u{\bf b}_{-\pi/4} + f^1(u){\bf b}_{\pi/4}$ of region $1$, corresponding to an angle $\theta = \theta(x)\in[0,2\pi]$, reads:
\begin{equation}
{\bf T}(\theta) 
= 
\cos(\theta){\bf b}_0+\sin(\theta){\bf b}_{\pi/2} 
= 
\big[1+(\partial_{u} f^1)^2\big]^{-1/2}\big({\bf b}_{-\pi/4} + \partial_{u} f^1 {\bf b}_{\pi/4}\big)
.
\label{eq_vect_tangent_et_dans_R_k}
\end{equation}
From~\eqref{eq_vect_tangent_et_dans_R_k} one can check, as done in Section 3.3 of~\cite{Lacoin2014a}, 
that the heat equation for $f^1$ corresponds to anisotropic motion by curvature~\eqref{eq_def_MCM} for the graph of $f^1$. 
On the other hand,~\eqref{eq_lien_rho_f} and~\eqref{eq_vect_tangent_et_dans_R_k} yield for $\rho^1(1-\rho^1)$:
\begin{align}
\rho^1(1-\rho^1) = \frac{1}{4}\big(1-(\partial_u f^1)^2\big) =\frac{1}{\sqrt{2\big[1+(\partial_u f^1)^2\big]}}\mu(\theta),
\end{align}
with $\mu$ the mobility coefficient obtained by Spohn \cite{Spohn1993}:
\begin{equation}
\mu(\vartheta) 
:= 
\frac{|\sin(2\vartheta)|}{2(|\sin(\vartheta)|+|\cos(\vartheta)|)} 
= 
\frac{|{\bf T}(\vartheta)\cdot {\bf b}_0| |{\bf T}(\vartheta)\cdot {\bf b}_{\pi/2}|}{\|{\bf T}(\vartheta)\|_1},
\qquad \vartheta\in[0,2\pi]
. 
\label{eq_def_mobility}
\end{equation}
Using the relation $ds = [1+(\partial_u f^1)^2]^{1/2}\, du$ between $x$ and the arclength coordinate $s$ and generalising the above discussion to the other three regions,~\eqref{eq_rate_function_SSEP} yields for the conjectured rate function $I^{\text{heur}}_\beta(\cdot|\gamma^{\mathrm{ref}})$ of the contour dynamics:
\begin{align}
I^{\text{heur}}_{\beta}\big(\gamma_\cdot|\gamma^{\mathrm{ref}}\big) 
= 
\sum_{k=1}^k I_{SSEP,k}\big((\rho^{k}_t)_{t\leq T}|(\rho^0)^k\big) 
=
\int_0^T \int_{\gamma_t}\frac{\big(v-ak)^2}{2\mu}\, ds\, dt
. 
\end{align}
This is indeed the rate function of Theorem~\ref{theo_large_dev} for trajectories satisfying the boundary conditions~\eqref{eq_bc_heuristics} (compare with the $\beta=\infty$ case in~\eqref{eq_rate_function_beta_infty}, where the formula is the same, but for smooth trajectories rather than those satisfying~\eqref{eq_bc_heuristics}). 

The analogy~\eqref{eq_rate_function_SSEP} with the SSEP thus gives the correct rate function at a formal level.  
To establish Theorem~\ref{theo_large_dev}, 
we will have to look at the contour dynamics both at and away from the poles simultaneously. 
The dynamics at the poles modifies the size of each of the SSEP. 
This makes it difficult to directly use the analogy with the SSEP in the proofs and a more global approach is necessary.
However, this analogy is used as a guideline throughout the article.
\subsection{Outline of the proof of large deviations}\label{sec_outline_proof}
The proof of Theorem~\ref{theo_large_dev} is structured as follows.
\begin{itemize}
	\item Before looking at rare events specifically, an understanding of the dynamics at the poles is required. This is perhaps the most difficult part of the paper and is the object of Section~\ref{app_behaviour_pole}. 
	In particular, we show there that poles behave like reservoirs in the sense of Proposition~\ref{prop_value_slope_at_poles}. 
	\item The proof of large deviations starts in Section~\ref{sec_relevant_martingales} where we compute the Radon-Nikodym derivative, 
	on a bounded time interval, 
	between the original dynamics and the dynamics tilted by a bias $H\in\C$ introduced in~\eqref{eq_def_jump_rates_H} 
	and express it in terms of the functional $J^\beta_H$ of~\eqref{eq_def_J_H}.  
	To avoid pathological issues with the contour dynamics such as non-locality, the computation is carried out for trajectories with values in the effective state space $\e$ of Definition~\ref{def_effective_state_space}.

	The computation at the microscopic level is inspired by the link with the SSEP as highlighted in Figure~\ref{fig_Ising_ssep_no_bords}. 
	This link is useful to perform discrete integration by parts and replacement lemma-type estimates. 
	
	The resulting expressions are not easily interpreted as line integrals involving tangent vectors as in $J^\beta_H$. This interpretation is carried out in a second time, 
	using similar ideas to what was done at the macroscopic level to go from SSEP to curves in Section~\ref{sec_heuristics}. 
	\item Section~\ref{sec_large_dev_upper_bound} contains upper bound large deviations. The proof technique is standard and consists in estimating the cost of tilting the dynamics by a bias $H\in\C$.

	The difficulty comes from the need to control the poles. The poles in particular prevent the functional $J^\beta_H$, from which the rate function is built, from having nice continuity properties, even for trajectories taking values in the effective state space $\e$ of Definition~\ref{def_effective_state_space}. Continuity is recovered by proving that trajectories must have kinks similarly to Proposition~\ref{prop_value_slope_at_poles}. 
	In fact a stronger version of the statement of Proposition~\ref{prop_value_slope_at_poles} is needed, with the corresponding proof carried out in Appendix~\ref{appen_tightness}.
	\item Section~\ref{sec_large_dev_lower_bound} contains the lower bound, which amounts to hydrodynamic limits for the tilted processes, i.e. Proposition~\ref{prop_limite_hydro}. As a first step, we need to make sure that the (tilted) contour dynamics takes a diffusive time to exit $\e$ as stated in Proposition~\ref{prop_short_time_existence_of_weak_solution}. The hydrodynamic limit results are then obtained in two steps: first in short time using the stability result of Proposition~\ref{prop_short_time_existence_of_weak_solution}. Secondly, by extending the hydrodynamic limit to longer times through an iteration procedure. 
\end{itemize}
\section{Change of measures}\label{sec_relevant_martingales}
\subsection{Motivations}
To investigate rare events, we consider tilted dynamics as in Chapter 10 of \cite{Kipnis1999}. 
Fix a time $T>0$ throughout the rest of Section~\ref{sec_relevant_martingales} and introduce a magnetic field $H\in\C$ ($\C$ is defined in~\eqref{eq_def_ensemble_test_functions}). 
When restricted to trajectories on $[0,T]$, 
$\Prob^{N}_{\beta,H}$ is absolutely continuous with respect to $\Prob^{N}_{\beta}$ (and vice versa). 
Let $D^{N,T}_{\beta,H} = \mathrm{d}\Prob^N_{\beta,H}/\mathrm{d}\Prob^N_{\beta}|_T$ denote their Radon-Nikodym derivative on $[0,T]^{\Omega^N_{\text{mic}}}$, 
so that for any measurable set $X\subset [0,T]^{\Omega^N_{\text{mic}}}$:
\begin{align}
\mathbb{P}^N_{\beta}(\gamma^N_\cdot\in X) 
= 
\E^N_{\beta}[{\bf 1}_{\gamma^N_\cdot\in X}] = \E^N_{\beta,H}\big[ (D^{N,T}_{\beta,H})^{-1}{\bf 1}_{\gamma^N_\cdot\in X}\big]
.
\end{align}
In the following, the dependence on $T$ will be clear from the context and we write $D^N_{\beta,H}$ for $D^{N,T}_{\beta,H}$. 
It acts on a trajectory $\gamma^N_\cdot=(\gamma^N_t)_{t\leq T}\subset \Omega^N_{\text{mic}}$, 
delimiting droplets $(\Gamma^N_t)_{t\leq T}$, 
according to 
(see Appendix A.7 in \cite{Kipnis1999}):
\begin{align}
N^{-1}\log D^N_{\beta,H}(\gamma^N_\cdot) 
&= 
\big<\Gamma^N_{T},H_{T}\big> - \big<\Gamma^N_0,H_0\big> 
\nonumber\\
&\quad - \int_0^{T} e^{-N\langle\Gamma^N_t,H_t\rangle}\big(\partial_t + N^2\lcal_{\beta}\big)e^{N\langle\Gamma^N_t,H_t\rangle}\, dt
.
\label{eq_def_generale_der_radon_niko}
\end{align}
In~\eqref{eq_def_generale_der_radon_niko}, recall that, for a domain $\Gamma^N$ with boundary $\gamma^N\in \Omega_{\text{mic}}^N$ and a bounded $J : \R^2\rightarrow \R$, 
$\big<\Gamma^N,J\big>$ stands for $\int_{\Gamma^N}J(u,v)\, du\, dv$. 
The rest of Section~\ref{sec_relevant_martingales} is devoted to the computation  of $N^2e^{-N\langle\Gamma^N_t,H_t\rangle}\lcal_{\beta}e^{N\langle\Gamma^N_t,H_t\rangle}$ for $t\leq T$.
\subsection{Action of the generator}\label{sec_action_gen_sur_volume}
Take an interface $\gamma^N\in\Omega^N_{\text{mic}}\cap \e$. As usual let $\Gamma^N$ denote the associated droplet.  
In view of the form~\eqref{eq_def_rate_functions} of the rate function, 
the quantity $N^2e^{-N\langle\Gamma^N,H_t\rangle}\lcal_{\beta}e^{N\langle\Gamma^N,H_t\rangle}$ will be expressed as line integrals on $\gamma^N$, 
as well as boundary terms involving the poles. 
We obtain such an expression in two steps. 
The first step relies on microscopic computations and replacement of local quantities by local averages, 
guided by the link of Section~\ref{sec_heuristics} with the SSEP. 
The second step is the interpretation of the resulting quantities in terms of line integrals. 
We first state a result involving discrete sums on vertices of a curve (Proposition~\ref{prop_action_gen_micro}). 
The counterpart in terms of line integrals, Proposition~\ref{prop_informal_stuff_to_prove_action_gen}, is presented and proven later.

To state Proposition~\ref{prop_action_gen_micro}, let us fix some notations. 
For $N\in\N_{\geq 1}$, $x\in V(\gamma^N)$ and $\epsilon >0$, 
the local density of vertical edges $\xi_x^{\epsilon N}$  around $x$ is defined as:
\begin{equation}
\xi_x^{\epsilon N} = \frac{1}{2\epsilon N +1}\sum_{y\in B_1(x,\epsilon)\cap V(\gamma^N)}\xi_y.\label{eq_def_xi_x_epsilon_N}
\end{equation}
The ball $B_1(x,\epsilon)$ is taken with respect to the $1$-norm $\|\cdot\|_1$ (recall~\eqref{eq_def_normes_1_2}). 
We assume that $\epsilon N$ is an integer for simplicity. 
In our case, it will be convenient to write $\xi_x^{\epsilon N}$ as a function of the tangent vector at $x$. Recall that we always enumerate elements of $V(\gamma^N)$ \emph{clockwise} and that ${\bf e}^+_x$ is the microscopic tangent vector with norm $1/N$ given below~\eqref{eq_def_vertices}. The direction of ${\bf e}^+_x$ is fixed by the region $x$ belongs to, 
see Figure~\ref{fig_flip_x_epsilon}. 
For instance, if $x$ belongs to the first region and $y$ to the second:
\begin{align}
N{\bf e}^+_x 
= 
(1-\xi_x){\bf b}_0 - \xi_x {\bf b}_{\pi/2},\qquad
N{\bf e}^+_y = -(1-\xi_y){\bf b}_0 - \xi_y {\bf b}_{\pi/2}
.
\end{align}
The following definition will be used below to keep track of the different signs depending on the region.
\begin{defi}\label{def_m(gamma)_gamma(eta)}
For $\gamma^N\in \Omega^N_{\text{mic}}$, 
recall that ${\bf b}_\theta$ ($\theta\in[0,2\pi]$) is defined in~\eqref{eq_def_b_theta}. 
Define then a vector ${\bf m}(\gamma^N) : \gamma^N\setminus \cup_k P_k(\gamma^N) \rightarrow \R^2$ to be constant on each region, with:
\begin{equation}
\forall x\in \gamma^N\setminus \cup_k P_k(\gamma^N),\qquad {\bf m}(x):= {\bf m}(\gamma^N,x) = \begin{cases}
(-1,-1)\quad &\text{if } x\text{ is in region }1,\\
(-1,1)\quad &\text{if } x\text{ is in region }2,\\
(1,1)\quad &\text{if } x\text{ is in region }3,\\
(1,-1)\quad &\text{if } x\text{ is in region }4.
\end{cases}\label{eq_def_m}
\end{equation}
\end{defi}
If $x$ is at 1-distance at least $\epsilon$ to the poles, then all vertices in $B_1(x,\epsilon)$ are in the same region. Define then the averaged tangent vector ${\bf t}^{\epsilon N}_x$ on the ball $B_1(x,\epsilon)$:
\begin{equation}
{\bf t}^{\epsilon N}_x 
= 
\frac{N}{2\epsilon N +1}\sum_{y\in B_1(x,\epsilon)\cap V(\gamma^N)}{\bf e}^ +_y \in \Big\{ \zeta_1 (1-\xi^{\epsilon N}_x){\bf b}_0 + \zeta_2\xi^{\epsilon N}_x{\bf b}_{\pi/2} : \zeta_1,\zeta_2\in\{-,+\}\Big\}
.
\label{eq_def_t_epsilonN_x}
\end{equation}
The signs in~\eqref{eq_def_t_epsilonN_x} only depend on the region of $\gamma^N$ that $x$ belongs to. 
For instance, if $B_1(x,\epsilon)\cap \gamma^N$ is included in the first region, 
\begin{align}
N{\bf e}^+_x 
= 
(1-\xi_x){\bf b}_0 - \xi_x {\bf b}_{\pi/2} \ \Rightarrow\ {\bf t}^{\epsilon N}_x 
= 
(1-\xi^{\epsilon N}_x){\bf b}_0 -\xi^{\epsilon N}_x {\bf b}_{\pi/2}
.
\end{align}
We stress the fact that ${\bf t}^{\epsilon N}_x$ is a unit vector in $1$-norm, but not in $2$-norm: $\|{\bf t}^{\epsilon N}_x\|_1 = 1 \neq \|{\bf t}^{\epsilon N}_x\|_2$. This has important consequences later on when expressing discrete sums as line integrals, 
see Section~\ref{sec_discrete_sums_to_line_integrals}. It will be useful to introduce the $2$-norm and $2$-normalised tangent vector:
\begin{equation}
\forall x\in V(\gamma^N),\qquad \mathsf{v}^{\epsilon N}_x := \|{\bf t}^{\epsilon N}_x\|_2,\qquad {\bf T}^{\epsilon N}_x = {\bf t}^{\epsilon N}_x/ \mathsf{v}^{\epsilon N}_x.\label{eq_def_v_epsilon_N_T_epsilon_N}
\end{equation}
As $\|{\bf t}^{\epsilon N}\|_1 = 1$, we get:
\begin{equation}
\mathsf{v}^{\epsilon N}_x = \|{\bf t}^{\epsilon N}_x\|_2 = \big(\|{\bf T}^{\epsilon N}_x\|_1\big)^{-1}.\label{eq_useful_relation_v_micro}
\end{equation}
We may now state Proposition~\ref{prop_action_gen_micro}.
\begin{prop}\label{prop_action_gen_micro}
Fix a time $T>0$ and $\beta>\log 2$. For any $\delta,\epsilon>0$ and any trajectory $(\gamma^N_t)_{t\leq T}\in E([0,T],\e)$ of microscopic curves (the set $E([0,T],\e)$ is defined in~\eqref{eq_def_E_0_T_Omega}), one has:
\begin{align}
\frac{1}{N}\int_0^{T}&N^2e^{-N\langle\Gamma^N_t,H_t\rangle}\lcal_{\beta}e^{N\langle\Gamma^N_t,H_t\rangle}\, dt \nonumber\\
&= 
-\bigg(\frac{1}{4}-\frac{e^{-\beta}}{2}\bigg)\int_0^{T}\sum_{k= 1}^4 \big[H(t,L_k(\gamma^N_t)) +H(t,R_k(\gamma^N_t))\big]\, dt 
\nonumber\\
&\quad 
+\frac{1}{4N}\int_0^{T} dt\sum_{x\in V^{\epsilon}(\gamma^N_t)}(\mathsf{v}^{\epsilon N}_x)^2 \big[{\bf T}^{\epsilon N}_x \cdot {\bf m}(x)\big] {\bf T}^{\epsilon N}_x\cdot\nabla H(t,x)\, dt 
\nonumber\\ 
&\quad 
+\frac{1}{2N}\int_0^{T}dt\sum_{x\in V^{\epsilon}(\gamma^N_t)}(\mathsf{v}^{\epsilon N}_x)^2|{\bf T}^{\epsilon N}_x\cdot {\bf b}_0|| {\bf T}^{\epsilon N}_x\cdot {\bf b}_{\pi/2}|H(t,x)^2 + \int_0^{T}\tilde \omega(H_t,\delta,\epsilon,\gamma^N_t)\, dt
.
\label{eq_version_microscopique_action_gen_sur_volume}
\end{align}
The vector ${\bf T}^{\epsilon N}_x$ and normalisation $\mathsf{v}^{\epsilon N}_x$ are defined in~\eqref{eq_def_v_epsilon_N_T_epsilon_N} 
and ${\bf m}(x) = (\pm 1,\pm 1)$ is the sign vector of Definition~\ref{def_m(gamma)_gamma(eta)}. 
For $\gamma^N\in \Omega^N_{\text{mic}}\cap \e$, $V^{\epsilon}(\gamma^N)\subset V(\gamma^N)$ is the subset of vertices at $1$-distance at least $\epsilon$ from the poles. 

The quantity $\tilde \omega(H_t,\delta,\epsilon,\cdot)$ is an error term controlled as follows: there is $C(H)>0$ and a set $\tilde Z=\tilde Z(\beta,H,\delta,\epsilon)\subset (\Omega^N_{\text{mic}})^{[0,T]}$ such that, for trajectories $\gamma^N_\cdot \in \tilde Z\cap E([0,T],\e)$: 
\begin{align}
\Big|\int_0^{T}\tilde \omega(H_t,\delta,\epsilon,\gamma^N_t)\, dt\Big|
\leq 
2\delta + C(H)\bigg(\epsilon T+ \frac{T}{N} + \frac{1}{N^2}\int_0^{T}|\gamma^N_t|\, dt\bigg)
.
\end{align}
Moreover, for each $A>0$, the following super-exponential estimate holds:
\begin{align}
\lim_{\epsilon\rightarrow 0}\limsup_{N\rightarrow\infty}
\frac{1}{N}\log \Prob^N_{\beta,H}\Big(\gamma^N_\cdot\in \tilde Z^c\cap E([0,T],\e)\cap \Big\{\int_0^T|\gamma^N_t|\, dt\leq AT\Big\}\Big) 
= 
-\infty
.
\end{align}
\end{prop}
The proof of Proposition~\ref{prop_action_gen_micro} 
(and the line integral version, Proposition~\ref{prop_informal_stuff_to_prove_action_gen}) takes up the rest of this section. 
It is obtained as a by-product of the study of the dynamics at (Section~\ref{sec_Bulk_terms}) and away from the poles (Section~\ref{sec_pole_term}). 
To lighten notation, we compute $N^2e^{-N\langle\Gamma^N,H_t\rangle}\lcal_{\beta}e^{N\langle\Gamma^N,H_t\rangle}$ at a time $t\in[0,T]$, fixed throughout the section, 
with $\Gamma^N$ the droplet associated with a curve $\gamma^N\in\Omega^N_{\text{mic}}$. 
We sometimes also omit the explicit dependence of $P_k,R_k,L_k$ ($1\leq k\leq 4$) on $\gamma^N$.\\

In the following, we only compute the action of the generator on curves that additionally belong to the effective state space $\e$. 
In that way, the jump rates of the contour dynamics are local
(see the last point of Remark~\ref{rmk_def_contour_dyn}). 
Moreover, for $\gamma^N\in\Omega^N_{\text{mic}}\cap \e$, 
each $(\gamma^N)^{-,k},1\leq k \leq 4$ is still in the state space $\Omega^N_{\text{mic}}$. 
Recalling the definition~\eqref{eq_def_generateur_H_is_0} of the generator of the contour dynamics, one can then write:
\begin{equation}
N^2e^{-N\langle\Gamma^N,H_t\rangle}\lcal_{\beta}e^{N\langle\Gamma^N,H_t\rangle} 
= 
\mathcal B_t(\gamma^N) + \mathcal P_t(\gamma^N)
.
\end{equation}
The bulk term $\mathcal B_t$ contains all updates affecting a single block, corresponding to the simple exclusion moves as discussed around~\eqref{eq_def_c_x_gamma}. 
It is convenient in the computations to also include, 
in $\mathcal B_t$, 
fictitious moves that delete just a single block in a pole containing exactly two blocks, so that:
\begin{align}
\mathcal B_t(\gamma^N) &:= \frac{N^2}{2}\sum_{x\in V(\gamma^N)}\big[\xi_{x+{\bf e}^-_x}(1-\xi_x) + \xi_x(1-\xi_{x+{\bf e}^{-}_x})\big]\Big[e^{N\langle(\Gamma^N)^x,H_t\rangle-N\langle\Gamma^N,H_t\rangle}-1\Big].\label{eq_def_bulk_term}
\end{align}
These fictitious moves (the last line of~\eqref{eq_def_pole_term} below) are  
not allowed in the contour dynamics, 
thus their contribution is subtracted in the term $\mathcal P_t$ which otherwise encompasses all contributions from the pole dynamics.  
Recalling that $p_k$ is the number of blocks in pole $k$ ($1\leq k\leq 4$), $\mathcal P_k$ reads:
\begin{align}
\mathcal P_t(\gamma^N) 
&:= 
N^2\sum_{k=1}^4\sum_{\substack{x\in P_k(\gamma^N)\cap V(\gamma^N)\\x+{\bf e}^{\pm}_x\in P_k(\gamma^N)}} 
\Big\{{\bf 1}_{p_k(\gamma^N)=2}\Big[e^{N\langle(\Gamma^N)^{-,k},H_t\rangle - N\langle\Gamma^N,H_t\rangle}-1\Big] \nonumber\\
&\hspace{5cm}
+e^{-2\beta}\Big[e^{N\langle(\Gamma^N)^{+,x},H_t\rangle-N\langle\Gamma^N,H_t\rangle}-1\Big]\Big\}
\nonumber\\
&\quad 
- \frac{N^2}{2}\sum_{k=1 }^4{\bf 1}_{p_k(\gamma^N)=2}\sum_{x\in \{R_k(\gamma^N),L_k(\gamma^N)\}} \Big[e^{N\langle(\Gamma^N)^x,H_t\rangle-N\langle\Gamma^N,H_t\rangle}-1\Big]
.
\label{eq_def_pole_term}
\end{align}
\subsubsection{Estimate of the pole terms}\label{sec_pole_term}
In this section, we estimate the pole term $\mathcal P_t$.
\begin{lemm}\label{lemm_pole_terms}
For each $\beta>\log 2$ and $\delta>0$, one has:
\begin{align}
\frac{1}{N}\int_0^{T} \mathcal P_t(\gamma^N_t)\, dt 
= 
\frac{e^{-\beta}}{2}\int_0^{T}\sum_{k=1}^4 \big[H(t,R_k(\gamma^N_t)) + H(t,L_k(\gamma^N_t))\big]\,dt 
+ \int_0^{T}\omega_P(H_t,\delta,\gamma^N_t)\, dt
.
\end{align}
The term $\omega_P$ is an error term, estimated as follows: there is a constant $C(H)>0$ and a set $Z_P=Z_P(\delta)$ of trajectories such that, for trajectories in $Z_P\cap E([0,T],\e)$: 
\begin{equation}
\Big|\int_0^{T}\omega_P(H_t,\delta,\gamma^N_t)\, dt\Big| \leq 2\delta + \frac{C(H)T}{N}.\label{eq_error_term_lemm_pole_terms}
\end{equation}
Moreover, the following super-exponential estimate holds:
\begin{align}
\lim_{N\rightarrow\infty}
\frac{1}{N}\log\Prob^N_{\beta,H}\big(\gamma^N_\cdot\in (Z_P)^c\cap E([0,T],\e)\big)
=
-\infty
.
\end{align}
\end{lemm}
\begin{proof}
Fix a time $t\in[0,T]$ and consider $\gamma^N\in \Omega^N_{\text{mic}}\cap \e$ as before. To estimate $\mathcal P_t$, let us first look at the difference $\big<(\Gamma^N)^{+,x},H_t\big>-\big<\Gamma^N,H_t\big>$ for one of the vertices $x$ appearing in the sum in the first line of~\eqref{eq_def_pole_term}. For concreteness, consider e.g. the north pole. A regrowth move $\Gamma^N\rightarrow(\Gamma^N)^{+,x}$ then amounts to adding the two blocks with centre $x + N^{-1}({\bf b}_{\pi/2}\pm {\bf b}_0)/2$ (recall Figure~\ref{fig_image_possible_jumps_beta_dynamics}), so that:
\begin{align}
\big<(\Gamma^N)^{+,x},H_t\big> -\big<\Gamma^N,H_t\big>= \int_{x + \frac{1}{N}[-1,1]\times [0,1]}H_t(z)\, dz = \frac{2}{N^2}H_t(x) + O(N^{-3}),
\end{align}
where we used the smoothness of $H$ to obtain the second equality. A similar estimate holds for the move $\gamma^N\rightarrow(\gamma^N)^{-,1}$ through which blocks in the north pole of $\gamma^N$ are deleted (recall the notation~\eqref{eq_def_gamma_-}); as well as for the other poles. As a result, the quantity $\mathcal P_t(\gamma^N)$ defined in~\eqref{eq_def_pole_term} reads:
\begin{align}
\frac{1}{N}\mathcal P_t(\gamma^N) &= \sum_{k=1}^4 \sum_{\substack{x\in P_k(\gamma^N)\cap V(\gamma^N) \\ x+{\bf e}^\pm_x\in P_k(\gamma^N)}} 2\big(e^{-2\beta} -{\bf 1}_{p_k(\gamma^N)=2}\big)H_t(x)\nonumber \\
&\quad+ \frac{1}{2}\sum_{k=1}^4{\bf 1}_{p_k(\gamma^N)=2}\sum_{x\in \{R_k(\gamma^N),L_k(\gamma^N)\}}H_t(x)+ \eta_P(t,\gamma^N),\label{eq_pole_term_0}
\end{align}
where the first term in the second line corresponds to the fictitious updates and $\eta_P(t,\gamma^N)$ satisfies:
\begin{align}
\big|\eta_P(t,\gamma^N)\big| \leq 2\|H\|^2_\infty \sum_{k=1}^4\frac{p_k(\gamma^N)}{N}.
\end{align}
To prove the claim of Lemma~\ref{lemm_pole_terms}, we need to estimate the time average of the number of blocks $p_k$ in pole $k$ ($1\leq k\leq 4$), of ${\bf 1}_{p_k=2}$ and of their difference. This is done in the following lemmas, the proof of which are postponed to Section~\ref{app_behaviour_pole}. The first lemma states that the pole contains a number of block that scales with $\beta$, but is independent of $N$, with large probability.
\begin{lemm}\label{lemm_size_pole_ds_calcul_action_gen}
For each pole $k\in\{1,...,4\}$,
\begin{align}
\lim_{N\rightarrow\infty}
\frac{1}{N}\log\Prob^N_{\beta,H} \bigg(&(\gamma^N_t)_{t\leq T}\in E([0,T],\e);
\nonumber\\
&\qquad\frac{1}{T}\int_0^{T} e^{-2\beta}\big(p_k(\gamma^N_t)-1\big) \,dt\geq 2\bigg)  
=
-\infty
.
\end{align}
\end{lemm}
The next lemma estimates the difference between growth or deletion of two blocks.
\begin{lemm}\label{lemm_local_eq_sec_martingales}
Let $G\in C_c(\R_+\times\R^2)$ be Lipschitz in space, uniformly in time. Let $W^{G_t}$ be defined, for $t\geq 0$ and $\gamma^N\in\Omega^N_{\text{mic}}$, by:
\begin{equation}
W^{G_t}(\gamma^N) = \sum_{k=1}^4\sum_{\substack{x\in P_k(\gamma^N)\cap V(\gamma^N) \\ x+{\bf e}^\pm_x\in P_k(\gamma^N)}} \big[{\bf 1}_{p_k(\gamma^N)=2}- e^{-2\beta}\big]G_t(x).
\end{equation}
Then:
\begin{align}
\forall\delta>0,\qquad 
\lim_{N\rightarrow\infty} 
\frac{1}{N}\log\Prob^N_{\beta,H} \biggr( (\gamma^N_t)_{t\leq T}\in E([0,T],\e); 
\bigg|\int_0^{T} W^{G_t}(\gamma^N_t) \,dt\bigg|>\delta\bigg)  
=
-\infty
.
\end{align}
\end{lemm}
It remains to compute the time integral of the ${\bf 1}_{p_k=2}$ terms ($1\leq k\leq 4$). Remarkably, this quantity is fixed by the dynamics in terms of $\beta$, as stated in the next lemma. The proof of this lemma, in Section~\ref{subsec_slope_around_the_poles}, is the main difficulty of the paper at the microscopic level.
\begin{lemm}
For pole $k\in\{1,...,4\}$ and each $\delta>0$:
\begin{align}
\lim_{N\rightarrow\infty}\frac{1}{N}\log \Prob^N_{\beta,H}\bigg(&(\gamma^N_t)_{t\leq T}\in E([0,T],\e);
\nonumber\\
&\qquad \bigg|\int_0^ {T} H(t,L_k(\gamma^N_t))\big({\bf 1}_{p_k(\gamma^N_t)=2}- e^ {-\beta}\big)\, dt\bigg|> \delta\bigg) 
= 
-\infty
.
\end{align}
\end{lemm}
Let us conclude the proof of Lemma~\ref{lemm_pole_terms}, using the last three lemmas to define the set $Z_P$, which controls the error term $\omega_P$ of Lemma~\ref{lemm_pole_terms}. 
Let $B^N_{P}(\beta)$ denote the set of trajectories with poles containing less than $2e^ {2\beta}$ blocks:
\begin{align}
B^N_{P}(\beta) := \bigcap_{k=1}^4&\Big\{ (\gamma^N_t)_{t\in[0,T]}\subset \Omega^N_{\text{mic}}: \frac{1}{T}\int_0^{T} \big(p_k(\gamma^N_t) - 1\big)e^{-2\beta}\, dt \leq 2\Big\}.\label{eq_def_B_p}
\end{align} 
On this set, the term $\int_0^ {T}dt\sum_k p_k(\gamma^N_t)/N$ is bounded by $4(2e^{2\beta}+1)T/N$ and therefore negligible. 
Define then $Z_P = Z_P(\beta,\delta)$ as:
\begin{align}
Z_P = B^ N_{P}(\beta)&\cap \Big\{\Big|\int_0^{T} W^{H_t}\, dt\Big|\leq \delta \Big\}\nonumber\\
&\cap \Big\{\sum_{k=1}^4\Big|\int_0^{T} H(t,L_k(\gamma^N_t))\big({\bf 1}_{p_k(\gamma^N_t)=2}- e^ {-\beta}\big)\,dt\Big|\leq \delta\Big\}.\label{eq_def_Z_p}
\end{align}
From~\eqref{eq_pole_term_0}, for a trajectory $(\gamma^N_t)_{t\leq T}\in Z_P\cap E([0,T],\e)$ of microscopic interfaces, we find:
\begin{align}
\frac{1}{N}\int_0^{T}\mathcal P_t(\gamma^N_t)\, dt = \frac{e^{-\beta}}{2}\int_0^{T} \big[H(t,R_k(\gamma^N_t)) + H(t,L_k(\gamma^N_t))\big]dt + \int_0^{T}\omega_P(H_t,\delta,\gamma^N_t)\, dt,
\end{align}
with $\omega_P(H_t,\delta,\cdot)$ an error term satisfying~\eqref{eq_error_term_lemm_pole_terms}. Moreover, the last three lemmas give the following super-exponential estimate:
\begin{align}
\lim_{N\rightarrow\infty}
\frac{1}{N}\log\Prob^N_\beta\Big(\gamma^N_\cdot\in (Z_P)^c\cap E([0,T],\e)\Big)
= 
-\infty
.
\end{align}
This completes the proof of Lemma~\ref{lemm_pole_terms}.
\end{proof}
\subsubsection{Estimate of the bulk terms at the microscopic level}\label{sec_Bulk_terms}
In this section, we compute the bulk term $\mathcal B_t$, introduced in~\eqref{eq_def_bulk_term}, 
expressing it in terms of discrete analogues of quantities that can be defined on a curve at the macroscopic level (such as the tangent vector and arclength derivative). 
To do so, we use the link of Section~\ref{sec_heuristics} between the dynamics in each region and the SSEP to perform discrete integration by parts and obtain a replacement lemma (Lemma~\ref{lemm_Replacement_lemma_sec_mart}). 
One then has to recover expressions that do not explicitly depend on the region any more.
\begin{lemm}\label{lemm_partieSSEP_calcul_micro}
Let $\delta,\epsilon>0$. For each trajectory $(\gamma^N_t)_{t\in[0,T]}\in E([0,T],\e)$ of microscopic curves:
\begin{align}
\frac{1}{N}\int_0^{T}\mathcal B_t(\gamma^N_t)\, dt&
= 
-\frac{1}{4}\int_0^{T}\sum_{k=1}^4 \big[H(t,L_k(\gamma^N_{t})) + H(t,R_k(\gamma^N_{t}))\big]\, dt+ \int_0^{T}\omega_B(H_t,\delta,\gamma^N_t)\, dt
\nonumber\\ 
&\quad+ 
\frac{1}{2N}\int_0^{T}\sum_{x\in V^\epsilon(\gamma^N_{t})} \big|{\bf t}_x^{\epsilon N} \cdot {\bf b}_0\big| \big|{\bf t}_x^{\epsilon N}\cdot {\bf b}_{\pi/2}\big| H(t,x)^2 \, dt
\nonumber\\
&\quad 
+\frac{1}{4N}\int_0^{T} \sum_{x\in V^{\epsilon}(\gamma^N_{t})} \big[{\bf t}_x^{\epsilon N}\cdot {\bf m}(x)\big]{\bf t}_x^{\epsilon N}\cdot \nabla H(t,x) \, dt
.
\label{eq_expression_L_H_tout_remplace}
\end{align}
Recall that ${\bf t}^{\epsilon N}_x$ is defined in~\eqref{eq_def_t_epsilonN_x} and ${\bf m}(x)$ is the sign vector of Definition~\ref{def_m(gamma)_gamma(eta)}. 
For $\tilde \gamma^N\in \Omega^N_{\text{mic}}$, 
the subset of vertices $V^{\epsilon }(\tilde \gamma^N)\subset V(\tilde \gamma^N)$ denotes all points of $V(\tilde \gamma^N)$ at $1$-distance at least $\epsilon$ from the poles of $\tilde \gamma^N$. 

In addition, there is a set $Z_B=Z_B(H,\delta,\epsilon)\subset E([0,T],\Omega^N_{\text{mic}})$ on which the error term $\omega_B$ can be controlled: for some constant $C(H)>0$ and all trajectories $\gamma^N_\cdot$ in $Z_B\cap E([0,T],\e)$,
\begin{align}
\Big|\int_0^{T}\omega_B(H_t,\delta,\epsilon,\gamma^N_t)\, dt\Big| \leq \delta + C(H)\Big(\epsilon T + \frac{1}{N^2} \int_0^{T}|\gamma^N_t|\, dt\Big).
\end{align}
The following super-exponential estimates holds for $\tilde Z_B$: for each $A>0$,
\begin{align}
\lim_{\epsilon\rightarrow 0}\limsup_{N\rightarrow\infty}\frac{1}{N}\log\Prob^N_{\beta,H}\Big(\gamma^N_\cdot\in (Z_B)^c\cap E([0,T],\e)\cap \Big\{\int_0^T|\gamma^N_t|\, dt\leq AT\Big\}\Big) = -\infty.
\end{align}
\end{lemm}
\begin{proof}[Proof of Lemma~\ref{lemm_partieSSEP_calcul_micro}]
As for the pole terms $\mathcal P_t$ in Section~\ref{sec_pole_term}, we work at fixed time $t\in[0,T]$ and fix $\gamma^N \in\Omega^N_{\text{mic}}\cap \e$. The starting point is the expression~\eqref{eq_def_bulk_term} of $\mathcal B_t$. 
Let us first compute the change in $\big<\Gamma^N,H_t\big>$ when $x\in V(\gamma^N)$ is flipped. 
Recall that ${\bf e}^+_x,{\bf e}^-_x$ are the vectors with origin $x$ and norm $1/N$ defined below~\eqref{eq_def_vertices} (see also Figure~\ref{fig_flip_x_epsilon}). 
One can then write:
\begin{align}
\big<(\Gamma^N)^x,H_t\big> - \big<\Gamma^N,H_t\big> &= \epsilon_x(\gamma^N)\int_{[x, x+{\bf e}^-_x]\times [x, x+{\bf e}^+_x]} H_t(z)\,  dz \label{eq_def_epsilon_N_x}\\
&= \frac{\epsilon_x(\gamma^N)}{N^2}\int_{[0,1]^2}H_t\Big(x + u{\bf e}^-_x + v{\bf e}^+_x\Big)\, du\, dv\nonumber.
\end{align}
Above, $\epsilon_x(\gamma^N)\in \{-1,1\}$ is set to $1$ if flipping $x$ means adding one block to $\Gamma^N$ and to $-1$ if it means deleting one (see Figure~\ref{fig_flip_x_epsilon}). 
Let us expand $H_t$ around the point $x$. 
Recall that the vectors ${\bf e}^\pm_x$ have norm $1/N$. 
As a result, e.g. if ${\bf e}^+_x = {\bf b}_{\pi/2}/N$:
\begin{equation}
\partial_{{\bf e}^+_x} H_t(x) 
= 
\frac{1}{N}\partial_{2}H_t(x)
, 
\label{eq_size_derivative_e_plus_x}
\end{equation}
and Equation~\eqref{eq_def_epsilon_N_x} becomes:
\begin{align}
\big<(\Gamma^N)^x,H_t\big> - \big<\Gamma^N,H_t\big>  = \frac{\epsilon_x(\gamma^N)}{N^ 2} \Big(H_t(x) + \frac{1}{2}\big(\partial_{{\bf e}^-_x} + \partial_{{\bf e}^+_x}\big)H_t(x)\Big) + \frac{\eta(H_t)}{N^4},
\end{align}
for an error term $\eta(H_t)$ satisfying $|\eta(H_t)| \leq \|\nabla^2 H_t\|_\infty$. Recalling:
\begin{align}
c_x(\gamma^N) := \frac{1}{2}\big[\xi_{x+{\bf e}^-_x}(1-\xi_{x}) + \xi_x(1-\xi_{x+{\bf e}^-_x})\big],
\end{align}
we find that the bulk term~\eqref{eq_def_bulk_term} can be written as follows:
\begin{figure}
\begin{center}
\includegraphics[width=8cm]{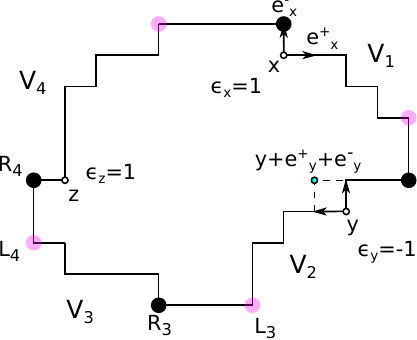} 
\caption{Definition of the $V_k$, $k\in\{1,...4\}$ for a curve $\gamma^N\in \Omega^N_{\text{mic}}$. 
The black dots are the first vertices and the light dots the last vertices of each $V_k$. 
Three points are marked by empty circles, with the corresponding value of $\epsilon_\cdot(\gamma^N)$. 
The block that is deleted if $y$ is flipped is materialised by dashed lines and the two arrows with origin $y$ correspond to ${\bf e}^+_y$ (left arrow) and ${\bf e}^-_y$ (up arrow). \label{fig_flip_x_epsilon}}
\end{center}
\end{figure}
\begin{align}
\frac{1}{N}\mathcal B_t(\gamma^N) &=\frac{1}{2N}\sum_{x\in V(\gamma^N)}c_x(\gamma^N) H_t(x)^2 + \sum_{x\in V(\gamma^N)}c_x(\gamma^N)\epsilon_x(\gamma^N) H_t(x) \nonumber\\
&\quad + \frac{1}{2}\sum_{x\in V(\gamma^N)}c_x(\gamma^N)\epsilon_x(\gamma^N)\big(\partial_{{\bf e}^-_x}+\partial_{{\bf e}^+_x}\big)H_t(x) +\frac{\eta'(H_t)|\gamma^N|}{N^2}, \label{eq_bulk_avant_affinage_pour_forcage_volumique}
\end{align}
where $\eta'(H_t)$ is bounded by a constant depending on $H_t$ and its derivatives. 
The first and third sums above are bounded by $|V(\gamma^N)|/N$,  
which we expect to be bounded with $N$ at each time for typical trajectories under the contour dynamics on $E([0,T],\e)$ (see Proposition~\ref{prop_short_time_existence_of_weak_solution}). 
At first glance however, the second sum in the first line of~\eqref{eq_bulk_avant_affinage_pour_forcage_volumique} appears to be of order $|V(\gamma^N)|$. 
To prove that it is in fact also of order 1 in $N$, we use the link with the SSEP in each region of $\gamma^N$ to perform integration by parts. 
This link is also used to compute the other terms~\eqref{eq_bulk_avant_affinage_pour_forcage_volumique}.\\

To compute~\eqref{eq_bulk_avant_affinage_pour_forcage_volumique} 
we are going to split $V(\gamma^N)$ into four pieces, essentially corresponding to the four regions of $\gamma^N$. 
In each region the mapping with the SSEP of Section~\ref{sec_heuristics} will be used to express $c_x(\gamma^N),\epsilon_x(\gamma^N)$ in terms of the local edge states $\xi_{x+{\bf e}^-_x},\xi_x$. 
Mirroring similar results for the SSEP, these $\xi_\cdot$ will then replaced by local averages thanks to a replacement lemma-type result, Lemma~\ref{lemm_Replacement_lemma_sec_mart}. 
For points in each region at 1-distance at least $\epsilon$ from the poles, 
these averages will be rewritten as components of the microscopic tangent vector ${\bf t}^{\epsilon N}_\cdot$ (defined in~\eqref{eq_def_t_epsilonN_x}), 
which will allow us below to recover a region-independent expression.\\
 
For $1\leq k \leq 4$, consider thus the set $V_k(\gamma^N)\subset V(\gamma^N)$ containing all vertices from $R_k$ to $L_{k+1}$ (comprised), see Figure~\ref{fig_flip_x_epsilon}, with $L_{k+1} := L_1$ if $k=4$.  
Then $V_k$ is included in region $k$ of $\gamma^N$ and:
\begin{align}
V(\gamma^N) = \bigcup_{k=1}^4 V_k(\gamma^N) \cup \bigcup_{k=1}^4 P_k(\gamma^N).
\end{align}
In the following, we often abbreviate $V_k(\gamma^N)$ as $V_k$ and similarly $R_k(\gamma^N),L_k(\gamma^N)$ as $R_k,L_k$. 

As $c_x(\gamma^N) = 0$ for each $x\in \cup_k P_k(\gamma^N)\setminus\{R_k(\gamma^N),L_k(\gamma^N)\}$, all sums in~\eqref{eq_bulk_avant_affinage_pour_forcage_volumique} reduce to sums on $V_k$, $1\leq k \leq 4$. With this splitting along the $V_k$,~\eqref{eq_bulk_avant_affinage_pour_forcage_volumique} becomes:
\begin{equation}
\frac{1}{N}\mathcal B_t(\gamma^N) =\frac{1}{2N}\sum_{x\in V(\gamma^N)} c_x(\gamma^N) H_t(x)^2 + \sum_{k=1}^4 B_t^k+ \sum_{k=1}^4 \tilde B_t^k + \frac{\eta'(H_t)|\gamma^N|}{N^2},\label{eq_bulk_term_as_fct_B_k_B'_k}
\end{equation}
where: 
\begin{equation}
B_t^k := \sum_{x\in V_k}c_x(\gamma^N)\epsilon_x(\gamma^N) H_t(x),\quad \tilde B_t^k := \frac{1}{2}\sum_{x\in V_k}c_x(\gamma^N)\epsilon_x(\gamma^N)\big(\partial_{{\bf e}^-_x}+\partial_{{\bf e}^+_x}\big)H_t(x).\label{eq_def_B_k_B'_k}
\end{equation}
For future reference note that, 
inside each $V_k$, the quantity $c_x(\gamma^N)\epsilon_x(\gamma^N)$ can be expressed in terms of the edge state $\xi$:
\begin{equation}
2c_x(\gamma^N)\epsilon_x(\gamma^N) = \begin{cases}
\xi_{x+{\bf e}^-_x}-\xi_x \quad \text{if }x\in V_1\cup V_3,\\
\xi_x - \xi_{x+{\bf e}^-_x} \quad \text{if }x\in V_2\cup V_4.
\end{cases}
\label{eq_value_c_x(gamma)epsilon_x(gamma)}
\end{equation}
Let us first treat the sums $B^k_t,\tilde B^k_t$ ($1\leq k \leq 4$) that involve $\epsilon_x(\gamma^N)$.\\

\noindent \textbf{1) $B_t^k$ terms.}
Using Equation~\eqref{eq_value_c_x(gamma)epsilon_x(gamma)}, one has e.g. for $k=1$:
\begin{align}
B_t^1 &= \sum_{x\in V_1(\gamma^N)}c_x(\gamma^N)\epsilon_x(\gamma^N) H_t(x) = \frac{1}{2}\sum_{x\in V_1(\gamma^N)}H_t(x) (\xi_{x+ {\bf e}^-_x}-\xi_x)\nonumber\\
&= \frac{1}{4}\sum_{x\in V_1(\gamma^N)} H_t(x) \big[\xi_{x+{\bf e}^-_x} - \xi_x + (1-\xi_x) - (1-\xi_{x+{\bf e}^-_x})\big].\label{eq_B_1_term_before_IPP}
\end{align}
The passage from first to second line is nothing more than a symmetrisation of the expression. \\
Recall that $V_1$ contains vertices between $R_1$ and $L_2$ (included). 
By definition, the edge with right extremity $R_1$, corresponding to $[R_1+{\bf e}^-_{R_1},R_1]$, is always horizontal: $1-\xi_{R_1+{\bf e}^-_{R_1}} = 1$. 
Similarly, $\xi_{L_2}=1$ by definition of $L_2$. 
Integrating~\eqref{eq_B_1_term_before_IPP} by parts, some of the boundary term vanish, whence:
\begin{align}
B_t^1 &=-\frac{1}{4}\big(H_t(R_1) + H_t(L_2)\big) \nonumber\\ 
&\hspace{2cm}+ \frac{1}{4}\sum_{x\in V_1\setminus\{R_1,L_2\}}\bigg[\xi_x\Big[H_t(x+{\bf e}^+_x)-H_t(x)\Big]-(1-\xi_x)\Big[H_t(x+{\bf e}^+_x)-H_t(x)\Big]\bigg]\nonumber\\
&= -\frac{1}{4}\big(H_t(R_1) + H_t(L_2)\big) \nonumber\\
&\hspace{2cm}+ \frac{1}{4}\sum_{x\in V_1\setminus\{R_1,L_2\}}\big[\xi_x \partial_{{\bf e}^+_x}H_t(x) -(1-\xi_x)\partial_{{\bf e}^+_x}H_t(x)\big] + \eta^{1}(H_t,\gamma^N),\label{eq_value_B1}
\end{align}
with $\eta_t^{1}(H_t,\gamma^N)$ an error term bounded by:
\begin{equation}
\big|\eta^{1}(H_t,\gamma^N)\big| \leq \frac{\|\nabla^2 H_t\|_\infty|V_1(\gamma^N)|}{4N^2}.\label{eq_erreur_B_k_t}
\end{equation}
Note that, as ${\bf e}^\pm_x$ has norm $1/N$, 
the sum in~\eqref{eq_value_B1} is bounded by $C(H)|V_1|/N$ for some $C(H)>0$. 
This is one factor of $1/N$ smaller than apparent in the expression~\eqref{eq_def_B_k_B'_k} of $B^1_t$ as desired.

The other $V_k$, $2\leq k \leq 4$ are treated similarly, with signs depending on the region due to both~\eqref{eq_value_c_x(gamma)epsilon_x(gamma)} and the fact that $N{\bf e}^+_x$ takes values in $\{\pm {\bf b}_0,\pm {\bf b}_{\pi/2}\}$ as the region varies. 
The point is now, from the expression of each $B_t^k$, to obtain an expression independent of the region of $\gamma^N$. To do so, 
introduce region-dependent signs $\sigma_1,\sigma_2$:
\begin{equation}
\sigma_1 := \begin{cases}
1\quad &\text{if } x\in V_1\cup V_4\\
-1 \quad &\text{if }x\in V_2\cup V_3
\end{cases},\quad \sigma_2 := \begin{cases}
-1 \quad &\text{if }x\in V_1\cup V_2\\
1\quad &\text{if } x\in V_3\cup V_4
\end{cases}.  \label{valeurs_signes_e^+_x}
\end{equation}
The idea behind~\eqref{valeurs_signes_e^+_x} is that $(\sigma_1,\sigma_2)$ is "the direction of the tangent vector to a curve" in each region, in the spirit of Definition~\ref{property_state_space} of $\Omega$. 
For instance, in the first region, the tangent vector can be either ${\bf b}_0$ or $-{\bf b}_{\pi/2}$, 
and $(\sigma_1,\sigma_2) = (1,-1)$. 
In region $2$ where the tangent vector is either $-{\bf b}_{0}$ or $-{\bf b}_{\pi/2}$, $(\sigma_1,\sigma_2) = (-1,-1)$, etc. One can then check that:
\begin{equation}
\forall x\in V(\gamma^N),\qquad 
N{\bf e}^+_x 
= 
(1-\xi_x)\sigma_1{\bf b}_0 + \xi_x\sigma_2{\bf b}_{\pi/2}
.
\label{eq_lien_vect_tangent_sigma}
\end{equation}
Compare with ${\bf m}(x)$ in Definition~\ref{def_m(gamma)_gamma(eta)}, which gives "the direction of the inwards normal" in the region $x$ belongs to:
\begin{equation}
{\bf m}(x) 
= 
-(-\sigma_2,\sigma_1) 
= 
(\sigma_2,-\sigma_1) 
= 
\sigma_2{\bf b}_0 - \sigma_1{\bf b}_{\pi/2}
.
\label{eq_relation_m_and_sigma}
\end{equation}
With Definition~\ref{valeurs_signes_e^+_x} and the error bound~\eqref{eq_erreur_B_k_t}, 
recalling also from~\eqref{eq_size_derivative_e_plus_x} that differentiating along ${\bf e}^+_x$ incurs a factor $1/N$ compared to differentiating with respect to ${\bf b}_0$ or ${\bf b}_{\pi/2}$, 
$B^1_t$ can be written as:
\begin{align}
B^1_t 
&= 
-\frac{1}{4}\big(H_t(R_1) + H_t(L_2)\big) 
\nonumber\\
&\qquad
+ \frac{1}{4N}\sum_{x\in V_1\setminus\{R_1,L_2\}}\Big(-\xi_x \sigma_1\partial_2+(1-\xi_x)\sigma_2\partial_1 \Big)H_t(x) + \eta^{1}(H_t,\gamma^N)
.
\end{align}
Recalling the sign change~\eqref{eq_value_c_x(gamma)epsilon_x(gamma)} between regions, one can check that the expression of the last summand is independent from the region and obtain:
\begin{align}
\sum_{k=1}^4B_t^k &= -\frac{1}{4}\sum_{k=1}^4\big[H_t(R_k(\gamma^N)) + H_t(L_k(\gamma^N))\big]\nonumber \\
&\qquad+ \frac{1}{4N}\sum_{x\in V(\gamma^N)\setminus \cup_k P_k(\gamma^N)}\Big(-\xi_x \sigma_1\partial_2+(1-\xi_x)\sigma_2\partial_1 \Big)H_t(x) + \sum_{k=1}^4\eta^{k}(H_t,\gamma^N). \label{eq_bulk_term_contribution_transitoire}
\end{align}
\noindent \textbf{2) $\tilde B^k_t$ terms (defined in~\eqref{eq_def_B_k_B'_k}).} 
Recall that $\epsilon_x(\gamma^N)$ is defined in~\eqref{eq_def_epsilon_N_x}. 
The key observation is the following: for $x\in V(\gamma^N)$, if $c_x(\gamma^N)\neq 0$, then $\epsilon_x(\gamma^N)(\partial_{{\bf e}^+_x} + \partial_{{\bf e}^-_x})$ is the same whether flipping $x$ corresponds to adding or deleting a block and it only depends on the region of the curve. 
Indeed, recall Definition~\eqref{valeurs_signes_e^+_x} of $\sigma_1,\sigma_2$. 
Using ${\bf e}^-_x = -{\bf e}^+_{x+{\bf e}^-_x}$ (see Figure~\ref{fig_flip_x_epsilon}) and the expression~\eqref{eq_lien_vect_tangent_sigma} of ${\bf e}^+_x$, one has:
\begin{align}
\partial_{{\bf e}^-_x} + \partial_{{\bf e}^+_x} = \big(\xi_{x+{\bf e}^-_x}-\xi_x\big)\Big(\frac{\sigma_1}{N}\partial_1 - \frac{\sigma_2}{N}\partial_2\Big).
\end{align}
Using the expression~\eqref{eq_value_c_x(gamma)epsilon_x(gamma)} for $c_x(\gamma^N)\epsilon_x(\gamma^N)$, elementary manipulations then yield:
\begin{align}
\forall x\in V(\gamma^N),\qquad c_x(\gamma^N)\epsilon_x(\gamma^N)(\partial_{{\bf e}^+_x} + \partial_{{\bf e}^-_x}) 
&= \frac{c_x(\gamma^N)}{N} \Big(-\sigma_2\partial_1 + \sigma_1\partial_2\Big) 
\nonumber\\
&= -\frac{c_x(\gamma^N)}{N} {\bf m}(x)\cdot \nabla,
\end{align}
with ${\bf m}(x)$ the sign vector of Definition~\ref{def_m(gamma)_gamma(eta)}. As a result:
\begin{equation}
\sum_{k=1}^4 \tilde B^k_t(\gamma^N) = -\frac{1}{2N}\sum_{x\in V(\gamma^N)\setminus \cup_k P_k(\gamma^N)}c_x(\gamma^N){\bf m}(x)\cdot \nabla H_t(x).\bigskip \label{eq_value_B_prime_k_recast}
\end{equation}
For $\epsilon>0$, 
Let $V^{\epsilon}(\gamma^N)$ be the subset of $V(\gamma^N)$ made of vertices at $1$-distance at least $\epsilon$ from each pole. 
At this point the bulk term $\mathcal B_t$ can be written as:
\begin{align}
\frac{1}{N}\mathcal B_t(\gamma^N) &= \frac{1}{2N}\sum_{x\in V^{\epsilon}(\gamma^N)}c_x(\gamma^N) H_t(x)^2 -\frac{1}{4}\sum_{k=1}^4\big[H_t(R_k(\gamma^N)) + H_t(L_k(\gamma^N))\big] \nonumber\\
&\qquad + \frac{1}{4N}\sum_{x\in V^{\epsilon}(\gamma^N)}\Big(-\xi_x \sigma_1\partial_2+(1-\xi_x)\sigma_2\partial_1 \Big)H_t(x) \label{eq_bulk_term_interm}\\
&\qquad-\frac{1}{2N}\sum_{x\in V^{\epsilon}(\gamma^N)}c_x(\gamma^N){\bf m}(x)\cdot \nabla H_t(x) + \eta(H_t,\epsilon,\gamma^N),\nonumber
\end{align}
where $\eta(H_t,\epsilon,\gamma^N)$ satisfies, for some constant $C(H)>0$ independent of $\gamma^N$:
\begin{align}
\Big|\eta(H_t,\epsilon,\gamma^N) -\frac{\eta'(H_t)|\gamma^N|}{N^2} - \sum_{k=1}^4 \eta^k(H_t,\gamma^N)\Big|\leq C(H)\epsilon.
\end{align}
To obtain the expression in Lemma~\ref{lemm_partieSSEP_calcul_micro} from~\eqref{eq_bulk_term_as_fct_B_k_B'_k}--\eqref{eq_bulk_term_contribution_transitoire}--\eqref{eq_value_B_prime_k_recast}, 
we will now replace $\xi_\cdot$ and $c_\cdot(\gamma^N)$ by local averages on boxes containing order $\epsilon N$ vertices; then express them in terms of the microscopic tangent vector ${\bf t}^{\epsilon N}_\cdot$ defined in~\eqref{eq_def_t_epsilonN_x}.

We start by replacing $c_\cdot(\gamma^N),\xi_\cdot$ by local averages. 
For $\xi_\cdot$, an integration by parts and the smoothness of $H_t$ yield the existence of an error term $\omega^{\nabla H_t}(\epsilon,\gamma^N)$ such that:
\begin{align}
&\frac{1}{4N}\sum_{x\in V^{\epsilon}(\gamma^N)}\Big(-\xi_x \sigma_1\partial_2+(1-\xi_x)\sigma_2\partial_1 \Big)H_t(x) \nonumber\\
&\hspace{3cm}= \frac{1}{4N}\sum_{x\in V^{\epsilon N}(\gamma^N)}\Big(-\xi^{\epsilon N}_x \sigma_1\partial_2+(1-\xi_x^{\epsilon N})\sigma_2\partial_1 \Big)H_t(x) + \omega^{\nabla H_t}(\epsilon,\gamma^N),
\end{align}
with, for a constant $C(H_t)>0$ involving $\nabla^2 H_t$ but independent of $\gamma^N$:
\begin{align}
|\omega^{\nabla H_t}(\epsilon,\gamma^N)\big|\leq C(H_t)\epsilon.
\end{align}
Replacing $c_\cdot(\gamma^N)$ by a local average is much more involved. It is the content of a so-called replacement lemma, stated below and proven in Appendix~\ref{sec_replacement_lemma}.
\begin{lemm}[Replacement lemma]\label{lemm_Replacement_lemma_sec_mart}
Let $G:\R_+\times\R^2\rightarrow\R$ be bounded, $A,\epsilon>0$ and $1\leq k\leq 4$. Define $W_{\epsilon N}^{G,k}(t,\cdot)$ for $t\geq 0$ by:
\begin{equation}
W_{\epsilon N}^{G,k}(t,\tilde \gamma^N) := \frac{1}{N}\sum_{x\in V_k(\tilde\gamma^N)}G(t,x)\bigg[c_x(\tilde\gamma^N) - \xi^{\epsilon N}_x\big(1-\xi^{\epsilon N}_x\big)\bigg],\quad \tilde\gamma^N\in\Omega^N_{\text{mic}}.
\end{equation}
Then, for each $\delta>0$, the following super-exponential estimate holds:
\begin{align}
\lim_{\epsilon\rightarrow 0}\limsup_{N\rightarrow\infty} 
\frac{1}{N}\log\Prob_{\beta,H}^N\Big((\gamma^N_t)_{t\leq T} \in &E([0,T],\e)\cap\big\{\sup_{t\leq T}|\gamma^N_t|\leq A\big\};
\nonumber\\
&\hspace{1.5cm} \bigg|\int_0^{T} W_{\epsilon N}^{G,k}(t,\gamma^N_t)\, dt\bigg|>\delta\Big) 
= 
-\infty
.
\end{align}
\end{lemm}
\noindent Using Lemma~\ref{lemm_Replacement_lemma_sec_mart}, we are going to conclude the proof of Lemma~\ref{lemm_partieSSEP_calcul_micro}. Define, for each bounded function $G$:
\begin{equation}
B_{G}(\delta,\epsilon) = \bigg\{(\gamma^N_t)_{t\leq T}\subset  \Omega^N_{\text{mic}} : \forall 1\leq k \leq 4,\  \bigg|\int_0^{T} W_{\epsilon N}^{G,k}(t,\gamma^N_t)\, dt \bigg|\leq\delta/4\bigg\}.\label{eq_def_B_i_rplcmt}
\end{equation}
By Lemma~\ref{lemm_Replacement_lemma_sec_mart}, for each $A,\delta>0$, one has:
\begin{align}
\lim_{\epsilon\rightarrow0}\limsup_{N\rightarrow\infty}
\frac{1}{N}\log\Prob^N_{\beta,H}\Big(B_G(\delta,\epsilon)^c\cap E([0,T],\e)\cap\Big\{\sup_{t\leq T}|\gamma^N_t|\leq A\Big\}\Big) 
= 
-\infty
.
\end{align}
Define then the set $Z_B$ appearing in Lemma~\ref{lemm_partieSSEP_calcul_micro} and controlling the error terms:
\begin{equation}
Z_B = Z_B(H,\delta,\epsilon) := B_{\partial_1 H}(\delta,\epsilon)\cap B_{\partial_2 H}(\delta,\epsilon)\cap B_{H^2}(\delta,\epsilon).\label{eq_def_Z_B}
\end{equation}
We may now write the bulk term $\mathcal B_t$ (recall~\eqref{eq_bulk_term_interm}) as:
\begin{align}
\frac{1}{N}\mathcal B_t(\gamma^N) &= \frac{1}{2N}\sum_{x\in V^{\epsilon}(\gamma^N)}\xi_x^{\epsilon N}\big(1-\xi^{\epsilon N}_x\big) H_t(x)^2 -\frac{1}{4}\sum_{k=1}^4\big[H_t(R_k(\gamma^N)) + H_t(L_k(\gamma^N))\big] \nonumber\\
&\qquad + \frac{1}{4N}\sum_{x\in V^{\epsilon}(\gamma^N)}\Big(-\xi^{\epsilon N}_x \sigma_1\partial_2+(1-\xi^{\epsilon N}_x)\sigma_2\partial_1 \Big)H_t(x) \label{eq_bulk_term_interm_2}\\
&\qquad-\frac{1}{2N}\sum_{x\in V^{\epsilon}(\gamma^N)}\xi_x^{\epsilon N}\big(1-\xi^{\epsilon N}_x\big){\bf m}(x)\cdot \nabla H_t(x) + \omega_B(H_t,\delta,\epsilon,\gamma^N),\nonumber
\end{align}
where $\omega_B(H_t,\delta,\epsilon,\gamma^N)$ is defined, with a slight abuse of notation, by: 
\begin{align}
\omega_B(H_t,\delta,\epsilon,\gamma^N) := \eta(H_t,\epsilon,\gamma^N) + \omega^{\nabla H_t}(\epsilon,\gamma^N) + \sum_{k=1}^4 \Big[W_{\epsilon N}^{{\bf m}\cdot\nabla H,k}(t,\gamma^N)+ W_{\epsilon N}^{H^2,k}(t,\gamma^N)\Big].
\end{align}
In particular, for microscopic trajectories $(\gamma^N_t)_{t\leq T}\in Z_B\cap E([0,T],\e)$, there is a constant $C(H)>0$ such that:
\begin{align}
\Big|\int_0^{T}\omega_B(H_t,\delta,\epsilon,\gamma^N_t) \, dt \Big|
\leq 
2\delta + C(H)\Big(\epsilon T + \frac{1}{N^2}\int_0^{T}|\gamma^N_t|\, dt\Big)
.
\end{align}
With the above estimate of $\omega_B(\cdot)$ and the expression~\eqref{eq_bulk_term_interm_2} of $\mathcal B_t$, we see that the proof of Lemma~\ref{lemm_partieSSEP_calcul_micro} now reduces to the third step in the program outlined below~\eqref{eq_bulk_avant_affinage_pour_forcage_volumique}, i.e. the interpretation of $\xi^{\epsilon N}_x$ and ${\bf m}(x)\cdot\nabla$ in terms of components of the microscopic tangent vector ${\bf t}^{\epsilon N}_x$, which we now perform.\\

Recall from~\eqref{eq_def_t_epsilonN_x} the following identity: for $\tilde \gamma^N\in \Omega^N_{\text{mic}}\cap \e$ and $x\in V(\tilde \gamma^N)$,
\begin{equation}
|{\bf t}^{\epsilon N}_x\cdot {\bf b}_{\pi/2}| 
= 
\xi^{\epsilon N}_x = 1-|{\bf t}^{\epsilon N}_x\cdot {\bf b}_0|
.
\label{eq_val_abs_composants_t_epsilon_N}
\end{equation}
As a result:
\begin{align}
\frac{1}{2N}\sum_{x\in V^{\epsilon}(\gamma^N)}\xi_x^{\epsilon N}\big(1-\xi^{\epsilon N}_x\big) H_t(x)^2 
= 
\frac{1}{2N}\sum_{x\in V^{\epsilon}(\gamma^N)} |{\bf t}^{\epsilon N}_x\cdot {\bf b}_0| |{\bf t}^{\epsilon N}_x \cdot {\bf b}_{\pi/2}| H_t(x)^2
.
\label{eq_replacement_HG_terms_by_tangent_vector}
\end{align}
To establish the expression~\eqref{eq_expression_L_H_tout_remplace}, we therefore only need to prove:
\begin{align}
&\frac{1}{4N}\sum_{x\in V^{\epsilon}(\gamma^N)}\Big(-\xi^{\epsilon N}_x \sigma_1\partial_2+(1-\xi^{\epsilon N}_x)\sigma_2\partial_1 \Big)H_t(x) \label{eq_replacement_xi_epsilon_N_by_tangent_vector}\\
&-\frac{1}{2N}\sum_{x\in V^{\epsilon}(\gamma^N)}\xi_x^{\epsilon N}\big(1-\xi^{\epsilon N}_x\big){\bf m}(x)\cdot \nabla H_t(x) = \frac{1}{4N}\sum_{x\in V^{\epsilon}(\gamma^N)}\big[{\bf t}^{\epsilon N}_x\cdot {\bf m}(x)\big] {\bf t}_x^{\epsilon N}\cdot \nabla H_t(x).\nonumber
\end{align}
To do so, we use the following shorthand notations:
\begin{align}
{\bf t}^\epsilon := {\bf t}^{\epsilon N}_x,\qquad {\bf t}^\epsilon_i := {\bf t}^{\epsilon N}_x\cdot {\bf b}_i,\quad i\in\{1,2\}.
\end{align}
Recalling that ${\bf m}(x) = (\sigma_2,-\sigma_1)$, the left-hand side of~\eqref{eq_replacement_xi_epsilon_N_by_tangent_vector} reads:
\begin{align}
\frac{1}{4N}&\sum_{x\in V^{\epsilon}(\gamma^N)}\Big[\big(-|{\bf t}^\epsilon_2|\sigma_1 + 2|{\bf t}^\epsilon_1||{\bf t}^\epsilon_2|\sigma_1\big)\partial_2 + \big(|{\bf t}^\epsilon_1|\sigma_2 - 2|{\bf t}^\epsilon_1||{\bf t}^\epsilon_2|\sigma_2\big)\partial_1\Big]H_t(x)
\nonumber\\
&\qquad 
=\frac{1}{4N}\sum_{x\in V^{\epsilon}(\gamma^N)}\Big[\sigma_1|{\bf t}^\epsilon_2|\big(-1 + 2|{\bf t}^\epsilon_1|\big)\partial_2 + \sigma_2|{\bf t}^\epsilon_1|\big(1 - 2|{\bf t}^\epsilon_2|\sigma_2\big)\partial_1\Big]H_t(x)
\nonumber\\
&\qquad 
= \frac{1}{4N}\sum_{x\in V^{\epsilon}(\gamma^N)}\Big[\sigma_1|{\bf t}^\epsilon_2|\big(|{\bf t}^\epsilon_1|-|{\bf t}^\epsilon_2|\big)\partial_2 + \sigma_2|{\bf t}^\epsilon_1|\big(|{\bf t}^\epsilon_1|-|{\bf t}^\epsilon_2|\big)\partial_1\Big]H_t(x)
.
\label{eq_pour_faire_apparaitre_der_tangente_et_m_0}
\end{align}
To obtain the third line, we used $|{\bf t}^\epsilon_1|+|{\bf t}^\epsilon_2|=1$, see~\eqref{eq_val_abs_composants_t_epsilon_N}. 
Recall from~\eqref{valeurs_signes_e^+_x} the definition of $(\sigma_1,\sigma_2)$ and that $V^{\epsilon}(\gamma^N)\subset V(\gamma^N)$ is the set of vertices at $1$-distance at least $\epsilon$ to the poles to obtain:
\begin{align}
\forall x \in V^{\epsilon}(\gamma^N),\qquad |{\bf t}^\epsilon_1| 
:=
|{\bf t}^{\epsilon N}_x\cdot {\bf b}_0|
=
\sigma_1  {\bf t}^\epsilon_1,
\qquad 
|{\bf t}^\epsilon_2| 
:= 
|{\bf t}^{\epsilon N}_x\cdot {\bf b}_{\pi/2}| 
= 
\sigma_2 {\bf t}^\epsilon_2
.
\end{align}
This is because all points in the $1$-norm ball $B_1(x,\epsilon)$ around $x$ are in the same region of $\gamma^N$ when $x\in V^{\epsilon }(\gamma^N)$, thus $\sigma_1,\sigma_2$ are constant on $V(\gamma^N)\cap B_1(x,\epsilon )$.
As a result, the last line of~\eqref{eq_pour_faire_apparaitre_der_tangente_et_m_0} is equal to:
\begin{align}
&\frac{1}{4N}\sum_{x\in V^{\epsilon}(\gamma^N)}\Big[\sigma_1\sigma_2{\bf t}^\epsilon_2\big(\sigma_1{\bf t}^\epsilon_1-\sigma_2{\bf t}^\epsilon_2\big)\partial_2 + \sigma_2\sigma_1{\bf t}^\epsilon_1\big(\sigma_1{\bf t}^\epsilon_1-\sigma_2{\bf t}^\epsilon_2\big)\partial_1\Big]H_t(x)\nonumber\\
&\qquad=\frac{1}{4N}\sum_{x\in V^{\epsilon}(\gamma^N)}\big[\sigma_2{\bf t}^\epsilon_1 - \sigma_1{\bf t}^\epsilon_2\big]\big[{\bf t}^\epsilon_1\partial_1+{\bf t}^\epsilon_2\partial_2\big]H_t(x)\nonumber\\
&\qquad= \frac{1}{4N}\sum_{x\in V^{\epsilon}(\gamma^N)}\big[{\bf t}^\epsilon\cdot {\bf m}(x)\big] {\bf t}^\epsilon\cdot \nabla H_t(x).\label{eq_tangent_derivative_term_lemm_SSEP}
\end{align}
This concludes the proof of Lemma~\ref{lemm_partieSSEP_calcul_micro}. 
Indeed, the set $Z_B$ was defined in~\eqref{eq_def_Z_B}. 
Equation~\eqref{eq_expression_L_H_tout_remplace} then follows from the expression~\eqref{eq_bulk_term_interm_2} with the two identities~\eqref{eq_replacement_HG_terms_by_tangent_vector}--\eqref{eq_tangent_derivative_term_lemm_SSEP}.
\end{proof}
Lemmas~\ref{lemm_pole_terms} and~\ref{lemm_partieSSEP_calcul_micro} yield the statement of Proposition~\ref{prop_action_gen_micro}, setting:
\begin{align}
\tilde \omega(H_t,\delta,\epsilon,\gamma^N) := \omega_B(H_t,\delta,\epsilon,\tilde\gamma^N) + \omega_P(H_t,\delta,\tilde\gamma^N),\qquad \tilde \gamma^N\in\Omega^N_{\text{mic}},
\end{align}
as well as (recall~\eqref{eq_def_Z_p}--\eqref{eq_def_Z_B}):
\begin{equation}
\tilde Z := Z_B \cap Z_P;\label{eq_def_tilde_Z}
\end{equation}
recalling also the normalisation ${\bf t}^{\epsilon N}_x = \mathsf{v}^{\epsilon N}_x{\bf T}^{\epsilon N}_x$ from~\eqref{eq_def_v_epsilon_N_T_epsilon_N}.\\

We conclude the section with a useful bound on the Radon-Nikodym derivative, obtained as a consequence of the computations in the proofs of Lemmas~\ref{lemm_pole_terms}--\ref{lemm_partieSSEP_calcul_micro}. 
We stress that the result below does not require an estimate of the error terms and is therefore valid for any trajectory in $E([0,T],\e)$.
\begin{coro}\label{coro_bound_RD}
Let $H\in\C$. Recall from~\eqref{eq_def_der_radon_nyk} the definition of the Radon-Nikodym derivative $D^N_{\beta,H}$ up to time $T>0$. There is a constant $C(H)>0$ such that, for each $T>0$ and each trajectory $(\gamma^N_t)_{t\leq T}$ with values in $\Omega^N_{\text{mic}}\cap \e$:
\begin{align}
-C(H)-C(H)\int_0^{T}|\gamma^N_t|\, dt\Big]
\leq 
\frac{1}{N}\log D^N_{\beta,H}((\gamma^N_t)_{t\leq T}) 
\leq 
C(H) + C(H)\int_0^{T}|\gamma^N_t|\, dt\Big]
.
\end{align}
The same bounds hold for $\big|\int_0^T e^{N\langle\Gamma^N_t,H\rangle}N^2\lcal_\beta e^{N\langle\Gamma^N_t,H\rangle}\, dt\big|$ without the $\pm C(H)$.
\end{coro}
\subsubsection{From microscopic sums to line integrals}\label{sec_discrete_sums_to_line_integrals}
The goal of this section is to turn the expression of Proposition~\ref{prop_action_gen_micro} into an $N$-independent object, with nice continuity properties with respect to the topology on $E([0,T],\e)$ (defined in~\eqref{eq_def_d_E}). The statement of the result, Proposition~\ref{prop_informal_stuff_to_prove_action_gen}, 
requires some notations, which we introduce together with an explanation of the difficulties.\\

In Proposition~\ref{prop_action_gen_micro}, for each $T>0$, we find a set $\tilde Z$ of trajectories such that, if $(\gamma^N_t)_{t\leq T}\in \tilde Z\cap E([0,T],\e)$, the action of the generator in the Radon-Nikodym derivative~\eqref{eq_def_der_radon_nyk} contains terms of the form:
\begin{equation}
\int_0^{T}\frac{1}{N}\sum_{x\in V(\gamma^N_t)} f(t,\gamma^N_t,x)\, dt
,
\label{eq_discrete_line_integral_0}
\end{equation}
with $f$ a bounded function that depends on a neighbourhood of the vertex $x$ inside $\gamma^N_t$ at each time $t\in[0,T]$. To make sense of such an expression when $N$ is large, we would like:
\begin{itemize}
	\item to prove that there is a set $Z\subset \tilde Z$ on which the Replacement Lemma~\ref{lemm_Replacement_lemma_sec_mart} holds and microscopic curves have length of order $N$ at each time, so that the sum in~\eqref{eq_discrete_line_integral_0} is typically of order $1$ when $N$ is large;
	\item to then prove that this discrete sum can be seen as the discretisation of a suitable line integral on $\gamma^N_t$ at each time $t\in[0,T]$. 
	Informally, this line integral should have the same continuity property as the discrete version: if $\gamma^N\in \Omega^N_{\text{mic}}\cap \e$, a small change of $\gamma^N$ in Hausdorff distance should correspond to a small change in the corresponding line integral.
\end{itemize}
The first point is treated in the following lemma, proven in Section~\ref{sec_short_time_stability}.
\begin{lemm}\label{lemm_control_length_sec_computations}
Let $\beta>\log 2$. For each $T'>0$, there is $C(\beta,H,T')>0$ such that:
\begin{equation}
\forall A>0,\qquad 
\limsup_{N\rightarrow\infty}\frac{1}{N}\log\Prob^N_{\beta}\Big(\sup_{t\leq T'}|\gamma^N_t|\geq A\Big) 
\leq 
-C(\beta,H,T')A
.
\end{equation}
\end{lemm}
It will thus be enough to define $Z$ as the intersection of $\tilde Z$ and a set where the length is well-controlled, as done below in~\eqref{eq_def_Z}.\\

Let us now focus on the second point. Let $\tilde \gamma\in\Omega\cap \e$ be a Lipschitz Jordan curve. Let $(\tilde\gamma(u))_{u\in[0,1)}$ denote a parametrisation of $\tilde \gamma$. The line integral of a continuous $f:\R^2\rightarrow\R$ on $\tilde \gamma$ is by definition:
\begin{align}
\int_{\tilde\gamma}f \, ds 
:= 
\int_0^1 f(\tilde\gamma(u))\|\tilde\gamma'(u)\|_2 \, du
,
\end{align}
where $s$ denotes the arclength coordinate on $\tilde\gamma$.
Assume that $\tilde \gamma = \gamma^N$ with $\gamma^N\in \Omega^N_{\text{mic}}\cap \e$. 
In this case, $\tilde \gamma'(u)$ is proportional to either ${\bf b}_0$ or ${\bf b}_{\pi/2}$ and the line integral reads:
\begin{equation}
\int_{\tilde\gamma} f \, ds 
= 
\sum_{x\in V(\gamma^N)} f(x) [s(x+{\bf e}^+_x) - s(x)]
= 
\frac{1}{N}\sum_{x\in V(\gamma^N)} f(x)
.
\label{eq_wrong_Riemann_sum_for_gamma}
\end{equation}
For a microscopic curve, the discrete sums of Proposition~\ref{prop_action_gen_micro} could therefore be replaced with line integrals without loss of information. \\

The problem with~\eqref{eq_wrong_Riemann_sum_for_gamma}, however, is that the right-hand side is continuous in Hausdorff distance, while this is in general not true for the left-hand side. 
Indeed, consider the simplest case $f\equiv 1$. 
Recall that $|\gamma^N|$ is the length of $\gamma^N$ in 1-norm and let $|\gamma^N|_2$ be its length in $2$-norm. 
Then:
\begin{equation}
|\gamma^N|_2 
:= 
\int_{\gamma^N} 1\, ds 
= 
\frac{1}{N}\sum_{x\in V(\gamma^N)} 1 
= 
|\gamma^N|
,
\label{eq_discontinuity_length}
\end{equation}
\begin{figure}
\begin{center}
\includegraphics[width=10cm]{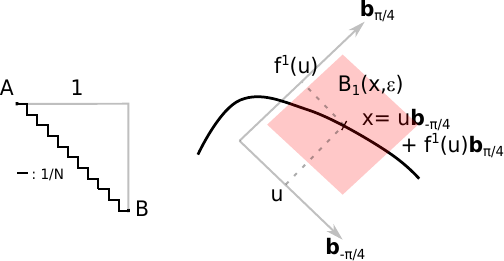} 
\caption{\emph{Left figure:} a lattice path on $(N^{-1}\Z)^2$ between the two extremities $A$ and $B$ of the unit square.  The path has length $2$ in both one- and two-norms. It converges in Hausdorff distance to the diagonal of the unit-square, which has 2-norm length $\sqrt 2$.\\
\emph{Right figure:} neighbourhood of a point $x$ at distance at least $\epsilon$ from the poles in one-norm. In the ball $B_1(x,\epsilon)$, the curve corresponds to the graph of a function $f^1$ in the reference frame $({\bf b}_{-\pi/4},{\bf b}_{\pi/4})$.\label{fig_steps_and_continuous_tangent_vector}}
\end{center}
\end{figure}
It is easy to see that the length $|\cdot|$ in one-norm (recall~\eqref{eq_def_normes_1_2}) is continuous in Hausdorff distance, using e.g.:
\begin{align}
|\gamma| 
= 
2\big[L_1(\gamma)-L_3(\gamma)\big]\cdot{\bf b}_{\pi/2} 
+ 2\big[L_2(\gamma)-L_4(\gamma)\big]\cdot{\bf b}_0,
\qquad \gamma\in\Omega
.
\end{align}
The continuity of the above functionals is established below~\eqref{eq_def_z_k_appB}. 
The length $|\cdot|_2$ in two-norm is however not continuous in Hausdorff distance. 
Indeed, assume that $(\gamma^N)_N$ converges in Hausdorff distance to a curve $\gamma^\infty$ and suppose $\gamma^\infty$ is not a lattice path: 
$\gamma^\infty\notin \bigcup_N\Omega_{\text{mic}}^N$. 
Then $|\gamma^\infty|_2\neq |\gamma^\infty|$, 
see Figure~\ref{fig_steps_and_continuous_tangent_vector}. 
However if $|\cdot|_2$ were continuous,~\eqref{eq_discontinuity_length} would yield:
\begin{align}
|\gamma^\infty|_2 = \int_{\gamma^\infty}1 ds = \lim_{N\rightarrow\infty}\frac{1}{N}\sum_{x\in V(\gamma^N)} 1 = |\gamma^\infty|,
\end{align}
which is a contradiction. 
We claim that, to preserve continuity of the right-hand side of~\eqref{eq_wrong_Riemann_sum_for_gamma} in Hausdorff distance, it must be written in terms of the following line integral:
\begin{equation}
\frac{1}{N}\sum_{x\in V(\gamma^N)} f(x) 
= 
\frac{1}{N}\sum_{x\in V(\gamma^N)} f(x) \|{\bf T}_x\|_1 = \int_{\gamma^N} f\mathsf{v}^{-1}\, ds
,
\label{eq_true_convergence_Riemann_sum_for_gamma}
\end{equation}
where, for $x\in V(\gamma^N)$, ${\bf T}_x$ is the unit vector in $\|\cdot\|_2$-norm, tangent to the edge $[x,x+{\bf e}^+_x]$. 
The quantity $\mathsf{v}$ is given by $\mathsf{v}^{-1} := \|{\bf T}\|_1$ and plays the same role as the $\|\gamma'(u) \|_2$ term in~\eqref{eq_wrong_Riemann_sum_for_gamma}. 
Note that $\mathsf{v}$ is identically equal to $1$ for microscopic curves, for which ${\bf T}_x$ is either $\pm {\bf b}_0$ or $\pm {\bf b}_{\pi/2}$.

The claim that~\eqref{eq_true_convergence_Riemann_sum_for_gamma} is the correct way to write the discrete sums is proven in Appendix \ref{sec_continuity_J_H_eps_ell_H_eps} 
where it is stated that, for continuous $f:\R^2\rightarrow\R$:
\begin{align}
\gamma\in \e\mapsto \int_\gamma f \mathsf{v}^{-1}ds\quad \text{is continuous in Hausdorff distance.}\label{eq_claim_continuity_line_integral}
\end{align}
The argument in Appendix~\ref{sec_continuity_J_H_eps_ell_H_eps} is actually carried out only for the integrands appearing in Proposition~\ref{prop_informal_stuff_to_prove_action_gen}, but could be generalised to the above setting.

Admitting the claim~\eqref{eq_claim_continuity_line_integral}, we now rewrite the expression of Proposition~\ref{prop_action_gen_micro} in terms of line integrals. 
To do so, we need some notations. 
Let $\gamma\in\Omega$ be a macroscopic interface. 
For $\epsilon>0$, let $\gamma(\epsilon)$ denote all points of $\gamma$ at $1$-distance at least $\epsilon$ from each poles and let $x\in \gamma(\epsilon)$. 
For definiteness, assume $x$ be in the first region of $\gamma$. 
By Definition~\ref{property_state_space} of $\Omega$, the portion $\gamma\cap B_1(x,\epsilon)$ of $\gamma$ is the graph of a $1$-Lipschitz function $f^1$ in the reference frame $({\bf b}_{-\pi/4},{\bf b}_{\pi/4})$ (see Figure~\ref{fig_steps_and_continuous_tangent_vector}): 
\begin{align}
\gamma\cap B_1(x,\epsilon) 
= 
\Big\{w{\bf b}_{-\pi/4}+ f^1(w){\bf b}_{\pi/4}:w\in u+ [-\epsilon/\sqrt{2},\epsilon/\sqrt{2}]\Big\},\quad u:= x\cdot {\bf b}_{-\pi/4}
.
\end{align}
The curve $\gamma$ has well-defined tangent vector at almost every point as it is Lipschitz. Let ${\bf t}$ and ${\bf T}$ denote two different normalisations of the same tangent vector, so that:
\begin{equation}
{\bf t} = \mathsf{v}{\bf T},\quad \|{\bf t}\|_1 = 1,\quad \|{\bf T}\|_2 = 1,\quad \mathsf{v} = \|{\bf t}\|_2 = (\|{\bf T}\|_1)^{-1},  \label{eq_relationship_t_T_v}
\end{equation}
with the tangent ${\bf t}$ at a point $w{\bf b}_{-\pi/4} + f^1(w){\bf b}_{\pi/4}$ given by:
\begin{align}
{\bf t}\Big(w{\bf b}_{-\pi/4} + f^1(w){\bf b}_{\pi/4}\Big) := \frac{\sqrt{2}}{2}\Big({\bf b}_{-\pi/4} + \partial_{w}f^1(w){\bf b}_{\pi/4}\Big).\label{eq_def_t_first_region}
\end{align}
Coming back to the point $x\in \gamma(\epsilon)$ in the first region of $\gamma$, define the continuous counterpart ${\bf t}^{\epsilon}(x)$ of the microscopic averaged tangent vector ${\bf t}^{\epsilon N}_\cdot$ introduced in~\eqref{eq_def_t_epsilonN_x} by:
\begin{equation}
{\bf t}^{\epsilon }(x) = \frac{1}{\sqrt{2}\epsilon}\int_{u-\epsilon/\sqrt{2}}^{u+\epsilon/\sqrt{2}} {\bf t}\Big(w{\bf b}_{-\pi/4} + f^1(w){\bf b}_{\pi/4}\Big)dw,\label{eq_exp_t_epsilon_N_as_integral_sur_un_cadrant}
\end{equation}
The corresponding object in other regions reads, for $y\in \gamma(\epsilon)$ in region $k$ of $\gamma$:
\begin{equation}
 {\bf t}^\epsilon(y) = \frac{1}{\sqrt{2}\epsilon}\int_{y\cdot {\bf b}_{\pi/4-k\pi/2}-\epsilon/\sqrt{2}}^{y\cdot {\bf b}_{\pi/4-k\pi/2}+\epsilon/\sqrt{2}}{\bf t}\Big(w{\bf b}_{\pi/4-k\pi/2} + f^k(w){\bf b}_{\pi/4-(k-1)\pi/2}\Big)dw.\label{eq_exp_t_epsilon_sans_N_as_integral_sur_un_cadrant}
\end{equation}
The vector ${\bf t}^\epsilon(x)$ indeed satisfies $\|{\bf t}^{\epsilon}(x)\|_1 = 1$ for $x\in\gamma(\epsilon)$. 
It coincides with ${\bf t}^{\epsilon N}_x$ if $\gamma$ is in fact in $\Omega^N_{\text{mic}}$ and $x$ is a vertex of $\gamma$. 
As in~\eqref{eq_relationship_t_T_v} for the tangent vectors at a single point, 
introduce finally a different normalisation ${\bf T}^\epsilon$ of the vector ${\bf t}^\epsilon$ and $\mathsf{v}^\epsilon$ as follows:
\begin{equation}
{\bf T}^ \epsilon 
:= 
{\bf t}^ \epsilon/\| {\bf t}^ {\epsilon}\|_2,\quad \mathsf{v}^ \epsilon := \|{\bf t}^ \epsilon\|_2 = \frac{1}{\|{\bf T}^ \epsilon\|_1}
.
\label{eq_def_cont_versions_T_v}
\end{equation}
Using~\eqref{eq_exp_t_epsilon_sans_N_as_integral_sur_un_cadrant}--\eqref{eq_def_cont_versions_T_v} and defining, for $A>0$:
\begin{equation}
Z
= 
Z(\beta,H,\delta,\epsilon,A) 
:= 
\tilde Z(\beta,H,\delta,\epsilon)\cap\Big\{(\gamma^N_t)_{t\leq T}:  \sup_{t\leq T}\int_0^T|\gamma^N_t|\, dt\leq AT\Big\}
,
\label{eq_def_Z}
\end{equation}
we obtain a version of Proposition~\ref{prop_action_gen_micro} where discrete sums are replaced with line integrals and the tangent vectors appearing are the usual 2-normed ones. 
\begin{prop}\label{prop_informal_stuff_to_prove_action_gen}
Let $A>0$. For each $\delta,\epsilon>0$ and trajectory $(\gamma^N_t)_{t\in[0,T]}\in E([0,T],\e)$, one has:
\begin{align}
\frac{1}{N}&\int_0^{T}N^2e^{-N\langle\Gamma^N_t,H_t\rangle}\lcal_{\beta}e^{N\langle\Gamma^N_t,H_t\rangle}\, dt 
\nonumber\\
&= 
\frac{1}{4}\int_0^{T}dt\int_{\gamma^N_t(\epsilon)}\frac{(\mathsf{v}^\epsilon)^2}{\mathsf{v}}\big[{\bf T}^\epsilon\cdot {\bf m}\big(\gamma^N_t(s)\big)\big] \ {\bf T}^\epsilon\cdot \nabla H\big(t,\gamma^N_t(s)\big)\, ds 
\nonumber\\
&\quad
+\frac{1}{2}\int_0^{T}dt \int_{\gamma^N_t(\epsilon)} \frac{(\mathsf{v}^\epsilon)^2}{\mathsf{v}}|{\bf T}^\epsilon\cdot {\bf b}_0||{\bf T}^\epsilon \cdot{\bf b}_{\pi/2}| H\big(t,\gamma^N_t(s)\big)^2\, ds + \int_0^{T}\omega(H_t,\delta,\epsilon,A,\gamma^N_t)\, dt
\nonumber\\
&\quad
-\frac{1}{2}\int_0^{T}\sum_{k=1}^4(1/2 - e^ {-\beta}) \big[H(t,L_k(\gamma^N_t)) + H(t,R_k(\gamma^N_t))\big]\, dt
, 
\label{eq_expression_L_avec_rplcmt_et_exp-beta_mais_avec_moyenne_vect_tangent} 
\end{align}
where $s$ is the arclength coordinate on $\gamma^N_t$, $\gamma(\epsilon)$ is the set of points in $\gamma$ at $1$-distance at least $\epsilon$ from the poles for each curve $\gamma\in\e$ and ${\bf m} = (\pm 1,\pm 1)$ is the sign vector in Definition~\ref{def_m(gamma)_gamma(eta)}. 
The vector ${\bf T}^\epsilon$ and $\mathsf{v}^\epsilon$ are defined in~\eqref{eq_def_cont_versions_T_v}. 
Distinguish $\mathsf{v}$ and $\mathsf{v}^\epsilon$ in~\eqref{eq_expression_L_avec_rplcmt_et_exp-beta_mais_avec_moyenne_vect_tangent}: 
the factor $\mathsf{v}^{-1}$ is the additional factor of~\eqref{eq_true_convergence_Riemann_sum_for_gamma} needed for continuity, 
while $\mathsf{v}^\epsilon$ comes from the averaging of the microscopic tangent vectors.

The error term $\omega$ can be controlled on the set $Z=Z(\beta,H,\delta,\epsilon,A)$ of~\eqref{eq_def_Z}: 
there is $C(H)>0$ such that, for trajectories $(\gamma^N_t)_{t\leq T}\in Z\cap E([0,T],\e)$, $\omega$ satisfies:
\begin{align}
\Big|\int_0^{T}\omega(H_t,\delta,\epsilon,A,\gamma^N_t)\, dt\Big|\leq 2\delta + C(H)T\Big(\epsilon  + \frac{A+1}{N}\Big)
.
\end{align}
Moreover,  
for each $\delta>0$, there is $C(\beta,H,T)>0$ and $\epsilon_0(\delta)>0$ such that, for each $A>0$, $Z=Z(\beta,H,\delta,\epsilon,A)$ satisfies: 
\begin{align}
&\sup_{0< \epsilon \leq \epsilon_0(\delta)}\limsup_{N\rightarrow\infty}\frac{1}{N}\log\Prob^N_{\beta,H}\Big(\gamma^N_\cdot\in Z(\beta,H,\delta,\epsilon,A)^c\cap E([0,T],\e)\Big) 
\nonumber\\
&\hspace{8cm}\leq 
\max\big\{-\delta^{-1},-C(\beta,H,T)A\big\}
.
 \label{eq_bound_proba_Z_G_H_epsilon_A_sec_3}
\end{align}
\end{prop}
\begin{rmk}\label{rmk_reason_for_volume_around_poles}
To connect the line integrals in~\eqref{eq_expression_L_avec_rplcmt_et_exp-beta_mais_avec_moyenne_vect_tangent} with those in the weak formulation~\eqref{eq_formulation_faible_avec_der_en_temps} of anisotropic motion by curvature with drift, take a curve $\gamma\in\e$. 
Notice from~\eqref{eq_exp_t_epsilon_sans_N_as_integral_sur_un_cadrant} that $\lim_{\epsilon\rightarrow 0}{\bf t}^\epsilon(x)={\bf t}(x)$ for almost every point $x$ of $\gamma$ at $1$-distance $\epsilon$ or more to the poles. 
Parametrise $\gamma$ by the tangent angle $\theta$, defined as the angle such that ${\bf T} = \cos(\theta){\bf b}_0+\sin(\theta){\bf b}_{\pi/2}$. 
Then, for almost every point that is not in the pole, 
${\bf T}^\epsilon,\mathsf{v}^\epsilon$, 
defined in~\eqref{eq_def_cont_versions_T_v}, converge a.e. to ${\bf T},\mathsf{v}$ respectively and: 
\begin{align}
\lim_{\epsilon\rightarrow 0} \bigg[\frac{(\mathsf{v}^\epsilon)^2}{\mathsf{v}}|{\bf T}^\epsilon_1{\bf T}^\epsilon_2|\bigg] 
= 
\mathsf{v}|{\bf T}_1{\bf T}_2|
= 
\frac{|\sin(2\theta)|}{2(|\sin(\theta)|+|\cos(\theta)|)}
.
\end{align}
This quantity is precisely the mobility $\mu(\theta)$, see~\eqref{eq_def_mobility}. 
Similarly, for almost every point associated with $\theta\in[0,2\pi]\setminus \frac{\pi}{2}\Z$:
\begin{align}
\lim_{\epsilon\rightarrow 0}\bigg[\frac{(\mathsf{v}^\epsilon)^2}{\mathsf{v}}[{\bf T}^\epsilon\cdot {\bf m}]\bigg]\big[{\bf T}^\epsilon \cdot\nabla\big] = \big[\mathsf{v}[{\bf T}\cdot {\bf m}]\big]\, \big[{\bf T}\cdot\nabla \big]= \alpha(\theta)\partial_s,
\end{align}
where $\alpha$ is defined in~\eqref{eq_def_alpha} and $\partial_s = \partial_{{\bf T}}$ is the derivative with respect to the arclength coordinate, well-defined almost everywhere on a Lipschitz curve.\demo
\end{rmk}
\section{Large deviation upper-bound and properties of the rate functions}\label{sec_large_dev_upper_bound}
In this section, we prove upper bound large deviations, i.e. the upper bound in Theorem~\ref{theo_large_dev}. This is done by adapting the method of \cite{Kipnis1989} to the present case, introducing the tilted dynamics $\Prob^N_{\beta,H}$, $H\in\C$ and quantifying the cost of tilting. 
A time $T>0$ is fixed throughout the section, as well as the value of $\beta>\log 2$. 
Before we start, let us fix and recall some notations.\\ 

For a bias $H\in\C$, the Radon-Nikodym derivative $D^N_{\beta,H} = \mathrm{d}\Prob^N_{\beta,H}/\mathrm{d}\Prob^N_{\beta}|_T$ until time $T$ reads:
\begin{align}
N^{-1}\log D^N_{\beta,H}((\gamma^N_t)_{t\leq T}) &= \big<\Gamma^N_{T},H_{T}\big> - \big<\Gamma^N_0,H_0\big> \nonumber\\
&\quad- N^{-1}\int_0^{T} e^{-N\langle\Gamma^N_t,H_t\rangle}\big(\partial_t + N^2\lcal_{\beta}\big)e^{N\langle\Gamma^N_t,H_t\rangle}\, dt
.
\label{eq_def_der_radon_nyk}
\end{align}
For each $A,\delta,\epsilon>0$, recall from~\eqref{eq_def_Z} the definition of: 
\begin{equation}
Z := Z(\beta,H,\delta,\epsilon,A), \label{eq_def_Z_H_epsilon_A_sec_4}
\end{equation}
the set of trajectories in which error terms arising in the computations of Section~\ref{sec_relevant_martingales} can be estimated. 
Recall also the expression of ${\bf m}$ from Definition~\ref{def_m(gamma)_gamma(eta)}. For a trajectory $\gamma^N_\cdot = (\gamma^N_t)_{t\leq T}$ in $E([0,T],\e)$, Proposition~\ref{prop_informal_stuff_to_prove_action_gen} tells us that there is a function $\omega$ such that $D^N_{\beta,H}$ reads:
\begin{equation}
N^{-1}\log D^N_{\beta,H}(\gamma^N_\cdot) = J^\beta_{H,\epsilon}(\gamma^N_\cdot) + \int_0^{T}\omega(H_t,\delta,\epsilon,\gamma^N_t)\, dt,\label{eq_def_der_radon_nyk_avec_J_eps_ell_eps}
\end{equation}
with, for some $C(H)>0$ and each trajectory $\gamma^N_\cdot\in Z\cap E([0,T],\e)$:
\begin{equation}
\Big|\int_0^{T}\omega(H_t,\delta,\epsilon,A,\gamma^N_t)\, dt\Big|\leq 2\delta + C(H)T\Big(\epsilon  + \frac{A+1}{N}\Big).\label{eq_error_term_RD_sec_4}
\end{equation}
The functional $J^\beta_{H,\epsilon}$ is defined on trajectories $\gamma_\cdot \in E([0,T],\e)$ by (refer to Appendix~\ref{sec_the_set_E(0,T0)} for properties of $E([0,T],\e)$):
\begin{equation}
J^\beta_{H,\epsilon}(\gamma_\cdot) := \ell^\beta_{H,\epsilon}(\gamma_\cdot) - \frac{1}{2}\int_0^{T}\int_{\gamma_t(\epsilon)} |{\bf T}^\epsilon\cdot {\bf b}_0||{\bf T}^\epsilon \cdot {\bf b}_{\pi/2}|\frac{(\mathsf{v}^\epsilon)^2}{\mathsf{v}}H^2(t,\gamma_t(s))\, ds\, dt.\label{eq_def_J_H_epsilon_zeta}
\end{equation}
Recall that $\mathsf{v} = \|{\bf T}\|_1^{-1} = \|{\bf t}\|_2$ and $\mathsf{v}^\epsilon = \|{\bf t}^\epsilon\|_2$. 
Moreover, for $t\in[0,T]$, $\gamma_t(\epsilon)$ is the set of points in $\gamma_t$ at $1$-distance at least $\epsilon$ from the poles, 
and $s$ is the arclength coordinate. Recall also that, for a curve $\gamma\in\e$, the letter $\Gamma$ denotes the associated droplet. 
The functional $\ell^\beta_{H,\epsilon}$ acts on trajectories $\gamma_\cdot\in E([0,T],\e)$ according to:
\begin{align}
\ell^\beta_{H,\epsilon}(\gamma_\cdot) 
&:= 
\big<\Gamma_{T},H_{T}\big> - \big<\Gamma^{\mathrm{ref}},H_0\big> - \int_0^{T} \big<\Gamma_t,\partial_t H_t\big>\, dt  
\nonumber\\
&\quad-\frac{1}{4}\int_0^{T}dt\int_{\gamma_t(\epsilon)}\frac{(\mathsf{v}^\epsilon)^2}{\mathsf{v}}\big[{\bf T}^\epsilon \cdot {\bf m}(\gamma_t(s ))\big] {\bf T}^\epsilon\cdot \nabla H(t,\gamma_t(s ))\, ds 
\label{eq_def_ell_H_epsilon_zeta}\\
&\quad
+ \Big(\frac{1}{4}-\frac{e^{-\beta}}{2}\Big)\int_0^{T}\sum_{k=1}^4 \big[H(t,L_k(\gamma_t))+H(t,R_k(\gamma_t))\big]\, dt
\nonumber
.
\end{align}
The proof of the upper bound large deviations in Theorem~\ref{theo_large_dev} is done in two steps.  
In Section~\ref{sec_upper_bound_around_trajectory}, we establish an upper bound on the probability of observing a given trajectory. This bound is then used, in Section~\ref{sec_upper_bound_closed_sets}, to establish an upper bound for closed sets.
\subsection{Upper bound around a given trajectory}\label{sec_upper_bound_around_trajectory}
In this section, all trajectories are defined on $[0,T]$, so we systematically write $\gamma_\cdot$ for $(\gamma_t)_{t\leq T}$. 

Let $\bar\gamma_\cdot\in E([0,T],\e)$ be fixed throughout. 
Let $B_{d_E}(\bar\gamma_\cdot,\zeta)$ denote the open ball of centre $\bar\gamma_\cdot$ and radius $\zeta>0$ in $d_E$-distance, defined in~\eqref{eq_def_d_E}, and 
let us estimate the quantity:
\begin{equation}
\lim_{\zeta\rightarrow 0}\lim_{N\rightarrow\infty}\frac{1}{N}\log\Prob^N_{\beta}\big(\gamma^N_{\cdot}\in B_{d_E}(\bar\gamma_\cdot,\zeta)\big).\label{eq_upper_bound_nice_traj_0}
\end{equation}
To highlight the important points and difficulties, we first estimate~\eqref{eq_upper_bound_nice_traj_0} in Section~\ref{sec_upper_bound_nice_trajectory} for "nice" trajectories, placing convenient assumptions on $\bar\gamma_\cdot$. 
General trajectories are treated in Section~\ref{sec_upper_bound_for_compact_sets}.
\subsubsection{Upper bound around nice trajectories\label{sec_upper_bound_nice_trajectory}}
Let us estimate~\eqref{eq_upper_bound_nice_traj_0}. 
Following \cite{Kipnis1989}, we estimate~\eqref{eq_upper_bound_nice_traj_0} using the expression~\eqref{eq_def_der_radon_nyk_avec_J_eps_ell_eps} 
of the Radon-Nikodym derivative $D^N_{\beta,H} = D^{N,T}_{\beta,H}$ between $\Prob^N_{\beta,H}$ and $\Prob^N_\beta$ for trajectories on the time-interval $[0,T]$. 
Let us first assume that, for some $\zeta>0$ henceforth fixed:
\begin{align}
B_{d_E}(\bar\gamma_\cdot,\zeta)\subset E([0,T],\e).
\label{eq_assumption_interior_upper_bound}
\end{align}
Take a bias $H\in\C$. For any measurable set $\tilde X$, we may write:
\begin{align}
\Prob^N_\beta\big(\gamma^N_\cdot\in B_{d_E}(\bar\gamma_\cdot,\zeta)\big) 
&= 
\Prob^N_\beta\big(\gamma^N_\cdot\in B_{d_E}(\bar\gamma_\cdot,\zeta)\cap \tilde X\big) 
+ \Prob^N_\beta\big(\gamma^N_\cdot\in B_{d_E}(\bar\gamma_\cdot,\zeta)\cap \tilde X^c\big)
\nonumber\\
&\leq
\E^N_{\beta,H}\Big[\big(D^N_{\beta,H}\big)^{-1}{\bf 1}_{\gamma^N_\cdot\in B_{d_E}(\bar\gamma_\cdot,\zeta)\cap \tilde X}\Big] 
+ \mathbb P^N_\beta\Big(\gamma^N_\cdot\in \tilde X^c\cap E([0,T],\e)\Big)
.
\label{eq_splitting_proba_with_set_tilde_X}
\end{align}
To estimate the right-hand side of~\eqref{eq_splitting_proba_with_set_tilde_X}, 
we choose the set $\tilde X$ to only contain trajectories on which the expression of Proposition~\ref{prop_informal_stuff_to_prove_action_gen} holds. 
In view of~\eqref{eq_def_der_radon_nyk_avec_J_eps_ell_eps}--\eqref{eq_error_term_RD_sec_4} set, for each $A,\delta,\epsilon>0$:
\begin{align}
\tilde X := Z = Z(\beta,H,\delta,\epsilon,A).
\end{align}
With this choice,~\eqref{eq_splitting_proba_with_set_tilde_X} becomes:
\begin{align}
\Prob^N_\beta\big(\gamma^N_\cdot\in B_{d_E}(\bar\gamma_\cdot,\zeta)\big) &\leq  \sup_{B_{d_E}(\bar\gamma_\cdot,\zeta)\cap Z}\exp\bigg[N\Big[ -J_{H,\epsilon}^\beta + 2\delta + C(H)T\Big(\epsilon  + \frac{A+1}{N}\Big)\Big]\bigg]\nonumber\\
&\hspace{2.5cm}+\mathbb P^N_\beta\Big(\gamma^N_\cdot\in Z^c\cap E([0,T],\e)\Big).\label{eq_splitting_proba_with_Z}
\end{align}
The first term in the right-hand side of~\eqref{eq_splitting_proba_with_Z} is typically of size $e^{-cN}$ for some $c>0$ as we shall see. For the decomposition into $Z$ and $Z^c$ to be useful, $Z^c$ must therefore have smaller probability. This is the case by~\eqref{eq_bound_proba_Z_G_H_epsilon_A_sec_3} provided $\epsilon$ is sufficiently small and $A$ sufficiently large: there is $c(\beta)>0$ and $\epsilon_0(\delta)>0$ such that, for each $A>0$:
\begin{equation}
\sup_{0<\epsilon\leq \epsilon_0(\delta)}\limsup_{N\rightarrow\infty}\frac{1}{N}\log\Prob^N_{\beta}\Big(\gamma^N_\cdot\in Z^c\cap E([0,T],\e)\Big) 
\leq 
\max\big\{-\delta^{-1},-c(\beta)A\big\}
.
\label{eq_bound_proba_Z_H_epsilon_A_sec_4} 
\end{equation}
For each $\delta>0$ and each $\epsilon\leq \epsilon_0(\delta)$, Equation~\eqref{eq_splitting_proba_with_Z} thus becomes:
\begin{align}
&\limsup_{N\rightarrow\infty}\frac{1}{N}\log \Prob_{\beta}^N\big(\gamma^N_\cdot\in B_{d_E}(\bar\gamma_\cdot,\zeta)\big)
\nonumber\\
&\hspace{3cm}
\leq 
\max\Big\{\sup_{B_{d_E}(\bar\gamma_\cdot,\zeta)}\big( -J_{H,\epsilon}^\beta\big) + 2\delta + C(H)\epsilon T, -\delta^{-1},-c(\beta) A\Big\}
.
\label{eq_upper_bound_nice_traj_1}
\end{align}
To relate~\eqref{eq_upper_bound_nice_traj_1} to the upper bound in terms of the functionals $J^\beta_H$ appearing in the definition~\eqref{eq_def_rate_functions} of the rate function $I_\beta(\cdot|\gamma^{\mathrm{ref}})$ of Theorem~\ref{theo_large_dev}, we need to know a bit more about the functional $J_{H,\epsilon}^\beta$. Let us momentarily make the following assumption:
\begin{equation}
\text{For each }\epsilon>0 \text{ and }H\in\C, \ \bar\gamma_\cdot\text{ is a point of continuity of the functional }J_{H,\epsilon}^\beta.\tag{$\star$}\label{assumption_continuity}
\end{equation}
Under Assumption~\eqref{assumption_continuity}, there is a modulus of continuity $m^\beta_{H,\epsilon,\bar\gamma_\cdot}(\zeta)\geq 0$ such that:
\begin{align}
\sup_{B_{d_E}(\bar\gamma_\cdot,\zeta)}\big(-J^\beta_{H,\epsilon}\big)\leq - J_{H,\epsilon}^\beta(\bar\gamma_\cdot)+ m^\beta_{H,\epsilon,\bar\gamma_\cdot}(\zeta),\qquad \lim_{\zeta'\rightarrow 0}m^\beta_{H,\epsilon,\bar\gamma_\cdot}(\zeta')
=
0
.
\end{align}
Thus, taking the small $\zeta$ limit in~\eqref{eq_upper_bound_nice_traj_1}, then the limits in $\epsilon,\delta,A$, one finds:
\begin{align}
\limsup_{\zeta\rightarrow0}\limsup_{N\rightarrow\infty}\frac{1}{N}\log \Prob_{\beta}^N\big(\gamma^N_\cdot\in B_{d_E}(\bar\gamma_\cdot,\zeta)\big) 
\leq 
-J_{H}^\beta(\bar\gamma)
,
\end{align}
where we used $\lim_{\epsilon\rightarrow 0}J_{H,\epsilon}^\beta(\bar\gamma_\cdot) = J_{H}^\beta(\bar\gamma_\cdot)$, see Proposition~\ref{prop_preuve_continuite_integrale_avec_v_moins_un} below. Optimising on $H\in\C$ then yields the desired upper bound under Assumption~\eqref{assumption_continuity}:
\begin{align}
\limsup_{\zeta\rightarrow0}\limsup_{N\rightarrow\infty}\frac{1}{N}\log \Prob_{\beta}^N\big(\gamma^N_\cdot\in B_{d_E}(\bar\gamma_\cdot,\zeta)\big) 
\leq 
-\sup_{H\in\C}J_{H}^\beta(\bar\gamma_\cdot)
.
\label{eq_upper_bound_nice_traj_2}
\end{align}
In full generality, however, 
Assumption~\eqref{assumption_continuity} is false: 
the functional $J^\beta_{H,\epsilon}$ is not continuous at $\bar\gamma_\cdot$ for every $H,\epsilon$ without further assumptions on $\bar\gamma_\cdot$. 
This can be seen by taking a uniform in space and time, small enough $H$ and a large $T$, in which case the dominating contribution in the expression~\eqref{eq_def_J_H_epsilon_zeta} of $J_{H,\epsilon}^\beta(\bar\gamma_\cdot)$ comes from the following pole term:
\begin{align}
\frac{1}{2}\int_0^T\sum_{k=1}^4(1/2-e^{-\beta})\big[H(t,L_k(\bar\gamma_t))+H(t,R_k(\bar\gamma_t))\big]\, dt
.
\end{align}
One can check that $L_k,R_k$ are not continuous functionals on $\e$ (this is discussed in Lemma~\ref{lemm_w_k}). 
The motion of the poles is thus responsible for a lack of continuity of $J_{H,\epsilon}^\beta$ on $E([0,T],\e)$, 
preventing Assumption~\eqref{assumption_continuity} from being true in general. 
The fact that the functionals $J_{H,\epsilon}^\beta$ are not continuous is a notable difference from the large deviations for the SSEP with reservoirs, where continuity does hold \cite{Bertini2003}.\\  

For Assumption~\eqref{assumption_continuity} and thus the upper bound~\eqref{eq_upper_bound_nice_traj_2} to hold, 
we therefore impose a further condition on the poles of $\bar\gamma_\cdot$, namely:
\begin{equation}
\text{for almost every }t\in[0,T], \bar\gamma_t\text{ has point-like poles: }R_k(\gamma_t)=L_k(\gamma_t) \text{ for each }1\leq k\leq 4
. 
\label{assumption_pointlike_poles}
\end{equation}
The sufficiency of this condition is stated next and proven in Appendix~\ref{sec_continuity_J_H_eps_ell_H_eps}.
\begin{prop}\label{prop_preuve_continuite_integrale_avec_v_moins_un}
Let $H\in\C$. For $\epsilon>0$, recall the definition~\eqref{eq_def_J_H_epsilon_zeta} of the functional $J^\beta_{H,\epsilon}$. 
Let $E_{pp}([0,T],\e)\subset E([0,T],\e)$ be the subset of trajectories with point-like poles at almost every time. 
Then each $\gamma_\cdot\in E_{pp}([0,T],\e)$ is a point of continuity of $J^\beta_{H,\epsilon}$ for the distance $d_E$ (defined in~\eqref{eq_def_d_E}), 
thus Assumption~\eqref{assumption_continuity} holds at $\gamma_\cdot$. 
In addition, the following convergence result holds on the whole of $E([0,T],\e)$:
\begin{equation}
\forall \gamma_\cdot\in E([0,T],\e),\qquad \lim_{\epsilon\rightarrow0}J_{H,\epsilon}^\beta(\gamma_\cdot) = J_{H}^\beta(\gamma_\cdot).
\label{eq_pointwise_convergence_J_H_epsilon_beta}
\end{equation}
\end{prop}
So far, we have proven the following upper bound. If $\bar\gamma_\cdot\in E_{pp}([0,T],\e)$,
\begin{align}
\limsup_{\zeta\rightarrow0}\limsup_{N\rightarrow\infty}\frac{1}{N}\log \Prob_{\beta}^N\big(\gamma^N_\cdot\in B_{d_E}(\bar\gamma_\cdot,\zeta)\big) 
\leq 
-\sup_{H\in\C}J_{H}^\beta(\bar\gamma_\cdot) 
= 
I_\beta(\bar\gamma_\cdot|\gamma^{\mathrm{ref}})
,
\label{eq_upper_bound_pointlike_poles}
\end{align}
with the rate function $I_\beta(\cdot|\gamma^{\mathrm{ref}})$ defined in~\eqref{eq_def_rate_functions}. In the next section, we explain how to extend this bound to trajectories that do not have point-like poles, thus do not satisfy~\eqref{assumption_pointlike_poles}.
\subsubsection{Upper bound around a general trajectory}\label{sec_upper_bound_for_compact_sets}
In Section~\ref{sec_upper_bound_around_trajectory}, we established upper bound large deviations around a trajectory having point-like poles at almost every time (and, for convenience, in the interior of $E([0,T],\e)$, recall~\eqref{eq_assumption_interior_upper_bound}).  
In this section, 
we explain how to estimate the probability of a ball around a trajectory that does not have these properties and prove:
\begin{align}
&\forall \bar\gamma_\cdot \in E([0,T],\e),
\nonumber\\
&\qquad
\limsup_{\zeta\rightarrow0}\limsup_{N\rightarrow\infty}\frac{1}{N}\log\Prob^N_\beta\Big(\gamma^N_\cdot\in B_{d_E}(\bar\gamma_\cdot,\zeta)\cap E([0,T],\e)\Big)
\leq 
-I_\beta(\bar\gamma_\cdot|\gamma^{\mathrm{ref}})
.
\label{eq_general_upper_bound_around_trajectory}
\end{align}
Note the presence of the set $E([0,T],\e)$ in the probability in~\eqref{eq_general_upper_bound_around_trajectory} to account for the fact that we no longer work under the assumption~\ref{eq_assumption_interior_upper_bound} that a ball around $\bar\gamma_\cdot$ is in $E([0,T],\e)$. 
This assumption merely simplified notations.

If $\bar\gamma_\cdot$ has almost always point-like poles,
~\eqref{eq_general_upper_bound_around_trajectory} is just~\eqref{eq_upper_bound_pointlike_poles}. 
When $\bar\gamma_\cdot$ does not have almost always point-like poles, 
one has $I_\beta(\bar\gamma_\cdot|\gamma^{\mathrm{ref}})=+\infty$ by definition. 
Proving~\eqref{eq_general_upper_bound_around_trajectory} thus boils down to proving:
\begin{align}
\forall &\bar\gamma_\cdot \in E([0,T],\e),
\label{eq_general_upper_bound_around_trajectory_bis}\\
&\qquad\bar\gamma_\cdot\notin E_{pp}([0,T],\e)
\quad \Rightarrow\quad 
\lim_{\zeta\rightarrow0}\limsup_{N\rightarrow\infty}\frac{1}{N}\log\Prob^N_\beta\Big(\gamma^N_\cdot\in B_{d_E}(\bar\gamma_\cdot,\zeta)\cap E([0,T],\e)\Big)
=
-\infty
.
\nonumber
\end{align}
To prove~\eqref{eq_general_upper_bound_around_trajectory_bis}, we show that, 
with probability super-exponentially close to $1$, 
microscopic trajectories are close to having almost always point-like poles 
(the precise statement is given in~\eqref{eq_proba_sous_exp_D_n_q} below). 
This is done in a similar spirit to energy estimates for the SSEP \cite{Bertini2009}, 
proving that, when $N$ is large, the proportion of the time interval during which poles are not reduced to a point vanishes. \\

\noindent\textbf{Controlling the pole dynamics.}  
We saw in Proposition~\ref{prop_value_slope_at_poles} that the time integrated slope at the pole only depends on the parameter $\beta$ with probability super-exponentially close to $1$. 
Here, we explain how to use an improved version of this result to define a large enough set $X$ (in fact a sequence of sets), 
on which trajectories will have almost point-like poles most of the time. 
This statement is made precise in Lemma~\ref{lemm_D_q_subset_E_pp} below. 
We then use this sequence of sets $X$~\eqref{eq_general_upper_bound_around_trajectory_bis}. \\

To control the poles, we start by reformulating the statement of Proposition~\ref{prop_value_slope_at_poles} in terms of a bound on the volume beneath a pole (rather than on the slope at the pole). 
For $\gamma\in\Omega$, let $(z_k)_{1\leq k \leq 4} = (z_k(\gamma))_{1\leq k\leq 4}$ denote the extremal coordinates of points in $\gamma$ (see Figure~\ref{fig_volume_under_pole}):
\begin{figure}
\begin{center}
\includegraphics[width=12cm]{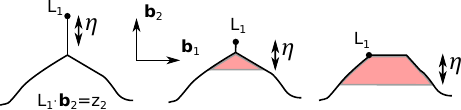} 
\caption{Neighbourhood of the north pole of three different curves. The position of the left extremity $L_1$ of the pole is marked by black dots. The volume $V_{1,\eta}$ at distance $\eta$ beneath the poles is the volume of the red area, which may vanish (left figure) if the ordinate $z_1 = L_1 \cdot {\bf b}_{\pi/2}$ of the north pole is at distance $\eta$ or more from the interior of the droplet associated with the curve. \label{fig_volume_under_pole}}
\end{center}
\end{figure}
\begin{align}
&z_1 = \sup \{ x\cdot{\bf b}_{\pi/2} : x\in\gamma\},\qquad z_3 = \inf \{ x\cdot{\bf b}_{\pi/2} : x\in\gamma\},\nonumber\\
&z_2 = \sup \{ x\cdot{\bf b}_0 : x\in\gamma\},\hspace{1.16cm} z_4 = \inf \{ x\cdot{\bf b}_0 : x\in\gamma\}.\label{eq_def_z_k}
\end{align}
For $\eta>0$ and $1\leq k \leq 4$, define then the volume $V_{k,\eta}$ beneath pole $k$ as follows (recall that $\Gamma$ is the droplet with boundary $\gamma$):
\begin{equation}
V_{k,\eta}(\gamma) = \begin{cases}
\big|\big\{x\in\Gamma : |z_k(\gamma)-x\cdot {\bf b}_{\pi/2}|\leq \eta\big\}\big|\quad \text{if }k\in\{1,3\},\\
\big|\big\{x\in\Gamma : |z_{k}(\gamma)-x\cdot {\bf b}_0|\leq \eta\big\}\big|\quad \text{if }k\in\{2,4\},
\end{cases}\qquad \gamma\in\e.\label{eq_def_V_epsilon_section_grd_devs}
\end{equation}
Compared to the slope, the volume $V_{k,\eta}$ is more robust to changes in the position of the pole: it is not hard to check that $V_{k,\eta}$ is continuous on $\e$ in Hausdorff distance $d_{\mathcal H}$ (see~\eqref{eq_def_distance_Hausroff}).  
Moreover, informally speaking, pole $k$ of a curve is point-like if and only if $V_{k,\eta}$ is of order $\eta^2$ for $\eta$ small, see the proof of Lemma~\ref{lemm_D_q_subset_E_pp}. 
We will therefore prove~\eqref{eq_general_upper_bound_around_trajectory_bis} by showing that the volume beneath a pole is controlled in a sufficiently strong sense.

The basic ingredient is the following bound established in Lemma~\ref{lemm_controle_deviations_volume}: 
for each $q,n\in\N_{\geq 1}$, there is $\eta(q,n)=\eta(q,n,T,\gamma^{\mathrm{ref}})>0$ such that:
\begin{align}
\sup_{\eta\leq \eta(q,n)}\limsup_{N\rightarrow\infty}\frac{1}{N}&\log \Prob^ N_{\beta}\bigg(\gamma^N_\cdot\in E([0,T],\e);\nonumber\\
&\quad\frac{1}{T}\int_{0}^{T}{\bf 1}\Big\{\big|\eta^{-2}V_{k,\eta}(\gamma^N_t)-(e^ {\beta}-1) \big|>\frac{1}{n}\Big\}\, dt>\frac{1}{n}\bigg)
\leq 
-q
.
\label{eq_proba_dev_volume_section_grd_devs}
\end{align}
In words and informally,~\eqref{eq_proba_dev_volume_section_grd_devs} states that it is very unlikely for trajectories to spend longer than $T/n$ without the pole dynamics fixing the volume of a sufficiently small portion beneath each pole. 
Simply by inclusion between the sets in the above probability, $n\mapsto\eta(q,n)$ can be taken to be decreasing. 
Up to reducing $\eta(q,n)$, we may also assume:
\begin{equation}
\forall q\in\N_{\geq 1},\qquad 
\lim_{n\rightarrow\infty}\eta(q,n)
=0
.
\label{eq_eta_vanishes_n_large}
\end{equation}
Define then a set $D_{q,n}$ with the following control of the poles:
\begin{align}
D_{q,n} 
&:= 
E([0,T],\e)\cap \bigcap_{k=1}^4\bigg\{\forall m\in \{1,...,n\},
\label{eq_def_D_p_n}\\ 
&\qquad 
\frac{1}{T}\int_{0}^{T}{\bf 1}\Big\{\Big|\eta(q,m)^{-2}V_{k,\eta(q,m)}(\gamma_t)-(e^ {\beta}-1) \Big|>\frac{1}{m}\Big\}\, dt\leq\frac{1}{m}\bigg\}
.
\nonumber
\end{align}
Since $V_{k,\eta}$ is continuous on $\e$ for the Hausdorff distance and the indicator function of an open set is lower semi-continuous, the set $D_{q,n}$ is closed in $E([0,T],\e)$ for each $q,n\in\N_{\geq 1}$. 
Moreover, by~\eqref{eq_proba_dev_volume_section_grd_devs}:
\begin{equation}
\limsup_{N\rightarrow\infty}\frac{1}{N}\log \Prob^{N}_{\beta}\big(\gamma^N_\cdot\in D_{q,n}^c\cap E([0,T],\e)\big)
\leq 
-q
.
\label{eq_proba_sous_exp_D_n_q}
\end{equation}
By construction, for $q\in\N_{\geq 1}$, $D_{q,n'}\subset D_{q,n}$ if $n\leq n'$. As a result, as $n$ increases, 
poles of trajectories in $D_{q,n}$ are controlled more and more precisely. 
For $q\in\N_{\geq 1}$, define then $D_q$ as:
\begin{equation}
D_q:= \bigcap_{n\geq 1}D_{q,n}.\label{eq_def_D_p}
\end{equation}
As stated in the next lemma, trajectories in each $D_q$, $q\in\N_{\geq 1}$ not only have almost always point-like poles 
(thus satisfy~\eqref{assumption_pointlike_poles}) but also have kinks with slope $e^{-\beta}$ at each pole. 
\begin{lemm}\label{lemm_D_q_subset_E_pp}
Fix $q\in\N_{\geq 1}$ and let $(\gamma_t)_{t\leq T}\in D_q$. Then, for each $k$ with $1\leq k\leq 4$:
\begin{align}
\text{for a.e. }t\in[0,T],\qquad 
\liminf_{\eta\rightarrow 0}\big|\eta^{-2}V_{k,\eta}(\gamma_t)-(e^{\beta}-1)\big|
=
0
.
\end{align}
In particular, $(\gamma_t)_{t\leq T}$ has almost always point-like poles. Note thus, for future reference, that $\gamma_t$ is almost always a Jordan curve.
\end{lemm}
Before proving the lemma, let us establish the general upper bound~\eqref{eq_general_upper_bound_around_trajectory_bis}. 
For $\zeta>0$, $n,q\in\N_{\geq 1}$ and $\beta>\log 2$, 
write using~\eqref{eq_proba_sous_exp_D_n_q}:
\begin{align}
\limsup_{N\rightarrow\infty}\frac{1}{N}\log\Prob^N_\beta\Big(&\gamma^N_\cdot\in B_{d_E}(\bar\gamma_\cdot,\zeta)\cap E([0,T],\e)\Big)
\nonumber\\
&\qquad\leq 
\max\Big\{\limsup_{N\rightarrow\infty}\frac{1}{N}\log\Prob^N_\beta\Big(\gamma^N_\cdot\in B_{d_E}(\bar\gamma_\cdot,\zeta)\cap D_{q,n}\Big), - q\Big\}
.
\label{eq_proof_general_upper_bound_trajectory_0}
\end{align}
By assumption, $\bar\gamma_\cdot$ does not have point-like poles. It thus does not belong to $D_{q}= \cap_n D_{q,n}$ by Lemma~\ref{lemm_D_q_subset_E_pp}. 
Since $D_{q,n}\subset D_{q,n'}$ for $n\geq n'$, 
there is $n_q\in\N_{\geq 1}$ such that $\bar\gamma_\cdot\notin D_{q,n}$ for each $n\geq n_q$. 
By construction, each $D_{q,n}$ is a closed set. There is thus $\zeta_q>0$ such that:
\begin{align}
\forall \zeta\in(0,\zeta_q),\forall n\geq n_q\qquad B_{d_E}\big(\bar\gamma_\cdot,\zeta\big) \cap D_{q,n} = \emptyset.
\end{align}
Injecting this in~\eqref{eq_proof_general_upper_bound_trajectory_0} proves the upper bound~\eqref{eq_general_upper_bound_around_trajectory_bis}:
\begin{align}
\forall q\in\N_{\geq 1},\forall \zeta\in(0,\zeta_q),\qquad \limsup_{N\rightarrow\infty}\frac{1}{N}\log\Prob^N_\beta\Big(&\gamma^N_\cdot\in B_{d_E}(\bar\gamma_\cdot,\zeta)\cap E([0,T],\e)\Big)
\leq 
-q
.
\end{align}
\begin{rmk}[Upper bound in Theorem~\ref{theo_large_dev_general}]
To obtain~\eqref{eq_general_upper_bound_around_trajectory}--\eqref{eq_general_upper_bound_around_trajectory_bis}, it is nowhere necessary that the trajectory $\bar\gamma_\cdot$ take values in $\e$ 
(i.e. that it be close to $\gamma^{\mathrm{ref}}$ in volume at each time, see Definition~\ref{def_effective_state_space}). 
The only property of $\bar\gamma_\cdot$ that is used is that it satisfies Property~\ref{prop_IC}. 
Rewriting the proof with this more general condition, Equations~\eqref{eq_general_upper_bound_around_trajectory}--\eqref{eq_general_upper_bound_around_trajectory_bis} correspond to the general upper bound in Theorem~\ref{theo_large_dev_general}.
\demo
\end{rmk}
We conclude the section with the proof of Lemma~\ref{lemm_D_q_subset_E_pp}.
\begin{proof}[Proof of Lemma~\ref{lemm_D_q_subset_E_pp}]
Consider the north pole $k=1$, the others being similar. 
Due to Definition~\ref{property_state_space} of $\Omega\supset\e$, 
a curve $\tilde\gamma\in \e$ does not have point-like north pole if and only if there is $c>0$ (the width of the north pole) such that, for any $\eta>0$ smaller than some $\eta_0(\tilde\gamma)$: 
\begin{equation}
V_{1,\eta}(\tilde\gamma)
\geq 
\eta c
.\label{eq_condition_not_pointlike}
\end{equation}
In particular, $\tilde\gamma$ has point-like north pole as soon as:
\begin{equation}
\liminf_{\eta\rightarrow 0}\eta^{-1}V_{1,\eta}(\tilde\gamma) = 0.\label{eq_suff_cond_volume_for_pointlike_pole}
\end{equation}
Fix a trajectory $(\gamma_t)_{t\leq T}\in D_q$ and let $\epsilon>0$. For each integer $n\geq 1/\epsilon$, one has by definition of $D_q$:
\begin{align}
\frac{1}{T}\int_{0}^{T}{\bf 1}\Big\{\big|\eta(q,n)^{-2}V_{1,\eta(q,n)}(\gamma_t)-(e^ {\beta}-1) \big|>\epsilon\Big\}\, dt\leq
\frac{1}{n}
.
\end{align}
Since $\eta(q,n)$ vanishes when $n$ is large by definition (see~\eqref{eq_eta_vanishes_n_large}), this implies:
\begin{align}
\liminf_{\eta\rightarrow 0}\int_{0}^{T}{\bf 1}\big\{\big|\eta^{-2}V_{1,\eta}(\gamma_t)-(e^ {\beta}-1) \big|>\epsilon\big\}\, dt 
=
0
.
\end{align}
Using Fatou inequality, we find:
\begin{align}
\liminf_{\eta\rightarrow 0}\big|\eta^{-2}V_{1,\eta}(\gamma_t)-(e^ {\beta}-1) \big|\leq \epsilon \quad \text{for a.e }t\in[0,T]
.
\end{align}
Since $\epsilon$ is arbitrary, $\gamma_\cdot$ has almost always point-like north pole recalling~\eqref{eq_suff_cond_volume_for_pointlike_pole}. Moreover, the last equation, also valid for pole $k\neq 1$, implies that there are $d_1,d_2>0$ such that, for almost every $t\in[0,T]$, 
there is a sequence $\eta_\ell(t)\in(0,1)$ ($\ell\in\N$) converging to $0$ such that:
\begin{equation}
\forall 1\leq k\leq 4,\forall\ell\in\N, \qquad 
\eta_\ell(t)^2 d_1\leq V_{k,\eta_\ell(t)}(\gamma_t)\leq \eta_\ell(t)^2d_2
.
\end{equation}
For each such time $t$, the poles are point-like by~\eqref{eq_condition_not_pointlike} and at zero distance to the interior of the droplet: $\gamma_t$ is thus a Jordan curve.
\end{proof}
\subsection{Upper bound on compact and closed sets}\label{sec_upper_bound_closed_sets}
Equipped with the bound~\eqref{eq_general_upper_bound_around_trajectory}, let us prove a large deviation bound for compact and closed sets in $E([0,T],\e)$. 
The arguments are classical and reproduced here for completeness.\\ 

\noindent\textbf{Upper bound for compact sets.} Let first $\mathcal{K}\subset E([0,T],\e)$ be a compact set. Let $\eta>0$. By~\eqref{eq_general_upper_bound_around_trajectory}, for each $\bar\gamma_\cdot\in E([0,T],\e)$, there is $\zeta(\bar\gamma_{\cdot})>0$ such that:
\begin{equation}
\limsup_{N\rightarrow\infty}\frac{1}{N}\log \Prob^N_\beta\Big(\gamma^N_\cdot\in B_{d_E}\big(\bar\gamma_\cdot, \zeta(\bar\gamma_{\cdot})\big)\cap E([0,T],\e)\Big)
\leq 
-I_\beta(\bar\gamma_\cdot|\gamma^{\mathrm{ref}})+\eta
.
\end{equation}
Cover the compact set $\mathcal K$ by $k_\eta\in\N_{\geq 1}$ balls $B_{d_E}\big(\bar\gamma^i_\cdot,\zeta^i(\bar\gamma_{\cdot})\big)$, to find:
\begin{align}
\limsup_{N\rightarrow\infty}\frac{1}{N}\log \Prob^N_\beta\Big(\gamma^N_\cdot\in\mathcal K\Big)&\leq \max_{1\leq i \leq k_\eta}\big(-I_\beta(\bar\gamma^i_\cdot|\gamma^{\mathrm{ref}})\big) + \eta\nonumber\\
&\leq -\inf_{\mathcal K}I_\beta(\cdot|\gamma^{\mathrm{ref}})+\eta.
\end{align}
This proves the upper bound for compact sets.\\

\noindent\textbf{Upper bound for closed sets.} The upper bound for closed sets follow from the exponential tightness of $\big(\mathbb{P}^N_{\beta}(\cdot, E([0,T],\e)\big)_N$ in $\mathcal M_1(E([0,T],\e))$, see Lemma 1.2.18 in \cite{Dembo2010}. Establishing exponential tightness is quite technical due to the poles, so we postpone it to Appendix~\ref{appen_tightness} and conclude here the upper bound of Theorem~\ref{theo_large_dev}. 
\section{Lower bound large deviations and hydrodynamic limits}\label{sec_large_dev_lower_bound}
In this section, we prove the lower bound in Theorem~\ref{theo_large_dev}. The method is classical (see \cite[Chapter 10]{Kipnis1999}). 
It consists in using Jensen inequality and an expression of the Radon-Nikodym derivative between the unbiased dynamics and the dynamics with bias $H\in\C$ on a finite time interval to turn the proof of the lower bound into a proof of the hydrodynamic limit for the tilted probability $\Prob^N_{\beta,H}$ (see Section~\ref{subsec_large_dev_lower_bound}). 

The subtlety is that the contour dynamics is only well-controlled inside the effective state space $\e$. 
As basic ingredient to prove hydrodynamics, 
we thus need to know that trajectories under tilted dynamics typically remain in the effective state space $\e$ for short time. 
This is proven in Section~\ref{sec_short_time_stability}. 
We then prove hydrodynamics in short time (Section~\ref{sec_short_time_hydrodynamics}) assuming there is only one solution $\gamma^H_{\cdot}$ to the weak formulation~\eqref{eq_formulation_faible_avec_der_en_temps} with bias $H\in \C$. 

If $\gamma^H_\cdot$ stays in the interior of the effective state space $\e$ on a longer time interval, 
then the hydrodynamic limit can correspondingly be extended to later time. 
This is carried out in Section~\ref{sec_up_to_time_T0}. 
\subsection{A first lower-bound}\label{subsec_large_dev_lower_bound}
In this section, we reduce the proof of the lower bound to the proof of hydrodynamic limits for the tilted processes.  
Recall the definition of the set $\tilde Z(\beta,H,\delta,\epsilon)$ for $\beta>\log 2$, $\epsilon,\delta,T>0$ and $H\in\C$ from Proposition~\ref{prop_action_gen_micro}. 
\begin{prop}\label{prop_lower_bound}
Let $\beta>\log 2$, $T>0$, $H\in\C$ and let $\gamma^H_\cdot\in E_{pp}([0,T],\e)$ solve the weak formulation~\eqref{eq_formulation_faible_avec_der_en_temps} of anisotropic motion by curvature with drift $H$ (uniqueness is not needed here). 
Assume that $\gamma^H_\cdot$ is in the interior of $E([0,T],\e)$:
\begin{equation}
\exists \zeta_0>0,\quad B_{d_E}(\gamma^H_\cdot,\zeta_0)\subset E([0,T],\e),\label{eq_gamma_H_in_interior}
\end{equation}
with $B_{d_E}(\gamma^H_\cdot,\zeta_0)$ the open ball of centre $\gamma^H_\cdot$ and radius $\zeta_0$ in $d_E$-distance. Then:
\begin{align}
&\liminf_{\zeta\rightarrow 0}\liminf_{N\rightarrow\infty}\frac{1}{N}\log\Prob^N_{\beta}\big(\gamma^N_\cdot\in B_{d_E}(\gamma^H_\cdot,\zeta)\big) \geq -I_\beta(\gamma^H_\cdot|\gamma^{\mathrm{ref}})\label{eq_lower_bound_dans_prop_lower_bound} \\
&\hspace{2cm} + \inf_{\delta>0}\liminf_{\epsilon\rightarrow 0}\liminf_{\zeta>0}\liminf_{N\rightarrow\infty}\frac{1}{N}\log\Prob^N_{\beta,H}\Big(\gamma^N_\cdot\in B_{d_E}(\gamma^H_\cdot,\zeta)\cap \tilde Z(\beta,H,\delta,\epsilon)\Big).\nonumber
\end{align}
\end{prop}
\begin{proof}
Let $H\in\C$ be as in the proposition and $\zeta,\delta,\epsilon>0$. Write for short:
\begin{align}
X := B_{d_E}(\gamma^H_\cdot,\zeta)\cap \tilde Z(\beta,H,\delta,\epsilon).
\end{align}
Then:
\begin{align}
\log \Prob^N_{\beta}\big(\gamma^N_\cdot\in B_{d_E}(\gamma^H_\cdot,\zeta)\big) &\geq \log \Prob^N_{\beta,H}\big(\gamma^N_\cdot\in X\big) \nonumber\\
&= \log\Bigg(\frac{\E^N_{\beta,H}\Big[\big(D^N_{\beta,H}\big)^{-1}{\bf 1}_{\gamma^N_\cdot\in X}\Big]}{\Prob^N_{\beta}\big(\gamma^N_\cdot\in X\big)}\Bigg) + \log \Prob^N_{\beta,H}\big(\gamma^N_\cdot\in X\big).
\end{align}
Jensen inequality applied to the logarithm then yields, dividing by $N$:
\begin{align}
\frac{1}{N}\log \Prob^N_{\beta}\big(\gamma^N_\cdot\in \mathcal O\big) \geq -\frac{\E^N_{\beta,H}\Big[N^{-1}\log D^N_{\beta,H}{\bf 1}_{\gamma^N\in X_\cdot}\Big]}{\Prob^N_{\beta,H}\big(\gamma^N_\cdot\in X\big)}+\frac{1}{N}\log \Prob^N_{\beta,H}\big(\gamma^N_\cdot\in X\big).\label{eq_lower_bound_interm_0}
\end{align}
Taking the limits and infima as in the statement of Proposition~\ref{prop_lower_bound}, the second term already has the desired form. Let us compute the expectation. Notice first that elements of $B_{d_E}(\gamma^H_\cdot,\zeta)$ have well controlled-length. Indeed, $\gamma^H_\cdot \in E([0,T],\e)$ means that its length is integrable in time, so that there is $c(H,T)>0$ with:
\begin{align}
\int_0^T|\gamma^H_t|\, dt\leq c(H,T).
\end{align}
One readily checks that the length of a curve $\gamma\in\e$ in $1$-norm is given in terms of the distance between the poles:
\begin{align}
|\gamma| = 2\big[L_1(\gamma)-L_3(\gamma)\big]\cdot{\bf b}_{\pi/2} + 2\big[L_2(\gamma)-L_4(\gamma)\big]\cdot{\bf b}_0.
\end{align}
Each functional in the right-hand side is $1$-Lipschitz in Hausdorff distance (see below~\eqref{eq_def_z_k_appB}), 
thus:
\begin{equation}
\forall \gamma_\cdot\in B_{d_E}(\gamma^H_\cdot,\zeta),\qquad \int_0^{T}|\gamma_t|\, dt 
\leq 
8\int_{0}^{T}d_{\mathcal H}(\gamma^H_t,\gamma_t)\, dt 
+ c(H,T)\leq 8\zeta + c(H,T)
.
 \label{eq_bound_length_dans_boule}
\end{equation}
Recall now the formula~\eqref{eq_def_der_radon_nyk_avec_J_eps_ell_eps} for $N^{-1}\log D^N_{\beta,H}$: for $\gamma^N_\cdot\in E([0,T],\e)$,
\begin{align}
N^{-1}\log D^N_{\beta,H}(\gamma^N_\cdot) = - J_{H,\epsilon}^{\beta}(\gamma^N_\cdot) + \int_0^{T} \omega(H_t,\delta,\epsilon,\gamma^N_t)\, dt
.
\end{align}
There is moreover $C(H)>0$ such that, for each $A>0$, on the set $\tilde Z(\beta,H,\delta,\epsilon)\cap E([0,T],\e)\cap \{\int_0^{T}|\gamma^N_t|\, dt\leq AT\}$, the quantity $\omega$ satisfies:
\begin{align}
\Big|\int_0^{T}\omega(H_t,\delta,\epsilon,\gamma^N_t)\, dt\Big|\leq 2\delta + C(H)T\Big(\epsilon  + \frac{A+1}{N}\Big).
\end{align}
In view of the bound~\eqref{eq_bound_length_dans_boule} on the length, 
the above bound on $\omega$ is valid on $X$. 
Taking the $\liminf$ in $N$, the expectation in~\eqref{eq_lower_bound_interm_0} is thus bounded from below as follows:
\begin{equation}
\liminf_{N\rightarrow\infty}-\frac{\E^N_{\beta,H}\Big[N^{-1}\log D^N_{\beta,H}{\bf 1}_{\gamma^N_\cdot\in X}\Big]}{\Prob^N_{\beta,H}\big(\gamma^N_\cdot\in X\big)}
\geq 
\liminf_{N\rightarrow\infty}\frac{\E^N_{\beta,H}\Big[\big(-J_{H,\epsilon}^\beta\big){\bf 1}_{\gamma^N_\cdot\in X}\Big]}{\Prob^N_{\beta,H}\big(\gamma^N_\cdot\in X\big)} - 2\delta  -C(H)\epsilon T.\label{eq_lower_bound_interm_1}
\end{equation}
Since $\gamma^H_\cdot\in E_{pp}([0,T],\e)$, it is a point of continuity of $J_{H,\epsilon}^\beta$ by Proposition~\ref{prop_preuve_continuite_integrale_avec_v_moins_un}. There is consequently a real function $m_{\beta,\gamma^H_\cdot,H,\epsilon}(\cdot)\geq 0$ such that:
\begin{align}
\sup_{\gamma_\cdot\in B_{d_E}(\gamma^H_\cdot,\zeta)}\big|J_{H,\epsilon}^\beta(\gamma_\cdot)-J^\beta_{H,\epsilon}(\gamma^H_\cdot)\big| = m_{\beta,\gamma^H_\cdot,H,\epsilon}(\zeta),
\qquad 
\lim_{\zeta\rightarrow 0}m_{\beta,\gamma^H_\cdot,H,\epsilon}(\zeta) 
= 0
.
\end{align}
As $X\subset B_{d_E}(\gamma^H_\cdot,\zeta)$, we deduce:
\begin{align}
\liminf_{\zeta\rightarrow 0}\liminf_{N\rightarrow\infty}-\frac{\E^N_{\beta,H}\Big[N^{-1}\log D^N_{\beta,H}{\bf 1}_{\gamma^N_\cdot\in X}\Big]}{\Prob^N_{\beta,H}\big(\gamma^N_\cdot\in X\big)}\geq -J_{H,\epsilon}^\beta(\gamma^H_\cdot)-2\delta - C(H)\epsilon T.
\end{align}
By Proposition~\ref{prop_preuve_continuite_integrale_avec_v_moins_un}, $J_{H,\epsilon}^\beta(\gamma^H_\cdot)$ converges to $J_{H}^\beta(\gamma^H_\cdot)$ when $\epsilon$ vanishes. Taking the liminf in $\epsilon$, then the infimum on $\delta$ in the last equation thus turns its right-hand side into:
\begin{align}
\inf_{\delta>0}\liminf_{\epsilon\rightarrow 0}\Big\{-J_{H,\epsilon}^\beta(\gamma^H_\cdot)-2\delta - C(H)\epsilon T\Big\} = -J_{H}^\beta(\gamma^H_\cdot).
\end{align}
To establish the claim of Proposition~\ref{prop_lower_bound}, it only remains to prove that $J^\beta_{H}(\gamma^{H}_{\cdot}) = I_\beta(\gamma^H_\cdot|\gamma^{\mathrm{ref}})$. For $G\in\C$, recall the definition~\eqref{eq_def_J_H} of $J^\beta_G$:
\begin{align}
J^\beta_G(\gamma^H_\cdot) = \ell_G(\gamma^H_\cdot) -\frac{1}{2}\int_0^T\int_{\gamma^H_t}G^2\big(t,\gamma^H_t(s))\mu(\theta(s)\big)\, ds\, dt,
\end{align}
where $\ell^\beta_G$ is the functional defined in~\eqref{eq_def_ell_H}. 
Above, recall that, for a point $\gamma^H_t(s)$, $\theta(s)$ is the angle of the tangent vector ${\bf T}(\theta(s)) = \cos(\theta(s)){\bf b}_0 + \sin(\theta(s)){\bf b}_{\pi/2}$ with the horizontal axis at $\gamma^H_t(s)$. \\
From the weak formulation~\eqref{eq_formulation_faible_avec_der_en_temps} of anisotropic motion by curvature, one has, for each $G\in\C$:
\begin{align}
\ell^\beta_G(\gamma^H_\cdot) 
= 
\int_0^T\int_{\gamma^H_t}G(t,\gamma^H_t(s))H\big(t,\gamma^H_t(s)\big)\mu(\theta(s))\, ds\, dt
.
\end{align}
As a result, 
\begin{align}
&I_\beta(\gamma^H_\cdot|\gamma^{\mathrm{ref}}) 
:= 
\sup_{G\in C}J^\beta_G(\gamma^H_\cdot) 
\nonumber\\
&\ = \frac{1}{2}\int_0^T\int_{\gamma^H_t}H^2\big(t,\gamma^H_t(s)\big)\mu(\theta(s))\, ds\, dt + \sup_{G\in\C}\bigg\{- \frac{1}{2}\int_0^T\int_{\gamma^H_t}\big[G-H\big]^2\big(t,\gamma^H_t(s)\big)\mu(\theta(s))\, ds\, dt\bigg\} \nonumber\\
&\ 
= \frac{1}{2}\int_0^T\int_{\gamma^H_t}H^2\big(t,\gamma^H_t(s)\big)\mu(\theta(s))\, ds\, dt  
= 
J^\beta_H(\gamma^H_\cdot)
.
\end{align}
This concludes the proof of Proposition~\ref{prop_lower_bound}.
\end{proof}
\subsection{Hydrodynamic limits for the tilted processes}\label{sec_hydro}
In this section, we prove hydrodynamics for $\Prob^N_{\beta,H}$ ($H\in\C$) under a uniqueness condition, thereby proving lower bound large deviations by showing that the probability in Proposition~\ref{prop_lower_bound} vanishes.

Fix $\beta>\log 2$ and $H\in\C$ throughout the section. 
Let $T>0$ and $\gamma^H_{\cdot}\in E_{pp}([0,T],\e)$ solve the weak formulation~\eqref{eq_formulation_faible_avec_der_en_temps} of anisotropic motion by curvature with drift $H$ on $[0,T]$.  Assume as in Proposition~\ref{prop_lower_bound} that $\gamma^H_{\cdot}$ is in the interior of $E([0,T],\e)$: 
for some $\zeta_0>0$,
\begin{equation}
B_{d_E}(\gamma^H_{\cdot},\zeta_0)\subset E([0,T],\e)
.
\end{equation}
\begin{prop}\label{prop_control_proba_lower_bound}
Under the above assumptions,
\begin{align}
\inf_{\delta>0}\liminf_{\epsilon\rightarrow 0}\liminf_{\zeta>0}\liminf_{N\rightarrow\infty}\frac{1}{N}\log\Prob^N_{\beta,H}\Big(\gamma^N_\cdot\in B_{d_E}(\gamma^H_\cdot,\zeta)\cap \tilde Z(\beta,H,\delta,\epsilon)\Big) = 0.
\end{align}
\end{prop}
The proof of Proposition~\ref{prop_control_proba_lower_bound} takes up Sections~\ref{sec_short_time_stability} to~\ref{sec_up_to_time_T0}. We proceed as follows. First, we get rid of the technical condition that trajectories belong to $\tilde Z(\beta,H,\delta,\epsilon)$, in Lemma~\ref{lemm_getting_rid_of_tilde_Z}. 
We then prove that trajectories typically stay in $E([0,T],\e)$ for sufficiently small time $T>0$, in Section~\ref{sec_short_time_stability}, 
thereby proving the first item of Proposition~\ref{prop_short_time_existence_of_weak_solution}. 
This result is used in Section~\ref{sec_short_time_hydrodynamics} to establish short time hydrodynamics for the tilted processes. 
This proves Proposition~\ref{prop_control_proba_lower_bound} for short time. 
Finally, in Section~\ref{sec_up_to_time_T0}, we extend the short time hydrodynamics to the whole time interval $[0,T]$, concluding the proof of Proposition~\ref{prop_control_proba_lower_bound}.\\

Let us first deal with $\tilde Z(\beta,H,\delta,\epsilon)$.
\begin{lemm}\label{lemm_getting_rid_of_tilde_Z}
With the notations of Proposition~\ref{prop_lower_bound},
\begin{align}
\inf_{\delta>0}\liminf_{\epsilon\rightarrow 0}\liminf_{\zeta>0}&\liminf_{N\rightarrow\infty}\frac{1}{N}\log\Prob^N_{\beta,H}\Big(\gamma^N_\cdot\in B_{d_E}(\gamma^H_\cdot,\zeta)\cap \tilde Z(\beta,H,\delta,\epsilon)\Big)\nonumber\\
&\qquad = \liminf_{\zeta>0}\liminf_{N\rightarrow\infty}\frac{1}{N}\log\Prob^N_{\beta,H}\Big(\gamma^N_\cdot\in B_{d_E}(\gamma^H_\cdot,\zeta)\Big).
\end{align}
\end{lemm}
\begin{proof}
Write first:
\begin{align}
&\Prob^N_{\beta,H}\Big(\gamma^N_\cdot\in B_{d_E}(\gamma^H_\cdot,\zeta)\cap \tilde Z(\beta,H,\delta,\epsilon)\Big) \nonumber\\
&\hspace{2cm}= \Prob^N_{\beta,H}\Big(\gamma^N_\cdot\in  B_{d_E}(\gamma^H_\cdot,\zeta)\Big) - \E^N_{\beta}\Big[D^N_{\beta,H}{\bf 1}_{\gamma^N_\cdot\in B_{d_E}(\gamma^H_\cdot,\zeta)\cap (\tilde Z(\beta,H,\delta,\epsilon))^c}\Big]
\end{align}
Corollary~\ref{coro_bound_RD} bounds the Radon-Nikodym derivative for trajectories in $E([0,T],\e)$ as follows. There is $C(H)>0$ such that:
\begin{align}
\forall \gamma^N_\cdot\in E([0,T],\e),\qquad D^N_{\beta,H}(\gamma^N_\cdot)\leq \exp\bigg[C(H)N+C(H)\int_0^{T}|\gamma^N_t|\, dt\bigg].
\label{eq_bound_RD_sec5}
\end{align}
Since $B_{d_E}(\gamma^H_\cdot,\zeta)\subset E([0,T],\e)$ for small enough $\zeta$ by Assumption~\eqref{eq_gamma_H_in_interior}, the above bound is valid for microscopic trajectories in $B_{d_E}(\gamma^H_\cdot,\zeta)$.

For trajectories in $B_{d_E}(\gamma^H_\cdot,\zeta)$, the time integral of the length in~\eqref{eq_bound_RD_sec5} is bounded by $c(H,T) + 8\zeta$, see~\eqref{eq_bound_length_dans_boule}. As a result, for $\delta,\epsilon>0$:
\begin{align}
&\Big|\Prob^N_{\beta,H}\Big(\gamma^N_\cdot\in B_{d_E}(\gamma^H_\cdot,\zeta)\cap \tilde Z(\beta,H,\delta,\epsilon)\Big) - \Prob^N_{\beta,H}\Big(\gamma^N_\cdot\in B_{d_E}(\gamma^H_\cdot,\zeta)\Big)\Big|\nonumber\\
&\hspace{2.0cm}\leq e^{C(H)N+C(H)N(c(H,T)+8\zeta)}\Prob^N_\beta\Big(\gamma^N_\cdot\in B_{d_E}(\gamma^H_\cdot,\zeta) \cap \big(\tilde Z(\beta,H,\delta,\epsilon)\big)^c\Big).\label{eq_lemm_delete_tilde_Z_0}
\end{align}
Moreover, we know by Proposition~\ref{prop_action_gen_micro} that, for each $A>0$:
\begin{align}
\lim_{\epsilon\rightarrow0}\limsup_{N\rightarrow\infty}\frac{1}{N}\log\Prob^N_\beta\Big(\gamma^N_\cdot\in E([0,T],\e)\cap \Big\{\int_0^{T}|\gamma_t|\, dt\leq AT\Big\}
\cap 
\big(\tilde Z(\beta,H,\delta,\epsilon)\big)^c\Big) 
= 
-\infty
.
\end{align}
Since Assumption~\eqref{eq_gamma_H_in_interior} and the bound~\eqref{eq_bound_length_dans_boule} on the length imply $B_{d_E}(\gamma^H_\cdot,\zeta)\subset E([0,T],\e)\cap \big\{\int_0^{T}|\gamma_t|\, dt\leq AT\big\}$ for any $AT\geq c(H,T)+8\zeta$, 
the above bound applies to estimate the right-hand side of~\eqref{eq_lemm_delete_tilde_Z_0}. 
Taking the logarithm, dividing by $N$ and taking the liminf in $N$, then in $\zeta$, then in $\epsilon$ in~\eqref{eq_lemm_delete_tilde_Z_0} yields the claim of the lemma.
\end{proof}
\subsubsection{The droplet moves on a diffusive scale}\label{sec_short_time_stability}
In this section, we prove that trajectories typically stay in the effective state space $\e$ on a short diffusive time scale, corresponding to item $2$ of Proposition~\ref{prop_short_time_existence_of_weak_solution}. Recall the convention that, for two interfaces $\gamma,\tilde\gamma\in\Omega$ with associated droplets $\Gamma,\tilde\Gamma$:
\begin{equation}
d_{L^1}(\gamma,\tilde\gamma) 
:= 
d_{L^1}(\Gamma,\tilde\Gamma),\quad \text{with}\quad d_{L^1}(\Gamma,\tilde\Gamma) = \int_{\R^2}|{\bf 1}_{\Gamma}-{\bf 1}_{\tilde\Gamma} |\, du\, dv.\label{eq_convention_volume_distance_sec_521}
\end{equation}
\begin{lemm}[Short-time stability of $\e$]\label{lemm_e_typical_short_time}
Recall that the effective state space $\e = \overline{B_{d_{L^1}}(\gamma^{\mathrm{ref}},r_0^2)}$ is defined in Definition~\ref{def_effective_state_space}. 
For each $\epsilon<r_0$, there is a time $t(\epsilon) = t(\beta,H,\epsilon,|\gamma^{\mathrm{ref}}|)>0$, 
independent of $\gamma^{\mathrm{ref}}$ except through its length, such that:
\begin{equation}
\lim_{N\rightarrow\infty}\Prob^N_{\beta,H}\Big(\sup_{t\leq  t(\epsilon)}d_{L^1}\big(\gamma^N_t,\gamma^{\mathrm{ref}}\big)\leq \epsilon^2\Big) 
= 
1
.\label{eq_short_time_stability_partant_de_Gamma_0}
\end{equation}
For other initial conditions, 
for any $\tilde r\in(0,r_0)$ and each $\kappa$ larger than some $\kappa(H,\gamma^{\mathrm{ref}})>0$, 
there is a time $t_0~:=~t_0(\beta,H,\tilde r,\kappa)>0$ such that:
\begin{equation}
\lim_{N\rightarrow\infty}\inf_{\substack{\gamma^N\in \Omega^N_{\text{mic}}:|\gamma^N|\leq \kappa \\ d_{L^1}(\gamma^{\mathrm{ref}},\gamma^N)\leq \tilde r^2}}\Prob^{\gamma^N}_{\beta,H}\Big(\gamma^N_\cdot\in  E\big([0,t_0],\e\big)\Big) = 1.\label{eq_stability_e_voisinage_Gamma_0}
\end{equation}
\end{lemm}
\begin{proof}
The proof of~\eqref{eq_short_time_stability_partant_de_Gamma_0} is similar to the proof of the same statement in \cite{Caputo2011} for the stochastic Ising model. In both cases, the idea is that changing the volume of the droplet requires adding or deleting a number of blocks of order $N^2$, which takes time. The additional difficulty in the present case comes from the pole dynamics: droplets can grow.

To deal with growth, we prove in Lemma~\ref{lemm_tightness_sup_length} that, under $\Prob^N_{\beta,H}$, the length of a curve typically stays of order $N$ on a diffusive time-scale.  
More precisely, for each $T>0$, 
there is $C(\beta,H,T)>0$ such that, for each $A>0$:
\begin{align}
\limsup_{N\rightarrow\infty}\frac{1}{N}\log\Prob^N_{\beta,H}\Big(\gamma^N_\cdot\in E([0,T],\e)\cap\Big\{&\sup_{t\leq T}|\gamma^N_t|\geq A\Big\}\Big)
\nonumber\\
&\qquad 
\leq 
-C(\beta,H,T)A+|\gamma^{\mathrm{ref}}|\beta
.
\label{eq_control_length_lemm_e_typical}
\end{align}
In the following, $\kappa_H>0$ is a constant such that the right-hand side of~\eqref{eq_control_length_lemm_e_typical} is strictly negative.

For trajectories with length bounded by $\kappa_H$, we will be able to use the following result. 
Recall the convention~\eqref{eq_convention_volume_distance_sec_521} that the volume distance between two curves is the volume distance between their respective droplets.
\begin{lemm}\label{lemm_volume_diff_as_test_function}
Let $\epsilon>0$. 
There are functions $J_1,J_2\in C^2_c(\R^2,[0,1])$, 
with gradient bounded in terms of $\epsilon,|\gamma^{\mathrm{ref}}|$ only,  
such that for any curve $\gamma\in\e$ with $d_{L^1}(\gamma,\gamma^{\mathrm{ref}})\geq \epsilon^2$ the following holds:
\begin{align}
\max_{i\in\{1,2\}}\big|\big<\Gamma,J_i\big>-\big<\Gamma^{\mathrm{ref}},J_i\big>\big|
\geq 
\epsilon^2/4,
\quad \text{where}\quad \big<\Gamma,J_i\big>:=\int_{\Gamma}J_i(u,v)\, du\, dv
.
\end{align}
\end{lemm}
Let us momentarily admit Lemma~\ref{lemm_volume_diff_as_test_function}, established at the end of the section, and prove~\eqref{eq_short_time_stability_partant_de_Gamma_0}. 
Let $\epsilon>0$, let $t^{\,\text{st}}_\epsilon$ denote the first time $t\geq 0$ such that:
\begin{align}
d_{L^1}\big(\gamma^N_{t},\gamma^{\mathrm{ref}}\big)
\geq 
\epsilon^2
.
\end{align}
Take $\epsilon\in(0,r_0)$ so that trajectories take values in $\e$ at least until time $t^{\,\text{st}}_\epsilon$. 
Introduce the dynamics $\Prob^{N,st}_{\beta,H}$, corresponding to $\Prob^N_{\beta,H}$, but stopped at time $t^{\,\text{st}}_\epsilon$. Then, for each $t\geq 0$:
\begin{align}
\Prob^N_{\beta,H}\big(t^{\,\text{st}}_\epsilon\leq t\big) &= \Prob^{N,st}_{\beta,H}\big(t^{\,\text{st}}_\epsilon\leq t\big) \nonumber\\
&= \Prob^{N,st}_{\beta,H}\Big(\gamma^N_\cdot\in E([0,t],\e)\cap\big\{t^{\,\text{st}}_\epsilon\leq t\big\}\Big)\nonumber\\
&= \Prob^{N,st}_{\beta,H}\Big(\gamma^N_\cdot\in E([0,t],\e)\cap\Big\{\sup_{t'\leq t}|\gamma^N_{t'}|\leq \kappa_H\Big\}\cap \big\{t^{\,\text{st}}_\epsilon\leq t\big\}\Big) + o_N(1)\label{eq_first_estimate_proba_tau_epsilon}.
\end{align}
The last equality follows from~\eqref{eq_control_length_lemm_e_typical}. 
By Lemma~\ref{lemm_volume_diff_as_test_function}, there are functions $J_1,J_2\in C^2_c(\R^2,[0,1])$ depending only on $\gamma^{\mathrm{ref}},\epsilon,\kappa_H$, such that:
\begin{align}
&\Prob^{N,st}_{\beta,H}\Big(\gamma^N_\cdot\in E([0,t],\e)\cap\Big\{\sup_{t'\leq t}|\gamma^N_{t'}|\leq \kappa_H \Big\}\cap \big\{t^{\,\text{st}}_\epsilon\leq t\big\}\Big) 
\\
&\quad\leq 
\Prob^{N,st}_{\beta,H}\Big(\gamma^N_\cdot\in E([0,t],\e)\cap\Big\{\sup_{t'\leq t}|\gamma^N_{t'}|\leq \kappa_H\Big\}\cap \Big\{\max_{i\in\{1,2\}}\sup_{t'\leq t}\big|\big<\Gamma^N_{t'},J_i\big>-\big<\Gamma^{\mathrm{ref}},J_i\big>\big|\geq \epsilon^2/4\Big\}\Big)
\nonumber\\
&\quad 
\leq 
\Prob^{N}_{\beta,H}\Big(\gamma^N_\cdot\in E([0,t],\e)\cap\Big\{\sup_{t'\leq t}|\gamma^N_{t'}|\leq \kappa_H\Big\}\cap \Big\{\max_{i\in\{1,2\}}\sup_{t'\leq t}\big|\big<\Gamma^N_{t'},J_i\big>-\big<\Gamma^{\mathrm{ref}},J_i\big>\big|\geq \epsilon^2/4\Big\}\Big)
.
\nonumber
\end{align}
To estimate the last probability, let us write, for each $i\in\{1,2\}$ and $t'\leq t$:
\begin{align}
\big<\Gamma^N_{t'},J_i\big>-\big<\Gamma^{\mathrm{ref}},J_i\big> = \frac{1}{N}\log A^{J_i}_{t'} + \frac{1}{N}\log D^{N,t'}_{\beta,H,J_i}\big((\gamma^N_u)_{u\leq t'}\big),
\end{align}
with $D^{N,t'}_{\beta,H,J_i} := \frac{\mathrm{d}\Prob^N_{\beta,H+J_i}}{\mathrm{d}\Prob^N_{\beta,H}}\Big|_{t'}$ the Radon-Nikodym derivative up to time $t'$ defined as in~\eqref{eq_def_der_radon_nyk} (setting $J_i(u,\cdot) = J_i(\cdot)$ for each $u\geq 0$) and:
\begin{align}
\forall i\in\{1,2\},\forall t'\leq t,\qquad 
\log A^{J_i}_{t'} 
:= 
\int_0^{t'}N^2e^{-N\langle\Gamma^N_u,J_i\rangle}\lcal_{\beta,H} e^{N\langle\Gamma^N_u,J_i\rangle}\, du
.
\end{align}
For trajectories with values in $\e$, Corollary~\ref{coro_bound_RD} can be used to estimate $A^{J_i}_{\cdot}$ (the bound in Corollary~\ref{coro_bound_RD} is for $\mathcal L_{\beta}$ rather than $\mathcal L_{\beta,H}$ but this latter case is identical): 
for some $C_i(H,\epsilon,|\gamma^{\mathrm{ref}}|)>0$,
\begin{align}
\forall \gamma^N_\cdot \in E([0,T],\e),
\qquad
\sup_{t'\leq t}\frac{1}{N}\big|\log A^{J_i}_{t'}(\gamma^N_\cdot)\big|
\leq 
C_i(H,\epsilon,|\gamma^{\mathrm{ref}}|)\int_0^t|\gamma^N_t|\, dt
.
\end{align}
As a result, for each $i\in\{1,2\}$, one has:
%
\begin{align}
&\Prob^{N}_{\beta,H}\Big(\gamma^N_\cdot\in E([0,t],\e)
\cap\Big\{\sup_{t'\leq t}|\gamma^N_{t'}|\leq \kappa_H\Big\}
\cap \Big\{\sup_{t'\leq t}\big|\big<\Gamma^N_{t'},J_i\big>-\big<\Gamma^{\mathrm{ref}},J_i\big>\big|\geq \epsilon^2/4\Big\}\Big)
\nonumber\\
&\ \,
= 
\Prob^{N}_{\beta,H}\Big(\gamma^N_\cdot\in E([0,t],\e)\cap\Big\{\sup_{t'\leq t}|\gamma^N_{t'}|\leq \kappa_H\Big\}
\cap \Big\{\sup_{t'\leq t}D^{N,t'}_{\beta,H,J_i}\geq e^{N\epsilon^2/4-C_i(H,\epsilon,|\gamma^{\mathrm{ref}}|)\, \kappa_Ht}\Big\}\Big)
.
\end{align}
Since $(D^{N,t'}_{\beta,H,J_i})_{t'\leq t}$ is a mean-1 martingale under $\Prob^N_{\beta,H}$, Doob's maximal inequality yields:
\begin{align}
\Prob^N_{\beta,H}\Big( \sup_{t'\leq t}D^{N,t'}_{\beta,H,J_i}\geq e^{N\epsilon^2/4-C_i(H,\epsilon,|\gamma^{\mathrm{ref}}|)\kappa_Ht}\Big)
&\leq 
e^{-N\epsilon^2/4 + C_i(H,\epsilon,|\gamma^{\mathrm{ref}}|)\, \kappa_Ht}\E^{N}_{\beta,H}\big[D^{N,t}_{\beta,H,J_i}\big] 
\nonumber\\
&= 
e^{-N\epsilon^2/4 + C_i(H,\epsilon,|\gamma^{\mathrm{ref}}|)\, \kappa_Ht}
.
\end{align}
Equation~\eqref{eq_first_estimate_proba_tau_epsilon} and the last inequality imply that $\Prob^N_{\beta,H}(t^{\,\text{st}}_\epsilon\leq t)$ vanishes with $N$ as soon as ${t<\epsilon^2/(4\kappa_H\max\{C_1(H,\epsilon,|\gamma^{\mathrm{ref}}|),C_2(H,\epsilon,|\gamma^{\mathrm{ref}}|)\})}$.  
This proves~\eqref{eq_short_time_stability_partant_de_Gamma_0} admitting Lemma~\ref{lemm_volume_diff_as_test_function}. 

To prove~\eqref{eq_stability_e_voisinage_Gamma_0}, 
notice that the functions $J_1,J_2$ built below in Lemma~\ref{lemm_volume_diff_as_test_function} for $\gamma^{\mathrm{ref}}$ can be built similarly for any $\gamma\in\e$ with $|\gamma|\leq \kappa$. 
The only change is that the constants $C_i(H,\epsilon)$ appearing above now depend on $\kappa$. 
\\

\noindent\emph{Proof of Lemma~\ref{lemm_volume_diff_as_test_function}.} Let $\epsilon>0$ and $\gamma\in\e$ be such that:
\begin{align}
d_{L^1}(\gamma,\gamma^{\mathrm{ref}})\geq 
\epsilon^2
.
\end{align}
Let $\Gamma$ be the associated droplet. One has:
\begin{align}
\text{either}\quad |\Gamma\setminus \Gamma^{\mathrm{ref}}|\geq \epsilon^2/2,
\qquad \text{or}\quad
|\Gamma^{\mathrm{ref}}\setminus\Gamma|\geq \epsilon^2/2
.
\end{align}
Let $J_1\in C^2_c(\R^2,[0,1])$ be a smooth approximation of ${\bf 1}_{\Gamma^{\mathrm{ref}}}$ equal to $1$ on $\Gamma^{\mathrm{ref}}$. 
Let also $J_2\in C^2_c(\R^2,[0,1])$ be a smooth approximation of ${\bf 1}_{(\Gamma^{\mathrm{ref}})^c}$ supported on the interior of $(\Gamma^{\mathrm{ref}})^c$. 
We claim that $J_1,J_2$ can be chosen as functions of $\Gamma^{\mathrm{ref}}$ and $\epsilon$ in such a way that:
\begin{align}
|\Gamma\setminus \Gamma^{\mathrm{ref}}|\geq \epsilon^2/2\quad \Rightarrow\quad \big<\Gamma,J_2\big>-\big<\Gamma^{\mathrm{ref}},J_2\big>
\geq 
\epsilon^2/4,
\nonumber\\
|\Gamma^{\mathrm{ref}}\setminus \Gamma|\geq \epsilon^2/2
\quad \Rightarrow\quad 
\big<\Gamma^{\mathrm{ref}},J_1\big> - \big<\Gamma,J_1\big>\geq \epsilon^2/4
.
\end{align}
Indeed, consider e.g. $J_1$. 
Due to the fact that $\gamma^{\mathrm{ref}}\in\Omega$ can be split into four pieces by Definition~\ref{property_state_space} of $\Omega$, 
for any curve $\gamma$ at Hausdorff distance at most $\eta>0$ from $\gamma^{\mathrm{ref}}$, 
there is a universal $C>0$ such that the distance $d_{L^1}(\gamma,\gamma^{\mathrm{ref}})$ is bounded by $C\eta|\gamma^{\mathrm{ref}}|$. 
It is therefore enough to ask for $J_1:\R^2\to[0,1]$ to be supported on points at distance at most $\epsilon^{2}/(4C|\gamma^{\mathrm{ref}}|)$ from $\Gamma^{\mathrm{ref}}$, 
as this implies:
\begin{equation}
\big<\Gamma^{\mathrm{ref}},J_1\big> - \big<\Gamma,J_1\big>
\geq 
\big|\Gamma^{\mathrm{ref}}\setminus\Gamma\big|
-\big|\Gamma\setminus \Gamma^{\mathrm{ref}}\big|
\geq 
\frac{\epsilon^2}{2}
- \frac{\epsilon^2}{4}
=
\frac{\epsilon^2}{4}
.
\end{equation}
$J_2$ is treated similarly, and we may choose the gradients of $J_1,J_2$ to be bounded in terms of $\epsilon,|\gamma^{\mathrm{ref}}|$ only.
This concludes the proof of Lemma~\ref{lemm_volume_diff_as_test_function}, thus of Lemma~\ref{lemm_e_typical_short_time}.
\end{proof}

\subsubsection{Short-time hydrodynamics}\label{sec_short_time_hydrodynamics}
We can now prove hydrodynamics for short time, under an additional uniqueness assumption.
\begin{prop}[Short time hydrodynamics]\label{prop_short_time_hydro}
Let $\tilde r\in[0,r_0]$, 
so that (recall Definition~\ref{def_effective_state_space} of $\e$):
\begin{align}
B_{d_{L^1}}\big(\gamma^{\mathrm{ref}},\tilde r^2\big)\subset \e.\label{eq_def_tilde_r}
\end{align}
Let also $t_0= t_0(\beta,H,\tilde r,\kappa)\in (0,T]$ be the time of Lemma~\ref{lemm_e_typical_short_time} with a large enough $\kappa>0$ in terms of $H,\gamma^{\mathrm{ref}}$. 
Assume that the weak formulation~\eqref{eq_formulation_faible_avec_der_en_temps} with drift $H$ has a unique solution $\gamma^H_\cdot\in E([0,t_0],\e)$ starting from the initial condition $\gamma^{\mathrm{ref}}$ of Definition~\ref{def_CI} (in other words, with the notations of Theorem~\ref{theo_large_dev}, 
assume $\gamma^H_{\cdot}\in \mathcal A_{\beta,t_0}$). 
Let $(\mu^N)_N$ be a sequence of probability measures on $\big(\e,d_{L^1}\big)$, converging weakly to $\delta_{\gamma^{\mathrm{ref}}}$ and such that:
\begin{equation}
\lim_{N\rightarrow\infty}\mu^N\big(|\gamma^N|\geq \kappa\big) = 0.\label{eq_bound_length_support_mu_N}
\end{equation}
Then:
\begin{equation}
\forall \zeta>0,\qquad \lim_{N\rightarrow\infty} \Prob^{\mu^N}_{\beta,H}\bigg( \gamma^N_\cdot\notin B_{d_E}\big((\gamma^H_{t})_{t\leq t_0},\zeta\big)\bigg) = 0.\label{eq_dans_prop_short_time_hydro}
\end{equation}
\end{prop}
Proposition~\ref{prop_short_time_hydro} is implied by the following lemma, in which hydrodynamics for $(\Prob^N_{\beta,H})_N$ are established in short time thanks to Lemma~\ref{lemm_e_typical_short_time} and the uniqueness assumption on solutions of~\eqref{eq_formulation_faible_avec_der_en_temps}. 
The proof of Lemma~\ref{lemm_large_dev_Q_N_H} in particular contains the proof of Proposition~\ref{prop_limite_hydro}.
\begin{lemm}\label{lemm_large_dev_Q_N_H}
Let $\kappa,t_0$ and the sequence $(\mu^N)_N$ be as in Proposition~\ref{prop_short_time_hydro}. 
Then $(\Prob^N_{\beta,H})_N$ converges to $\delta_{(\gamma^H_t)_{t\leq t_0}}$ in the weak topology of probability measures on $(E([0,t_0],\Omega),d_E)$, with $\Omega$ the general state space given in Definition~\ref{def_state_space}. 
In particular~\eqref{eq_dans_prop_short_time_hydro} holds.
\end{lemm}
\begin{proof}
The complementary of the open ball $B_{d_E}((\gamma^H_t)_{t\leq t_0},\zeta)$ is closed in $E([0,t_0],\Omega)$ for each $\zeta>0$. 
Equation~\eqref{eq_dans_prop_short_time_hydro} is therefore a direct consequence of the weak convergence result, which we now prove.

The hypothesis~\eqref{eq_bound_length_support_mu_N} on the initial law $\mu^N$ ensures, by Lemma~\ref{lemm_e_typical_short_time}, that:
\begin{equation}
\lim_{N\rightarrow\infty}\Prob^{\mu^N}_{\beta,H}\Big(\gamma^N_\cdot \in E([0,t_0],\e)\Big) = 1.
\end{equation}
Under this condition, in Appendix~\ref{appen_tightness} (see Corollary~\ref{coro_tightness_small_time}), the sequence $\{\Prob^{\mu^N}_{\beta,H}: N\in \N_{\geq 1}\}$ is proven to be relatively compact in $E([0,t_0],\Omega)$, with limit points supported on $E([0,t_0],\e)$ that are continuous in $d_{L^1}$ distance. Let $\Prob^*_{\beta,H}$ be one of its limit points. In view of~\eqref{eq_proba_sous_exp_D_n_q} and Lemma~\ref{lemm_e_typical_short_time}, 
$\Prob^*_{\beta,H}$ is supported on trajectories starting from a curve $\gamma^0$ with $d_{L^1}(\gamma^0,\gamma^{\mathrm{ref}})=0$ and with almost always point-like poles (for each $q\in\N_{\geq 1}$, they are in the set $D_q$ defined as in~\eqref{eq_def_D_p} with $t_0$ instead of $T$). 

To prove that $\Prob^*_{\beta,H} = \delta_{(\gamma^H_t)_{t\leq t_0}}$, let us prove that $\Prob^*_{\beta,H}$ concentrates on trajectories that satisfy the weak formulation~\eqref{eq_formulation_faible_avec_der_en_temps} of anisotropic motion by curvature on $[0,t_0]$. 
This is sufficient to conclude the proof of Lemma~\ref{lemm_large_dev_Q_N_H}, because we have assumed that $(\gamma^H_{t})_{t\leq t_0}$ is the unique solution of~\eqref{eq_formulation_faible_avec_der_en_temps} on $[0,t_0]$.

To prove this concentration property, 
the standard idea (see e.g. Chapter 4 in \cite{Kipnis1999}) is to start from the following semi-martingale representation: 
if $t\geq 0$, $G\in\C$ and $(\Gamma^N_t)_{t\geq 0}$ is as usual the trajectory of droplets associated with microscopic curves $(\gamma^N_t)_{t\geq 0}$,
\begin{align}
\big<\Gamma^N_t,G_t\big> 
= 
\big<\Gamma^N_0,G_0\big>  + \int_0^t \big<\Gamma_u,\partial_u G_u\big>\, du +  \int_0^t N^2\lcal_{\beta,H}\big<\Gamma^N_u,G_u\big>\, du + M^{N,G}_t
,
\end{align}
where $(M^{N,G}_t)_t$ is a martingale. 
By assumption, $\gamma^N_0$ converges in $d_{L^1}$-distance to the initial condition $\gamma^{\mathrm{ref}}$ of the trajectory $\gamma^H_\cdot$. 
Computing the action of the generator $\lcal_{\beta,H}$, (defined with the jump rates $c^H$ of~\eqref{eq_def_jump_rates_H}) 
is done in exactly the same way as the computation of the Radon-Nikodym derivative in Section~\ref{sec_action_gen_sur_volume}. 
This gives, for each $\zeta>0$ and $G\in\C$, the existence of $\epsilon_0(\zeta)>0$ with:
\begin{equation}
\forall\epsilon\in(0,\epsilon_0(\zeta)),\forall t\leq t_0,
\qquad
\lim_{N\to\infty}\Prob^N_{\beta,H}\big(|X_{G,\epsilon,t}|\leq \zeta\big) 
= 
1
,
\label{eq_limit_proba_weak_formulation}
\end{equation}
where for $\gamma_\cdot\in E([0,t_0],\e)$, and $t\leq t_0$, 
recalling the definition of $\alpha,\mu$ in~\eqref{eq_def_mobility_first}--\eqref{eq_def_alpha}:
\begin{align}
X_{G,\epsilon,t}(\gamma_\cdot)
&=
\big<\Gamma_{t},G_t\big>-\big<\Gamma^{\mathrm{ref}},G_0\big> - \int_0^t \big<\Gamma_{t'},\partial_{t'}G_{t'}\big>\, dt'
\nonumber\\
&\quad
- \frac{1}{4}\int_0^{t}\int_{\gamma_{t'}(\epsilon)}\frac{(\mathsf{v}^\epsilon)^2}{\mathsf{v}}\big[{\bf T}^\epsilon\cdot {\bf m}\big(\gamma_{t'}(s)\big)\big] \ {\bf T}^\epsilon\cdot \nabla G\big(t,\gamma_{t'}(s)\big)\, ds \, dt'
\nonumber\\
&\quad 
+ \sum_{k=1}^4\int_0^{t}\Big(\frac{1}{4}-\frac{e^{-\beta}}{2}\Big)\big[G({t'},L_k(\gamma_{t'})) +G({t'},R_k(\gamma_{t'}))\big] \, d{t'}
\nonumber\\
&\quad
-\int_0^{t} \int_{\gamma_{t'}(\epsilon)} \frac{(\mathsf{v}^\epsilon)^2}{\mathsf{v}}|{\bf T}^\epsilon\cdot {\bf b}_0||{\bf T}^\epsilon \cdot{\bf b}_{\pi/2}| (HG)\big(t,\gamma_{t'}(s)\big)^2\, ds\, dt'
.
\end{align}
In view of the $\epsilon\to 0$ convergence results of Proposition~\ref{prop_preuve_continuite_integrale_avec_v_moins_un} (see more precisely the proof in Section~\ref{sec_continuity_J_H_eps_ell_H_eps}), 
if we can prove~\eqref{eq_limit_proba_weak_formulation} also holds under $\Prob^*_{\beta,H}$, then we are done. 
This is not immediate because $X_{G,\epsilon,t}$ is not continuous on $E([0,t],\Omega)$ due to the poles. 
Recall however that $\Prob^*_{\beta,H}$ is supported on the closed set $D_q$ (defined as in~\eqref{eq_def_D_p} with $t_0$ instead of $T$ there) 
for all $q$ larger than some $q(H,t_0)$ as follows from the estimate~\eqref{eq_proba_sous_exp_D_n_q} and Corollary~\ref{coro_change_measure_sous_exp} to transfer this estimate to $\Prob^N_{\beta,H}$. 
Trajectories in $D_q$ have almost always point-like poles and are therefore continuity points of $X_{G,\epsilon,t}$ by Proposition~\ref{prop_preuve_continuite_integrale_avec_v_moins_un}. 
For $q\geq q(H,t_0)$ henceforth fixed, 
this implies:
\begin{equation}
\overline{\big\{|X_{G,\epsilon,t}|\leq \zeta\big\}} \cap D_q
=
\big\{|X_{G,\epsilon,t}|\leq \zeta\big\}\cap D_q
,
\end{equation}
where $\overline U$ denotes the closure of a set $U\subset E([0,t_0],\Omega)$ for $d_E$ (defined as in~\eqref{eq_def_d_E} with $t_0$ instead of $T$).  
Thus, for each $\epsilon\in(0,\epsilon_0)$ and each $t\leq t_0$:
\begin{align}
1 
=
\limsup_{N\to\infty}\Prob^N_{\beta,H}\Big(\overline{\big\{|X_{G,\epsilon,t}|\leq \zeta\big\}}\Big) 
&\leq 
\Prob^*_{\beta,H}\Big(\overline{\big\{|X_{G,\epsilon,t}|\leq \zeta\big\}}\Big) 
=
\Prob^*_{\beta,H}\Big(\overline{\big\{|X_{G,\epsilon,t}|\leq \zeta\big\}}\cap D_q\Big) 
\nonumber\\
&= 
\Prob^*_{\beta,H}\big(|X_{G,\epsilon,t}|\leq \zeta\big) 
.
\end{align}
This concludes the proof of Lemma~\ref{lemm_large_dev_Q_N_H}, 
thus of Proposition~\ref{prop_short_time_hydro}.
\end{proof}
\subsubsection{Extension to later times}\label{sec_up_to_time_T0}
In this section, 
we extend the hydrodynamic limit result of Proposition~\ref{prop_short_time_hydro} to the whole time interval $[0,T]$ on which $\gamma^H_\cdot$ is assumed to take values in the interior of $\e$ in the sense of~\eqref{eq_gamma_H_in_interior}. 
This concludes the proof of Proposition~\ref{prop_control_proba_lower_bound}, 
which together with Lemma~\ref{lemm_getting_rid_of_tilde_Z} and Proposition~\ref{prop_lower_bound} concludes the proof of lower bound large deviations in Theorem~\ref{theo_large_dev}. 
\begin{prop}\label{prop_time_up_to_T0}
Let $T>0$ and assume that there is a unique solution $(\gamma^H_{t})_{t\leq T}\in E([0,T],\e)$ of the weak formulation~\eqref{eq_formulation_faible_avec_der_en_temps} of anisotropic motion by curvature with drift $H$. 
Assume that $\gamma^H_\cdot$ stays in the interior of $\e$ until time $T$ in the following sense:
\begin{equation}
\exists r_H>0,\forall t\in[0,T], \qquad B_{d_{L^1}}\big(\gamma^H_t,r_H^2\big) \subset \e. \label{eq_condition_interieur_de_e}
\end{equation}
Then:
\begin{align}
\forall \zeta>0,\qquad \lim_{N\rightarrow\infty}\Prob^N_{\beta,H}\Big(\gamma^N_\cdot\in E([0,T],\e)\cap B_{d_E}(\gamma^H_\cdot,\zeta)\Big) = 1.
\end{align}
\end{prop}
\begin{rmk}[Lower bound in Theorem~\ref{theo_large_dev_general}]
Proposition~\ref{prop_time_up_to_T0} states that, 
if $\gamma^H_\cdot$ stays inside $\e$ up to a time $T$ larger than the time $t_0$ of Lemma~\ref{lemm_e_typical_short_time}, 
then hydrodynamics are valid up to time $T$. 

However, the only properties of $\e$ that were used are 1) that curves in $\e$ satisfy Property~\ref{prop_IC} which enables us to compute Radon-Nikodym derivatives and 2) that if a curve is in $\e$, then all curves in a small volume neighbourhood satisfy Property~\ref{prop_IC} (see Lemma~\ref{lemm_local_jump_rates}). 
Thus Proposition~\ref{prop_time_up_to_T0} extends to any trajectory satisfying Property~\ref{prop_IC} at each time. 
This in particular implies the general lower bound in Theorem~\ref{theo_large_dev_general}.
\demo
\end{rmk}
\begin{proof}
The claim of Proposition~\ref{prop_time_up_to_T0} is proven in Proposition~\ref{prop_short_time_hydro} up to the time $t_0(\beta,H,r_H,|\gamma^{\mathrm{ref}}|)\in(0,T]$ of Lemma~\ref{prop_short_time_hydro}. 
The point here is to show that the result holds up to time $T$. 
To do so, we iterate the results of Section~\ref{sec_short_time_hydrodynamics} on small time intervals, 
built so that the length and volume of the droplet stay well controlled on each of these intervals.

The intervals are built thanks to the condition~\eqref{eq_condition_interieur_de_e} and a bound on the length as follows. 
By Lemma~\ref{lemm_tightness_sup_length}, there is $\kappa_H>0$ such that, for each $t\leq T$:
\begin{align}
\lim_{N\rightarrow\infty}\Prob^N_{\beta,H}\Big(\gamma^N_\cdot\in E([0,t],\e)\cap\Big\{\sup_{t'\leq t}|\gamma^N_{t'}|\geq \kappa_H \Big\}\Big) 
= 
0
.
\label{eq_bound_length_time_up_to_T}
\end{align}
Recalling the times in Lemma~\ref{lemm_e_typical_short_time} and the definition of $r_H$ from~\eqref{eq_condition_interieur_de_e}, one can then consider intervals of length:
\begin{equation}
t_{H} 
:= 
\min\Big\{t(\beta,H,r_H,|\gamma^{\mathrm{ref}}|), t_0(\beta,H,r_{H},\kappa_{H})\Big\}
>
0
\label{eq_def_r_min_t_min}
\end{equation}
and apply Proposition~\ref{prop_short_time_hydro} on each of these intervals as we shall see. 
Let $n := \lfloor T/t_H\rfloor +1$ and let us prove by recursion on $1\leq i\leq n$ that hydrodynamics hold up to time $it_H$, i.e.:
\begin{equation}
\forall \zeta>0,\qquad \lim_{N\rightarrow\infty}\Prob^N_{\beta,H}\Big(\gamma^N_\cdot \in E([0,it_H],\e)\cap B_{d_E}\big((\gamma^H_t)_{t\leq it_H},\zeta\big)\Big) = 1.
\end{equation}
On $[0,t_H]$, Proposition~\ref{prop_short_time_hydro} yields:
\begin{align}
\forall \zeta>0,\qquad \lim_{N\rightarrow\infty}\Prob^N_{\beta,H}\Big(\gamma^N_\cdot \in E([0,t_H],\e)\cap B_{d_E}\big((\gamma^H_t)_{t\leq t_H},\zeta\big)\Big) = 1,
\end{align}
which is the $i=1$ claim. Assume the claim holds up to $i-1<n$. To prove that it holds at rank $i$, it is enough to prove:
\begin{equation}
\forall\zeta>0,\qquad \lim_{N\rightarrow\infty}\Prob^N_{\beta,H}\Big(E([(i-1)t_H,it_H],\e)\cap B_{d_E}\big((\gamma^H_t)_{(i-1)t_H\leq t\leq it_H},\zeta\big)\Big)=1.\label{eq_recursion_proof_prop_up_to_time_T}
\end{equation}
To prove the last equation, we would like to use the Markov property, then apply Proposition~\ref{prop_short_time_hydro}. To do so, we need to check that the law $\Prob^N_{\beta,H}(\gamma^N_{(i-1)t_H}\in\cdot)$ of $\gamma^N_{(i-1)t_H}$ concentrates on curves with bounded lengths in the sense of~\eqref{eq_bound_length_support_mu_N} and converges to $\delta_{\gamma^H_{(i-1)t_H}}$ for the weak topology associated with $d_{L^1}$. We prove it as follows. 

Equation~\eqref{eq_bound_length_time_up_to_T} applied to $t=(i-1)t_H$ and the recursion hypothesis bounds the length of supported trajectories:
\begin{align}
\lim_{N\rightarrow\infty}\Prob^N_{\beta,H}\Big(\sup_{t\leq (i-1)t_H}|\gamma^N_t|\geq \kappa_H \Big) = 0.
\end{align}
Moreover, as $\gamma^H_\cdot$ is continuous in $d_{L^1}$-distance, the mapping $\gamma_\cdot\in E([0,T],\e)\mapsto d_{L^1}(\gamma_t,\gamma^H_t)$ is continuous for the distance $d_E$ for each time $t\in[0,T]$. The hydrodynamic limit up to time $(i-1)t_H$, given by the recursion hypothesis, then yields the desired convergence (in fact also in probability rather than only weakly):
\begin{align}
\forall \eta>0,\qquad \lim_{N\rightarrow\infty}\Prob^N_{\beta,H}\Big(\gamma^N_{(i-1)t_H}\in B_{d_{L^1}}\big(\gamma^H_{(i-1)t_H},\eta\big)\Big) =1.
\end{align}
For short, write $\mu^N_{i-1}$ for the law of $\gamma^N_{(i-1)t_H}$. As a result of the last two estimates and the Markov property,~\eqref{eq_recursion_proof_prop_up_to_time_T} holds as soon as:
\begin{align}
\forall\zeta>0,\qquad &\lim_{N\rightarrow\infty}\Prob^{\mu^N_{i-1}}_{\beta,H}\Big(\gamma^N_\cdot \in E([0,it_H],\e)\cap B_{d_E}\big((\gamma^H_t)_{(i-1)t_H\leq t\leq it_H},\zeta\big)\Big) =1.\label{eq_extension_later_times_final}
\end{align}
By Assumption~\eqref{eq_condition_interieur_de_e}, $\gamma^H_{(i-1)t_H}$ satisfies the property~\eqref{eq_def_tilde_r} demanded of $\gamma^{\mathrm{ref}}$ in Proposition~\ref{prop_short_time_hydro}, with $\tilde r$ there replaced by $r_H$. We also just checked that the initial condition $\mu^N_{i-1}$ satisfies the same properties as the initial condition of Proposition~\ref{prop_short_time_hydro}, with $\gamma^{\mathrm{ref}}$ replaced by $\gamma^H_{(i-1)t_H}$. Proposition~\ref{prop_short_time_hydro} thus applies to prove~\eqref{eq_extension_later_times_final}. This completes the induction step and the proof of Proposition~\ref{prop_time_up_to_T0}.
\end{proof}
\section{Behaviour of the poles and ${\bf 1}_{p_k=2}$ terms }\label{app_behaviour_pole}
In this section, we focus on the specificity of the contour dynamics: the behaviour of the poles.  
There are two main results. The first is the control of the length of a curve, which is the first item of Proposition~\ref{prop_short_time_existence_of_weak_solution}. 
The second is the proof of Proposition~\ref{prop_value_slope_at_poles}, which states that the regrowth, $e^{-2\beta}$ term in the generator~\eqref{eq_def_generateur_H_is_0} can be seen as the action of a moving reservoir of particles, fixing the density of vertical edges in its vicinity (i.e. the tangent vector at each pole) in terms of $\beta$. 
This fact is proven in Subsection~\ref{subsec_slope_around_the_poles}. 
Preliminary estimates are established in Subsection~\ref{subsec_size_pole_local_eq} 
which presents a useful bijection argument, 
used to both bound the pole size and establish local equilibrium at the poles. 
In addition to being useful for Section~\ref{subsec_slope_around_the_poles}, 
these two results were used in Sections~\ref{sec_relevant_martingales}--\ref{sec_large_dev_upper_bound}. 
A bias $H\in\C$ is fixed throughout.
\subsection{Control of the length of a curve}
This section is devoted to the proof of the first item of Proposition~\ref{prop_short_time_existence_of_weak_solution}, i.e. the control of the supremum of the length of a trajectory. This estimate is central to the proof of large deviations: it enables one to prove, in Corollary~\ref{coro_change_measure_sous_exp}, that if an event has probability decaying super-exponentially fast under $\Prob^N_\beta$, then this remains true under the tilted dynamics $\Prob^N_{\beta,H}$, $H\in\C$.
\begin{lemm}\label{lemm_tightness_sup_length}
Let $\beta>\log 2$ and $T>0$. 
There is then $C(\beta)>0$ such that:
\begin{align}
\forall A>0,\qquad 
\limsup_{N\rightarrow\infty}\frac{1}{N}\log\Prob^N_{\beta}\Big(\sup_{t\leq T}|\gamma^N_t|\geq A\Big) 
\leq  
-C(\beta)A +|\gamma^{\mathrm{ref}}|\beta
.
\label{eq_controle_length_H_is0}
\end{align}
Moreover, take a bias $H\in\C$. There are constants $C(\beta,H,T),C(H)>0$ with $t\mapsto C(\beta,H,t)$ increasing, such that:
\begin{align}
\forall T,A>0,\qquad &\limsup_{N\rightarrow\infty}\frac{1}{N}\log\Prob^N_{\beta,H}\Big(E([0,T],\e)\cap \Big\{\sup_{t\leq T}|\gamma^N_t|\geq A\Big\}\Big) 
\nonumber\\
&\hspace{5cm}
\leq 
-C(\beta,H,T)A + |\gamma^{\mathrm{ref}}|\beta + C(H)
.
\label{eq_bound_sup_length_H}
\end{align}
\end{lemm}
\begin{proof}
The proof relies on the structure of the invariant measure. We start with the $H\equiv 0$ case.\\
 First, as $\beta>\log 2$, the partition function $\mathcal Z^N_\beta$ normalising $\nu^N_\beta$ (see~\eqref{eq_def_nu_beta}) is bounded. One has, for some $c^0>0$ such that $|\gamma^{N,0}|\leq |\gamma^{\mathrm{ref}}| +c^0/N$ for each $N$:
\begin{align}
\Prob^N_{\beta}\Big(\sup_{t\leq T}|\gamma^N_t|\geq A\Big) &\leq  \nu^N_\beta(\gamma^{\mathrm{ref},N})^{-1}\Prob^{\nu^N_\beta}_{\beta}\Big(\sup_{t\leq T}|\gamma^N_t|\geq A\Big)
\nonumber\\
&\leq 
\mathcal Z^N_\beta e^{\beta (N|\gamma^{\mathrm{ref}}| + c^0) } \Prob^{\nu^N_\beta}_{\beta}\Big(\sup_{t\leq T}|\gamma^N_t|\geq A\Big)
.
\end{align}
Let $b>0$ to be fixed later and split the time interval $[0,T]$ into $N^b$ slices of length $TN^{-b}$ to obtain, using the invariance of $\nu^N_\beta$:
\begin{align}
\Prob^N_{\beta}\Big(\sup_{t\leq T}|\gamma^N_t|\geq A\Big)\leq  \mathcal Z^N_\beta N^b e^{\beta N(|\gamma^{\mathrm{ref}}|+c^0)} \Prob^{\nu^N_\beta}_{\beta}\Big(\sup_{t\leq N^{-b}T}|\gamma^N_t|\geq A\Big).\label{eq_to_estimate_sup_length_His0}
\end{align}
To estimate the last probability, let us decompose $|\gamma^N_\cdot|$: for each $t\geq 0$,
\begin{equation}
|\gamma^N_t| 
= 
|\gamma^N_0| + \frac{1}{N}\log V_t+ \frac{1}{N}\log D_t
,
\label{eq_mart_decomp_length}
\end{equation}
where $D_\cdot$ is a mean-$1$ exponential martingale and $V_\cdot$ is the finite variation process given by:
\begin{align}
\forall t\geq 0,\qquad \frac{1}{N}\log V_t &:= \frac{1}{N}\int_0^t e^{-N|\gamma^N_s|}N^2\lcal_{\beta}e^{N|\gamma^N_s|}dt \label{eq_def_V_length}\\
&= N\int_0^t \sum_{k=1}^4 \big[e^{-2\beta}(p_k(\gamma^N_u)-1)\big(e^2-1\big) + {\bf 1}_{p_k(\gamma^N_u)=2}{\bf 1}_{(\gamma^N)^{-,k}\in\Omega_{\text{mic}}^N}\big(e^{-2}-1)\big]du.\nonumber
\end{align}
To estimate the probability in~\eqref{eq_to_estimate_sup_length_His0}, 
it is enough to separately estimate the probability that $|\gamma^N_0|\geq A$ and the probability of the suprema of each of the other two terms in~\eqref{eq_mart_decomp_length}. 
Let us start with $|\gamma^N_0|$. Since the number of curves in $\Omega^N_{\text{mic}}$ with $n\in\N_{\geq 1}$ edges is bounded by $cn^42^n$ for some $c>0$, the following equilibrium estimate holds:
\begin{equation}
\nu^N_\beta\Big(|\gamma^N|\geq A\Big) \leq \frac{1}{\mathcal Z^N_\beta} \sum_{n\geq AN}cn^4 2^n e^{-\beta n} = O\big(e^{-AN\beta'}\big),\quad 0<\beta'<\beta-\log 2.\label{eq_bound_gamma_0}
\end{equation}
Consider now the finite variation term~\eqref{eq_def_V_length}. Bounding each $p_k$ by $CN$ for some $C>0$ and using Chebychev inequality to obtain the second line below, we find:
\begin{align}
\Prob^{\nu^N_\beta}_\beta\Big(\sup_{t\leq N^{-b}T}\frac{1}{N}\log V_t\geq A\Big) &\leq \Prob^{\nu^N_\beta}_\beta\Big(\frac{N^2\big(e^2-1\big)}{2}\int_0^{TN^{-b}}\sum_{k=1}^4 p_k(\gamma^N_s)\, ds \geq \frac{ANe^{2\beta}}{4}\Big)\nonumber\\
&\leq \exp\Big[-\frac{ANe^{2\beta}}{4(e^2-1)}\Big] e^{CTN^{3-b}/2}\nonumber\\
&\leq e^{CT/2}e^{-ANe^{2\beta}(e^2-1)^{-1}/4}\quad \text{ for }b\geq 3.\label{eq_bound_V_length}
\end{align}
Consider finally the martingale term $\frac{1}{N}\log D_t$ in~\eqref{eq_mart_decomp_length}. As $D_t$ is a mean-1 positive martingale, Doob's martingale inequality gives:
\begin{equation}
\Prob^{\nu^N_\beta}_\beta\Big(\sup_{t\leq N^{-b}T}\frac{1}{N}\log D_t \geq A\Big)\leq e^{-AN}\E^{\nu^N_\beta}_\beta\big[D_{N^{-b}T}\big] = e^{-AN}.\label{eq_bound_martingale_length}
\end{equation}
Putting~\eqref{eq_bound_gamma_0}--\eqref{eq_bound_V_length}--\eqref{eq_bound_martingale_length} together yields the claim of Lemma~\ref{lemm_tightness_sup_length} when $H\equiv 0$.\\

Take now $H\in\C$ and let us prove~\eqref{eq_bound_sup_length_H}. 
Recall from Corollary~\ref{coro_bound_RD} that there is $C(H)>0$ such that the Radon-Nikodym derivative $D^N_{\beta,H} = \mathrm{d}\Prob^N_{\beta,H}/\mathrm{d}\Prob^N_\beta|_T$ until time $T$ satisfies, for each $T>0$:
\begin{align}
\forall \gamma^N_\cdot\in E([0,T],\e),\qquad \log D^N_{\beta,H}(\gamma^N_\cdot) \leq \exp\Big[C(H)N+ C(H)N\int_0^{T}|\gamma^N_t|\, dt\Big].\label{eq_bound_RD_dans_coro_sous_exp}
\end{align}
We will prove the following: until time $t_{*} = t_{*}(\beta,H) := C(\beta)/(2C(H))$, for each $A>0,\kappa>0$ and each initial condition $\gamma^N_0\in \e$ with length bounded by $\kappa$:
\begin{align}
\limsup_{N\rightarrow\infty}\frac{1}{N}\log\Prob^N_{\beta,H}\Big(E([0,t_{*}],\e)\cap\Big\{\sup_{t\leq t_{*}}|\gamma_t^N|\geq A\Big\}\Big) 
\leq 
-\frac{C(\beta)A}{2}+\kappa\beta+C(H)
.
\label{eq_small_length_until_t_gamma}
\end{align}
Assuming~\eqref{eq_small_length_until_t_gamma}, 
let us conclude the proof of Lemma~\ref{lemm_tightness_sup_length}. Let $n_*\in\N$ be such that $(n_*-1)t_{*}\leq T\leq n_* t_{*}$ and define a sequence $(b_i)_{i\in\N}$ through:
\begin{align}
b_{0}=0,\qquad \frac{C(\beta)b_{i+1}}{2} = C(\beta)+ b_{i}(\beta+\log 2),\quad i\in\N.
\end{align}
Write then:
\begin{align}
&\Prob^N_{\beta,H}\Big(E([0,T],\e)\cap\Big\{\sup_{t\leq T}|\gamma^N_t|\geq b_{n_*}A\Big\}\Big) \nonumber\\
&\quad\leq  \Prob^N_{\beta,H}\Big(\gamma_{(n_*-1)t_{*}}\in \e,\sup_{t\leq n_{*}t_{*}}|\gamma^N_t|\geq b_{n_*}A, \sup_{t\leq (n_*-1)t_{*}}|\gamma^N_t|< b_{n_*-1}A\Big) \nonumber\\
&\hspace{2cm} +  \Prob^N_{\beta,H}\Big(E([0,(n_*-1)t_{*}],\e)\cap\Big\{\sup_{t\leq (n_*-1)t_{*}}|\gamma^N_t|\geq b_{n_*-1}A\Big\}\Big) \nonumber\\
&\quad\leq \sum_{i=2}^{n_*} \Prob^N_{\beta,H}\Big(\gamma^N_{(i-1)t_{*}}\in \e, \sup_{t\leq it_{*}}|\gamma^N_t|\geq b_{i}A, \sup_{t\leq (i-1)t_{*}}|\gamma^N_t|< b_{i-1}A\Big) \nonumber\\
&\hspace{2cm} + \Prob^N_{\beta,H}\Big(E([0,t_{*}],\e)\cap\Big\{ \sup_{t\leq t_{*}}|\gamma^N_t|\geq b_{1}A\Big\}\Big).
\end{align}
The last probability is estimated by~\eqref{eq_small_length_until_t_gamma}. On the other hand, apply Markov inequality to each term of the sum to find, for $2\leq i \leq n_*$:
\begin{align}
&\Prob^N_{\beta,H}\Big(\gamma^N_{(i-1)t_{*}}\in \e, \sup_{t\leq it_{*}}|\gamma^N_t|\geq b_{i}A,  \sup_{t\leq (i-1)t_{*}}|\gamma^N_t|< b_{i-1}A\Big)\nonumber\\
&\qquad \leq \sup_{\substack{\gamma^N\in \Omega^N_{\text{mic}}\cap \e \\ |\gamma^N|<b_{i-1}A}}\Prob^{\gamma^N}_{\beta,H}\Big(\sup_{t\leq t_{*}}|\gamma^N_t|\geq b_{i}A\Big).
\end{align}
Using~\eqref{eq_small_length_until_t_gamma} and the fact that there is $c>0$ such that the number of curves with $\ell$ edges in $\Omega^N_{\text{mic}}$ is less than $c\ell^4 2^{\ell}$, we find for $2\leq i\leq n_*$:
\begin{align}
&\limsup_{N\rightarrow\infty}\frac{1}{N}\log\Prob^N_{\beta,H}\Big( \sup_{t\leq it_{*}}|\gamma^N_t|\geq b_{i}A, \sup_{t\leq (i-1)t_{*}}|\gamma^N_t|< b_{i-1}A\Big)
\nonumber\\
&\hspace{3cm}
\leq 
-\frac{C(\beta) b_{i}A}{2} + b_{i-1}(\beta+\log 2)A + C(H) 
= 
-C(\beta)A+C(H)
.
\end{align}
This estimate and~\eqref{eq_small_length_until_t_gamma} corresponding to $i=1$ prove~\eqref{eq_bound_sup_length_H} assuming~\eqref{eq_small_length_until_t_gamma}, with $C(\beta,H,T) := b_{n_*}^{-1}$. \\
Let us now prove the short time estimate~\eqref{eq_small_length_until_t_gamma}. Starting again from~\eqref{eq_bound_RD_dans_coro_sous_exp}, one has, for each $T>0$:
\begin{align}
\Prob^N_{\beta,H}\Big(\Big\{\sup_{t\leq T}|\gamma^N_t|\geq A\Big\}\cap E([0,T],\e)\Big) &\leq e^{C(H)N}\E^N_{\beta}\Big[e^{C(H)NT\sup_{t\leq T}|\gamma^N_t|}{\bf 1}_{\sup_{t\leq T}|\gamma^N_t|\geq A}\Big].
\end{align}
By~\eqref{eq_controle_length_H_is0}, this expectation reads:
\begin{align}
\E^N_{\beta}\Big[e^{C(H)NT\sup_{t\leq T}|\gamma^N_t|}{\bf 1}_{\sup_{t\leq T}|\gamma_t|\geq A}\Big] &\leq \int_{A}^{\infty} e^{C(H)N T\lambda }\Prob^N_{\beta}\Big(\sup_{t\leq T}|\gamma^N_t|\geq \lambda\Big)\,d\lambda
\nonumber\\
&\leq \int_{A}^\infty e^{C(H)NT\lambda -c(\beta)N\lambda +|\gamma^{\mathrm{ref}}|\beta N}\, d\lambda
.
\end{align}
Setting $t_{*}(\beta,H) := C(\beta)/(2C(H))$ concludes the proof of~\eqref{eq_small_length_until_t_gamma}, thus of Lemma~\ref{lemm_tightness_sup_length}.
\end{proof}
The following corollary explains how to use Lemma~\ref{lemm_tightness_sup_length} to argue that events with super-exponentially small probability under $\Prob^N_\beta$ also have super-exponentially small probability under the tilted dynamics $\Prob^N_{\beta,H}$ for $H\in\C$. 
Typical examples are the sets $(D_{n,q})_q$ for fixed $n$, see~\eqref{eq_def_D_p_n} and $Z=Z(\beta,H,\epsilon,\delta)$ for $\epsilon\leq\epsilon_0(\delta)$ and $\delta>0$, 
see~\eqref{eq_bound_proba_Z_H_epsilon_A_sec_4}.
\begin{coro}[Sub-exponential estimates for tilted dynamics]\label{coro_change_measure_sous_exp}
For a time $T>0$, let $(\chi_{A,T})_{A>0}\subset E([0,T],\Omega)$ be a family of sets such that, for some $C(\beta)>0$:
\begin{align}
\forall T>0,\forall A>0,\qquad \limsup_{N\rightarrow\infty}\frac{1}{N}\log \Prob^{\nu^N_\beta}_\beta\Big(\chi_{A,T}^c\cap E([0,T],\e)\Big)
\leq 
- C(\beta)A
.\label{eq_hyp_sur_chi_A_T}
\end{align}
Then, for each $H\in\C$ and each time $T>0$, there are constants $C(\beta,H,T),C(H)>0$ (different from those of Lemma~\ref{lemm_tightness_sup_length}) with $t\mapsto C(\beta,H,t)$ increasing, such that:
\begin{align}
\forall T,A>0,\quad 
\limsup_{N\rightarrow\infty}\frac{1}{N}\log\Prob^N_{\beta,H}\Big(\chi_{A,T}^c\cap E([0,T],\e)\Big) 
\leq 
-C(\beta,H,T)A + C(H) + |\gamma^{\mathrm{ref}}|\beta
.
\end{align}
\end{coro}
\begin{proof}
Let $A'>0$ and write first, using Corollary~\ref{coro_bound_RD}:
\begin{align}
\Prob^N_{\beta,H}\Big(\chi_{A,T}^c\cap E([0,T],\e)\Big) &\leq \E^{N}_{\beta}\Big[{\bf 1}_{\chi_{A,T}^c\cap E([0,T],\e)\cap \{\sup_{t\leq T}|\gamma^N_t|\leq A'\}}D^N_{\beta,H}\Big] + \Prob^N_{\beta,H}\Big(\sup_{t\leq T}|\gamma^N_t|\geq A'\Big)\nonumber\\
&\leq e^{C(H)N(1+A'T)}\Prob^N_\beta\Big(\chi_{A,T}^c\cap E([0,T],\e) \cap\big\{\sup_{t\leq T}|\gamma^N_t|<A'\big\}\Big) \nonumber\\
&\quad+ \Prob^N_{\beta,H}\Big(\sup_{t\leq T}|\gamma_t|\geq A'\Big).\label{eq_bound_length_prob_H_interm}
\end{align}
The first probability is controlled by~\eqref{eq_hyp_sur_chi_A_T}, the second by the tail estimates for the length obtained in Lemma~\ref{lemm_tightness_sup_length}, so we conclude the proof here.
\end{proof}
\subsection{Size of the poles and local equilibrium}\label{subsec_size_pole_local_eq}
In this section, we prove Lemmas~\ref{lemm_size_pole_ds_calcul_action_gen}--\ref{lemm_local_eq_sec_martingales}, 
i.e. we estimate the time integral of the number of blocks $p_1$ in the north pole $P_1$ and of the following term, for any test function $G\in\C$:
\begin{equation}
W^G_t(\gamma^N) 
:= 
\sum_{\substack{x\in P_1(\gamma^N)\\x+{\bf e}^\pm_x\in P_1(\gamma^N)}}\big({\bf 1}_{p_1(\gamma^N)=2,(\gamma^N)^{-,1}\in\Omega^N_{\text{mic}}} - e^{-2\beta}\big) G(t,x),
\qquad 
\gamma^N\in \Omega^N_{\text{mic}},
\ t\geq 0
.
\label{eq_terme_representing_local_eq_at_the_pole}
\end{equation}
\noindent \textbf{Notation:} As we only work with the north pole in the following, we drop the subscript $1$ and simply write $P=P_1$ and $p$ for its number of blocks. Moreover, as we only consider microscopic curves, we remove the superscript $N$ on curves, writing $\gamma$ for $\gamma^N\in\Omega_{\text{mic}}^N$. 
We also write $\nu$ for $\nu^N_\beta$. 
\begin{lemm}\label{lemm_pole_size_and_local_eq_at_the_pole}
Let $T>0$. For each $N\in\N_{\geq 1}$ and each $0<a<N/2$,
\begin{align}
\E^{\nu}_\beta\Big[\exp\Big[aN\int_0^{T}e^{-2\beta}(p(\gamma_t)-1)\, dt\Big]\Big]
\leq 
\exp\Big[NaT + \frac{2Na^2 T}{N-2a}\Big]
,
\label{eq_controle_moment_exponentiuel_pole_size}
\end{align}
so that for each $A>1$:
\begin{equation}
\lim_{N\rightarrow\infty}\frac{1}{N}\log\Prob^N_{\beta} \left( \frac{1}{T}\int_0^{T}e^{-2\beta}(p(\gamma_t)-1) \, dt\geq A\right)  
=
-\infty
.
\label{eq_proba_size_pole_gtr_A}
\end{equation}
Moreover, for each $\delta>0$ and $G\in\C$,
\begin{equation}
\lim_{N\rightarrow\infty} \frac{1}{N}\log\Prob^N_{\beta} \biggr( \bigg|\frac{1}{T}\int_0^{T} W^G_t(\gamma_t) \, dt\bigg|>\delta\bigg)  =-\infty.\label{eq_proba_terme_pole_grtr_epsilon_is_minus_infty}
\end{equation}
By Corollary~\ref{coro_change_measure_sous_exp}, the limits~\eqref{eq_proba_size_pole_gtr_A}--\eqref{eq_proba_terme_pole_grtr_epsilon_is_minus_infty} hold also under $\Prob^{N}_{\beta,H}$ with the additional condition that trajectories belong to $E([0,T],\e)$.
\end{lemm}
The proof of Lemma~\ref{lemm_pole_size_and_local_eq_at_the_pole} relies on a bijection argument, stated in the following lemma, for which more notations are required.\\
If $\gamma\in\Omega^N_{\text{mic}}$ (recall that we drop the $N$ superscript on microscopic curves), let $p'(\gamma)$ denote the number with centre at height $z_1(\gamma)-\frac{3}{2N}$ in the associated droplet $\Gamma$ (i.e. blocks directly below those that form the poles, 
see Figure~\ref{fig_possible_excroissances}):
\begin{equation}
p'(\gamma) := \Big|\Big\{\text{blocks in }\Gamma\text{ with centre }i \text{ with }i\cdot{\bf b}_{\pi/2} = z_1(\gamma)-\frac{3}{2N}\Big\}\Big|,\label{eq_def_p_prime}
\end{equation}
where $z_1$ is the largest ordinate of points in a curve:
\begin{equation}
\quad z_1(\gamma) := \sup\big\{ x\cdot{\bf b}_{\pi/2}: x\in\gamma\big\}.\label{eq_def_y_max_sec_6}
\end{equation}
Write $E_\nu$ for the expectation under the static measure $\nu=\nu^N_\beta$ and $E_{\nu_f}$ for the expectation under $f\nu$ when $f$ is a density for $\nu$ and $\nu_f$ denotes the associated probability. 
Define the Dirichlet form $D_N$ of the contour dynamics:
\begin{equation}
D_N(g) := -E_{\nu}\Big[g \lcal_\beta g\Big] = \frac{1}{2}\sum_{\gamma,\tilde\gamma\in\Omega^N_{\text{mic}}}\nu(\gamma)c(\gamma,\tilde\gamma)\big[g(\gamma)-g(\tilde\gamma)\big]^2,\qquad g:\Omega^N_{\text{mic}}\rightarrow\R.\label{eq_def_Dirichlet_form}
\end{equation}
\begin{lemm}\label{lemm_bijection_argument}
Let $f$ be a density with respect to the contour measure $\nu$. Then, for any integer $A\geq 2$,
\begin{equation}
\Big[\nu_f\big(p=2,\gamma^{-,1}\in\Omega^N_{\text{mic}}, p'\geq A\big)^{1/2} - E_{\nu_f}\big[(p-1)e^{-2\beta}{\bf 1}_{p\geq A}\big]^{1/2}\Big]^2 \leq 2D_N(f^{1/2}).\label{eq_majoration_carre_terme_de_pole}
\end{equation}
The indicator function ${\bf 1}_{\gamma^{-,1}\in\Omega^N_{\text{mic}}}$ in the first probablity ensures that the deletion of a pole is a dynamically allowed move 
(it is not true for all curves as the point $0$ must belong to droplets associated with curves $\gamma\in\Omega^N_{\text{mic}}$ 
and deleting the north pole could make this fail).\\
Equation~\eqref{eq_majoration_carre_terme_de_pole} also holds with $p'\leq A, p\leq A$ instead of $p'\geq A,p\geq A$ respectively in the probability and in the expectation.
\end{lemm}
\begin{proof}
We prove the result with $A=2$ (i.e. without constraint on $p'$), the general case is similar. 
Fix a density $f$ for $\nu$ and define $U$ on $\Omega^N_{\text{mic}}$ as follows:
\begin{equation}
\forall \gamma\in \Omega^N_{\text{mic}},\qquad U(\gamma) = e^{-2\beta}(p(\gamma)-1).\label{eq_def_U}
\end{equation}
Let us first prove that $\nu_f(p=2,\gamma^{-,1}\in\Omega^N_{\text{mic}})$ and $E_{\nu_f}[U]$ are comparable, up to an error that can be expressed in terms of the Dirichlet form $D_N(f^{1/2})$.\\
\begin{figure}[H]
\begin{center}
\includegraphics[width=12cm]{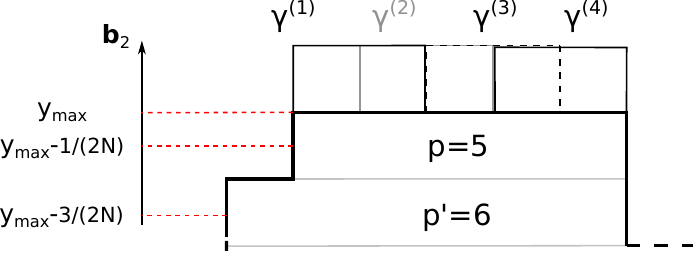} 
\caption{Neighbourhood of the north pole of a curve $\gamma \in \Omega^N_{\text{mic}}$ (thick line) and the $\gamma^{(n)}$, $n\leq p-1 = 4$. $\gamma^{(3)}$ is the curve $\gamma$ to which the two blocks delimited by dashed lines are added. 
Conversely, any of the $\gamma^{(n)}$, $n\leq 4$ is in $\{p=2,\gamma^{-,1}\in\Omega^N_{\text{mic}}\}$ and deleting the two blocks constituting their poles turns them into $\gamma$. 
The curve $\gamma$ has $p=5$ blocks in the pole, corresponding to blocks with centre at height $z_1-\frac{1}{2N}$, 
and $p'=6$ blocks on the level below, i.e. with centre at height $z_1-\frac{3}{2N}$.\label{fig_possible_excroissances}}
\end{center}
\end{figure}
Each $\gamma\in \Omega^N_{\text{mic}}$ can be turned into any one of the curves $\gamma^{(1)},...,\gamma^{(p-1)}$ with two blocks added atop the north pole, where $\gamma^ {(n)}$ is identical to $\gamma$ except that two blocks sitting on the edges $n,n+1$ are added, counting the edges from the left extremity of the pole (see Figure~\ref{fig_possible_excroissances}). 
Note that the $\gamma^{(n)}$ correspond to the $\gamma^{+,x}$ with $x,x+{\bf e}^\pm_x \in P$ in Definition~\ref{def_contour_dynamics}, 
where $n$ stands for the number of the block at $x$.\\
Conversely, the north pole of each curve $\tilde\gamma$ with $p=2$ and $\tilde\gamma^{-,1}\in\Omega^N_{\text{mic}}$ can be deleted 
to obtain a curve $\gamma = (\tilde\gamma)^{-,1}$ which is still in $\Omega^N_{\text{mic}}$. 
The resulting curve $\gamma$ has length $|\gamma| = |\tilde\gamma|-2/N$. 
The same curve $\gamma$ can be obtained $p-1=p(\gamma)-1$ times by deleting the pole of size $2$ of a certain $\tilde\gamma$, 
with these $\tilde\gamma$ corresponding to curves in $\{\gamma^{(n)}:1\leq n\leq p-1\}$, see Figure~\ref{fig_possible_excroissances}. 
Thus:
\begin{align}
\nu_f\big(p=2,\gamma^{1,-}\in\Omega^N_{\text{mic}}\big) &= \sum_{\tilde\gamma\in \{p=2,\tilde\gamma^{-,1}\in\Omega^N_{\text{mic}}\}} \nu(\tilde\gamma)f(\tilde\gamma) \nonumber\\
&= \sum_{\tilde\gamma\in \{p=2,\tilde\gamma^{-,1}\in\Omega^N_{\text{mic}}\}}\sum_{\gamma\in \Omega^N_{\text{mic}}} {\bf 1}_{\{\exists n\leq p-1:\tilde\gamma = \gamma^{(n)}\}}\nu(\gamma) e^{-2\beta}f(\gamma^{(n)}) \nonumber\\
&=\sum_{\gamma\in \Omega^N_{\text{mic}}} \nu(\gamma) e^{-2\beta}\sum_{n=1}^{p-1}f(\gamma^{(n)}).
\end{align}
Add and subtract the quantities needed to bound the second line by the Dirichlet form $D_N(f^{1/2})$:
\begin{align}
\nu_f\big(p=2,\gamma^{-,1}\in\Omega^N_{\text{mic}}\big)&= \sum_{\gamma\in \Omega^N_{\text{mic}}}\nu(\gamma) e^{-2\beta}\sum_{n=1}^{p-1} \big[f(\gamma^{(n)}) + f(\gamma) -2f^{1/2}(\gamma)f^{1/2}(\gamma^{(n)})\big] \nonumber\\
&\hspace{1cm}- \sum_{\gamma\in \Omega^N_{\text{mic}}}\nu(\gamma)e^{-2\beta}\Big[(p-1)f(\gamma) - 2\sum_{n=1}^{p-1}f^{1/2}(\gamma)f^{1/2}(\gamma^{(n)})\Big].\label{eq_line_to_estimate_bijec_argument_size_pole}
\end{align}
To estimate the second line of~\eqref{eq_line_to_estimate_bijec_argument_size_pole}, apply Cauchy-Schwarz inequality to the sum $\sum_{k=1}^{p-1 }$ to obtain:
\begin{align}
\nu_f\big(p=2,\gamma^{-,1}\in\Omega^N_{\text{mic}}\big)&\leq 2D_N(f^{1/2}) - E_{\nu_f}\big[e^{-2\beta} (p-1)\big] \nonumber\\
&\qquad + 2\sum_{\gamma\in \Omega^N_{\text{mic}}}\nu(\gamma)e^{-2\beta}(p-1)^{1/2}f^{1/2}(\gamma)\Big[\sum_{n=1}^{p-1}f(\gamma^{(n)}) \Big]^{1/2}.
\end{align}
Recall the definition of $U$ from~\eqref{eq_def_U} and again use Cauchy-Schwarz on the sum on $\gamma\in \Omega^N_{\text{mic}}$ to find:
\begin{align}
\nu_f\big(p=2,\gamma^{-,1}\in\Omega^N_{\text{mic}}\big)&\leq 2D_N(f^{1/2}) - E_{\nu_f}[U] \nonumber\\
&\qquad+2\Big[\sum_{\gamma\in \Omega^N_{\text{mic}}}\nu(\gamma)e^{-2\beta}(p-1)f(\gamma)\Big]^{1/2}\Big[\sum_{\gamma\in \Omega^N_{\text{mic}}}\nu(\gamma)e^{-2\beta}\sum_{n=1}^{p-1}f(\gamma^{(n)}) \Big]^{1/2}\nonumber\\
&= 2D_N(f^{1/2}) - E_{\nu_f}[U] + 2 E_{\nu_f}[U]^{1/2}\nu_f\big(p=2,\gamma^{-,1}\in\Omega^N_{\text{mic}}\big)^{1/2}.
\end{align}
Putting things together yields the claim of the lemma:
\begin{equation}
\Big[\nu_f\big(p=2,\gamma^{-,1}\in\Omega^N_{\text{mic}}\big)^{1/2} -E_{\nu_f}[U]^{1/2}\Big]^2 
\leq 
2D_N(f^{1/2})
.
\label{eq_bound_difference_pole_as_fct_D}
\end{equation}
\end{proof}
\begin{proof}[Proof of Lemma~\ref{lemm_pole_size_and_local_eq_at_the_pole}]
We now explain how to obtain Lemma~\ref{lemm_pole_size_and_local_eq_at_the_pole} from Lemma~\ref{lemm_bijection_argument}. We need to do two things:
\begin{enumerate}
	\item Bound from above the probabilities appearing in the claim by an expression involving the measure $\nu_f$ as in~\eqref{eq_majoration_carre_terme_de_pole}.
	\item Prove that~\eqref{eq_proba_terme_pole_grtr_epsilon_is_minus_infty} holds for $W_\cdot^G$, with $G\in\C$. The first point only gives the result for ${\bf 1}_{p=2,\gamma^{-,1}\in\Omega^N_{\text{mic}}} - U$, which corresponds to $W^1_\cdot$;.
\end{enumerate}
The first point relies on a classical Feynman-Kac estimate. Since a similar reasoning is used repeatedly in the article, we present it here once and for all. The second point, however, requires some care; in order to apply the bounds of Lemma~\ref{lemm_bijection_argument} to integrals depending on a function $G\in\C$.\\

Let us explain the general idea for the first point using~\eqref{eq_proba_size_pole_gtr_A} as an example. We wish to estimate:
\begin{align}
\Prob^N_{\beta}\left(\frac{1}{T}\int_0^{T} e^{-2\beta}(p(\gamma_t)-1)\, dt \geq A\right).
\end{align}
We do so using Feynman-Kac formula. 
Let $a\in(0,N/2)$ and apply the exponential Chebychev inequality to obtain
\begin{align}
\frac{1}{N}\log\Prob^N_{\beta} &\left(\frac{1}{T}\int_0^{T} e^{-2\beta}(p(\gamma_t)-1)\, dt \geq A\right) 
\label{eq_prob_size_pole_exp_tchebycheff}\\
&\leq 
-a A T+ \frac{1}{N}\log \E^N_{\beta}\left[\exp\left[aN\int_0^{T}  e^{-2\beta}(p(\gamma_t)-1)\,  dt\right]\right]\nonumber
.
\end{align}
Consider the generator $N^2\lcal_{\beta} + aNU$, with $U$ defined in~\eqref{eq_def_U}. 
This generator is self-adjoint for the contour measure $\nu_\beta$ and Feynman-Kac inequality plus a representation theorem for the largest eigenvalue of a symmetric operator (Lemma A.1.7.2 in \cite{Kipnis1999}) yield that, 
with the equilibrium measure $\nu = \nu^N_\beta$ as an initial condition:
\begin{align}
\E^{\nu}_\beta\biggr[\exp\left[aN\int_0^{T} U(\gamma_t)\, dt\right] \biggr]
\leq 
\exp \bigg[\int_0^T\sup_{f\geq 0 : E_{\nu}[f]= 1} \Big\{aNE_{\nu_f}[U] - N^2D_N(f^{1/2})\Big\}\, dt\bigg]
.
\label{eq_variational_formula_largest_eigenvalue}
\end{align}
In the present, $G\equiv 1$ case, the supremum in~\eqref{eq_variational_formula_largest_eigenvalue} does not depend on time. \\
One can bound $\Prob^N_{\beta}$, where the initial condition is the deterministic curve $\gamma^{N,\mathnormal{ref}}$, by the probability $\Prob^\nu_\beta$ starting from the equilibrium measure $\nu$:
\begin{equation}
\Prob^N_{\beta}(\cdot) \leq \mathcal Z^N_{\beta}e^{\beta N|\gamma^{N,ref}|}\Prob^{\nu}_\beta(\cdot) \leq e^{C\beta N}\Prob^{\nu}_\beta(\cdot),\label{eq_proba_equilibrium_bounds_proba_PNrbeta}
\end{equation}
for some constant $C>0$. Using~\eqref{eq_variational_formula_largest_eigenvalue}--\eqref{eq_proba_equilibrium_bounds_proba_PNrbeta},~\eqref{eq_prob_size_pole_exp_tchebycheff} becomes:
\begin{align}
\frac{1}{N}\log\Prob^N_{\beta}\biggr(\frac{1}{T}\int_0^{T} U(\gamma_t)dt \geq A\biggr) &\leq -aAT+ C\beta +  T\sup_{\substack{ f\geq 0, \\ E_\nu[f]=1}} \Big\{aE_{\nu_f}[U] - N D_N(f^{1/2})\Big\}.\label{eq_var_formula_for_pole_size}
\end{align}
At this point, we can use Lemma~\ref{lemm_bijection_argument} to bound the supremum in the right-hand side of~\eqref{eq_var_formula_for_pole_size}: 
by~\eqref{eq_majoration_carre_terme_de_pole}, 
\begin{equation}
E_{\nu_f}[U]\leq \Big[1+ (2D_N(f^{1/2}))^{1/2}\Big]^2.\label{eq_bound_U}
\end{equation}
As a result, the supremum in~\eqref{eq_var_formula_for_pole_size} satisfies (recall $a<N/2$):
\begin{align}
\sup_{\substack{f\geq 0, \\ E_{\nu}[f]=1}}\Big\{aE_{\nu_f}[U] -ND_N(f^{1/2})\Big\} 
&\leq 
\sup_{u\geq 0}\big\{a + 2a\sqrt{2u} + 2au-Nu\big\}
= 
a + \frac{2a^{2}}{N-2a}
.
\end{align}
Injecting this result in~\eqref{eq_variational_formula_largest_eigenvalue} gives~\eqref{eq_controle_moment_exponentiuel_pole_size}. On the other hand, injecting it in~\eqref{eq_var_formula_for_pole_size}, then taking $N$ large, then $a$ large concludes the proof of~\eqref{eq_proba_size_pole_gtr_A}. \\
We claim that Equation~\eqref{eq_proba_terme_pole_grtr_epsilon_is_minus_infty} in the $G \equiv 1$ case follows similarly. 
Indeed, for $W^1$, 
using the identity $x-y = (\sqrt{x} - \sqrt{y})(\sqrt{x}+\sqrt{y})$ valid for $x,y\geq 0$, 
the quantity in the supremum in~\eqref{eq_var_formula_for_pole_size} is now $a\E_{\nu_f}[W^1] - N D_N(f^{1/2})$, where by definition:
\begin{align}
W^1 &= \sum_{\substack{x\in P \\ x+{\bf e}^\pm_x \in P}}\big[{\bf 1}_{p=2,\gamma^{-,1}\in\Omega^N_{\text{mic}}} - e^{-2\beta}\big] \nonumber\\
&= {\bf 1}_{p=2,\gamma^{-,1}\in\Omega^N_{\text{mic}}} - (p-1) e^{-2\beta} = {\bf 1}_{p=2,\gamma^{-,1}\in\Omega^N_{\text{mic}}} - U.
\end{align}
As a result, $E_{\nu_f}[W^1]$ can be bounded from above as follows using~\eqref{eq_bound_difference_pole_as_fct_D} and~\eqref{eq_bound_U}:
\begin{align}
\big|E_{\nu_f}\big[W^1\big]\big| &= \Big|\nu_f\big(p=2,\gamma^{-,1}\in\Omega^N_{\text{mic}}\big)^{1/2} - E_{\nu_f}[U]^{1/2}\Big|\Big[\nu_f\big(p=2,\gamma^{-,1}\in\Omega^N_{\text{mic}}\big)^{1/2} + E_{\nu_f}[U]^{1/2}\Big]\nonumber\\
&\leq (2D_N(f^{1/2}))^{1/2}\big[2+(2D_N(f^{1/2}))^{1/2}\big].
\end{align}
Elementary computations again yield:
\begin{align}
\sup_{\substack{f\geq 0, \\ E_{\nu}[f]=1}}\Big\{aE_{\nu_f}[W^1] -ND_N(f^{1/2})\Big\} &\leq \frac{2a^2}{N-2a}.
\end{align}
Using this estimate in~\eqref{eq_var_formula_for_pole_size}, with $W^1$ there instead of $U$; taking the large $N$, then the large $a$ limits conclude the proof of the first point.\\

Let us now deal with the second point, i.e. proving~\eqref{eq_proba_terme_pole_grtr_epsilon_is_minus_infty} for any $G$ and not just $G\equiv 1$. As $G$ may not have constant sign, one cannot directly use the bounds in the proof of Lemma~\ref{lemm_bijection_argument}. However, if $G$ is positive, it is not complicated to repeat the bijection argument of Lemma~\ref{lemm_bijection_argument} to obtain, for each $t\leq T$:
\begin{align}
\bigg[E_{\nu_f}\Big[e^{-2\beta}\sum_{\substack{x\in P\\x+{\bf e}^\pm_x\in P }}  G(t,x)\Big]^{1/2} - E_{\nu_f}&\Big[{\bf 1}_{p=2,\gamma^{-,1}\in\Omega^N_{\text{mic}}} \sum_{\substack{x\in P\\x+{\bf e}^\pm_x\in P }} G(t,x)\Big]^{1/2}\bigg]^2\nonumber\\
&\leq \|G_t\|_\infty D_N(f^{1/2}) + \frac{\|\nabla G_t\|_\infty}{N}E_{\nu_f}[U],\label{eq_local_eq_pole_G_positive}
\end{align}
where the second term in the right-hand side comes from the fact that the point at which $G_t$ is evaluated depends on the position of the pole. 
Recall also that the summation on $x\in P$ such that $x+{\bf e}^\pm_x\in P$ is just a way of enumerating all positions where two blocks can appear atop the pole. \\
For general $G\in\C$, the result then follows by applying~\eqref{eq_local_eq_pole_G_positive} to the positive and negative parts $G^+$ and $G^-$ of $G$, i.e. $G := G^+-G^-$ with $G^+,G^-\geq 0$.
\end{proof}
\subsection{Convergence of the ${\bf 1}_{p_k=2}$ term at fixed $\beta$ and slope around the poles}\label{subsec_slope_around_the_poles}
This section is devoted to the proof of Proposition~\ref{prop_value_slope_at_poles}: poles act as reservoirs that fix the value $e^{-\beta}$ of the slopes $\xi^{\pm,\epsilon N}_{L_1},1-\xi^{\pm,\epsilon N}_{L_1}$ at the poles. 

We prove this statement in several steps. 
First, we explain how to use the condition that trajectories belong to $E([0,T],\e)$ to project the contour dynamics onto a local one. 
This is a key technical argument to compare the contour dynamics to simpler $1$-dimensional ones. 

We then prove that the ${\bf 1}_{p_k=2}$ term fixes the slope around the poles, in the sense that the time integrals of ${\bf 1}_{p_k=2}$ and $\xi_{L_k}^{\pm,\epsilon N}$ are close, see Section~\ref{subsec_1_pi2_is_slope_around_poles}. This should not come as a surprise if one remembers that, in a Symmetric Simple Exclusion Process (SSEP) with reservoirs, the density close to the reservoirs is fixed. 
The time average of ${\bf 1}_{p_k=2}$ is then proven to be equal to $e^{-\beta}$ in Section~\ref{sec_value_slope_at_the_pole}. 
Preliminary microscopic estimates used in this computation are carried out in Section~\ref{subsec_compactness_estimate_1pis2}.
\subsubsection{Turning the contour dynamics into a local dynamics}\label{sec_turning_contour_into_local_dyn}
The next lemma shows that if trajectories belong to the set $E([0,T],\e)$ where the contour dynamics is local (recall the discussion following Definition~\ref{def_CI}), 
then exponential moments of time-integrated observables can be estimated in terms of quantities defined on the effective state space $\e$ only. 
This Lemma is proven in Section~\ref{sec_local_dynamics}.
\begin{lemm}[Projection onto a local dynamics in the effective state space $\e$]\label{lemm_projec_dynamics_onto_local}
Let $\psi:[0,T]\times \Omega^N_{\text{mic}}\rightarrow \R$ be bounded. Then, for some $C= C(\gamma^{\mathrm{ref}})>0$:
\begin{align}
\frac{1}{N}\log \E^ N_{\beta}&\bigg[{\bf 1}_{\gamma^N_\cdot\in E([0,T],\e)}\exp\bigg[N\int_0^ {T} \psi(t,\gamma^N_t)\, dt\bigg]\bigg]\nonumber\\
&\leq 
C\beta +\int_0^ {T}\sup_{f\geq 0 : \nu_f(\e) = 1}\Big\{ \E_{\nu_f}\big[\psi(t,\cdot)\big] - ND_N(f^{1/2})\Big\}\, dt
.
\label{eq_projec_dynamics_into_local}
\end{align}
\end{lemm}
\begin{rmk}
Equation~\eqref{eq_projec_dynamics_into_local} looks like a standard Feynman-Kac estimate. Note however that the supremum in~\eqref{eq_projec_dynamics_into_local} is on densities with full support in $\e$. In general, if $f$ is a $\nu$-density, there is no way to control $D_N(f^{1/2})$ by $D_N\big(f^{1/2}{\bf 1}_{\e}\big)$. Indeed, if $\tilde f = f{\bf 1}_{\e}$, $D_N(\tilde f^{1/2})$ contains terms of the form:
\begin{equation}
\sum_{\substack{\gamma^N\in \Omega^N_{\text{mic}}\cap \e \\ \tilde\gamma^N\notin \e}}\big[\nu(\gamma^N)c(\gamma^N,\tilde\gamma^N) f(\gamma^N)+ \nu(\tilde\gamma^N)c(\tilde\gamma^N,\gamma^N) f(\gamma^N)\big],
\end{equation}
which have a priori no reason to be comparable to differences $[f(\gamma^N)^{1/2} - f(\tilde\gamma^N)^{1/2}]^2$.

Note also that Lemma~\ref{lemm_projec_dynamics_onto_local} is not a statement about the contour dynamics \emph{conditioned} to stay inside $\e$, but about the \emph{full} dynamics. This is an important point: the jump rates of a conditioned dynamics would be non-local, whereas we really need locality to later project the dynamics onto $1$-dimensional particle dynamics.\demo 
\end{rmk}
\subsubsection{The ${\bf 1}_{p=2}$ term coincides with the slope around the pole}\label{subsec_1_pi2_is_slope_around_poles}
Inside each region, the contour dynamics has the same updates as an SSEP, as explained in Section~\ref{sec_heuristics} and presented more thoroughly in the proof of Lemma~\ref{lemm_1_block_2_block_pour_pis2} below. 
Here, we treat the poles as the extremal sites of an SSEP, 
viewing ${\bf 1}_{p_1=2}$ as the edge state of the first (in region $1$) or last (in region $4$) site of a SSEP, 
see~\eqref{eq_observation_1_p_is_2_is_occupation_number} below. 
We use this observation to prove estimates of the slope at the pole in terms of ${\bf 1}_{p_1=2}$.  
Recall that, for $\gamma^N\in\Omega^N_{\text{mic}}$, $x\in V(\gamma^N)$ and $\ell\in\N_{\geq 1}$:
\begin{align}
\xi_x^{+, \ell} 
= 
\frac{1}{\ell+1}\sum_{\substack{y\geq x \\ \|y-x\|_1\leq \ell/N}}\xi_y,\qquad \xi_x^{-, \ell} = \frac{1}{\ell+1}\sum_{\substack{y\leq x \\ \|y-x\|_1\leq \ell/N}}\xi_y
.
\end{align}
Recall also the definition of the space $E([0,T],\e)$ from~\eqref{eq_def_E_0_T_Omega}. 
We focus on the north pole, using the notations $P := P_1$ and $p:= p_1$ as well as $\nu:=\nu^N_\beta$. 
The superscript $N$ is also dropped for microscopic curves.
\begin{lemm}\label{lemm_1_block_2_block_pour_pis2}
For each $T>0$, $\beta>\log 2$, $\delta>0$ and each $G\in\C$, the slope on each side of the pole satisfies a one block estimate:
\begin{align}
\lim_{\ell\rightarrow\infty}&\limsup_{N\rightarrow\infty}\frac{1}{N}\log\Prob^N_{\beta}\bigg(\gamma^N_\cdot \in E([0,T],\e);
\nonumber\\
&\qquad \bigg|\frac{1}{T}\int_0^{T}G(t,L_1(\gamma^N_t))\big({\bf 1}_{p_1(\gamma^N_t)=2} - \xi^ {\pm,\ell}_{L_1(\gamma^N_t)+2{\bf b}_0/N}\big)\, dt\bigg|\geq \delta\bigg) 
= 
-\infty
,
\label{eq_1block_pour_1_pis2}
\end{align}
and a two block estimate:
\begin{align}
\lim_{\epsilon\rightarrow0}\limsup_{N\rightarrow\infty}\frac{1}{N}&\log\Prob^N_{\beta}\bigg(\gamma^N_\cdot \in E([0,T],\e);
\nonumber\\
&\hspace{1.5cm} \bigg|\frac{1}{T}\int_0^{T} G(t,L_1(\gamma^N_t))\big({\bf 1}_{p_1(\gamma^N_t)=2} - \xi^ {\pm,\epsilon N}_{L_1(\gamma^N_t)+2{\bf b}_0/N}\big) \, dt\bigg|\geq \delta\bigg) 
= 
-\infty
.
\label{eq_2block_pour_1_pis2}
\end{align}
Both estimates are valid under $\Prob^N_{\beta,H}$ by Corollary~\ref{coro_change_measure_sous_exp}.
\end{lemm}
\begin{proof}
The proof relies on the key observation that the quantity ${\bf 1}_{p=2}$ can be controlled in terms of the edges at the extremities $L_1,R_1$ of the poles. 
Indeed, abusing notations and respectively writing $L_1+2$, 
$R_1-3$ for the vertex at distance $2/N$ to $L_1$ clockwise and the vertex at distance $3/N$ from $R_1$ anticlockwise:
\begin{equation}
{\bf 1}_{p=2} 
=
\xi_{L_1+2} 
= 
\xi_{R_1-3}
.\label{eq_observation_1_p_is_2_is_occupation_number}
\end{equation}
Here, we focus on the slope to the right of the pole, 
for which we use the identity ${\bf 1}_{p=2} = \xi_{L_1+2}$. The slope to the left of the pole is treated similarly using ${\bf 1}_{p=2}=\xi_{R_1-3}$.

As long as no growth/deletion move at the poles occurs, 
${\bf 1}_{p=2}$ can thus be thought of as the occupation number of the closest site to a reservoir in a SSEP, in which case~\eqref{eq_1block_pour_1_pis2}--\eqref{eq_2block_pour_1_pis2} are well-known (see \cite{Eyink1990}). 
We first prove~\eqref{eq_1block_pour_1_pis2}. Building on the observation~\eqref{eq_observation_1_p_is_2_is_occupation_number}, define $\phi_\ell$ as the function:
\begin{equation}
\phi_\ell(\gamma) = \xi_{L_1+2} - \xi^ {+,\ell}_{L_1+2},\qquad \gamma\in\Omega^N_{\text{mic}}.\label{eq_def_phi_1_pis2_minus_densite}
\end{equation}
To estimate the probability in~\eqref{eq_1block_pour_1_pis2}, it is enough to consider, for each $a>0$, the quantity:
\begin{align}
&\frac{1}{N}\log \Prob^N_{\beta}\bigg(\gamma_\cdot \in E([0,T],\e);\exp\bigg[a N\int_0^{T}G(t,L_1(\gamma_t))\phi_\ell(\gamma_t) \, dt\bigg]
\geq 
\exp[aNT\delta]\bigg)
\nonumber\\
&\leq 
-aT\delta +\frac{1}{N}\log\E^N_{\beta}\bigg[{\bf 1}_{E([0,T],\e)}\exp\bigg[a N\int_0^{T}{\bf 1}_{\gamma_t\in \e}\, G(t,L_1(\gamma_t))\phi_\ell(\gamma_t) \, dt\bigg]\bigg]
.
\label{eq_ad_hoc_counter_term_SSEP_1_p_is_2}
\end{align}
Let $D_N^{\text{ex}}\leq D_N$ be the Dirichlet form of the contour dynamics excluding moves at the pole:
\begin{equation}
D_N^{\text{ex}}(g) 
= 
\frac{1}{2}\sum_{\gamma\in\Omega^N_{\text{mic}}}\nu(\gamma)\sum_{x\in V(\gamma)\setminus \cup_k P_k(\gamma)}c(\gamma,\gamma^x)\big[g(\gamma^x)-g(\gamma)\big]^2,
\qquad g:\Omega^N_{\text{mic}}\rightarrow\R
.
\label{eq_def_D_N_S}
\end{equation}
If $g$ is supported in $\e$, 
then the $c(\gamma,\gamma^x)$ are local (see Definition~\ref{def_effective_state_space}) and:
\begin{equation}
D^{\text{ex}}_N(g) 
= 
\frac{1}{2}\sum_{\gamma\in\Omega^N_{\text{mic}}}\nu(\gamma)\sum_{x\in V(\gamma)\setminus \cup_kP_k(\gamma)}c_x(\gamma)\big[g(\gamma^x)-g(\gamma)\big]^2
,
\label{eq_Dirichlet_form_SSEP_dans_e}
\end{equation}
with:
\begin{equation}
c_x(\gamma) 
= 
\frac{1}{2}\big[\xi_x(1-\xi_{x+{\bf e}^-_x}) + \xi_{x+{\bf e}^-_x}(1-\xi_x)\big]
.
\end{equation}
Apply Lemma~\ref{lemm_projec_dynamics_onto_local} to $\psi = aG\phi_\ell$ to obtain that~\eqref{eq_ad_hoc_counter_term_SSEP_1_p_is_2} is bounded from above by:
\begin{align}
&-a\delta T + C\beta +\int_0^{T}dt\sup_{f\geq 0: \nu_f(\e) = 1}\Big\{ aE_{\nu_f}[G(t,L_1)\phi_\ell] - N D^{\text{ex}}_N(f^{1/2})\Big\}.\label{eq_qtite_ds_supremum_pentopole_reduction_a_e_r}
\end{align}
Let us now compare the contour dynamics in the first region to a SSEP. 
To do so, we partition curves in $\Omega^N_{\text{mic}}$ according to their first region.

Fix $t\in[0,T]$ and a $\nu$-density $f$ with support in $\e$. We first split the expectation in the supremum in~\eqref{eq_qtite_ds_supremum_pentopole_reduction_a_e_r} depending on the possible positions of $L_1$ so that $G(t,L_1)$ becomes deterministic. 
For $x\in(N^{-1}\Z)^2$, define:
\begin{equation}
M(x) := \Big\{\gamma\in\Omega^N_{\text{mic}} : L_1(\gamma) +\frac{2{\bf b}_0}{N} = x\Big\}.
\end{equation}
Note that, since $\gamma\in\Omega^N_{\text{mic}}\cap \e$ is associated to a droplet that is by definition in a volume neighbourhood of the bounded droplet $\Gamma^{\mathrm{ref}}$, 
only a finite number of $M(x)$ are actually non-empty. Then:
\begin{equation}
E_{\nu_f}[G(t,L_1)\phi_\ell] = \sum_{x\in(N^{-1}\Z)^2}G\Big(t,x-\frac{2{\bf b}_0}{N}\Big)\bigg[\sum_{\gamma\in M(x)}\nu(\gamma) f(\gamma)\phi_\ell(\gamma)\bigg].\label{eq_transit_penteopole_1}
\end{equation}
In~\eqref{eq_transit_penteopole_1}, recall that the constraint that curves belong to $\e$ is enforced by the density $f$. \\

In the following, for $\gamma\in M(x)$, we refer to the edge $[x,x+e^+_x]$ as edge $1$, to the one following it as edge $2$, etc, up to edge $\ell$; 
and write $\xi_1(\gamma),...,\xi_\ell(\gamma)$ for the corresponding values of the edge labels (as usual, curves are oriented clockwise). 
As we work with curves in $\e$ for which each region contains a number of sites of order $N$ at least, all these edges are in region $1$. 
Configurations in $\{0,1\}^\ell =: \Omega_\ell$, are denoted by the letter $\xi$. 
The function $\phi_\ell$ depends only on edges $1$ to $\ell$, so that the expectation in~\eqref{eq_transit_penteopole_1} reads, 
letting $f_{\ell,x}$ denote the marginal for the uniform measure on $\Omega_\ell$ of $f$ restricted to $M(x)$:
\begin{equation}
E_{\nu_f}[G(t,L_1)\phi_\ell] = \sum_{x\in(N^{-1}\Z)^2}\nu_f(M(x))G\Big(t,x-\frac{2{\bf b}_0}{N}\Big)\frac{1}{|\Omega_\ell|}\sum_{\xi\in\Omega_{\ell}} f_{\ell,x}(\xi)\phi_\ell(\xi),\label{eq_transit_penteopole_2}
\end{equation}
where $|\Omega_\ell|=2^\ell$ and, if $\xi(\gamma)$ denotes the collection $\xi_1(\gamma),...,\xi_{\ell}(\gamma)$ for a given $\gamma\in \Omega^N_{\text{mic}}$,
\begin{equation}
\forall\xi\in\Omega_\ell,\qquad f_{\ell,x} (\xi) := \frac{1}{\nu_f(M(x))}\sum_{\gamma\in M(x) : \xi(\gamma) = \xi} |\Omega_\ell|\nu(\gamma) f(\gamma).\label{eq_def_densite_conditionne_SSEP_taille_k}
\end{equation}
Note that we need only consider points $x$ and densities $f$ with $\nu_f(M(x))>0$. 
This ensures that $f_{\ell,x}$ is unambiguously defined. Moreover, $f_{\ell,x}$ is a density for the uniform measure on $\Omega_\ell$. 

Let us do the same splitting on the Dirichlet form $D^{\text{ex}}_N$ in~\eqref{eq_qtite_ds_supremum_pentopole_reduction_a_e_r} in order to bound it from below by the Dirichlet form of a SSEP on configurations with $\ell$ sites. 
The mapping to go from a portion of length $\ell$ of a region of a curve $\gamma\in \Omega^N_{\text{mic}}\cap\e$ to an associated SSEP configuration $\xi(\gamma)\in\Omega_\ell$ is represented on Figure~\ref{fig_corner_flip_ssep} for the first region: 
each edge is tilted clockwise by $\pi/4$, turning the portion of $\gamma$ into the graph of a $1$-Lipschitz function, constant on segments of the form $[(j-1)\sqrt{2},j\sqrt{2}]$, $1\leq j \leq \ell$. 
A particle is then put at site $1\leq j\leq \ell$ if the path goes down between $j\sqrt{2}$ and $(j+1)\sqrt{2}$, or this site is left empty if the path goes up.

\begin{figure}[H]
\begin{center}
\includegraphics[width=11cm]{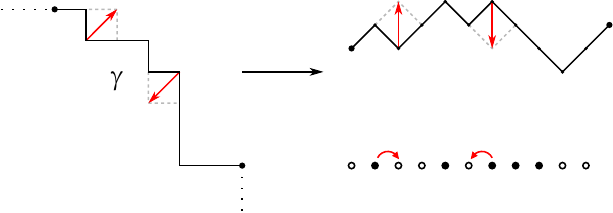} 
\caption{On the left, a portion of region $1$ of an interface $\gamma\in\Omega^N_{\text{mic}}$, delimited by the two black dots. On the right, the corresponding path and simple exclusion particle configuration. The mapping is possible if the left-extremity of the interface as well as its length are fixed. \label{fig_corner_flip_ssep}}
\end{center}
\end{figure}
Recall the definition~\eqref{eq_Dirichlet_form_SSEP_dans_e} of the SSEP part of the Dirichlet form of the contour dynamics. 
Define then the Dirichlet form $D^{\text{ex}}_{\ell}$ associated with the SSEP on $\Omega_\ell$ (we use the same notation as for $D^{\text{ex}}_N$ in~\eqref{eq_Dirichlet_form_SSEP_dans_e} to emphasize the analogy): 
for any $g:\Omega_\ell\rightarrow\R$, 
\begin{equation}
D^{\text{ex}}_\ell(g) = \frac{1}{2|\Omega_\ell|}\sum_{\xi\in \Omega_\ell}\sum_{\substack{1\leq u\leq v\leq \ell \\ |u-v|=1}} \frac{1}{2}[\xi_{u}(1- \xi_v) + \xi_v(1-\xi_{u})][g(\xi^{u,v}) - g(\xi)]^2
,
\end{equation}
with $\xi^{u,v}$ the configuration where the state of sites $u,v$ are exchanged. 
In view of the expression~\eqref{eq_Dirichlet_form_SSEP_dans_e} of $D_N^{\text{ex}}(g)$, a simple upper-bound and convexity yield:
\begin{align}
D_N^{\text{ex}}(f) 
&\geq 
\frac{1}{2}\sum_{x\in(N^{-1}\Z)^2}\sum_{\gamma\in M(x)}\nu(\gamma) \sum_{\substack{y\in V(\gamma) \\ x\leq y,\, \|y-x\|_1\leq \ell-1}} c_y(\gamma)\big[f^{1/2}(\gamma^{y})-f^{1/2}(\gamma)\big]^{1/2} 
\nonumber\\
&\geq 
\sum_{x\in(N^{-1}\Z)^2}\nu_f(M(x))D^{\text{ex}}_\ell(f_{\ell,x}).
\label{eq_transit_pentopole_Dirichlet_form_0}
\end{align}
Let $\mathcal U_\ell$ denote the uniform measure on $\Omega_\ell$ with associated expectation $E_{\mathcal U_\ell}$. Using the expressions~\eqref{eq_transit_penteopole_2}--\eqref{eq_transit_pentopole_Dirichlet_form_0}, the supremum in~\eqref{eq_qtite_ds_supremum_pentopole_reduction_a_e_r} at time $t\in[0,T]$ can be bounded from above by:
\begin{align}
&\sup_{f\geq 0 : \nu_f(\e)=1}\bigg\{ \sum_{x\in(N^{-1}\Z)^2} \nu_f(M(x)) \Big[ aG\Big(t,x-\frac{2{\bf b}_0}{N}\Big) E_{\mathcal U_\ell}\big[f_{\ell,x}\phi_\ell\big] - ND^{\text{ex}}_\ell(f_{\ell,x})\Big]\bigg\}
\nonumber\\
&\qquad\leq \sup_{f\geq 0 : \nu_f(\e)=1}\bigg\{ \sum_{x\in(N^{-1}\Z)^2} \nu_f(M(x))\sup_{g\geq 0 : E_{\mathcal U_\ell}[g]=1}\Big\{aG\Big(t,x-\frac{2{\bf b}_0}{N}\Big) E_{\mathcal U_\ell}[g\phi_\ell] - ND^{\text{ex}}_\ell(g)\Big\}\bigg\}
\nonumber\\
&\qquad
\leq 
\sup_{f\geq 0 : \nu_f(\e)=1}\bigg\{ \sum_{x\in(N^{-1}\Z)^2}\nu_f(M(x))a\|G\|_\infty \bigg\}\sup_{\substack{g\geq 0 : E_{\mathcal U_\ell}[g]=1 \\ D^{\text{ex}}_\ell(g)\leq 2a\|G\|_\infty/N}}  E_{\mathcal U_\ell}[g\phi_\ell]
.
\label{eq_qtite_ds_supremum_pour_penteopole_2}
\end{align}
The first supremum is bounded by $a\|G\|_\infty$. 
Bounding~\eqref{eq_1block_pour_1_pis2} therefore reduces to a one-block estimate for a SSEP of size $\ell$, which is well known. 
Indeed, 
the expectation in~\eqref{eq_qtite_ds_supremum_pour_penteopole_2} satisfies (see e.g. \cite{Eyink1990}): 
\begin{equation}
\limsup_{N\rightarrow\infty}\sup_{\substack{g\geq 0 : E_{\mathcal U_\ell}[g]=1 \\ D^{\text{ex}}_\ell(g)\leq 2a\|G\|_\infty/N}} \Big|E_{\mathcal U_\ell}[g\phi_\ell]\Big| 
= 
O(\ell^ {-1})
.
\end{equation}
This concludes the proof of the one block estimate~\eqref{eq_1block_pour_1_pis2}.  
The two block estimate~\eqref{eq_2block_pour_1_pis2} is proven similarly using \cite{Eyink1990}.
\end{proof}
Now that we know that the time integral of ${\bf 1}_{p=2}$ and of the slope at the poles are close, it remains to compute their common value. This is the goal of the next two sections.
\subsubsection{A compactness result}\label{subsec_compactness_estimate_1pis2}
In the previous section, the ${\bf 1}_{p=2}$ term was viewed as an occupation number in a SSEP. 
In this section and the next, we compute its time average by looking directly at the pole dynamics and comparing it with well-chosen zero-range dynamics. 
A similar comparison is made in the proof of the hydrodynamic limit in \cite{Lacoin2014a} using the monotonicity of the zero temperature Glauber dynamics (with a different zero-range dynamics as the parameter $\beta$ was not present there). 

Monotonicity is very useful to cut out a portion of the interface around the pole that can be compared to a zero-range process. 
The lack of monotonicity in the contour dynamics makes the definition of such a portion challenging. 
To do so, we give below a number of estimates on the shape of the interface around the poles. 
In the SSEP picture, 
these estimates say that the slope around the pole is bounded away from $0$ and $1$. 
Slope $0$ and $1$ respectively correspond to very flat/very peaked curves around the pole. 
In terms of zero-range configurations, 
the estimates on the slope in particular imply a compactness result. 
Indeed, they imply that the number of particles in a zero-range process of size $\ell\in\N_{\geq 1}$ is bounded by $C(\ell)$ independently of the scaling parameter $N$. \\

The first estimate is a control of the ${\bf 1}_{p=2}$ term. As shown in Section~\ref{subsec_1_pi2_is_slope_around_poles}, this term coincides with the slope around each pole, so that the next result can be understood as proving that poles are typically not flat, a statement made precise afterwards.
\begin{lemm}[Tail estimate on the flatness of the pole]\label{lemm_required_compactness_pis2}
For $\gamma\in \Omega^N_{\text{mic}}$ with associated droplet $\Gamma$, let $p'(\gamma)$ be the number of blocks in $\Gamma$ composing the next level below the north pole, as defined in~\eqref{eq_def_p_prime}. If $C>0$ and $A\geq 2$ is an integer:
\begin{equation}
\lim_{N\rightarrow\infty}\sup_{\substack{f\geq 0:E_\nu[f]=1 \\ D_N(f^{1/2})\leq C/N}}\nu_f\Big(p=2, \gamma^{-,1}\in\Omega^N_{\text{mic}},p'\geq A\Big)
\leq 
\frac{1}{\log A}
.
\label{eq_tail_to_prove_for_compactness_pis_2}
\end{equation}
\end{lemm}
\begin{proof}
Fix a density $f$ with $D_N(f^{1/2})\leq C/N$ throughout.   
By definition of $\Omega^N_{\text{mic}}$, 
one has $p'\geq p$ (recall that $p$ is the number of blocks in the pole). 
The idea is to estimate $\nu_f\big(p=2,\gamma^{-,1}\in\Omega^N_{\text{mic}},p'\geq A\big)$ for $A\geq 2$ in terms of $\nu_f(p=A,\gamma^{-,1}\in\Omega^N_{\text{mic}})$ by a bijection argument similar to the one of Lemma~\ref{lemm_bijection_argument}. 
The claim will then follow from a summation using the straightforward bound:
\begin{equation}
\sum_{B\geq 2}\nu_f\Big(p=B,\gamma^{-,1}\in\Omega^N_{\text{mic}}\Big) 
\leq 
1
.
\label{eq_avt_def_tilde_gamma}
\end{equation}
Let us present the aforementioned bijection argument. Fix an integer $A\geq 2$. 
A curve $\gamma$ in $\big\{p=2,\gamma^{-,1}\in\Omega^N_{\text{mic}},p'\geq A\big\}$ can be turned into a curve $F(\gamma)$ of $\big\{p=A,\gamma^{-,1}\in\Omega^N_{\text{mic}}\big\}$ as follows. 
Blocks in the north pole of $\gamma$ have centre at height $z_1(\gamma)-\frac{1}{2N}$ by definition. 
By definition there are at least $A$ blocks in the level below the pole, 
i.e. with centre at height $z_1(\gamma)-\frac{3}{2N}$.

Add up to $A-2$ blocks with centre at height $z_1(\gamma)-\frac{1}{2N}$ to the left of the north pole of $\gamma$ in such a way that the resulting curve is in $\Omega^{N}_{\text{mic}}$. 
If exactly $A-2$ such blocks can be added, an element of $\big\{p=A,\gamma^{-,1}\in\Omega^N_{\text{mic}}\}$ has been created. 
If $B<A-2$ blocks only can fit to the left of the pole, add the remaining $A-2-B$ blocks to the right of the pole. The mapping $F$ is illustrated on Figure~\ref{fig_mapping_F}.

\begin{figure}
\begin{center}
\includegraphics[width=11cm]{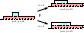} 
\caption{The mapping $F$ for different values of $A$. Here, there are $p'=10$ blocks with centre at height $z_1-\frac{3}{2N}$, in red and delimited by dashed lines. The two cyan blocks mark the position of the pole before the mapping by $F$.\label{fig_mapping_F}}
\end{center}
\end{figure}
Since $p'\geq A$, the above procedure always makes sense. 
In order to compare the dynamical cost of turning $\gamma$ into $F(\gamma)$, 
let us list properties of $F$.

Notice first that $F$ leaves the equilibrium measure $\nu$ invariant: $\nu(\gamma) = \nu(F(\gamma))$. 
This is because the length of $\gamma\in \Omega^N_{\text{mic}}$ and $F(\gamma)$ are the same.

Notice next that the mapping $F$ is nearly bijective in the following sense. Label each of the $p'$ blocks with centre at height $z_1(\gamma)-\frac{3}{2N}$ from $1$ to $p'$, starting from the left. 
We shall say that a block $C$ in the pole is \emph{above the block with label }$i$ if the centre of $C$ has the same abscissa and ordinate higher by $1$ than the centre of the block with label $i$.
\begin{itemize}
	\item Let $\{P\geq A\}\subset\Omega^N_{\text{mic}}$ be the event that the first block in the pole of $\gamma$ is above a block with label $n\geq A$. 
	If $\gamma\in \{P\geq A\}$, 
	then the leftmost block in the pole of $F(\gamma)$ is above the one with label $n-(A-2)>1$. 
	In that case $F(\gamma)$ is different for different values of $n$, 
	and deleting the leftmost $A-2$ blocks in the pole of $F(\gamma)$ transforms it back into $\gamma$.  
	\item Let $\gamma\in\Omega^N_{\text{mic}}$ satisfy $\gamma^{-,1}\in\Omega^N_{\text{mic}}$ and have pole of size $2$. 
	Let $1\leq n\leq A-1$ and suppose instead that $\gamma$ has pole located above the blocks with labels $n,n+1$. 
	Denote this event by $P = \{n,n+1\}$. 
	Then, for any $n$ with $1\leq n\leq A$ and any $\gamma\in\{P=\{n,n+1\}\}$, 
	the procedure described above turns $\gamma$ into the same $F(\gamma)\in\{p=A,\gamma^{-,1}\in\Omega^N_{\text{mic}}\}$. 
	This $F(\gamma)$ is in $\{P=1,...,A\}$, defined as the event that a curve has north pole composed of the blocks above those with labels $1,...,A$, see Figure~\ref{fig_mapping_F}.
\end{itemize}
We shall use the following compact reformulation of the above two cases: for each $n\leq A-1$,
\begin{equation}
F\Big(\Big\{p=2,\gamma^{-,1}\in\Omega^N_{\text{mic}},p'\geq A,\ P = \{n,n+1\}\Big\}\Big) = \Big\{\gamma^{-,1}\in\Omega^N_{\text{mic}},\ P= \{1,...,A\}\Big\}
\label{eq_first_event_mapping_F}
\end{equation}
and, more generally writing $\{P\geq n \}$ (resp.: $\{P\leq n\}$) for the events that all blocks in the pole are above blocks with labels at least (resp.: at most) $n$:
\begin{equation}
F\Big(\Big\{p=2,\gamma^{-,1}\in\Omega^N_{\text{mic}},p'\geq A,\ P \geq A\Big\}\Big) = \Big\{p=A,\gamma^{-,1}\in\Omega^N_{\text{mic}},\ P\geq 2\Big\}
.
\end{equation}
The event on the right-hand side is disjoint from the event on the right-hand side of~\eqref{eq_first_event_mapping_F}.

We now prove~\eqref{eq_tail_to_prove_for_compactness_pis_2}. 
To fix ideas, consider first the case where $f$ is the density constant equal to $1$. 
As $F$ leaves the equilibrium measure $\nu$ invariant,	
\begin{align}
\nu\Big(&p=2,\gamma^{-,1}\in\Omega^N_{\text{mic}},p'\geq A\Big) 
= 
\nu\Big(p=2,\gamma^{-,1}\in\Omega^N_{\text{mic}},p'\geq A,\ P\leq A-1\Big) 
\nonumber\\
&\hspace{6cm}
+ \nu\Big(p=2,\gamma^{-,1}\in\Omega^N_{\text{mic}},p'\geq A,\ P\geq A\Big)
\nonumber\\
&\qquad
= (A-1)\nu\Big(\gamma^{-,1}\in\Omega^N_{\text{mic}},P=\{1,...,A\}\Big) 
+ 
\nu\Big(p=A,\gamma^{-,1}\in\Omega^N_{\text{mic}},P\geq 2\Big)
.
\end{align}
Each of the above two events only contains curves with $p=A$, thus:
\begin{align}
\nu\Big(p=2,\gamma^{-,1}\in\Omega^N_{\text{mic}},p'\geq A\Big)\leq (A-1)\nu\Big(p=A,\gamma^{-,1}\in\Omega^N_{\text{mic}}\Big)
.
\label{eq_link_nu_pis2_and_nu_pisA}
\end{align}
To obtain~\eqref{eq_tail_to_prove_for_compactness_pis_2} in the $f= 1$ case,  
fix an integer $B\geq 2$ and apply~\eqref{eq_link_nu_pis2_and_nu_pisA} to each $A\in\{2,...,B\}$ to find:
\begin{align}
1\geq \sum_{A=2}^B\nu\Big(p=A,\gamma^{-,1}\in\Omega^N_{\text{mic}}\Big) \geq \sum_{A=2}^B\frac{1}{A-1}\nu\Big(p=2,\gamma^{-,1}\in\Omega^N_{\text{mic}},p'\geq A\Big).\label{eq_debut_estim_pis_2_sous_nu}
\end{align}
For $\ell\geq 2$, let $H_\ell = \sum_{n=2}^\ell(n-1)^{-1}$, $H_1 := 0$ and integrate the right-hand side of the last equation by parts:
\begin{equation}
1\geq \nu\Big(p=2,\gamma^{-,1}\in\Omega^N_{\text{mic}},p'\geq B\Big) H_B + \sum_{A=2}^{B-1}H_{A}\nu\Big(p=2,\gamma^{-,1}\in\Omega^N_{\text{mic}},p'=A\Big)
\end{equation}
Equation~\eqref{eq_tail_to_prove_for_compactness_pis_2} when $f=1$ follows (in fact also at each $N$ and not just in the limit):
\begin{equation}
\limsup_{N\rightarrow\infty}\nu\Big(p=2,\gamma^{-,1}\in\Omega^N_{\text{mic}},p'\geq B\Big) 
\leq 
H_B^{-1}\leq \frac{1}{\log B} 
.
\label{eq_limite_nu_f_p_is_2_p'_gtr_than_A_is_0}
\end{equation}
We now prove~\eqref{eq_tail_to_prove_for_compactness_pis_2} for a general density $f$ for $\nu$ satisfying $D_N(f^{1/2})\leq C/N$. 
By the above discussion, 
it is enough to prove that, 
up to an error that vanishes for $N$ large,
~\eqref{eq_link_nu_pis2_and_nu_pisA} holds also under $\nu_f$. 
To prove this,  
the idea is similar to the one used in Lemma~\ref{lemm_bijection_argument}: 
For each suitable $\gamma\in\Omega^N_{\text{mic}}$, 
the transformation $\gamma\to F(\gamma)$ (defined below~\eqref{eq_avt_def_tilde_gamma}) is decomposed into a succession of dynamical moves. 
One then notices that the number of required moves does not depend on $N$. 
In this way the difference between $f(\gamma)$ and $f(F(\gamma))$ can be expressed in terms of the Dirichlet form.

We only carry out the argument for the $\{P=\{1,...,A\},\gamma^{-,1}\in\Omega^N_{\text{mic}}\}$ term in~\eqref{eq_link_nu_pis2_and_nu_pisA}, 
the $\{P\geq 2\}$ term is similar. 
To lighten notation, defined the event $E_{P\leq A-1}$ as follows (as well as $E_{P\geq A}$ for future reference):
\begin{align}
E_{P\leq A-1} 
&:=
\Big\{ p=2,\gamma^{-,1}\in\Omega^N_{\text{mic}},p'\geq A,\ P\leq A-1\Big\}
,\nonumber\\
E_{P\geq A} 
&:=
\Big\{ p=2,\gamma^{-,1}\in\Omega^N_{\text{mic}},p'\geq A,\ P\geq A\Big\}
.
\label{eq_def_E_PgeqA}
\end{align}
Start from~\eqref{eq_first_event_mapping_F} to write:
\begin{align}
(A-1)\nu_f&\Big(P=\{1,...,A\},\gamma^{-,1}\in\Omega^N_{\text{mic}}\Big) 
= 
(A-1)\sum_{\tilde\gamma\in \Omega^N_{\text{mic}}} {\bf 1}_{\tilde\gamma\in F(E_{P\leq A-1})}\nu(\tilde\gamma)f(\tilde\gamma)
.
\end{align}
As explained above, there are $A-1$ different curves in $E_{P\leq A-1}$ with the same image $\tilde\gamma$ by $F$. Moreover, $F$ leaves the measure $\nu$ invariant, thus:
\begin{align}
(A-1)\nu_f\Big(\gamma^{-,1}\in\Omega^N_{\text{mic}},P=\{1,...,A\}\Big) 
&=
\sum_{\tilde\gamma\in \Omega^N_{\text{mic}}} {\bf 1}_{\tilde\gamma\in F(E_{P\leq A-1})}\, \nu(\tilde\gamma)f(\tilde\gamma)\sum_{\gamma\in E_{P\leq A-1}}{\bf 1}_{F(\gamma) = \tilde\gamma} 
\nonumber\\
&=
\sum_{\gamma\in E_{P\leq A-1}}\nu(\gamma)f(F(\gamma))
.
\label{eq_estim_nu_p_is_A_0}
\end{align}
Adding and subtracting appropriate terms,~\eqref{eq_estim_nu_p_is_A_0} can be written as:
\begin{align}
(A-1)\nu_f&\Big(\gamma^{-,1}\in\Omega^N_{\text{mic}},P=\{1,...,A\}\Big) =\sum_{\gamma\in E_{P\leq A-1}}\nu(\gamma) \big[f^{1/2}(F(\gamma))-f^{1/2}(\gamma)\big]^2\nonumber\\
&\qquad+\sum_{\gamma\in E_{P\leq A-1}} \nu(\gamma)\big[-f(\gamma) +2f^{1/2}(\gamma)f^{1/2}(F(\gamma))\big].
\end{align}
Cauchy-Schwarz inequality applied to the terms involving $f^{1/2}(\cdot) f^{1/2}(F(\cdot))$ then yields:
\begin{align}
\Big[(A-1)^{1/2}\nu_f\Big(P=&\{1,...,A\},\gamma^{-,1}\in\Omega^N_{\text{mic}}\Big)^{1/2}  - \nu_f\big(E_{P\leq A-1}\big)^{1/2}\Big]^{2} \nonumber\\
&\leq \sum_{\gamma\in E_{P\leq A-1}} \nu(\gamma) \big[f^{1/2}(F(\gamma))-f^{1/2}(\gamma)\big]^2.\label{eq_transit_compactness_arg_pis2_0}
\end{align}
It remains to bound the right-hand side of~\eqref{eq_transit_compactness_arg_pis2_0} in terms of the Dirichlet form. 
Decompose the transformation $\gamma\to F(\gamma)$ into flips adding a single block to the pole: $\gamma = \gamma_0 \rightarrow\gamma_1\rightarrow...\rightarrow\gamma_{A-2} = F(\gamma)$ and apply Cauchy-Schwarz inequality to find:
\begin{align}
\sum_{\gamma\in E_{P\leq A-1}}\nu(\gamma) \big[f^{1/2}(F(\gamma))-f^{1/2}(\gamma)\big]^2 
\leq 
(A-2)  \sum_{\gamma\in E_{P\leq A-1}}\nu(\gamma)\sum_{j=1}^{A-2}\big[f^{1/2}(\gamma_{j})-f^{1/2}(\gamma_{j-1})\big]^2
.
\end{align}
Each transition $\gamma_{j-1}\to\gamma_{j}$ is authorised in the contour dynamics, at rate $1/2$. A given curve corresponding to one of the $\gamma_j$ can occur at most $A-1$ times in all paths $\gamma\rightarrow F(\gamma)$ for $\gamma\in E_{P\leq A-1}$. 
As a result and since $\nu(\gamma_j) = \nu(\gamma)$ for all $1\leq j\leq A-2$:
\begin{align}
\Big[(A-1)^{1/2}\nu_f\Big(P = \{1,...,A\},\gamma^{-,1}\in\Omega^N_{\text{mic}}\Big)^{1/2} - \nu_f\big(E_{P\leq A-1}\big)^{1/2}\Big]^{2} 
\leq 
4(A-1)^{2}D_N(f^{1/2})
.
\label{eq_first_terme_majore_bij_pour_estim_compactness}
\end{align}
Similar computations give the same kind of bound for the second term in~\eqref{eq_link_nu_pis2_and_nu_pisA} under $\nu_f$ (recall the definition~\eqref{eq_def_E_PgeqA} of $E_{P\geq A}$):
\begin{align}
\Big[\nu_f\Big( P\geq 2 ,p=A,\gamma^{-,1}\in\Omega^N_{\text{mic}}\Big)^{1/2} - \nu_f\big(E_{P\geq A}\big)^{1/2}\Big]^{2} 
\leq 4(A-1)D_N(f^{1/2})
.
\label{eq_second_terme_majore_bij_pour_estim_compactness}
\end{align}
Let us use~\eqref{eq_first_terme_majore_bij_pour_estim_compactness}--\eqref{eq_second_terme_majore_bij_pour_estim_compactness} to prove that~\eqref{eq_link_nu_pis2_and_nu_pisA} still holds under $\nu_f$ with a small error in $N$ (recall that $D_N(f^{1/2})\leq C/N$). 
Equation~\eqref{eq_first_terme_majore_bij_pour_estim_compactness} yields:
\begin{align}
\nu_f\big(E_{P\leq A-1}\big)\nonumber\\
&\leq (A-1)\nu_f\Big(P = \{1,...,A\},\gamma^{-,1}\in\Omega^N_{\text{mic}}\Big)+ C(A)\big[D_N(f^{1/2})^{1/2} + D_N(f^{1/2})\big]  \nonumber\\
&\leq (A-1)\nu_f\Big(P = \{1,...,A\},\gamma^{-,1}\in\Omega^N_{\text{mic}}\Big)+ C(A)N^{-1/2},
\end{align}
where the constant $C(A)>0$ changes between inequalities. Similarly,~\eqref{eq_second_terme_majore_bij_pour_estim_compactness} yields:
\begin{align}
\nu_f\big(E_{P\geq A}\big)
\leq 
\nu_f\Big( P\geq 2 ,p=A,\gamma^{-,1}\in\Omega^N_{\text{mic}}\Big)+C(A)N^{-1/2}
,
\end{align}
whence the following counterpart of~\eqref{eq_link_nu_pis2_and_nu_pisA} for $\nu_f$:
\begin{align}
\nu_f\Big(p=2,\gamma^{-,1}\in\Omega^N_{\text{mic}},p'\geq A\Big) 
&= 
\nu_f\big(E_{P\leq A-1}\big) 
+ \nu_f\big(E_{P\geq A}\big)
\nonumber\\
&\leq 
(A-1)\nu_f\Big(p=A,\gamma^{-,1}\in\Omega^N_{\text{mic}}\Big) + C(A)N^{-1/2}
.
\label{eq_estim_prelim_lien_pis2_pis_A}
\end{align}
Equation~\eqref{eq_estim_prelim_lien_pis2_pis_A} is sufficient to conclude the proof of the upper bound as in the $f\equiv 1$ case, see~\eqref{eq_debut_estim_pis_2_sous_nu} to~\eqref{eq_limite_nu_f_p_is_2_p'_gtr_than_A_is_0}. Indeed, the bound in~\eqref{eq_limite_nu_f_p_is_2_p'_gtr_than_A_is_0} requires only the use of $A$ independent from $N$, so that $C(A)N^{-1/2}$ vanishes when $N$ is large. We therefore conclude the proof here. 
\end{proof}
In the next two lemmas, we use Lemma~\ref{lemm_required_compactness_pis2} to control the number of horizontal edges in a curve as a function of the vertical distance to the north pole (Lemma~\ref{lemm_width_compactness_pis2}) as well as, conversely, the vertical distance to the north pole as a function of the number of blocks (Lemma~\ref{lemm_height_compactness_pis2}). 
More precisely, let $\gamma\in\Omega^N_{\text{mic}}\cap \e$ and let $n\in\N_{\geq 1}$. 
The line $y = z_1(\gamma) - n/N$ contains a certain number of horizontal edges in $\gamma$ (recall that $z_1$ is the ordinate of the highest points in $\gamma$, defined in~\eqref{eq_def_y_max_sec_6}).  
Let $\ell(n)$ be the number of these edges to the right of $L_1$ and $\ell(-n)$ be the number of edges to the left of $L_1$. 
Define also $\ell(0) := p(\gamma)-2$. For $N$ large enough, $\gamma\in \Omega^N_{\text{mic}}\cap\e$ implies that each of the $\ell(i)$, $|i|\leq n\ll N$ are well defined, see Figure~\ref{fig_estim_compactness}.
\begin{lemm}[Width of a curve at depth $n$ below the pole]\label{lemm_width_compactness_pis2}
For $n\in\N_{\geq 1},C>0,A\geq 2$,
\begin{equation}
\forall |i|\leq n,\qquad \limsup_{N\rightarrow\infty}\sup_{\substack{f\geq 0 : E_\nu[f] =1 \\ D_N(f^{1/2})\leq C/N}}\nu_f\Big(\e,\ell(i)\geq A\Big) \leq \frac{e^{2\beta}}{(A+1)\log (A+2)}.\label{eq_increment_width_lvl_k}
\end{equation}
As a result, the numbers $w_n^+ =2+ \sum_{i=0}^ n\ell(i)$ and $w^ -_k = \sum_{i=1}^{n}\ell(-i)$ of blocks with centres at height $z_1(\gamma)-(n+1/2)/N$ in a droplet $\Gamma$ associated to $\gamma\in \e$, respectively to the right/to the left of $L_1$ (see Figure~\ref{fig_estim_compactness}), satisfy:
\begin{equation}
 \limsup_{N\rightarrow\infty}\sup_{\substack{f\geq 0 : E_\nu[f] =1 \\ D_N(f^{1/2})\leq C/N}}\nu_f\Big(\e,w^\pm_n\geq n^2 \Big)
 \leq 
 \frac{3e^{2\beta}}{\log n}
 .
 \label{eq_width_level_k}
\end{equation}
\end{lemm}
\begin{figure}
\begin{center}
\includegraphics[width=10cm]{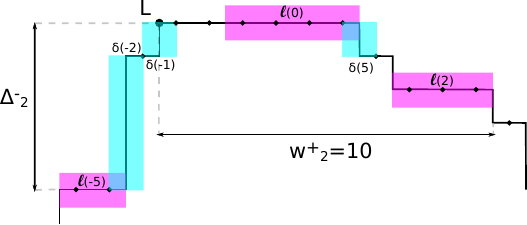} 
\caption{Definition of the $\delta(\pm i),\ell(\pm i),\Delta^{\pm}_n,w^{\pm}_n$. Here $w^+_2=10$, the number of blocks to the right of $L_1$ with centre at height $z_1 - \frac{5}{2N}$. The small black dots mark the centre of each horizontal edge. The shaded areas highlight which edges make up the represented $\ell(\pm i),\delta(\pm i)$. Here, $\ell(-2) = \delta(-3) = 0$.\label{fig_estim_compactness}}
\end{center}
\end{figure}
\begin{proof}
Equation~\eqref{eq_width_level_k} follows from~\eqref{eq_increment_width_lvl_k} by a union bound. Let us prove~\eqref{eq_increment_width_lvl_k} by recursion on $|i|\leq n$. 

The set $\e$, given in Definition~\ref{def_effective_state_space} as a volume neighbourhood of a fixed curve, 
is not stable under the contour dynamics around the poles. Rather than work directly with $\e$, it is convenient to consider a subset that has this stability property. 
Introduce then the set $Q_n$ of curves for which there is at least $n$ vertical edges between pole $1$ and poles $2,4$ (recall that $L_k/R_k$ is the left/right extremity of pole $k$ with $1\leq k\leq 4$):
\begin{equation}
Q_n := \Big\{\gamma\in\Omega^N_{\text{mic}} : [R_1(\gamma)-L_2(\gamma)]\cdot{\bf b}_{\pi/2}>\frac{n}{N}, [R_1(\gamma)-R_4(\gamma)]\cdot{\bf b}_{\pi/2}>\frac{n}{N}\Big\}.\label{eq_def_Q_k}
\end{equation}
Note the factor $1/N$ as elements of $\Omega^N_{\text{mic}}$ are rescaled by definition. 
Contrary to $\e$, $Q_n$ is stable under any dynamical move 1) affecting points at vertical distance at most $n$ below the pole and any horizontal distance; 2) that do not change the height of the north pole. 
Moreover, for large enough $N\in\N_{\geq 1}$, it holds that $\e\subset Q_n$ as $\e$ only contains curves with positive distance between each pole, recall Definition~\ref{def_effective_state_space}. 
Fix one such $N$, $C>0$ and a density $f$ for $\nu$ with $D_N(f^{1/2})\leq C/N$ throughout. 
We are going to prove~\eqref{eq_increment_width_lvl_k} with $Q_n$ replacing $\e$. 
To start the recursion, consider $\ell(0)$.
Recall that $p = \ell(0)+2$ by definition. Thus:
\begin{align}
(A+1)\nu_f\Big(Q_n,\ell(0)\geq A\Big)&\leq e^{2\beta} E_{\nu_f}\big[(\ell(0)+1)e^{-2\beta}{\bf 1}_{\ell(0)\geq A}\big] \nonumber\\
&= e^{2\beta} E_{\nu_f}\big[(p-1)e^{-2\beta}{\bf 1}_{p\geq A+2}\big]
.
\label{eq_rewriting_proba_ell_larger_A}
\end{align}
Lemma~\ref{lemm_bijection_argument} gives:
\begin{align}
&\Big|\nu_f\Big(p=2,\gamma^{-,1}\in\Omega^N_{\text{mic}},p'\geq A+2\Big)^{1/2} -  E_{\nu_f}\Big[(p-1)e^{-2\beta}{\bf 1}_{p\geq A+2}\Big]^{1/2}\Big|^2
\leq 
2D_N(f^{1/2})
,
\nonumber\\
&E_{\nu_f}\Big[(p-1)e^{-2\beta}{\bf 1}_{p\geq A+2}\Big]^{1/2}
\leq 
1+\big(2D_N(f^{1/2})\big)^{1/2}
.
\end{align}
Using the identity $a-b = (\sqrt{a}-\sqrt{b})(\sqrt{a}+\sqrt{b})$ for $a,b\geq 0$, we then find (here $Q_n$ plays no special role):
\begin{align}
\Big|\nu_f\Big(&p=2,\gamma^{-,1}\in\Omega^N_{\text{mic}},p'\geq A+2\Big) -  E_{\nu_f}\Big[(p-1)e^{-2\beta}{\bf 1}_{p\geq A+2}\Big]\Big|
\label{eq_lien_l_0_p}
\\
&\quad\leq 
\big(2D_N(f^{1/2})\big)^{1/2}\Big[2 + \big(2D_N(f^{1/2})\big)^{1/2}\Big] 
\leq 
C'N^{-1/2},
\qquad 
C' 
= 
2C+2\sqrt{2C}
. 
\nonumber
\end{align}
The bound~\eqref{eq_increment_width_lvl_k} for $i=0$ then follows:
\begin{align}
(A+1)\nu_f\Big(\e,\ell(0)\geq A\Big)
&\ \, \leq 
(A+1)\nu_f\Big(Q_n,\ell(0)\geq A\Big)
\nonumber\\
&\hspace{-0.5cm}\overset{\eqref{eq_rewriting_proba_ell_larger_A}-\eqref{eq_lien_l_0_p}}{\leq} e^{2\beta}\nu_f\Big(p=2,\gamma^{-,1}\in\Omega^N_{\text{mic}},p'\geq A+2\Big) + \frac{e^{2\beta}C'}{N^{1/2}}\nonumber\\
&\overset{\eqref{eq_tail_to_prove_for_compactness_pis_2}}{\leq} \frac{e^{2\beta}}{\log(A+2)}+o_N(1).
\end{align}
Assume now that, for some integer $i$ with $|i|<n$:
\begin{align}
\limsup_{N\rightarrow\infty}\sup_{\substack{f\geq 0 : E_\nu[f] =1 \\ D_N(f^{1/2})\leq C/N}}\nu_f\Big(Q_n,\ell(i)\geq A\Big) \leq \frac{e^{2\beta}}{(A+1)\log(A+2)}.\label{eq_recursion_relation_length}
\end{align}
For definiteness, assume $i\geq 0$. To show~\eqref{eq_recursion_relation_length} for $i+1$, we are going to prove:
\begin{equation}
\nu_f\Big(Q_n,\ell(i+1)\geq A\Big) = \nu_f\Big(Q_n, \ell(i+1)\geq 0,\ell(i)\geq A\Big) + O\big(D_N(f^{1/2})^{1/2}+D_N(f^{1/2})\big).\label{eq_recurrence_width_increment}
\end{equation}
Equation~\eqref{eq_recurrence_width_increment} implies~\eqref{eq_increment_width_lvl_k} up to $i+1$, 
since $\e\subset Q_n$ for large $N$ and the Dirichlet form vanishes. 

The idea behind~\eqref{eq_recurrence_width_increment} is the following. 
Take a curve $\gamma\in \{\ell(i+1)\geq A\}\cap Q_n$ (see Figure~\ref{fig_estim_compactness} for a representation of $\ell(i)$). 
One can then add at least $A$ blocks to $\Gamma$ with centre at height $z_1(\gamma)-(i+1/2)/N$, one above each of the edges ensuring $\ell(i+1)\geq A$. 
This is where working with curves in $Q_n$ instead of $\e$ is convenient: 
the resulting curve $F'(\gamma)$ is still in $Q_n$ although it might not belong to $\e$. 
The above procedure therefore bijectively yields a curve $F'(\gamma)\in Q_n$ with $\ell(i)\geq A$. 
Moreover, $\gamma$ and $F'(\gamma)$ have the same length. 
The two events on either side of~\eqref{eq_recurrence_width_increment} thus have the same $\nu$-measure. 

Under $\nu_f$, though, the two events in~\eqref{eq_recurrence_width_increment} may have different probability. However, the mapping $\gamma\mapsto F'(\gamma)$ can be decomposed into a sequence of $A$ curves $\gamma=\gamma_1\rightarrow...\rightarrow\gamma_A = F'(\gamma)$, each differing from the previous one by a single block. Each curve $\gamma_j$ ($j\leq A$) appears at most $A+1$ times when effecting the procedure for all curves in $\{\ell(i+1)\geq A\}\cap Q_n$. The cost of turning $\gamma$ into $F'(\gamma)$ under $\nu_f$ is then estimated by the Dirichlet form in a very similar fashion to the bijection argument of Lemma~\ref{lemm_required_compactness_pis2}, so we give no more details.
\end{proof}
The next lemma controls the depth at fixed horizontal distance to the pole. For $n\in\N_{\geq 1}$ and $|i|\leq n$, define $\delta(i)$ as the number of vertical edges with abscissa $L_1\cdot {\bf b}_0 + \frac{i+1}{N}$ that belong to either region 4 (if $i\leq 0$) or region 1 ($i\geq 0$), see Figure~\ref{fig_estim_compactness}. 
Note that $\delta(0) = 0$, corresponding to the fact that the pole contains at least two horizontal edges.
\begin{lemm}[Height of a curve at horizontal distance $n$ to the pole]\label{lemm_height_compactness_pis2}
Let $\beta>\log 2$. For $n\in\N_{\geq 1}$ and $C>0,A\geq 1$,
\begin{equation}
\forall 1\leq i,j\leq n,\qquad \limsup_{N\rightarrow\infty}\sup_{\substack{f\geq 0:E_\nu[f]=1 \\ D_N(f^{1/2})\leq C/N}}\nu_f\Big(\e,\delta(i)\geq A,\delta(-j)\geq A\Big) \leq e^{-2\beta (A-1)}.\label{eq_increment_height_distance_k}
\end{equation}
Let $\Delta^\pm_n = \sum_{ i= 1}^n\delta(\pm i)$ be the number of vertical edges in a curve between the point $L_1+\frac{{\bf b}_0}{N}$ and the first point at horizontal distance $n+1$ to the right (for $\Delta^+_n$) or to the left (for $\Delta^-_n$), 
see Figure~\ref{fig_estim_compactness}. 
Then:
\begin{equation}
 \limsup_{N\rightarrow\infty}\sup_{\substack{f\geq 0:E_\nu[f]=1 \\ D_N(f^{1/2})\leq C/N}}\nu_f\Big(\e,\Delta^+_n\geq n(1+2\log n),\Delta^-_n\geq n(1+2\log n)\Big)
 \leq 
 \frac{1}{n^{4\beta-2}}
 =
 o_n(1)
 .
 \label{eq_height_distance_k}
\end{equation}
\end{lemm}
\begin{proof}
Let $n\in\N_{\geq 1}$ be fixed. Equation~\eqref{eq_height_distance_k} follows from~\eqref{eq_increment_height_distance_k} by a union bound. 
The proof of~\eqref{eq_increment_height_distance_k} resembles that of Lemma~\ref{lemm_width_compactness_pis2}. 
We first treat the case $i=j=1$. $\{\delta(1)\geq A,\delta(-1)\geq A\}$ is the event that the north pole is on top of a stack which contains two blocks in width and at least $A$ blocks in height, 
the highest two blocks being those in the pole. 
With each $\gamma\in \{\delta(1)\geq A,\delta(-1)\geq A\}$ associate a curve $\tilde F(\gamma) \in\{\delta(1)\geq 1,\delta(-1)\geq 1\}$ in which the top $A-1$ pairs of two blocks have been deleted. 
The curve $\tilde F(\gamma)$ has length $|\gamma|-2(A-1)$, thus has higher equilibrium probability. 
In fact, if $\gamma^{-q}$ denotes the curve $\gamma$ in which the highest $q\in\N_{\geq 1}$ groups of two blocks have been deleted, 
$\tilde F$ is a bijection between the sets $\big\{\delta(1)\geq A,\delta(-1)\geq A,\gamma^{-(A-1)}\in\Omega^N_{\text{mic}}\big\}$ and $\big\{\delta(1)\geq 1,\delta(-1)\geq 1\big\}$. In addition:
\begin{equation}
\nu\Big(\delta(1)\geq A,\delta(-1)\geq A,\gamma^{-(A-1)}\in\Omega^N_{\text{mic}}\Big) 
= 
e^ {-2\beta(A-1)}\nu\Big( \delta(1)\geq 1,\delta(-1)\geq 1\Big)
  \leq 
e^ {-2\beta(A-1)}
.
\label{eq_true_equality_shrinking_curve_height_estimate}
\end{equation}
In the same way as in Lemma~\ref{lemm_required_compactness_pis2},~\eqref{eq_true_equality_shrinking_curve_height_estimate} holds under $\nu_f$ for any $\nu$-density $f$ up to an error term which quantifies the cost of consecutively deleting the two blocks in the pole of a curve $A-1$ times. As a result:
\begin{align}
&\sup_{\substack{f\geq 0 : E_{\nu}[f]=1 \\ D_N(f^{1/2})\leq C/N}}\nu_f\Big(\delta( 1)\geq A, \delta(-1)\geq A,\gamma^{-(A-1)}\in\Omega^N_{\text{mic}}\Big) \leq e^{-2\beta(A-1)}+O\big(N^{-1/2}\big). \label{eq_controle_height_increments_i_j_are_1}
\end{align}
As $\e \subset \{\gamma^{-(A-1)}\in\Omega^N_{\text{mic}}\}$ for all large enough $N$,~\eqref{eq_increment_height_distance_k} holds for $i=j=1$. \\

To prove~\eqref{eq_increment_height_distance_k} for each $(i,j)\in\{1,...,n\}$, let us first prove it for $j=1$, $i>1$. As for Lemma~\ref{lemm_width_compactness_pis2}, it is convenient to not work directly with $\e$, which does not have nice stability properties under the contour dynamics, but with the set $\tilde Q_n$ of curves with regions $1,4$ each containing at least $n$ horizontal edges in addition to those in the north pole:
\begin{equation}
\tilde Q_n := \Big\{\gamma\in\Omega^N_{\text{mic}} : \big[L_1(\gamma)- R_4(\gamma)\big]\cdot{\bf b}_0>\frac{n}{N}, \big[L_2(\gamma)-R_1(\gamma)\big]\cdot{\bf b}_0>\frac{n}{N}\Big\}.\label{eq_def_Q_prime_k}
\end{equation}
We claim that, for each $C>0$:
\begin{equation}
\lim_{N\rightarrow\infty}\sup_{\substack{f\geq 0 : E_{\nu}[f]=1 \\ D_N(f^{1/2})\leq C/N}}\Big|\nu_f\Big(\tilde Q_n,\delta(i)\geq A,\delta(-1)\geq A\Big) - \nu_f\Big(\tilde Q_n,\delta(i-1) \geq A,\delta(-1)\geq A\Big)\Big| = 0.\label{eq_incr_height_from_i_to_i_minus_one}
\end{equation}
The idea is the same as in Lemma~\ref{lemm_width_compactness_pis2}. 
The set $\tilde Q_n$ ensures that $\delta(i),|i|\leq n$ are well defined and involve only edges in region $4$ ($i\leq 0$) or $1$ ($i\geq 0$). 
One has again $\e\subset \tilde Q_n$ for large enough $N$ as curves in $\e$ have positive distance between each pole. 
In addition, $\tilde Q_n$ is stable under any dynamical move that 1) involves points at horizontal distance at most $n$ from the north pole and 2) do not change the lateral position of the north pole. 

A curve in $\tilde Q_n$ with $\delta(i)\geq A$ is transformed into one with $\delta(i-1)\geq A$ by deleting $A$ blocks with centres at abscissa $L_1\cdot{\bf b}_0-\frac{i}{N}+\frac{1}{2N}$. 
These deletions are SSEP moves, which do not change the length of the curve (and under which $\tilde Q_n$ is stable). 
Their cost is estimated in terms of the Dirichlet form, which vanishes with $N$.

Iterating~\eqref{eq_incr_height_from_i_to_i_minus_one} from $i$ to $1$ and using~\eqref{eq_controle_height_increments_i_j_are_1} yields~\eqref{eq_increment_height_distance_k} for indices $(i,-1)$. 
Now if $j\neq 1$, the same argument applies to go from $-j$ to $-1$. This concludes the proof of~\eqref{eq_increment_height_distance_k}.
\end{proof}
\subsubsection{Value of the slope at the pole}\label{sec_value_slope_at_the_pole}
We now have all prerequisites to prove that the motion of the north pole imposes a particle density of $e^{-\beta}$ on each side, as stated in Lemma~\ref{lemm_valeur_p_is_2}. 
The proof relies in a central way on the fact that the contour dynamics around the pole is irreducible, owing to the $e^{-2\beta}$ regrowth updates. 
These updates are the main difference with the zero temperature stochastic Ising model and the irreducibility is the main technical reason for the introduction of the parameter $\beta$. 
\begin{lemm}\label{lemm_valeur_p_is_2}
Let $\beta>\log 2$. 
For each $T,\delta>0$ and each test function $G\in\C$, 
\begin{align}
\lim_{N\rightarrow\infty}\frac{1}{N}\log \Prob^N_{\beta}\bigg(&\gamma^N_\cdot \in E([0,T],\e);\nonumber\\
\qquad &\bigg|\frac{1}{T}\int_0^ {T} G(t,L_1(\gamma^N_t))\big({\bf 1}_{p=2}- e^ {-\beta}\big)\, dt\bigg|\geq \delta\bigg) = 
-\infty
.
\label{eq_valeur_p_is_2}
\end{align}
The claim is also valid under $\Prob^N_{\beta,H}$ for $H\in\C$ by Corollary~\ref{coro_change_measure_sous_exp}.
\end{lemm}
\begin{proof}
The proof only deals with $H=0$ as generalisations to $\Prob^N_{\beta,H}$ for $H\in\C$ follow as in the proof of Lemma~\ref{lemm_1_block_2_block_pour_pis2}. 
We first write the proof in the case where the test function $G$ is equal to $1$ and explain at the end how to adapt it to non-constant $G$. 
To lighten notations, we do not explicitly write integer parts and drop the superscript $N$ for microscopic curves and trajectories.

The proof is structured as follows. We first use Lemma~\ref{lemm_projec_dynamics_onto_local} to project the dynamics inside $\e$. The compactness results provided by Section~\ref{subsec_compactness_estimate_1pis2} are then incorporated to the probability in~\eqref{eq_valeur_p_is_2}. This enables us to define a proper frame around the pole. After conditioning to this frame, the quantity to estimate in~\eqref{eq_valeur_p_is_2} can be retrieved from an equilibrium computation, which is the last step of the proof.\\

Let $\phi = {\bf 1}_{p=2} - e ^{-\beta}$ and $a>0$. 
By Markov inequality and Lemma~\ref{lemm_projec_dynamics_onto_local}, the left-hand side of~\eqref{eq_valeur_p_is_2} without the limits is bounded from above by:
\begin{equation}
-a\delta T+ C \beta + T\sup_{f\geq 0 :\nu_f(\e) =1}\Big\{aE_{\nu_f}[\phi] - ND_N(f^{1/2})\Big\}
.
\label{eq_qtite_interm_a_estime_val_pis2_0}
\end{equation}
It is therefore enough to estimate the supremum in~\eqref{eq_qtite_interm_a_estime_val_pis2_0}.\\

\noindent\textbf{Step 1: definition of a suitable frame around the pole.}\\
The first step consists in writing the expectation in~\eqref{eq_qtite_interm_a_estime_val_pis2_0} as a quantity that depends only on the dynamics around the pole. The idea is to compare the contour dynamics to a zero-range process with two species of particles. The number of particles is given by the height difference between consecutive columns around the pole. The species is determined by the sign of the height difference. This process is irreducible and its invariant measure can be made explicit. More is said on this dynamics below, see also Figure~\ref{fig_ZRP}. To make such a comparison, we define a frame around the pole, in which to study the pole dynamics. This is done as follows.\\

Fix an integer $n\in\N_{\geq 1}$, which will be the typical size of the frame around the pole. In the following, for a curve $\gamma\in \Omega^N_{\text{mic}}\cap \e$, we talk of \emph{blocks at level }$q\in\N$ to denote all blocks in $\Gamma$ with centre at height $z_1(\gamma) - N^{-1}(q+1/2)$, see Figure~\ref{fig_def_h_k}. With this notation, blocks at level $0$ correspond to blocks in the poles.\\
Consider the following partition of $\Omega^N_{\text{mic}}\cap \e$. For any curve $\gamma\in \Omega^N_{\text{mic}}\cap \e$, let $h_n(\gamma)$ be the smallest integer such that the number of blocks in $\Gamma$ (the droplet delimited by $\gamma$) at level $z_1- N^{-1}\big(h_n(\gamma)+\frac{1}{2}\big)$ is strictly larger than $n$ (see Figure~\ref{fig_def_h_k}):
\begin{equation}
h_n(\gamma) = \min\big\{ q \in \N : N_q(\gamma)> n\},\label{eq_def_h_k}
\end{equation}
where: 
\begin{equation}
N_q(\gamma) = \Big|\Big\{\text{blocks in } \Gamma \text{ at level }q,\text{ i.e. with centre at height }z_1(\gamma)-\frac{(q+1/2)}{N}\Big\}\Big|.
\end{equation}
\begin{figure}[H]
\begin{center}
\includegraphics[width=9cm]{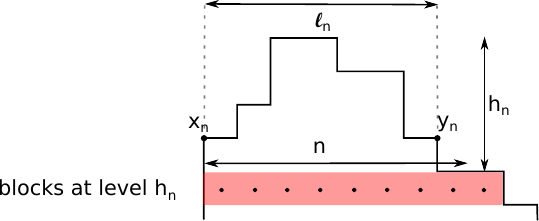} 
\caption{Definition of $h_n,\ell_n$ and $x_n,y_n$ for a given curve. The first level of blocks with width strictly larger than $n$ corresponds to the filled area. In this case there are $n+1$ such blocks, with centres indicated by black dots. The number $\ell_n$ of blocks in the last level containing at most $n$ blocks is equal here to $n-1$. The portion of the curve affected by the ZRP dynamics (see Figure~\ref{fig_ZRP} below) is delimited by dashed lines and the segment $[x_n,y_n]$. \label{fig_def_h_k}}
\end{center}
\end{figure}
These objects are well defined for elements of $\Omega_{\text{mic}}^N\cap\e$ as soon as $N$ is large enough compared to $n$, which we henceforth assume. Let $x_n(\gamma)\leq y_n(\gamma)\in V(\gamma)$ denote the extremal vertices of level $h_n(\gamma)-1$, i.e. the last level of $\Gamma$ with at most $n$ blocks; and let $\ell_n = \ell_n(\gamma)$ denote the number of blocks of this level (see Figure~\ref{fig_def_h_k}):
\begin{equation}
\ell_n(\gamma) := N\|y_n(\gamma)-x_n(\gamma)\|_1.
\end{equation}
The rescaling by $N$ comes from the fact that $x_n(\gamma),y_n(\gamma)$ are points of $(N^{-1}\Z)^2$. The quantity $\ell_n(\gamma)$ is thus an integer. For $2\leq \ell \leq n$, consider the set:
\begin{equation}
M_\ell = \big\{\gamma \in \Omega^N_{\text{mic}}\cap \e: \ell_n(\gamma) = \ell\big\}.\label{eq_def_M_ell}
\end{equation}
Then $(M_\ell)_{2\leq\ell\leq n}$ is a disjoint family which partitions $\Omega^N_{\text{mic}}\cap \e$ by construction. The expectation in~\eqref{eq_qtite_interm_a_estime_val_pis2_0} thus reads, for each $\nu$-density $f$ supported on $\e$:
\begin{equation}
E_{\nu_f}[\phi] = \sum_{2\leq \ell\leq n} E_{\nu_f}[{\bf 1}_{M_\ell}\phi].\label{eq_qtite_interm_a_estime_val_pis2_1}
\end{equation}
At this point, the splitting of curves in the different $M_\ell$ in~\eqref{eq_qtite_interm_a_estime_val_pis2_1} suffers from two flaws. 
On the one hand, the width $\ell$, which will correspond to the number of sites in a zero-range process, 
may be bounded independently of $n$. 
This makes a local equilibrium argument impossible to apply. 
On the other hand, the pole may be macroscopically higher than the points $x_n(\gamma),y_n(\gamma)$. 
The point is thus to find diverging sequences $h_{\max}(n),\ell_{\min}(n)$ such that, for any $C>0$:
\begin{align}
\limsup_{N\rightarrow\infty}\sup_{\substack{f\geq 0 : E_\nu[f]=1 \\ D_N(f^{1/2})\leq C/N}}&\nu_f\Big(\gamma\in\e,h_n(\gamma)\geq h_{\max}(n) \text{ or }\ell_n(\gamma)\leq \ell_{\min}(n)\Big)
=
o_n(1)
.
\end{align}
Lemmas~\ref{lemm_width_compactness_pis2}--\ref{lemm_height_compactness_pis2} enable the construction of such sequences, as we now explain.\\

Consider first the height $h_n(\gamma)$, defined in~\eqref{eq_def_h_k}. 
Then either $h_n(\gamma) = 0$, which corresponds to having at least $n$ blocks in the north pole $P$: $p(\gamma)\geq n$. 
Or $h_n(\gamma)\geq 1$ and there are strictly less than $n$ blocks at level $h_n(\gamma)-1$, 
thus strictly less than $n$ blocks both with abscissa smaller and larger than $L_1\cdot {\bf b}_0$ at this level. 
In both cases, 
recalling from Lemma~\ref{lemm_height_compactness_pis2} that $\Delta^ +_n(\gamma)$ (respectively: $\Delta^-_n(\gamma)$) is the number of vertical edges between the pole and the first point at horizontal distance $n+1$ of $L_1+\frac{{\bf b}_0}{N}$ to its right (respectively: to its left), we find:
\begin{equation}
h_{n}(\gamma)-1\leq  \min\big\{\Delta^+_n(\gamma),\Delta^-_n(\gamma)\big\}.
\end{equation}
Lemma~\ref{lemm_height_compactness_pis2} then gives for each $C>0$:
\begin{align}
&\limsup_{N\rightarrow\infty}\sup_{\substack{f\geq 0: E_{\nu}[f]=1 \\ D_N(f^{1/2})\leq C/N}}\nu_f\Big(\gamma\in \e,h_n(\gamma)\geq h_{\max}(n)\Big) 
=
o_n(1)
,\quad 
h_{\max}(n)
:= 
n(1+2\log n)
.\label{eq_height_typical}
\end{align}
We now turn to an estimate of the number of blocks $\ell_n = \ell_n(\gamma)$ at level $h_n(\gamma)-1$. Recalling the definition of the widths $w^\pm$ from Lemma~\ref{lemm_width_compactness_pis2}, notice first the identity:
\begin{equation}
\forall\gamma\in \Omega^N_{\text{mic}}\cap \e,\qquad \ell_n(\gamma) = w^+_{h_n-1}(\gamma) + w^-_{h_n-1}(\gamma).\label{eq_width_ell_as_width_w_of_h}
\end{equation}
Let us use~\eqref{eq_width_ell_as_width_w_of_h} and a bound on $h_n$ to estimate $\ell_n$. Let $a_n>0$ to be chosen later, fix $C>0$ and a $\nu$-density $f$ with $D_N(f^{1/2})\leq C/N$. According to~\eqref{eq_width_ell_as_width_w_of_h}, one has:
\begin{align}
\Big\{\gamma\in \e,\ell_n(\gamma)\leq a_n\Big\}\subset  \Big\{\gamma\in \e,w^-_{h_n-1}(\gamma) \leq a_n\Big\}\cap\Big\{\gamma\in \e,w^+_{h_n-1}(\gamma) \leq a_n\Big\}.\label{eq_upper_bound_ell_small0}
\end{align}
Consider e.g. the event involving $w^-_{h_n-1}$. In that event, if one travels a horizontal distance $a_n$ to the left of $L_1$, it must then be that at least $h_n-1$ vertical edges have been encountered, so that:
\begin{equation}
\Big\{\gamma\in\e: w^-_{h_n-1}(\gamma)\leq a_n\Big\}\subset \Big\{\gamma\in\e: \Delta^-_{a_n}(\gamma)\geq h_n(\gamma)-1\Big\}.\label{eq_upper_bound_ell_small1}
\end{equation}
Proceeding similarly for $w^+_{h_n-1}$ implies:
\begin{equation}
\Big\{\gamma\in \e,\ell_n(\gamma)\leq a_n\Big\}
\subset
\Big\{\gamma\in\e: \min\big\{\Delta^-_{a_n}(\gamma),\Delta^+_{a_n}\big\}\geq h_n(\gamma)-1\Big\}
.
\end{equation}
To bound the probability of the event $\{\gamma\in\e,\ell_n(\gamma)\leq a_n\}$, 
the idea is then to show that $h_n-1$ has to be larger than some $b_n>0$ with large probability, 
then to choose $a_n$ as a function of $b_n$ such that $\{\gamma\in\e,\Delta^{\pm}_{a_n}(\gamma)\geq b_n\}$ is unlikely by Lemma~\ref{lemm_height_compactness_pis2}. 
To choose $b_n>0$, we make use of Lemma~\ref{lemm_width_compactness_pis2} and the following observation which holds by definition of $h_n$:
\begin{equation}
h_n\leq b_n \quad \Rightarrow\quad w^-_{b_n}+w^+_{b_n} > n.
\end{equation}
Lemma~\ref{lemm_width_compactness_pis2} controls the width on either side of $L_1$: choosing $b_n = \sqrt{n/2}$,
\begin{align}
\limsup_{N\rightarrow\infty}\sup_{\substack{f\geq 0 : E_\nu[f]=1 \\ D_N(f^{1/2})\leq C/N}}&\nu_f\big(\gamma\in\e,h_n(\gamma)\leq b_n\big)\nonumber\\
&\leq  2\limsup_{N\rightarrow\infty}\max_{\epsilon\in\{-,+\}}\sup_{\substack{f\geq 0 : E_\nu[f]=1 \\ D_N(f^{1/2})\leq C/N}}\nu_f\Big(\gamma\in \e,w^{\epsilon}_{b_n}(\gamma)\geq n/2\Big) = o_n(1). \label{eq_h_k_not_too_small}
\end{align}
From~\eqref{eq_upper_bound_ell_small0} and~\eqref{eq_upper_bound_ell_small1} one can therefore write:
\begin{align}
\Big\{\gamma\in \e,\ell_n(\gamma)\leq a_n\Big\} 
&\subset 
\Big\{\gamma\in\e, \min\big\{\Delta^-_{a_n}(\gamma),\Delta^+_{a_n}(\gamma)\geq h_{n}(\gamma)-1\Big\}
\nonumber\\
&\quad 
\cap
\Big(\big\{\gamma\in\e, h_n(\gamma)-1\geq b_n\big\}\cup \big\{\gamma\in\e, h_n(\gamma)\leq b_n\big\}\Big)
.
\end{align}
The event involving $\{h_n\leq b_n\}$ is estimated through~\eqref{eq_h_k_not_too_small}, so that it only remains to bound $\min\{\Delta^{-}_{a_n},\Delta^{-}_{a_n}\}$ when $h_n\geq b_n+1$:
\begin{align}
\limsup_{N\rightarrow\infty}\sup_{\substack{f\geq 0 : E_\nu[f]=1 \\ D_N(f^{1/2})\leq C/N}}&\nu_f\Big(\gamma\in\e,\ell_n\leq a_n\Big)\nonumber\\
&\leq  \limsup_{N\rightarrow\infty}\sup_{\substack{f\geq 0 : E_\nu[f]=1 \\ D_N(f^{1/2})\leq C/N}}\nu_f\Big(\gamma\in\e,\Delta^-_{a_n}(\gamma)\geq h_n(\gamma)-1\geq b_n\Big) + o_n(1). \label{eq_width_small_bounded_by_Delta_big}
\end{align}
We now choose $a_n$ as a function of $b_n =  \sqrt{n/2}$ so that the probability in the right-hand side of~\eqref{eq_width_small_bounded_by_Delta_big} vanishes for large $n$. 
By Lemma~\ref{lemm_height_compactness_pis2}, 
$\min\{\Delta^-_{a_n},\Delta^{+}_{a_n}\}$ is typically smaller than $a_n(1+2\log(a_n))$. 
It thus suffices to take $a_n$ with $a_n(1+2\log(a_n)) \leq b_n$, e.g. for large enough $n$:
\begin{equation}
a_n 
= 
\frac{\sqrt{n}}{4\log n}
=: 
\ell_{\min}(n)
.
\label{eq_def_ell_min_de_k}
\end{equation}
With this choice of $a_n = \ell_{\min}(n)$,~\eqref{eq_width_small_bounded_by_Delta_big} yields the desired control on $\ell_n$:
\begin{align}
\limsup_{N\rightarrow\infty}\sup_{\substack{f\geq 0 : E_\nu[f]=1 \\ D_N(f^{1/2})\leq C/N}}\nu_f\Big(\gamma\in \e,\ell_{n}(\gamma)\leq \ell_{\min}(n)\Big)=o_n(1).\label{eq_width_typical}
\end{align}
We now use the bounds~\eqref{eq_height_typical}--\eqref{eq_width_typical} on $h_n$ and $\ell_n$ to restrict admissible configurations around the pole, thus concluding the definition of the frame around the pole. Recall that $h_{\max}(n) = n(1+2\log n)$ and $\ell_{\min}(n) := \sqrt{n}/(4\log n)$. 
From the splitting~\eqref{eq_qtite_interm_a_estime_val_pis2_1} of curves in the different $M_\ell$ ($2\leq \ell \leq n$) and the above discussion on bounds of $h_n,\ell_n$, as also $\phi = {\bf 1}_{p=2} - e^{-\beta}$ is bounded,~\eqref{eq_qtite_interm_a_estime_val_pis2_0} is bounded from above by:
\begin{equation}
-a\delta T + C\beta + T\ \sup_{f\geq 0 : \nu_f(\e) =1}\Big\{ a\sum_{\ell_{\min}(n) \leq\ell\leq n}E_{\nu_f}\big[{\bf 1}_{M_\ell}{\bf 1}_{h_n\leq h_{\max}(n)}\phi\big] - \frac{N}{2} D_N(f^{1/2}) \Big\}+T\omega_{N,n},\label{eq_qtite_interm_a_estimer_1pis2_is_exp-beta}
\end{equation}
where $\omega_{N,n}$ satisfies, by~\eqref{eq_height_typical} and~\eqref{eq_width_typical}:
\begin{align}
&\limsup_{N\rightarrow\infty}\omega_{N,n} \nonumber\\
&\quad \leq  a\|\phi\|_\infty\limsup_{N\rightarrow\infty}\sup_{\substack{f\geq 0 : E_\nu[f]=1 \\ D_N(f^{1/2})\leq 2\|\phi\|_\infty a/N}}\nu_f\Big(\gamma\in\e,h_n>h_{\max}(n)\text{ or }\ell_n<\ell_{\min}(n)\Big)=o_n(1).
\end{align}
It is thus sufficient to estimate the supremum in~\eqref{eq_qtite_interm_a_estimer_1pis2_is_exp-beta}.\\

\noindent\textbf{Step 2: conditioning and mapping to a two-species zero-range process.}\\
\begin{figure}[H]
\begin{center}
\includegraphics[width=11cm]{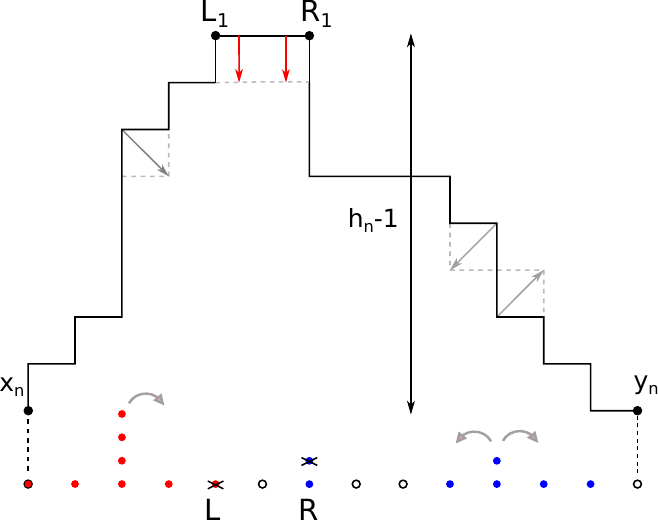} 
\caption{Portion of the interface of a curve around the north pole $[L_1,R_1]$ and associated path and particle configurations. Particles are represented by red dots, antiparticles by blue dots; and empty sites are white with a dark contour. 
In the particle configuration, the rightmost particle (in red) is at site $L$, the leftmost antiparticle (in blue) at site $R$. 
The quantity $|\eta_j|$ at site $0\leq j\leq \ell$ is the number of particles/antiparticles at site $j$, with $\eta_j\geq 0$ for particles, $\eta_j\leq 0$ for antiparticles. Here, one has e.g. $\eta_0 =0, \eta_2 = 4, \eta_R=-2$. The height $h_n-1$ is the total number of particles (or of antiparticles).\\ 
The grey arrows on the particle configuration correspond to jumps allowed by the contour dynamics that conserve the particle number. A move reducing the length of the curve, materialised on the curve by the vertical arrows, corresponds to a particle-antiparticle pair annihilation, represented by the black crosses. No particle creation is represented here.\label{fig_ZRP}}
\end{center}
\end{figure}
We now study the expectation in~\eqref{eq_qtite_interm_a_estimer_1pis2_is_exp-beta} on each $M_\ell$ for $\ell_{\min}(n)\leq \ell\leq n$, 
where the set $M_\ell$ is defined in~\eqref{eq_def_M_ell}. The goal is to obtain a local description of the contour dynamics around the pole,  
rewriting the expectation in~\eqref{eq_qtite_interm_a_estimer_1pis2_is_exp-beta} and the variables it contains in terms of quantities that only depend on possible shapes of curves in a neighbourhood of the pole.  
We claim that to each configuration in $M_{\ell}$ corresponds a unique particle configuration in $\Omega_{\ell} = \Z^{\ell+1}$. 
The mapping goes as follows. If $\gamma\in M_\ell$, define, for $0\leq j \leq \ell$, a particle number $\eta_j$ corresponding to the number of vertical edges with abscissa $x_n(\gamma)+j/N$ that correspond to points above $x_n(\gamma)$, 
i.e. to points $z\in\gamma$ with ordinate $[z-x_n(\gamma)]\cdot{\bf b}_{\pi/2}\geq 0$ (see Figure~\ref{fig_ZRP}):
\begin{equation}
\eta_j 
= 
\epsilon_j\sum_{\substack{z\in V(\gamma):z\cdot {\bf b}_0 = x_n(\gamma)\cdot {\bf b}_0 +j/N\\ [z-x_n(\gamma)]\cdot {\bf b}_{\pi/2} \geq 0}}\xi_z,\quad \epsilon_j = \begin{cases}
1&\text{if }j/N \leq L_1\cdot {\bf b}_0,\\
-1 &\text{if }j/N>L_1\cdot {\bf b}_0.
\end{cases}
\end{equation}
In words, $\eta_j\geq 0$ if $j$ is smaller than the abscissa $L_1\cdot{\bf b}_0$ of $L_1$ and $\eta_j\leq 0$ if $j$ is larger. 
If $\eta_j<0$ for some $j$, we say that there are $|\eta_j|$ antiparticles at site $j$.  
Note that the $(\eta_j)_j$ are related to the  $(\delta(i))_i$ of Lemma~\ref{lemm_height_compactness_pis2} through:
\begin{equation}
\eta_j
=
(-1)^{{\bf 1}_{j\geq L_1\cdot{\bf b}_0+1}} \delta(j-1-L_1\cdot{\bf b}_0),
\qquad j\in \{1,...,\ell-1\}
.
\end{equation}
We prefer the parametrisation in terms of $(\eta_j)_j$ here for two reasons: 
we do not wish to label points according to the position of the poles and only vertical edges at level $h_n -1$ count (recall Figure~\ref{fig_def_h_k}) towards the particle/antiparticle number whereas $\delta$ is not cut off.

The constraint $z\cdot {\bf b}_{\pi/2} > z_1(\gamma)-N^{-1}(h(\gamma)-1)$ guarantees that only the vertical edges above $x_n(\gamma),y_n(\gamma)$ are counted as particles. 
We let $\eta(\gamma)$ denote the unique particle configuration in $\Omega_\ell$ associated with $\gamma\in M_\ell$ (see Figure~\ref{fig_ZRP}). 
Conversely, with each $\eta\in\Omega_\ell$ can be associated  one or more interface in $M_\ell$ that are identical above $x_n(\gamma)$.  
Note that, as $x_n(\gamma)$ and $y_n(\gamma)$ have the same ordinate, the number of particles and antiparticles is the same. 
Note also, importantly, that to \emph{any} configuration $\eta\in\Omega_\ell$ corresponds a curve in $M_\ell$. 
Said differently, the condition that configurations are associated to curves in $M_\ell$ has no other effect on the resulting ZRP configurations than fixing the number of sites.

The quantity $h_n$ corresponds to the number of particles (or equivalently of antiparticles). 
For curves in $M_\ell$, the event $\{h_n\leq h_{\max}(n)\}$ appearing in~\eqref{eq_qtite_interm_a_estimer_1pis2_is_exp-beta} can thus be rewritten as: 
\begin{align}
W_\ell
:=
\{\rho^\ell \leq C_\ell\},
\quad \text{where}\quad 
\rho^\ell 
= 
\frac{1}{\ell+1}\sum_{j=0}^\ell |\eta_j|,\quad C_\ell = C_{\ell,n} 
&= 
\frac{2}{\ell+1}\big(h_{\max}(n)-1\big) 
\nonumber\\
&= 
\frac{4n\log n + 2(n-1)}{\ell+1} 
.
\label{eq_def_event_Well}
\end{align}
For $\ell\in\{\ell_{\min}(n),...,n\}$, 
define the probability measure $\bar \mu_\ell$ on $\Omega_\ell$:
\begin{equation}
\forall\eta\in \Omega_\ell,\qquad 
\bar\mu_\ell (\eta) 
= 
\bar{\mathcal{Z}}_\ell^{-1} \exp\Big[-\beta\ell-\beta\sum_{j=0}^{\ell}|\eta_j|\Big]
,
\label{eq_def_tilde_mu_ell}
\end{equation}
where $\bar{\mathcal{Z}}_\ell$ is a normalisation factor and $\sum_{j=0}^{\ell}|\eta_j|+\ell$ is the number of edges in the portion of $\gamma$ which is mapped to the particle configuration $\eta$.  
Though we could factor it out as it is common to all $\eta$, the $e^{-\beta\ell}$ factor in the definition of $\bar \mu_\ell$ will be convenient later on.

Let $f$ be a $\nu$-density supported on $\e$ with $\nu_f(M_\ell)>0$ and define its marginal $\bar f_\ell$ for $\bar\mu_\ell$:
\begin{equation}
\forall\eta\in \Omega_\ell,\qquad \bar f_\ell(\eta) := \frac{1}{\nu_f(M_\ell)}\sum_{\gamma\in  M_{\ell}: \eta= \eta(\gamma)} \mathcal Z_{\beta}^{-1}f(\gamma) \exp\Big[-\beta\Big(N|\gamma|-\ell-\sum_{j=0}^\ell|\eta_j|\Big)\Big].\label{eq_def_tilde_f_ell}
\end{equation}
With the above definitions, 
the expectation in~\eqref{eq_qtite_interm_a_estimer_1pis2_is_exp-beta} can be rewritten in terms of particle configurations:
\begin{align}
E_{\nu}\Big[f{\bf 1}_{M_\ell}{\bf 1}_{h_n\leq h_{\max}(n)}\phi\Big] = \nu_f(M_\ell)E_{\bar\mu_\ell}\Big[\bar f_\ell{\bf 1}_{W_\ell} \phi\Big]
.
\label{eq_qtite_interm_a_estimer_1pis2_is_exp-beta_2}
\end{align}
The constraint ${\bf 1}_{M_\ell}$ is implicitly contained in the definition of $\mu_\ell$ as a measure on $ell+1$ sites.

From~\eqref{eq_qtite_interm_a_estimer_1pis2_is_exp-beta_2} it follows that we know how to estimate the supremum in~\eqref{eq_qtite_interm_a_estime_val_pis2_0} as soon as we can estimate (recall $\ell_{\min}(n)=\sqrt{n}/(4\log n)$):
\begin{align}
\sup_{f \geq 0 :\nu_f(\e) =1}\Big\{\sum^{n}_{\ell= \ell_{\min}(n)} a\nu_f(M_\ell)E_{\bar \mu_\ell}\big[\bar f_\ell{\bf 1}_{W_\ell}\phi\big] - \frac{N}{2}D_N(f^{1/2})\Big\},\qquad \phi := {\bf 1}_{p=2} - e^{-\beta}.\label{eq_qtite_interm_a_estimer_1pis2_is_exp-beta_3}
\end{align}
\noindent\textbf{Step 3: local equilibrium.}\\
We now prove that estimating the supremum in~\eqref{eq_qtite_interm_a_estimer_1pis2_is_exp-beta_3} reduces to an equilibrium computation. At this stage, the technique is the same as in \cite{Kipnis1999}.
Denote by $\bar D_\ell$ the reduced Dirichlet form on $\Omega_\ell$, defined as follows. 
For $\eta\in \Omega_\ell$, 
let $P(\eta)$ denote the vertices making up the "pole" of $\eta$, 
i.e. $P(\eta) = \{L,...,R\}$, 
with $L,R$ such that $\eta_L$ is the last $\eta_j$ that is strictly positive (or $L=0$ if there are no such $\eta_j$) and $\eta_R$ the first to be strictly negative (or $R=\ell$ if none exist). Let also $p(\eta)= |P(\eta)|-1$. 
For any two configurations $\eta,\tilde\eta\in\Omega_\ell$, let $\gamma,\tilde\gamma\in M_\ell$ respectively be two associated interfaces and define a jump rate:
\begin{equation}
c(\eta,\tilde\eta) 
:= 
c(\gamma,\tilde\gamma),\qquad 
\text{with }
c(\gamma,\tilde\gamma)\text{ given in Definition }\ref{def_contour_dynamics}
.
\end{equation}
Since the positions of the extremal sites $0$, $\ell$ (corresponding for curves $\gamma$ compatible with a given configuration to the points $x_n(\gamma),y_n(\gamma)$) are unchanged by dynamical updates involving vertices with ordinate higher than that of $x_n(\gamma)$ or $y_n(\gamma)$, the last level of $\gamma$ with less than $n$ blocks, which defines the position of $x_n(\gamma),y_n(\gamma)$, is never modified. 
The jump rates $c(\eta,\eta')$ therefore only depend on $\eta,\eta'$ and not on the rest of the curves $\gamma,\gamma'$. 
For any $\bar\mu_\ell$-density $g$, let $\bar D_\ell$ denote the associated Dirichlet form:
\begin{align}
\bar D_\ell(g^{1/2}) &= \frac{1}{2}\sum_{\eta,\tilde\eta\in \Omega_\ell} \bar\mu_\ell(\eta)c(\eta,\tilde\eta)\big[g^{1/2}(\tilde\eta)-g^{1/2}(\eta)\big]^2. \label{eq_def_dir_form_ZRP}
\end{align}
Convexity then yields, recalling that $\ell_{\min}(n)$ is defined in~\eqref{eq_def_ell_min_de_k}:
\begin{equation}
D_N(f^{1/2}) \geq \sum^{n}_{\ell = \ell_{\min}(n)}\nu_f(M_\ell)\bar D_\ell(\bar f_\ell^{1/2}).\label{eq_D_n_f_majore_les_D_ell_f_ell_valeur_1_pis2}
\end{equation}
Injecting~\eqref{eq_D_n_f_majore_les_D_ell_f_ell_valeur_1_pis2} into the supremum in~\eqref{eq_qtite_interm_a_estimer_1pis2_is_exp-beta_3}, we see that it is enough to estimate:
\begin{align}
\sup_{f\geq 0:\nu_f(\e)=1} \bigg\{\sum_{\ell =\ell_{\min}(n)}^ {n} \nu_f(M_\ell)\Big[aE_{\bar\mu_\ell}\big[\bar f_\ell{\bf 1}_{W_\ell} \phi\big] - \frac{N}{2}\bar D_\ell(\bar{f}_\ell^{1/2})\Big]\bigg\}.\label{eq_qtite_to_estimate_after_discard_e_r_d_ZRP}
\end{align}
We are nearly done with conditioning to a frame where we can compute the expectation in~\eqref{eq_qtite_to_estimate_after_discard_e_r_d_ZRP}. The remaining step is to reduce the state space $\Omega_\ell = \Z^{\ell+1}$ to something that is compact. By definition of $\bar f_\ell,\bar \mu_{\ell},\bar D_\ell$ in~\eqref{eq_def_tilde_f_ell}--\eqref{eq_def_tilde_mu_ell}--\eqref{eq_def_dir_form_ZRP} respectively, it is enough to delete all jumps that increase the number of particles above what is allowed by $W_\ell$ (defined in~\eqref{eq_def_mu_ell}). Indeed, define $\mu_\ell$ as a measure on $W_\ell$ as follows:
\begin{equation}
\forall \eta \in W_{\ell} = \Big\{\rho^\ell\leq C_\ell =:\frac{2h_{\max}(n)}{\ell+1}\Big\},\qquad \mu_{\ell}(\eta) :=  \mathcal{Z}_\ell^{-1} \exp\Big[-\beta\ell-\beta\sum_{j=0}^{\ell}|\eta_j|\Big] = \frac{\bar{\mathcal{Z}}_\ell}{\mathcal{Z}_\ell}\bar \mu_\ell(\eta),\label{eq_def_mu_ell}
\end{equation}
where $\mathcal Z_\ell$ is a normalisation factor on $W_\ell$. The marginal $\bar f_\ell$ is correspondingly modified into a $\mu_\ell$-density $f_\ell$:
\begin{equation}
\forall \eta \in W_\ell,\qquad f_\ell(\eta) := \frac{\mathcal Z_\ell}{\bar{\mathcal{Z}}_\ell}\frac{1}{E_{\bar \mu_\ell}\big[\bar f_\ell{\bf 1}_{W_\ell}\big]}\bar f_\ell(\eta).
\end{equation}
Finally, the Dirichlet form $D_\ell$ for the reduced dynamics reads, for any $\mu_\ell$-density $g$:
\begin{align}
&D_\ell(g^{1/2}) = \sum_{\eta,\tilde\eta\in W_\ell}\mu_\ell(\eta)c(\eta,\tilde\eta)[g^{1/2}(\eta')-g^{1/2}(\eta)]^2.\label{eq_def_dir_form_ZRP_pour_mu_ell}
\end{align}
Since we simply restricted allowed jumps, one has:
\begin{equation}
\bar D_\ell(\bar{f_\ell}^{1/2}) 
\geq 
D_\ell(f_\ell^{1/2})E_{\bar \mu_\ell}[\bar{f}_\ell{\bf 1}_{W_\ell}]\bar{\mathcal Z_\ell}/\mathcal Z_\ell
.
\end{equation}
Under $\mu_\ell$, the supremum to estimate in~\eqref{eq_qtite_to_estimate_after_discard_e_r_d_ZRP} is then bounded from above by:
\begin{align}
&\sup_{f\geq 0:\nu_f(\e)=1} \bigg\{\sum^{n}_{\ell =\ell_{\min}(n)} \nu_f(M_\ell)\frac{\bar{\mathcal Z_\ell}}{\mathcal Z_\ell}E_{\bar\mu_\ell}\big[\bar f_\ell{\bf 1}_{W_\ell}\big]\Big[a E_{\mu_\ell}\big[f_\ell \phi\big] - \frac{N}{2} D_\ell(f_\ell^{1/2})\Big]\bigg\}\nonumber\\
&\leq a\sup_{f\geq 0:\nu_f(\e)=1}\bigg\{\sum^{n}_{\ell =\ell_{\min}(n)}\nu_f(M_\ell)\frac{\bar{\mathcal Z_\ell}}{\mathcal Z_\ell}E_{\bar\mu_\ell}\big[\bar f_\ell{\bf 1}_{W_\ell}\big] \bigg[\sup_{\substack{g\geq 0 : E_{\mu_\ell}[g] = 1 \\ D_\ell(g^{1/2})\leq 2a\|\phi\|_\infty/N}}E_{\mu_\ell}\big[g \phi\big]\bigg]\bigg\}\nonumber\\
&\leq a\max_{\ell_{\min}(n)\leq \ell\leq n}\bigg[\sup_{\substack{g\geq 0 : E_{\mu_\ell}[g] = 1 \\ D_\ell(g^{1/2})\leq 2a\|\phi\|_\infty/N}}E_{\mu_\ell}\big[g \phi\big]\bigg]\times \sup_{f\geq 0:\nu_f(\e)=1}\bigg\{\sum^{n}_{\ell =\ell_{\min}(n)}\nu_f(M_\ell)\frac{\bar{\mathcal Z_\ell}}{\mathcal Z_\ell}E_{\bar\mu_\ell}\big[\bar f_\ell{\bf 1}_{W_\ell}\big] \bigg\}.\label{eq_qtite_to_estimate_after_conditioning_valeur_1_pis2}
\end{align}
The second supremum is bounded by $1$. The proof of Lemma~\ref{lemm_valeur_p_is_2} will therefore be concluded if we can prove that, for fixed $n$ and $N$ large, the supremum on $g$ in the right-hand side of~\eqref{eq_qtite_to_estimate_after_conditioning_valeur_1_pis2} is bounded by $o_n(1)$ uniformly in $\ell$ with $\ell_{\min}(n)\leq \ell\leq n$.\\
Fix $\ell\in \{\ell_{\min}(n),...,n\}$. As $W_\ell$ is compact, the supremum on $g$ in~\eqref{eq_qtite_to_estimate_after_conditioning_valeur_1_pis2} is achieved by a density $g^N_\ell$ for each $N$. Up to taking a subsequence, by lower semi-continuity of $D_\ell$ and continuity of the expectation in~\eqref{eq_qtite_to_estimate_after_conditioning_valeur_1_pis2} with respect to weak convergence, we can take the large $N$ limit and restrict ourselves to studying:
\begin{align}
\sup_{\substack{g^\infty\geq 0 : E_{\mu_\ell}[g^\infty]=1 \\ D_\ell((g^\infty)^{1/2})=0}}E_{\mu_\ell}[g^\infty \phi].
\end{align}
By definition of $D_\ell$, the zero-range dynamics is irreducible on $W_\ell$. 
This is the major difference between the contour dynamics and the $0$-temperature stochastic Ising model, 
which motivated the introduction of the temperature-like parameter $\beta$ to allow for regrowth at the poles. 
Irreducibility means that any $g^\infty$ satisfying $D_\ell(g^\infty)=0$ is constant equal to $1$ and we are left with the estimate of:
\begin{equation}
E_{\mu_\ell}[\phi]\qquad \text{with}\quad \phi = {\bf 1}_{p=2} - e^{-\beta}.\label{eq_to_compute_at_equilibrium_poles_exp-beta}
\end{equation}
\noindent \textbf{Step 4: equilibrium computations}\\ 
The expectation~\eqref{eq_to_compute_at_equilibrium_poles_exp-beta} is taken under the equilibrium measure of the zero-range dynamics. Properties of the measure $\mu_\ell$ are analysed in Appendix~\ref{sec_equilibrium_estimates_at_the_pole}. In particular, it is proven there that, recalling the definition~\eqref{eq_def_ell_min_de_k} of $\ell_{\min}(n)$:
\begin{equation}
\lim_{n\rightarrow\infty}\max_{\ell_{\min}(n)\leq \ell \leq n}\big|E_{\mu_\ell}[\phi]\big| = 0.\label{eq_estimate_1_pis2_minus_exp-_beta_under_mu_ell}
\end{equation}
Equation~\eqref{eq_estimate_1_pis2_minus_exp-_beta_under_mu_ell} concludes the proof of Lemma~\ref{lemm_valeur_p_is_2} when $G=1$. \\

Let us now discuss the case $G\neq 1$. 
In this case we need to bound:
\begin{equation}
-a\delta T+ C \beta + \int_0^T\sup_{f\geq 0 :\nu_f(\e) =1}\Big\{aE_{\nu_f}[\phi G(t,L_1)] - ND_N(f^{1/2})\Big\}
\, dt
.
\label{eq_supremum_penteopole_with_test_function0}
\end{equation}
Fix $t\in[0,T]$ and consider the supremum at time $t$. 
Since $G$ is bounded, it is still possible to localise around the pole through bounds on $h_n,\ell_n$, 
so that it is sufficient to estimate the following analogue of~\eqref{eq_qtite_interm_a_estimer_1pis2_is_exp-beta} for each $t\leq T$ and $\nu$-density $f$ supported in $\e$:
\begin{align}
\sum_{\ell=\ell_{\min}(n)}^n &E_{\nu_f}\big[{\bf 1}_{M_\ell}{\bf 1}_{h_n\leq h_{\max}(n)} \phi\ G(t,L_1)\big] \nonumber\\
&= 
\sum_{\ell=\ell_{\min}(n)}^n E_{\nu_f}\big[{\bf 1}_{M_\ell}{\bf 1}_{h_n\leq h_{\max}(n)}{\bf 1}_{x_n(\gamma) = x}\phi\ G(t,x_n)\big]+o_N(1)
,
\label{eq_decomp_calcul_slope_with_test_function}
\end{align}
with an error term uniform in $f$, accounting for the difference $G(t,L_1)-G(t,x_n)$. This error term vanishes for $N$ large due to the conditions $\ell_n\leq n$, $h_n\leq h_{\max}(n)$ which imply $\|x_n-L_1\|_1\leq C(G,n)/N$ for some $C(G,n)>0$ independent of $N$.

To reduce to the $G=1$ case, we further split the expectation depending on the position of $x_n(\gamma)$. 
The first term in the right-hand side of~\eqref{eq_decomp_calcul_slope_with_test_function} is then equal to:
\begin{align}
\sum_{\ell=\ell_{\min}(n)}^n \sum_{x\in (N^{-1}\Z)^2} \|G\|_{\infty}\Big|E_{\nu_f}\big[{\bf 1}_{M_\ell}{\bf 1}_{h_n\leq h_{\max}(n)}{\bf 1}_{x_n(\gamma) = x}\phi\big]\Big|
.
\end{align}
At this point $G$ dose not play a role any more and we can proceed as in the $G=1$ case. 
Indeed, 
the position of the left extremity $x_n(\gamma)$ of the frame around the pole is unchanged by the zero-range dynamics, see the discussion following~\eqref{eq_def_dir_form_ZRP}. 
Conditioning on its position therefore does not affect the projection on the ZRP. 
As such, the rest of the arguments in the proof of Lemma~\ref{lemm_valeur_p_is_2} go through unchanged, 
except that one has to rewrite everything with $x$ fixed, e.g. to consider $M_{\ell,x} = M_{\ell}\cap\{x_n(\gamma)=x\}$ instead of $M_\ell$, $f_{\ell,x},\bar f_{\ell,x}$ instead of $f_\ell,\bar f_\ell$ for the new marginal of $f$ under $\mu_\ell,\bar\mu_\ell$ respectively, etc. 
In particular, the supremum in~\eqref{eq_supremum_penteopole_with_test_function0} is bounded at each time $t\leq T$ by the following variant of~\eqref{eq_qtite_to_estimate_after_conditioning_valeur_1_pis2}:
\begin{align}
&
\sup_{f\geq 0:\nu_f(\e)=1} \bigg\{\sum^{n}_{\ell =\ell_{\min}(n)}\sum_{x\in (N^{-1}\Z)^2} \nu_f(M_{\ell,x})\frac{\bar{\mathcal Z_\ell}}{\mathcal Z_{\ell}}E_{\bar\mu_\ell}\big[\bar f_{\ell,x}{\bf 1}_{W_\ell}\big]\Big[a\|G\|_\infty \Big|E_{\mu_\ell}\big[f_{\ell,x} \phi\big]\Big| - \frac{N}{2} D_\ell(f_{\ell,x}^{1/2})\Big]\bigg\}
\nonumber\\
&\quad \leq
a\|G\|_{\infty}\max_{\ell_{\min}(n)\leq \ell\leq n}\bigg[\sup_{\substack{g\geq 0 : E_{\mu_\ell}[g] = 1 \\ D_\ell(g^{1/2})\leq 2a\|\phi\|_\infty\|G\|_{\infty}/N}}\Big|E_{\mu_\ell}\big[g \phi\big]\Big|\bigg]\nonumber\\
&\hspace{4cm}\times \sup_{f\geq 0:\nu_f(\e)=1}\bigg\{\sum^{n}_{\ell =\ell_{\min}(n)}\sum_{x\in (N^{-1}\Z)^2}\nu_f(M_{\ell,x})\frac{\bar{\mathcal Z_\ell}}{\mathcal Z_\ell}E_{\bar\mu_\ell}\big[\bar f_{\ell,x}{\bf 1}_{W_\ell}\big] \bigg\}.
\end{align}
The second supremum is bounded by $1$ and the supremum on the zero-range process is the same quantity as in~\eqref{eq_qtite_to_estimate_after_conditioning_valeur_1_pis2}.
\end{proof}
The method of proof of Lemma~\ref{lemm_valeur_p_is_2} can be used to obtain tighter estimates on the slope at the poles. An example is given in the following corollary, used in Appendix~\ref{appen_tightness} to obtain exponential tightness.
\begin{coro}[One and two block estimates for deviations from the average]\label{coro_1_2_blocks_deviations_slope}
Let $\beta>\log 2$. For each $\delta,\eta>0$:
\begin{align}
\lim_{q\rightarrow\infty}\limsup_{N\rightarrow\infty}\frac{1}{N}\log \Prob^ {N}_{\beta}\bigg(&\gamma^N_\cdot\in E([0,T],\e);\frac{1}{T}\int_0^ {T} {\bf 1}_{|\xi^{\pm, q}_{L_1(\gamma^N_t)}-e^ {-\beta}|\geq \delta}\, dt >\eta\bigg)=-\infty.\label{eq_1block_deviations_slope_average}
\end{align}
and:
\begin{align}
\lim_{q\rightarrow\infty}\limsup_{\epsilon\rightarrow 0}\limsup_{N\rightarrow\infty}\frac{1}{N}\log \Prob^ {N}_{\beta}\bigg(&\gamma^N_\cdot\in E([0,T],\e);\nonumber\\
&\hspace{2cm}\frac{1}{T}\int_0^ {T} {\bf 1}_{|\xi^{\pm, q}_{L_1(\gamma^N_t)}-\xi^ {\pm,\epsilon N}_{L_1(\gamma^N_t)}|\geq \delta}\, dt >\eta\bigg)= -\infty.\label{eq_2blocks_deviations_slope_average}
\end{align}
Equations~\eqref{eq_1block_deviations_slope_average}--\eqref{eq_2blocks_deviations_slope_average} are valid under $\Prob^N_{\beta,H}$ for $H\in\C$ by Corollary~\ref{coro_change_measure_sous_exp}.
\end{coro}
\begin{rmk}
Note that ${\bf 1}_{|\xi^{\pm,q}_{L_1}-e^{-\beta}|\geq \delta}$ is simply a cylindrical function, which has average $o_q(1)$ under the invariant measure $\nu$. Corollary~\ref{coro_1_2_blocks_deviations_slope} thus says no more than the usual replacement lemmas.\demo
\end{rmk}
\begin{proof}
Consider first~\eqref{eq_2blocks_deviations_slope_average}. 
Using Feynman-Kac inequality as usual, 
it is enough to prove that, for each $a>0$:
\begin{equation}
\lim_{q\to\infty}\limsup_{\epsilon\to 0}\limsup_{N\to\infty}\sup_{f\geq 0 :\nu_f(\e)=1}\Big\{ aE_{\nu_f}\big[{\bf 1}_{|\xi^{\pm, q}_{L_1(\gamma^N)}-\xi^ {\pm,\epsilon N}_{L_1(\gamma^N)}|\geq \delta}\big] - ND_N(f^{1/2})\Big\}
=
0
.
\end{equation}
For each density $f$, the expectation is bounded above using Markov inequality by $\delta^{-1}E_{\nu_f}[|\xi^{\pm, q}_{L_1(\gamma^N)}-\xi^ {\pm,\epsilon N}_{L_1(\gamma^N}|]$. 
At this point the proof of Equation~\eqref{eq_2blocks_deviations_slope_average} boils down to the proof of a two block estimate using only the SSEP part of the dynamics. 
The method of proof is then the same as in Lemma~\ref{lemm_1_block_2_block_pour_pis2}.

Consider now~\eqref{eq_1block_deviations_slope_average}. 
By the same approach, we only need to prove:
\begin{equation}
\lim_{q\to\infty}\limsup_{\epsilon\to 0}\limsup_{N\to\infty}
\sup_{f\geq 0 :\nu_f(\e)=1}\Big\{ \frac{a}{\delta}E_{\nu_f}\big[|\xi^{\pm, q}_{L_1(\gamma^N)}-e^{-\beta}|\big] - ND_N(f^{1/2})\Big\}
=
0
.
\end{equation}
Since we need to compute the exact value of the slope, 
we have to project on the ZRP around the pole as in the proof of Lemma~\ref{lemm_valeur_p_is_2}. 
This is done as for Lemma~\ref{lemm_valeur_p_is_2}, 
with the only difference that we need to project on a frame around the pole that has at least $q$ edges to the right (for $\xi^{+,q}$) or to the left (for $\xi^{-,q}$) of the pole. 

To do so, it is enough to choose the parameter $n$ in the definition of the reference frame such that there are typically at least $q$ vertical edges (as there will then be at least $q$ edges to either side of the pole). 
Written formally, we only need to prove:
\begin{equation}
\limsup_{N\rightarrow\infty}\sup_{\substack{f\geq 0:E_\nu[f]=1 \\ D_N(f^{1/2})\leq C/N}}\nu_f\Big(\e, h_n(\gamma)-1\leq q\Big) 
=
o_q(1).
\end{equation}
This estimate was already established in~\eqref{eq_h_k_not_too_small}. 
With the only difference that $n$ is taken large before $q$, 
the proof of~\eqref{eq_1block_deviations_slope_average} thus reduces, as in Lemma~\ref{lemm_valeur_p_is_2}, 
to an elementary (though more involved) equilibrium computation under the measure $\mu_\ell$ ($\ell_{\min}(n)\leq \ell \leq n$) 
defined in~\eqref{eq_def_mu_ell}.
\end{proof}

\begin{appendices}
\renewcommand{\thesection}{\Alph{section}}

\section{Projection of the dynamics, replacement lemma and equilibrium estimates}\label{sec_replacement_lemma}
\subsection{Projection of the contour dynamics on the good state space}\label{sec_local_dynamics}
In this section, we prove Lemma~\ref{lemm_projec_dynamics_onto_local} which states that the contour dynamics can be projected to the effective state space $\e$. We state and prove a more general result. \\
Let $(X_t)_{t\geq 0}$ be a continuous time Markov chain on a finite state space $E$, reversible with respect to a measure $\nu$.  
If $x_0\in E$, let $\Prob^X_{x_0},\E^X_{x_0}$ be the associated probability and expectation. The jump rates of the chain between states $(x,y)\in E^2$ are denoted by $c(x,y)$, with associated Dirichlet form:
\begin{equation}
\forall g : E\rightarrow\R,\qquad D(g) = \frac{1}{2}\sum_{(x,y)\in E^2}\nu(x)c(x,y)\big[g(y)-g(x)\big]^2.\label{eq_def_D_lemm_6_3}
\end{equation}
\begin{lemm}\label{lemm_projec_dynamics_onto_local_appendix}
Let $B\subset E$ and $x_0\in B$. Let also $T>0$ and $\psi : [0,T]\times E\rightarrow \R$ be bounded. Then:
\begin{align}
\E^X_{x_0}&\bigg[{\bf 1}_{\{\forall t\in [0,T],X_t\in B\}}\exp\bigg[\int_0^{T}\psi(t,X_t)\, dt\bigg]\bigg] 
\nonumber\\
&\qquad \leq 
\frac{1}{\nu(x_0)^{1/2}}\exp\bigg[\int_0^T \sup_{f \geq 0: \nu_f(B) = 1}\Big\{ E_{\nu}\big[f\psi(t,\cdot)\big] - D(f^{1/2})\Big\}\, dt\bigg]
.
\end{align}
\end{lemm}
\begin{proof}
As $(X_t)_{t\geq 0}$ is right-continuous with left limits,
\begin{align}
&\E^X_{x_0}\bigg[{\bf 1}_{\{\forall t\in [0,T],X_t\in B\}}\exp\bigg[\int_0^{T}\psi(t,X_t)\, dt\bigg]\bigg]
\nonumber\\
&\hspace{2cm}
=
\lim_{n\to\infty}
\E^X_{x_0}\bigg[{\bf 1}_{\{X_T\in B\}}\exp\bigg[-n\int_0^T{\bf 1}_{X_t\in B}\, dt +\int_0^{T}\psi(t,X_t)\, dt\bigg]\bigg]
.
\label{eq_proj_lemma1}
\end{align}
For $n\in\N\cup\{\infty\}$ and $t\in[0,T]$, 
introduce the $E\times E$ matrix $L^B$ defined by:
\begin{equation}
L^n_t(x,y) 
= 
\begin{cases}
-c(x,y)\quad &\text{if }x\neq y\in E\\
\sum_{z\in E}c(x,z) + \psi(t,x) - n{\bf 1}_{x\notin B}\quad &\text{if }x=y.
\end{cases}
\end{equation}
Write also $P^n_t$ ($t\in[0,T]$) for the associated semigroup defined through the Chapman-Kolmogorov equation:
\begin{equation}
\partial P^n_t 
= 
L^n_tP^n_t,
\quad 
t\in[0,T]
,\qquad 
P^n_0(x,y)
=
{\bf 1}_{y=x,x\in B},\quad (x,y)\in E^2
.
\end{equation}
The right-hand side of~\eqref{eq_proj_lemma1} is then equal to:
\begin{equation}
\lim_{n\to\infty}\sum_{y\in B}P^n_T(x_0,y) 
=
\sum_{y\in B}P^\infty_T(x_0,y) 
=
\big\|\frac{{\bf 1}_{x_0}}{\nu(x_0)}P^\infty_T\big\|_{L^1(\nu)}
\leq 
\big\|\frac{{\bf 1}_{x_0}}{\nu(x_0)}P^\infty_T\big\|_{L^2(\nu)}
.
\end{equation}
Note that $L^n_t$ is self adjoint in $L^2(\nu)$. 
If $\lambda^n_t$ denotes its largest eigenvalue, 
one has in particular when $n=\infty$:
\begin{align}
\lambda^\infty_t
&=
\sup_{f\geq 0:\nu(f)=1}\Big\{E_\nu[f\psi(t,\cdot)] - \infty E_\nu[f{\bf 1}_{x\notin B}] - D(f^{1/2})\Big\}
\nonumber\\
&=
\sup_{f\geq 0:\nu(f{\bf 1}_B)=1}\Big\{E_{\nu}[f\psi(t,\cdot)] - D(f^{1/2})\Big\}
.
\end{align}
As a result, for $t\in[0,T]$,
\begin{equation}
\frac{\mathrm{d}}{\mathrm{d}t}\big\|\frac{{\bf 1}_{x_0}}{\nu(x_0)}P^\infty_t\big\|^2_{L^2(\nu)}
=
2E_{\nu}\Big[\frac{{\bf 1}_{x_0}}{\nu(x_0)}P^\infty_t L_t \Big(\frac{{\bf 1}_{x_0}}{\nu(x_0)}P^\infty_t\Big)\Big]
\leq 
2\lambda^\infty_t\big|\frac{{\bf 1}_{x_0}}{\nu(x_0)}P^\infty_T\big\|_{L^2(\nu)}^2
.
\end{equation}
Gronwall inequality and $\|\frac{{\bf 1}_{x_0}}{\nu(x_0)}P^\infty_0\big\|_{L^2(\nu)} = \nu(x_0)^{-1/2}$ conclude the proof.
\end{proof}
\subsection{Replacement lemma}\label{subsec_replacement_lemma}
In this section, we prove the Replacement Lemma~\ref{lemm_Replacement_lemma_sec_mart}. Let us first introduce and recall some notations. Fix a time $T>0$ and $A>0$ throughout the section, such that all trajectories considered here will be in the set:
\begin{equation}
\Big\{(\gamma^N_t)_{t\leq T} \subset \Omega^N_{\text{mic}}\cap \e: \sup_{t\leq T}|\gamma^N_t|\leq A\Big\}.
\end{equation}
For each $\epsilon >0$ and $x\in \R^2$, recall that $B_1(x,\epsilon)\subset \R^2$ is the subset of points at distance less than $\epsilon$ to $x$ in $1$-norm. For $\gamma^N\in \Omega^N_{\text{mic}}$ and $x\in V(\gamma^N)$, define
\begin{equation}
\phi(\gamma^N,x) = c_x(\gamma^N) = \frac{1}{2}\big[\xi_{x+e^-_x}(1-\xi_x) + \xi_x(1-\xi_{x+e^-_x})\big].\label{eq_def_phi_replacement}
\end{equation}
Recall from~\eqref{eq_def_xi_x_epsilon_N} that $\xi^ {\epsilon N}_x$ is the quantity
\begin{equation}
\xi^ {\epsilon N}_x = \frac{1}{2\epsilon N +1}\sum_{y \in V(\gamma^N)\cap B_1(x,\epsilon)} \xi_y,
\end{equation}
and define $\tilde \phi$ by:
\begin{equation}
\tilde \phi(\rho) = \rho(1-\rho),\qquad \rho\in[0,1].
\end{equation}
Let $G\in C_c(\R_+\times\R^2)$ be a bounded function. By Chebychev exponential inequality and the Projection lemma~\ref{lemm_projec_dynamics_onto_local_appendix} applied to the set $E([0,T],\e)\cap \big\{\sup_{t\leq T}|\gamma^N_t|\leq A\big\}$, Lemma~\ref{lemm_Replacement_lemma_sec_mart} holds if, uniformly on $t>0$ and for each $a>0$:
\begin{align}
&\lim_{\epsilon\rightarrow 0}\limsup_{N \rightarrow\infty}\label{eq_lemme_de_remplacement}\\
&\hspace{0.5cm}\sup_{f\geq 0 : \nu_f(\e_A)=1}\Bigg\{E_{\nu_f}\Bigg[a\biggr\{\frac{1}{N}\sum_{x\in V(\gamma^N)}G(t,x)\biggr[\phi(\gamma^N,x) -\tilde\phi(\xi_x^ {\epsilon N})\biggr]\biggr\}^2\Bigg] - N D_N(f^{1/2})\Bigg\} = 0.\nonumber
\end{align}
Above, $\e_A$ is the set:
\begin{equation}
\e_A := \e\cap\big\{|\gamma|\leq A\big\}.
\end{equation}
Following \cite{Eyink1990}, it is sufficient to prove the following two estimates.
\begin{lemm}(One and two block estimates)\label{lemm_amended_rplct_lemma}
\\
Let $\epsilon>0,\ell\in\N_{\geq 1}$ and let $(V_j)_{1\leq j \leq J}$ denote a partition of $\{-\epsilon N,..., \epsilon N\}$ in $J$ intervals of length $\ell$ (except maybe the last one that is of size at most $2\ell$), 
such that $\max V_j = \min V_{j+1} -1$ for $j\leq J-1$. 
For $\gamma^N\in \Omega^N_{\text{mic}}$, $x\in V(\gamma^N)$ and $1\leq j \leq J$, 
let $V_j(x)$ be the set of vertices in $B_1(x,\epsilon) \cap V(\gamma^N)$, 
with numbering relative to $x$ corresponding to elements of $V_j$ (i.e. $x$ corresponds to $0$, $x+{\bf e}^{\pm}_x$ to $\pm 1$, etc.). Define also:
\begin{align}
S(\phi,V_j(x)) := \frac{1}{|V_j(x)|} \sum_{y\in V_j(x)} \phi(\gamma^N,y),\quad \xi^{V_j(x)} := \frac{1}{|V_j(x)|}\sum_{y\in V_j(x)}\xi_y.
\end{align}
For any $a>0$, one has then (one block estimate):
\begin{align}
&\lim_{\ell\rightarrow\infty}\limsup_{\epsilon\rightarrow 0}\limsup_{N\rightarrow\infty}\max_{1\leq j\leq J}\nonumber \\
&\hspace{1cm}\sup_{f\geq 0 : \nu_f(\e_A)=1}\bigg\{ aE_{\nu_f}\bigg[ \frac{1}{N}\sum_{x\in V(\gamma^N)} \Big| S(\phi,V_j(x)) -\tilde\phi(\xi^{V_j(x)})\Big|^2\bigg] - ND_N(f^{1/2})\bigg\} = 0,\label{eq_one_block}
\end{align}
and (two block estimate):
\begin{align}
&\lim_{\ell\rightarrow\infty}\limsup_{\epsilon\rightarrow 0}\limsup_{N\rightarrow\infty}\max_{1\leq b,c\leq J} \nonumber\\
&\hspace{0.2cm}\sup_{f\geq 0 : \nu_f(\e_A)=1}\bigg\{a E_{\nu_f}\bigg[\frac{1}{N}\sum_{x\in V(\gamma^N)} \Big| S(\phi,V_b(x)) -S(\phi,V_c(x))\Big|^2\bigg] -ND_N(f^{1/2})\bigg\}= 0.\label{eq_two_blocks}
\end{align}
\end{lemm}
\begin{proof}
Only microscopic curves occur in this proof, so we drop the superscript $N$ and write $\gamma\in\Omega^N_{\text{mic}}$. 
All distances are in $1$-norm.\\ 
The proof is written for a function $\phi$ of range $R\in\N$, i.e. $\phi(\gamma,x)$ depends only on $B_1(x,R/N)\cap V(\gamma)$ for each $\gamma\in\Omega^N_{\text{mic}}$ and $x\in V(\gamma)$. It in particular applies to~\eqref{eq_def_phi_replacement}, for which $R=1$. 
The proof of~\eqref{eq_one_block}--\eqref{eq_two_blocks} consists in showing that the one and two block estimates for the contour dynamics amount to the same estimates for the SSEP, which are well known \cite{Eyink1990}. 
We do it for~\eqref{eq_one_block},~\eqref{eq_two_blocks} is similar. 

The first step is to discard all points in the sum in~\eqref{eq_one_block} that are close to the poles so that the pole dynamics can be neglected. 
Define thus, for $u>0$, the set $W^u(\gamma)$, which contains all points of $V(\gamma)$ at distance at least $u/N$ from each $L_k(\gamma)$, $k\in\{1,...,4\}$ 
(compare with $V^u(\gamma)$, see Figure~\ref{fig_flip_x_epsilon}, 
which contains points at $1$-distance at least $u/N$ from the whole poles rather than their left extremities). For any $\gamma\in \Omega^N_{\text{mic}}$, 
\begin{align}
\frac{1}{N}\sum_{x\in V(\gamma)} \Big| S(\phi,V_j(x)) -\tilde\phi(\xi^{V_j(x)})\Big|^2  &\leq \frac{1}{N}\sum_{x\in W^{\epsilon N +R+3}(\gamma)} \Big| S(\phi,V_j(x)) -\tilde\phi(\xi^{V_j(x)}_x)\Big|^2\nonumber\\
&\quad + C\|\phi\|_\infty \epsilon .\label{eq_qtite_a_estimer_1_blocs_apres_virer_coins}
\end{align}
The second term in the right-hand side of~\eqref{eq_qtite_a_estimer_1_blocs_apres_virer_coins} is independent of $N,\ell$ and vanishes for $\epsilon$ small. 
We now estimate the sum. 
To do so, we split curves depending on their four regions. 
We then use the mapping to the SSEP for the dynamics on each region.

Let $1\leq k \leq 4$ and let $M_k$ denote the set of all lattice paths compatible with region $1\leq k\leq 4$, defined as follows. 
For $\gamma\in \Omega^N_{\text{mic}}$, let $\gamma_k$ denote the part of $\gamma$ that comprises the edges between the vertex $L_k +2{\bf e}^+_{L_k}$, 
and the vertex before $L_{k+1}$, 
these two vertices included (with $L_{k+1}:=L_1$ when $k=4$). 
Define then:
\begin{equation}
M_k := \big\{\rho\subset (N^{-1}\Z)^2 : \exists \gamma\in\Omega^N_{\text{mic}}, \gamma_k = \rho\big\}.
\end{equation}
One can check that, e.g. when $k=1$, any lattice path on $(N^{-1}\Z)^2$ allowed to only go right or down and starting in the upper half plane $x\cdot{\bf b}_{\pi/2}\geq 0$ is an element of $M_1$. 
A similar statement holds for other values of $k$ for the corresponding lattice paths directions. 

Define now $\mu_k$ as the marginal of the contour measure $\nu = \nu^N_\beta$ (defined in~\eqref{eq_def_nu_beta}) on $M_k$:
\begin{equation}
\forall \rho \in M_k,\qquad \mu_k(\rho) = \frac{e^{-\beta N |\rho|}}{\mathcal Z_k},\quad \mathcal Z_k = \sum_{\rho\in M_k}e^{-\beta N |\rho|}.
\end{equation}
Let $f$ be a $\nu$-density supported on $\e_A$. Define the corresponding $\mu_k$-marginal $f_k$:
\begin{equation}
\forall \rho \in M_k,\qquad f_k(\rho)= \frac{1}{\mu_k(\rho)}\sum_{\gamma\in \Omega^N_{\text{mic}}}{\bf 1}_{\gamma_k = \rho}f(\gamma)\nu(\gamma).
\end{equation}
For $\gamma\in\Omega^N_{\text{mic}}$, if $\gamma\setminus\gamma_k$ is fixed, then so are all poles. Moreover, if $\gamma$ is in $\e$, then the contour dynamics is local. 
The definition of $M_k$ then implies that dynamical updates affecting an edge of $\gamma_k$ correspond to SSEP moves. 
As a result, the Dirichlet form $D_N(f^{1/2})$ is bounded from below by convexity according to:
\begin{align}
&D_N(f^{1/2})\geq \frac{1}{2}\sum_{k=1}^4\sum_{\rho\in M_k}\mu_i(\rho)\sum_{x\in V(\rho)}c_x(\rho) \big[f_k^{1/2}(\rho^{x,x+{\bf e}^-_x})-f^{1/2}_k(\rho)\big]^2 
=: 
\sum_{k=1}^4D_{N}^{\text{ex},k}(f^{1/2}_k)
,
\label{eq_lower_bound_Dirichlet_form_par_les_Di_temp}
\end{align}
where $V(\rho)$ is the set of vertices in $\rho$ and, for $k\in\{1,...,4\}$ and a $\mu_k$-density $h$, the Dirichlet form $D_{N}^{\text{ex},k}$ corresponding to the SSEP dynamics in region $k$ is given by:
\begin{align}
D_{N}^{\text{ex},k}(h^{1/2}) 
= 
\sum_{\rho\in M_k} \mu_k(\rho) \sum_{x\in V(\rho)}c_x(\rho) \big[h^{1/2}(\rho^{x,x+{\bf e}^-_x}) - h^{1/2}(\rho) \big]^2
.
\label{eq_def_dirichlet_corner_flip_toutes_length}
\end{align}
We now use the decomposition on the $M_k,1\leq k\leq 4$ to estimate the sum appearing in the right-hand side of~\eqref{eq_qtite_a_estimer_1_blocs_apres_virer_coins}. 
For short, define $\Phi_j$ for $1\leq j \leq J$ by:
\begin{equation}
\Phi_j(\gamma,x) = \Big|S(\phi,V_j(x)) - \tilde\phi(\xi^{V_j(x)})\Big|^2,\qquad \gamma\in \Omega^N_{\text{mic}}, x\in V(\gamma).\label{eq_def_Phi_j}
\end{equation}
Note that $\Phi_j(\gamma,x)$ depends only on the orientation (horizontal or vertical) of the edges of $\gamma$ at $1$-distance at most $\ell/N$ from $x$ and in particular does not depend on the absolute position of $x$ as a point of $\R^2$. 
We thus only need to keep track of the label of $x$ in a well chosen parametrisation of $\gamma$.
We have:
\begin{align}
(E) &:=\sum_{\gamma\in \Omega^N_{\text{mic}}}\nu(\gamma)f(\gamma) \frac{1}{N}\sum_{x\in W^{\epsilon N +R+3}(\gamma)}\Phi_j(\gamma,x) \nonumber \\
&\ \leq   \sum_{k=1}^4\sum_{\rho\in M_k} \mu_k(\rho) f_k(\rho) \frac{1}{N}\sum_{x\in W^{\epsilon N +R}(\rho)}\Phi_j(\rho,x), 
\end{align}
where $\Phi_j(\rho,\cdot)$ is defined as in~\eqref{eq_def_Phi_j} replacing $\gamma\in \Omega^N_{\text{mic}}$ by a path $\rho\in M_k$ ($1\leq k \leq 4$). 
Since $\phi$ depends only locally on the curve, this is not ambiguous for $x\in W^{\epsilon N+R}(\rho)$.
 
So far, we proved that the one block estimate~\eqref{eq_one_block} holds as soon as:
\begin{align}
&\lim_{\ell\rightarrow\infty}\limsup_{\epsilon\rightarrow 0}\limsup_{N\rightarrow\infty}\max_{1\leq j\leq J}\nonumber \\
&\hspace{1cm}\sup_{f\geq 0 :\nu_f(\e_A)=1}\bigg\{\sum_{k=1}^4\Big[a\sum_{\rho\in M_k}\mu_k(\rho)f_k(\rho)\frac{1}{N}\sum_{x\in W^{\epsilon N +R}(\rho)}\Phi_j(\rho,x) - ND_{N}^{\text{ex},k}(f^{1/2}_k)\Big]\bigg\}
= 0
.
\label{eq_one_block_interm}
\end{align}
The estimate for each $1\leq k\leq 4$ is identical, so we only do it for $k=1$. 
Further split paths in $M_1$ according to their number of vertices and let $M_1(n)$ be the subset of $M_1$ of paths with $n+1$ vertices. 
All such paths have the same probability under $\mu_1$, 
thus the marginal of $\mu_1$ on $M_1(n)$ is the uniform measure $U_n$ on paths with $n+1$ vertices or, equivalently, by the correspondence expounded in Section~\ref{subsec_1_pi2_is_slope_around_poles} (see Figure~\ref{fig_corner_flip_ssep}), of SSEP configurations with $n$ sites. Define $f_{1,n}$ as the corresponding $U_n$-marginal of $f_1$ on $M_1(n)$:
\begin{equation}
\forall \rho\in M_1(n),\qquad f_{1,n}(\rho):= E_{\mu_1}\big[f_1{\bf 1}_{M_1(n)}\big]^{-1}f_1(\rho) \mu_1(\rho)|M_1(n)|\quad \text{if}\quad E_{\mu_1}\big[f_1{\bf 1}_{M_1(n)}\big]>0.
\end{equation}
It is a density for $U_n$, thus convexity of the Dirichlet form yields:
\begin{equation}
D^{\text{ex},1}_{N}(f_1^{1/2})\geq \sum_{n\geq 2} E_{\mu_1}\big[f_1{\bf 1}_{M_1(n)}\big]D^{\text{ex},1}_{N,n}(f_{1,n}^{1/2})
,
\label{eq_Dirichlet_form_to_estimate_rplct_lemma_after_conditioning}
\end{equation}
where $D^{\text{ex},1}_{N,n}$ is defined as in~\eqref{eq_def_dirichlet_corner_flip_toutes_length}, but with $U_n$ instead of $\mu_1$ and paths in $M_1(n)$ rather than $M_1$. \\
Note also that, as $f$ is supported on $\e_A$, $f_{1,n}$ is supported on paths with at most $AN$ edges. In addition, for any $\rho\in M_1(n)$ with $n<2\epsilon N +2R$, $W^{\epsilon N+R}(\rho)$ is empty. Thus:
\begin{align}
&\sum_{\rho\in M_1}\mu_1(\rho)f_1(\rho)\frac{1}{N}\sum_{x\in W^{\epsilon N +R}(\rho)}\Phi_j(\rho,x) \nonumber\\
&\hspace{2.5cm}=
\sum_{n=2\epsilon N+2R}^{AN} E_{\mu_1}\big[f_1{\bf 1}_{M_1(n)}\big] \frac{1}{|M_1(n)|}\sum_{\rho\in M_1(n)}\frac{1}{N}\sum_{x\in W^{\epsilon N +R}(\rho)}\Phi_j(\rho,x).\label{eq_decomp_paths_fixed_length}
\end{align}
Now that paths appearing in~\eqref{eq_decomp_paths_fixed_length} have fixed length, it is possible to give a numerical label $i\in\{\epsilon N+R+1,... ,n-\epsilon N - R\}$ to each point in $W^{\epsilon N +R}(\rho)$, independent from the choice of the path $\rho \in M_1(n)$. One can then associate a SSEP configuration $\sigma\in \{0,1\}^n$ with each $\rho$ (see Figure~\ref{fig_corner_flip_ssep}) and rewrite the quantity $\Phi_j(\rho,x)$ for $x\in W^{\epsilon N +R}(\rho)$ as:
\begin{equation}
\forall x\in W^{\epsilon N +R}(\rho),\qquad \Phi_j(\rho,x) = \Phi(\tau_i\sigma),
\end{equation}
where $i\in\{\epsilon N+R+1,...,n-\epsilon N - R\}$ is the label of the point $x$ and $\tau_i\sigma(i') = \sigma(i'-i)$ is the translation operator. 
The average on $M_1$ in~\eqref{eq_decomp_paths_fixed_length} is then equal to:
\begin{align}
\sum_{n=2\epsilon N +2R}^{AN} E_{\mu_1}\big[f_1{\bf 1}_{M_1(n)}\big]\frac{1}{|M_1(n)|}\sum_{\sigma \in \{0,1\}^n} g_{1,n}(\sigma)\frac{1}{N}\sum_{i= \epsilon N +R+1}^{n - \epsilon N -R}\Phi_j(\tau_i \sigma).\label{eq_to_estimate_rplct_lemma_after_conditioning}
\end{align}
In the last line, $g_{1,n}$ is defined for $\sigma\in\{0,1\}^n$ by $g_{1,n}(\sigma) = g_{1,n}(\rho(\sigma))$, with $\rho(\sigma)$ the unique path in $M_1(n)$ corresponding to the particle configuration $\sigma$, as represented in Figure~\ref{fig_corner_flip_ssep}. In view of~\eqref{eq_one_block_interm}--\eqref{eq_Dirichlet_form_to_estimate_rplct_lemma_after_conditioning}--\eqref{eq_to_estimate_rplct_lemma_after_conditioning}, to prove the one block estimate~\eqref{eq_one_block}, it is sufficient to prove:
\begin{align}
&\lim_{\ell\rightarrow\infty}\limsup_{\epsilon\rightarrow 0}\limsup_{N\rightarrow\infty}\max_{1\leq j\leq J}\nonumber \\
&\hspace{1cm}\max_{n\in\{2\epsilon N+2R,...,AN\}}\sup_{g\geq 0 : E_{U_n}[g]=1}\bigg\{\frac{a}{N}\ E_{U_n}\Big[g\sum_{i= \epsilon N +R+1}^{n - \epsilon N -R} \Phi_j(\tau_i\cdot)\Big] - ND^{\text{ex}}_{n}(g^{1/2})\bigg\}
= 0
.
\label{eq_1_block_is_1block_SSEP}
\end{align}
The notation $D_n^{\text{ex}}$, already used in Section~\ref{app_behaviour_pole}, stands for the Dirichlet form associated with a SSEP on $n$ sites. 
We are left with a usual one block estimate for a SSEP of size $n$, proven e.g. in \cite{Eyink1990}. 
The size $n$ of the SSEP becomes irrelevant in the large $N$ limit since only the $2\ell+1$ site closest to each $i$ matter. This concludes the proof of~\eqref{eq_one_block}. 
The two block estimate~\eqref{eq_two_blocks} is proven similarly.
\end{proof}
\subsection{Equilibrium estimates at the pole}\label{sec_equilibrium_estimates_at_the_pole}
In this section, we investigate the equilibrium measure $\mu_\ell$ (see~\eqref{eq_def_mu_ell}) of the zero-range process at the poles. We prove:
\begin{prop}\label{prop_equil_properties_mu_ell}
The sequence of the laws of the top height of a path under $\mu_\ell$ ($\ell\in\N_{\geq 1}$) satisfies a large deviation principle at speed $\ell$ (equivalently: the number of particles or of antiparticles) with good, convex rate function given by:
\begin{equation}
\forall u\geq 0,\qquad C(u) = 2\beta u - 2u \log\big(1+1/(2u)\big) - \log(1+2u)-\log\big(1-e^{-\beta}\big).\label{eq_def_free_energy_at_the_pole}
\end{equation}
In particular, recalling that $\ell_{\min}(n) := \sqrt{n}/(4\log n)$:
\begin{equation}
\lim_{n\rightarrow\infty}\max_{\ell_{\min}(n)\leq \ell \leq n}\big|E_{\mu_\ell}[\phi]\big| = 0,\qquad \phi = {\bf 1}_{p=2} - e^{-\beta}.\label{eq_estimate_1_pis2_minus_exp-_beta_under_mu_ell_appendix}
\end{equation}
\end{prop}
\begin{proof}
We prove~\eqref{eq_def_free_energy_at_the_pole} first. 
We say that a path is a north-east path if it goes either up or right, a south-east path if it goes either down or right and an up-down path if it is the concatenation of a north-east path followed by a south-east path (see Figure~\ref{fig_bijection_path_pole}).

Recall that up-down paths correspond to possible configurations of the neighbourhood of the north pole of microscopic interfaces $\gamma^N\in\Omega^N_{\text{mic}}$. In contrast, the north-east paths appearing on the left of Figure~\ref{fig_bijection_path_pole} below do not have any interpretation in terms of microscopic interfaces. 

We speak alternately of up-down paths or of particle/antiparticle configurations in the proof depending on what is easier to use, the height of a path corresponding to $\sum_{x\leq L_1}\eta_x = -\sum_{x> L_1}\eta_x$. Here, the point $L_1$ is the left extremity of the pole of an up-down path. This pole and $L_1$ for an up-down path are defined analogously to the north pole of a curve $\gamma^N\in\Omega^N_{\text{mic}}$ and its left extremity $L_1(\gamma)$. We similarly write $p$ for the length of the pole of an up-down path. 

Fix $\ell_{\min}(n)\leq \ell\leq n$ throughout. Let us first study the probability to observe a given height under $\mu_{\ell}$. There are exactly $\binom{2q+\ell-2}{2q}$ configurations with height $q\in\N$. To see it, notice that this is the number of north-east paths with length $2q+\ell-2$ and $2q$ vertical edges. To each such path $\rho$, one can associate a unique up-down path of length $2q+\ell$ as follows (see also Figure~\ref{fig_bijection_path_pole}).
\begin{itemize}
	\item Travelling on the path $\rho$ from its origin, stop at the first point at height $q$, call it $X$, and cut the path there in two parts $\rho_{\leq X}$ and $\rho_{> X}$. 
	\item Add two horizontal edges to $\rho_{\leq X}$ immediately after $X$, call $\rho_{X+2}$ the resulting path.
	\item Change $\rho_{>X}$ into its symmetrical $\tilde \rho_{>X}$ with respect to the horizontal, i.e. change every upwards edge into a downwards one, leaving the horizontal edges unchanged. Stitch the last edge of $\rho_{X+2}$ to the first of $\tilde\rho_{>X}$ to obtain an up-down path of height $q$ and length $2q+\ell$. 
\end{itemize}
One easily checks that this mapping is a bijection, mapping the point $X$ onto the left extremity $L_1$ of the pole of the up-down path, whence:
\begin{equation}
\forall q\leq h_{\max}(n) = n(1+2\log n),\qquad \mu_{\ell}\Big(\sum_{j\leq L_1}\eta_{j} = q\Big) = \binom{2q+\ell-2}{2q}e^{-2\beta q-\beta\ell}/\mathcal Z_\ell.
\end{equation}
Let us investigate the dependence of this quantity in $q< h_{\max}(n)$:
\begin{equation}
\mu_{\ell}\Big(\sum_{j\leq L_1}\eta_{j} = q+1\Big)/\mu_{\ell}\Big(\sum_{j\leq L_1}\eta_{j} = q\Big) = e^{-2\beta}\frac{(2q+\ell)(2q+\ell-1)}{(2q+2)(2q+1)}.\label{eq_tx_variation_height_sous_mu_ell}
\end{equation}
This quantity increases until some value $q_c$ of $q$, given by
\begin{equation}
q_c = \frac{1}{2}(e^{\beta}-1)^{-1}\ell + o(\ell) =: u_c\ell + o(\ell).
\end{equation}
In particular, due to the logarithm in the large deviation bounds for the measure $\mu_\ell$ that we are trying to prove, only the maximum value of $\binom{2q+\ell-2}{2q}e^{-2\beta q-\beta\ell}$ will matter. One thus needs only consider heights of order $\ell$ in the large $\ell$ limit. For fixed $u>0$, elementary computations give:
\begin{equation}
\frac{1}{\ell}\log \bigg[\binom{2\lfloor \ell u\rfloor +\ell-2}{2\lfloor \ell u\rfloor }e^{-2\beta \lfloor \ell u\rfloor -\beta\ell}\bigg] = -\beta-2\beta u +2u\log\big(1+1/(2u)\big) + \log(1+2u) + o_\ell(1).\label{eq_value_mu_ell_de_height_of_order_ell}
\end{equation}
Define the function $D(\cdot)$ on $\R^*_+$ by;
\begin{equation}
D(u) = \beta+2\beta u -2u\log\big(1+1/(2u)\big) - \log(1+2u) \geq 0,
\qquad 
u\geq 0
.
\end{equation}
From~\eqref{eq_value_mu_ell_de_height_of_order_ell} and with $u_c = \frac{1}{2}(e^{\beta}-1)^{-1}$, we obtain for the normalisation $\mathcal Z_\ell$:
\begin{equation}
\lim_{\ell\rightarrow\infty}\frac{1}{\ell}\log \mathcal Z_\ell = D(u_c) = \beta+\log\big(1-e^{-\beta}\big).\label{eq_valeur_tau_P_ds_proof}
\end{equation}
We now turn to the large deviation principle for the height of a path. From~\eqref{eq_value_mu_ell_de_height_of_order_ell} and~\eqref{eq_valeur_tau_P_ds_proof}, we obtain
\begin{equation}
\frac{1}{\ell}\log \mu_{\ell}\Big(\sum_{j\leq L_1}\eta_{j} = \lfloor \ell u\rfloor \Big) = -(D(u)-D(u_c))+ o_\ell(1)
.
\label{eq_ldp_height_path}
\end{equation}
Define the rate function $C(\cdot)$ on $\R^*_+$ by
\begin{equation}
C(u) = D(u) - D(u_c)\geq 0
,\qquad
u\geq 0
.
\label{eq_rate_function_height_path}
\end{equation}
The function $C$ is $C^\infty$ on $\R^*_+$ and satisfies:
\begin{equation}
C(u_c) = 0 =C'(u_c),\qquad C''(u) = \frac{2}{u+2u^2}>0\text{ for each }u>0,
\end{equation}
so that $C$ is strictly convex and a good rate function. The large deviation principle follows from~\eqref{eq_ldp_height_path}. \\

\begin{figure}[H]
\begin{center}
\includegraphics[width=12cm]{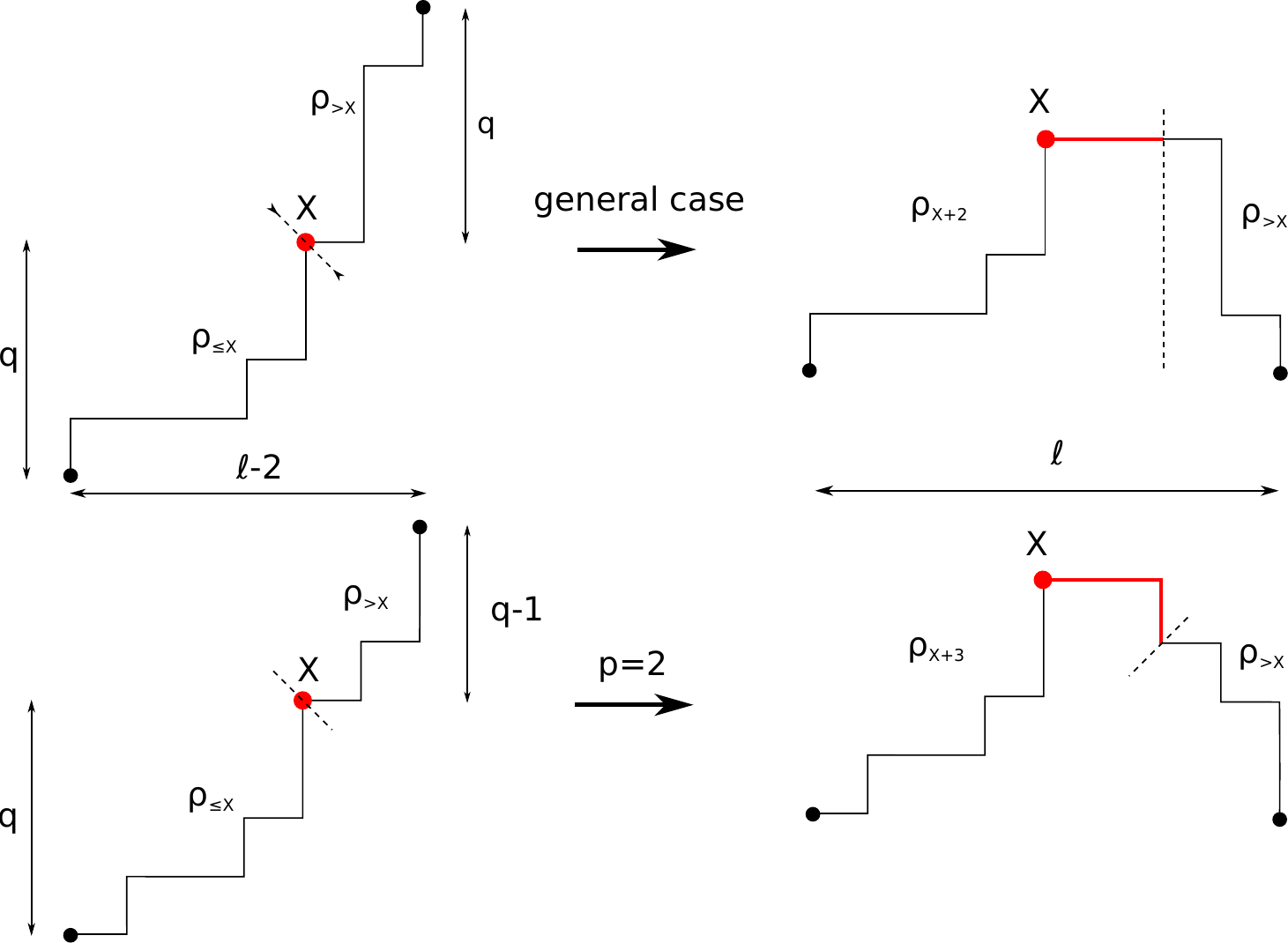} 
\caption{Bijection argument to count the number of up-down paths with height $q$ and length $2q+\ell$, with (bottom figure) or without (top figure) conditions on the size $p$ of the pole. Black dots delimit the extremities of the paths. Without conditions on $p$, an up-down path is obtained by transforming a north-east path with length $2q+\ell-2$ and height $2q$ (top left figure). If $p=2$, then instead the length is $2q+\ell-3$ and height $2q-1$ (bottom left figure). Dashed lines delimit the two portions $\rho_{\leq X}$ and $\rho_{>X}$ of the north-east paths. The red dot is the place at which the initial north-east path is split and the red, thick lines on the right-hand side are the edges added to the initial path to obtain an up-down configuration with height $q$ and length $2q+\ell$ (with also $p=2$ in the bottom figure). \label{fig_bijection_path_pole}}
\end{center}
\end{figure}
\indent It remains to prove~\eqref{eq_estimate_1_pis2_minus_exp-_beta_under_mu_ell_appendix}. This follows from the large deviations principle~\eqref{eq_ldp_height_path} and the following observation.
Constructing a path with $p=2$ and height $q\in\N_{\geq 1}$ is done by building a north-east path of length $2q-1+\ell-2$ with $2q-1$ vertical edges, then cutting it as described previously and taking the symmetric part of the path after the first point $X$ at height $q$. The only difference is that one now sticks not just two horizontal edges after $X$, but two horizontal edges followed by a vertical one hanging from below, before stitching back the two parts of the path (see Figure~\ref{fig_bijection_path_pole}). There are thus $\binom{2q+\ell-3}{2q-1}$ configurations with $p=2$ and height $q\in\N_{\geq 1}$ and:
\begin{equation}
\mu_\ell\Big(p=2,\sum_{j\leq L_1}\eta_{j} = q\Big) 
= 
\mathcal Z_\ell^{-1}e ^{-\beta\ell-2\beta q}\binom{2q+\ell-3}{2q-1} 
= 
\frac{2q}{2q+\ell-2}\mu_\ell\Big(\sum_{j\leq L_1}\eta_{j} = q\Big)
.
\label{eq_mes_pis2_height_q_is_propto_height_q}
\end{equation}
As the height $q$ of a path corresponds to $\sum_{j\leq L_1}\eta_j$, 
the expectation in~\eqref{eq_estimate_1_pis2_minus_exp-_beta_under_mu_ell_appendix} at $\ell$ then reads: 
\begin{align}
E_{\mu_\ell}[\phi] &=-e^{-\beta} + \sum_{q\geq 1}\mu_\ell\Big(p=2,\sum_{j\leq L_1}\eta_j=q\Big) \overset{\eqref{eq_mes_pis2_height_q_is_propto_height_q}}{=}  E_{\mu_\ell}\bigg[ \frac{2\sum_{j\leq L_1}\eta_j}{2\sum_{j\leq L_1}\eta_j+\ell-2}-e^{-\beta}\bigg].\label{eq_expectation_equilibrium_to_compute_1pis2_0}
\end{align}
Let $\zeta>0$. The integrand in~\eqref{eq_expectation_equilibrium_to_compute_1pis2_0} is bounded and, for all $\ell$ large enough,
\begin{equation}
\frac{1}{\ell}\log\mu_{\ell}\Big(\frac{1}{\ell}\sum_{j\leq L_1}\eta_j \notin[u_c-\zeta,u_c+\zeta]\Big) \leq -C(u_c+\zeta)/2<0.\label{eq_prob_mu_ell_deviating_mean_exp_small}
\end{equation}
As a result, since $2u_c/(2u_c+1) = e^{-\beta}$, the expectation in~\eqref{eq_expectation_equilibrium_to_compute_1pis2_0} is recast as follows:
\begin{align}
E_{\mu_\ell}[\phi] 
&= 
E_{\mu_\ell}\bigg[ 
\bigg(\frac{2\ell^ {-1}\sum_{j\leq L_1}\eta_j}{2\ell^ {-1}\sum_{j\leq L_1}\eta_j+1}-e^{-\beta}\bigg){\bf 1}_{u_c-\zeta \leq \frac{1}{\ell}\sum_{j\leq L_1}\eta_{j} \leq u_c+\zeta}
\bigg] +O(\ell^{-1})\nonumber\\
&=O(\zeta) + O(\ell^{-1}).
\end{align}
The $O(\zeta)$ is independent of $\ell \in \{\ell_{\min}(n),...,n\}$ and $\ell_{\min}(n) = \sqrt{n}/(4\log n)$ diverges with $n$. 
This proves~\eqref{eq_estimate_1_pis2_minus_exp-_beta_under_mu_ell_appendix}.
\end{proof}

\section{Topology results}\label{app_prop_e_r}
At the microscopic level, elements of $\Omega^N_{\text{mic}}$ are simple curves. 
Macroscopically, however, curves may be non-simple, for instance when the situation of Figure~\ref{fig_non_simple_curve} occurs. 
Microscopic estimates on the poles, 
such as Proposition~\ref{prop_value_slope_at_poles}, 
indicate that macroscopic trajectories should have well-behaved poles and thus be simple at almost every time.

To turn the information provided by Proposition~\ref{prop_value_slope_at_poles} into a property of limiting curves (e.g. for lower bound large deviations), 
it is necessary to first be able to define a limiting object (for lower bound large deviations, a limiting probability measure) on general, possibly pathological trajectories. 

This means that we have to define a topology on trajectories that is strong enough to handle pathological cases, 
yet weak enough for microscopic estimates to be available 
(which for instance excludes pointwise in time estimates in Hausdorff distance). 

The following appendix builds this topological setting, 
first on curves (Section~\ref{sec_omega_e}), 
then on trajectories (Section~\ref{sec_the_set_E(0,T0)}); 
concluding with a proof of exponential tightness (Section~\ref{appen_tightness}). 
\subsection{Topological properties of $\Omega$ and $\e$}
\label{sec_omega_e}
In preparation for the definition of the topology on trajectories, 
we investigate in this section topological properties of $\Omega,\e$ for Hausdorff and volume distances, 
expressing the Hausdorff distance between curves in $\Omega$ as a function of the volume distance and the distance between the poles only.  
We also prove Lemma~\ref{lemm_local_jump_rates} on the locality of jump rates of the contour dynamics.

The Hausdorff distance $d_{\mathcal H}(A,B)$ between two non-empty compact sets $A,B$ is defined in~\eqref{eq_def_distance_Hausroff} 
and equivalently by:
\begin{align}
\qquad d_{\mathcal H}\big(A,B) = \max\Big\{\sup_{y\in B}\inf_{x\in A}\|x-y\|_1, \sup_{x\in A}\inf_{y\in B}\|x-y\|_1\Big\} 
= 
d_{\mathcal{H}}\big(\partial A,\partial B)
.
\label{eq_def_hausdorff_distance}
\end{align}
Recall also the convention:
\begin{equation}
\forall\gamma,\tilde\gamma\in\Omega,\qquad d_{L^1}(\gamma,\tilde\gamma) := d_{L^1}(\Gamma,\tilde\Gamma) = \int_{\R^2}|{\bf 1}_{\Gamma} - {\bf 1}_{\tilde \Gamma}|\, du\, dv,\label{eq_convention_volume_appendice}
\end{equation}
where $\Gamma,\tilde\Gamma$ are the droplets associated with $\gamma,\tilde\gamma$ respectively.

For $1\leq k\leq 4$, recall that the $z_k=z_k(\gamma)$ are the extremal coordinates of a curve $\gamma\in\Omega$ (see Figure~\ref{fig_gamma_prime_z_k_w_k} and recall that $L_k$ is the left extremity of pole $k$):
\begin{align}
&z_1 = \sup \{ x\cdot{\bf b}_{\pi/2} : x\in\gamma\} = L_1\cdot {\bf b}_{\pi/2},
\qquad 
z_3 = \inf \{ x\cdot{\bf b}_{\pi/2} : x\in\gamma\} = L_3\cdot {\bf b}_{\pi/2},\nonumber\\
&z_2 = \sup \{ x\cdot{\bf b}_0 : x\in\gamma\} = L_2\cdot {\bf b}_0,
\hspace{1.51cm} z_4 = \inf \{ x\cdot{\bf b}_0 : x\in\gamma\} = L_4\cdot {\bf b}_0.
\label{eq_def_z_k_appB}
\end{align}
Equation~\eqref{eq_def_hausdorff_distance} directly yields that the $z_k$ are $1$-Lispchitz functions in Hausdorff distance. 
In addition, since each curve $\gamma\in\Omega$ surrounds $0$, 
it holds that $z_1,z_2\geq 0$ and $z_3,z_4\leq 0$. 

We first prove that $\Omega,\e$ are closed with respect to convergence in Hausdorff and/or volume distances.
\begin{prop}\label{prop_e_closed}
The sets $\Omega$ and $\e$ are closed in the topology associated with the Hausdorff distance $d_{\mathcal H}$. 
Moreover, $\e$ is closed in $\Omega$ for the volume distance $d_{L^1}$, defined in~\eqref{eq_convention_volume_appendice}. 
In addition, $\Omega\cap \{|\gamma|\leq \kappa\}$ is compact in Hausdorff topogy (thus $\e\cap \{|\gamma|\leq \kappa\}$ as well) for each $\kappa>0$.
\end{prop}
\begin{proof}
If a sequence converges in Hausdorff distance, then it converges in volume distance $d_{L^1}$. 
In particular $\e$, 
defined as a closed ball in $\Omega$ in $d_{L^1}$-distance,
is a closed subset of $\Omega$ for both volume and Hausdorff distances. 
It is thus enough to prove that $\Omega$ is closed and $\Omega\cap\{\|\gamma\|\leq \kappa\}$ is compact in Hausdorff distance for $\kappa>0$. We start with the closeness.

Let $\gamma^n\in\Omega$ ($n\in\N$) converge to $\gamma$ in Hausdorff distance. 
Then by definition~\eqref{eq_def_distance_Hausroff} or~\eqref{eq_def_hausdorff_distance} $\gamma$ delimits a bounded region that must contain $0$. 
The curve $\{0\}$ is in $\Omega$ by assumption, 
so let us assume that $\gamma$ is not reduced to a point. 
Convergence of the $(z_k(\gamma^n))_n$ ($1\leq k\leq 4$) guarantees that $\gamma$ has poles, 
i.e. that regions with extremal ordinate or abscissa are connected. 
It thus also makes sense to talk of the regions of $\gamma$ (some of them possibly reduced to a point).

Let $k\in\{1,...,4\}$ be such that region $k$ of $\gamma$ is not reduced to a point. 
To prove $\gamma\in\Omega$, 
it remains to prove that its region $k$ is a Lipschitz curve with tangent vector ${\bf T}$ of $\gamma$ satisfying ${\bf T}\cdot {\bf b}_{-(k-1)\pi/2}\geq 0$ and ${\bf T}\cdot {\bf b}_{-k\pi/2}\geq 0$. 
This is equivalent to proving that region $k$ of $\gamma$ is the graph of a $1$-Lipschitz function in the reference frame $({\bf b}_{-k\pi/4},{\bf b}_{-k\pi/4+\pi/2})$ (in particular this is true for the $\gamma^n$). 
Let $x$ be in region $k$ of $\gamma$. 
If $x$ is in one of the poles, there is nothing to prove since the poles are horizontal or vertical segments (possibly reduced to a point). 
Suppose instead $d(x,P(\gamma))> 0$ and assume $k=1$ so that $x$ is in the first region, the others being similar.  
Let $\epsilon>0$ be such that: 
\begin{equation}
x\cdot{\bf b}_{\pi/2}\leq z_1(\gamma)-\epsilon
,\qquad
x\cdot{\bf b}_0<z_2(\gamma)-\epsilon
.
\end{equation}
Then all points $y\in\gamma$ with:
\begin{equation}
y\cdot{\bf b}_{-\pi/4}\in I
:=
\Big\{u\in\R: u\in x\cdot{\bf b}_{-\pi/4} + \Big[-\frac{\epsilon}{2\sqrt{2}},\frac{\epsilon}{2\sqrt{2}}\Big]\Big\}
\end{equation}
satisfy $y\cdot{\bf b}_{\pi/2}<z_1(\gamma)-\epsilon/2$ and $y\cdot{\bf b}_0<z_2(\gamma)-\epsilon/2$. 
This implies that, for $n$ large enough, all points in $I$ correspond to a point $y^n$ in region $1$ of $\gamma^n$. 
Write then the set of such $y^n\in \gamma^n$ as the graph of a $1$-Lipschitz function:
\begin{equation}
\big\{y^n\text{ in region 1 of }\gamma^n:y^n\cdot{\bf b}_{-\pi/4}\in I\big\}
:=
\big\{u{\bf b}_{-\pi/4}+f^n(u){\bf b}_{\pi/4}:u\in I\big\}
.
\end{equation} 
Hausdorff convergence of $\gamma^n$ to $\gamma$ implies uniform convergence of the $f^n$ on $I$. 
It follows that their limit $f$ is $1$-Lipschitz, 
with:
\begin{equation}
\big\{y\text{ in region 1 of }\gamma:y\cdot{\bf b}_{-\pi/4}\in I\big\}
=
\big\{u{\bf b}_{-\pi/4}+f(u){\bf b}_{\pi/4}:u\in I\big\}
.
\end{equation}
The fact that $x\notin P(\gamma)$ was an arbitrary point of region $1$ concludes the proof: $\gamma\in\Omega$.\\

Consider now the compactness claim. For $\kappa>0$, $\gamma\in\Omega\cap \{|\gamma|\leq \kappa\}$ implies that $|z_k(\gamma)|\leq \kappa/2$ for each $k$ with $1\leq k \leq 4$. Since $\gamma$ surrounds $0$ by Definition~\ref{def_state_space} of $\Omega$, the set $\Omega\cap \{|\gamma|\leq \kappa\}$ is closed in Hausdorff topology in the set of non-empty compact sets in $[-\kappa/2,\kappa/2]^2$, which is compact. 
This concludes the proof.
\end{proof}
\begin{figure}
\begin{center}
\includegraphics[width=11cm]{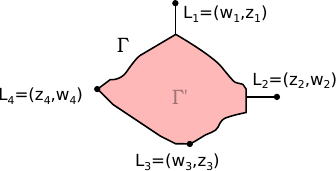}
\caption{A droplet $\Gamma$ associated with a curve in $\Omega$ with point-like north and east poles. The left extremities $L_k = (w_k,z_k)$ ($1\leq k\leq 4$) of each pole are indicated by black dots. Here, for $u,v\in\R$, $(u,v)$ is the point $u{\bf b}_0 + v{\bf b}_{\pi/2}$. The droplet $\Gamma'$ defined in Lemma~\ref{lemm_continuite_z_k'} corresponds to the red shaded area.\label{fig_gamma_prime_z_k_w_k}}
\end{center}
\end{figure}
To control Hausdorff convergence, 
we will need to control the position of the poles. 
The next lemma characterises the continuity of the other coordinate $w_k$ of the left extremity of pole $k\in\{1,...,4\}$
(see Figure~\ref{fig_gamma_prime_z_k_w_k}). 
Continuity properties of the $z_k,w_k$ are summarised on Figure~\ref{fig_lsc}.
\begin{lemm}\label{lemm_w_k}
For $\gamma\in\Omega$, let $w_k(\gamma)$ ($1\leq k \leq 4$) be the other four coordinates of the left extremities $L_k(\gamma)$ of each pole: $w_i(\gamma) = L_i(\gamma)\cdot {\bf b}_0$ for $i\in\{1,3\}$, 
$w_i(\gamma) = L_i(\gamma)\cdot {\bf b}_{\pi/2}$ for $i\in\{2,4\}$. 

The function $\gamma\mapsto w_k(\gamma)$ ($1\leq k\leq 4$) satisfies $w_k\in [z_{k-1},z_{k+1}]$ with the convention $k+1:=1$ if $k=4$ and $k-1 :=4$ if $k=1$. 
It is not continuous in Hausdorff distance on $\Omega$, 
but $w_k$ is lower semi-continuous if $k\in\{1,4\}$, upper semi-continuous if $k\in\{2,3\}$. 

In addition, if $\gamma\in\Omega$ has point-like pole $k\in\{1,...,4\}$, 
then $\gamma$ is a point of continuity of $w_k$ for the Hausdorff distance. 
That is, for $(\gamma^n)_n\subset \Omega$ converging to $\gamma$ in Hausdorff distance:
\begin{align}
\lim_{n\rightarrow\infty} \big\|L_k(\gamma^n)-L_k(\gamma)\big\|_1\vee \big\|R_k(\gamma^n)-R_k(\gamma)\big\|_1 = 0.\label{eq_convergence_pole_at_fixed_time}
\end{align}
\end{lemm}
\begin{proof}
Note that the $z_k$ $(1\leq k \leq 4)$ are the extremal coordinates of points of an interface by definition, see Figure~\ref{fig_gamma_prime_z_k_w_k}. 
E.g. for $k=1$, $w_1$ is the abscissa of the left extremity $L_1$ of the north pole while $z_4$ is the lowest abscissa and $z_2$ the highest. 
Thus $w_1\in[z_4,z_2]$ and similarly for $w_k$ for $k\in\{2,3,4\}$.

The lack of continuity of $w_k$ is best explained on a picture (see right picture in Figure~\ref{fig_z_k_w_k}). 
The idea is the following. 
E.g. for the north pole $k=1$, let $\gamma\in\e$ be a simple curve with north pole not reduced to a point: $|P_1(\gamma)|>0$. 
Then $w_1(\gamma)$ is always the abscissa of the leftmost point of the pole by definition, 
but one can build a sequence $(\gamma^n)_n$ of curves with point-like north pole at ordinate $z_1(\gamma)+1/n$ and abscissa $R_1(\gamma)\cdot {\bf b}_0\neq w_1(\gamma)$. 
Then $\lim_nd_{\mathcal H}(\gamma^n,\gamma) =0$, 
but $w(\gamma^n)$ converges to $R_1(\gamma)\cdot {\bf b}_0\neq w_1(\gamma)$.  

The fact that $w_1(\gamma)$ is the abscissa of the leftmost point of the pole also implies the lower semi-continuity of $\gamma\in\Omega\mapsto w_1(\gamma)$ for the Hausdorff distance. 
Indeed, let $\gamma^n\in\Omega$ ($n\in\N$) converge to $\gamma\in\Omega$ for $d_{\mathcal H}$. 
Then $L_1(\gamma^n)\cdot {\bf b}_{\pi/2} = z_1(\gamma^n)$ converges to $z_1(\gamma)$ while $w_1(\gamma^n) = L_1(\gamma^n)\cdot {\bf b}_{0}\in [z_4(\gamma^n),z_2(\gamma^n)]$ is bounded. 
Let $L:= \liminf_{n\to\infty}L_1(\gamma^n)\in\gamma$. 
As $L\in\gamma$ has maximal ordinate $z_1(\gamma)$, it belongs to the north pole of $\gamma$, 
hence $L\cdot{\bf b}_0\geq w_1(\gamma)$ and the lower semi-continuity. 

Finally, let us prove that $\gamma\in\Omega$ with point-like pole $1$ is a point of continuity of $w_1$. 
Let $(\gamma^n)_n\subset\e$ converge to $\gamma$ in Hausdorff distance. 
Recall that $z_1=L_1\cdot {\bf b}_{\pi/2}$ is continuous in Hausdorff distance as noticed below~\eqref{eq_def_z_k_appB}.  
We just saw that $w_1 = L_1\cdot{\bf b}_0$ is lower semi-continuous. 
One can similarly show that $R_1\cdot{\bf b}_0$ is upper semi-continuous (see Figure~\ref{fig_lsc}). 
Thus, since $\Gamma$ has point-like pole 1:
\begin{equation}
L_1(\gamma)\cdot {\bf b}_0 
\leq 
\liminf_{n\rightarrow\infty} L_1(\gamma^n)\cdot {\bf b}_0
\leq 
\limsup_{n\rightarrow\infty} R_1(\gamma^n)\cdot {\bf b}_0
\leq 
R_1(\gamma)\cdot {\bf b}_0 = L_1(\gamma)\cdot {\bf b}_0
.
\end{equation}
This concludes the proof of~\eqref{eq_convergence_pole_at_fixed_time} for the north pole. 
The other poles are similar.
\end{proof}
The last ingredient we need to compare convergence in Hausdorff distance and convergence of the poles and volume is a control of the poles of the largest droplet at 0 volume distance of a given droplet. 
This is stated next (see also Figure~\ref{fig_lsc}), together with continuity estimates in volume distance that will be used in the proof of the locality of the dynamics in Lemma~\ref{lemm_stability_initial_condition_appendix}.
\begin{figure}
\begin{center}
\includegraphics[width=8.5cm]{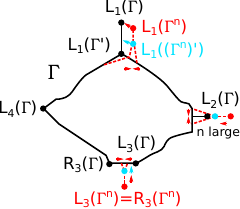} 
\caption{Continuity properties of the $z_k,w_k$ and $\gamma\mapsto z_k(\Gamma')$. 
A droplet $\Gamma$ and an element $\Gamma^n$ of a sequence of droplets  converging to $\Gamma$ in Hausdorff distance ($\Gamma^n$ differs from $\Gamma$ in the portions materialised by red lines) are represented.
The extremities of the poles of $\Gamma,\Gamma^n$ are given by dark dots and the positions of $L_k(\Gamma'),L_k((\Gamma^n)')$ by light blue dots. 
Arrows indicate the evolution of $\Gamma^n$ when $n$ is large. 
The convergence at each pole illustrates different possible behaviours.\\
The droplet $\Gamma$ has point-like pole 1, thus $\lim_nL_1(\Gamma^n)= L_1(\Gamma)$. 
However, on this figure, 
$\Gamma^n$ is chosen in such a way that $\liminf_nz_1((\Gamma^n)')> z_1(\Gamma')$. 
Similarly, at pole 2 of $\Gamma$, $\Gamma^n$ is such that $\lim L_2(\Gamma^n) = L_2(\Gamma)$, but $\inf_nz_2((\Gamma^n)')>z_2(\Gamma')$: $\gamma\mapsto z_i(\Gamma')$ is lower semi-continuous if $i\in\{1,2\}$.\\
At pole $3$, where $L_3(\Gamma^n) = R_3(\Gamma^n)$, 
the point $L_3(\Gamma^n)$ does not converge to $L_3(\Gamma^n) = L_3(\Gamma)$, as  
$w_3(\Gamma^n) = L_3(\Gamma^n)\cdot{\bf b_0}$ satisfies $\lim_n w_3(\Gamma^n)<w_3(\Gamma)$. 
The functionals $R_3\cdot{\bf b}_0$ and $L_3\cdot {\bf b}_0 = w_3$ are respectively lower- and upper semi-continuous (while instead $w_1 = L_1\cdot {\bf b}_0$ is lower- and $R_1\cdot {\bf b}_0$ upper semi-continuous, etc.).
\label{fig_lsc}}
\end{center}
\end{figure}

\begin{lemm}\label{lemm_continuite_z_k'}
For $\gamma=\partial\Gamma\in\Omega$ with volume $|\Gamma|>0$, 
define $\Gamma'\subset \Gamma$ as the largest droplet with simple boundary such that $d_{L^1}(\Gamma,\Gamma')=0$. 
In other words, $\Gamma'$ is the closure of the interior of $\Gamma$ (see Figure~\ref{fig_gamma_prime_z_k_w_k}). 
Define then $\Omega'\subset \Omega$ as the set of boundaries of all such droplets:
\begin{equation}
\Omega' := \{\partial(\Gamma') : \gamma = \partial\Gamma\in\Omega\}.
\label{eq_def_eprime}
\end{equation}
Then, in both Hausdorff and volume distances, 
$\gamma\in\{\tilde\gamma\in\Omega:|\tilde\Gamma|>0\}\mapsto z_k(\Gamma')$ is lower semi-continuous if $k\in\{1,2\}$, 
upper semi-continuous if $k\in\{3,4\}$. 
In addition $\gamma \mapsto w_k(\Gamma')$ is upper semi-continuous in volume distance on the same set of curves if $k\in\{1,4\}$, lower semi-continuous if $k\in\{2,3\}$, 
and similar statements hold for the $R_k\cdot{\bf b}_{(k-1)\pi/2}$ ($1\leq k\leq 4$).
\end{lemm}
\begin{proof}
Note that upper semi-continuity of $z'_k := \gamma\mapsto z_k(\Gamma')$ ($k\in\{3,4\}$) is the same as lower semi-continuity of $z'_1, z'_2$ up to rotating all curves by $\pi$.  
We focus on the lower semi-continuity of $z'_1$ on $\Omega$, 
$z'_2$ being similar. 
Since convergence in Hausdorff distance implies convergence in volume, it is enough to work with the latter. 
Let $\gamma^n\in\Omega$ ($n\in\N$) converge to $\gamma$ in volume distance. 
Then either $(z'_1(\gamma^n))_n$ diverges to $+\infty$, 
in which case the lower semi-continuity holds, 
or it is bounded along a subsequence that we still denote by $(z'_1(\gamma^n))_n$. 
Taking yet another subsequence, we may assume $(z'_1(\gamma^n))_n$ converges. 
Suppose by contradiction that its limit $\bar z'_1$ satisfies $\bar z'_1\leq  z'(\gamma)-\epsilon$ for some $\epsilon>0$. 
Then, for all large enough $n$:
\begin{equation}
z'_1(\gamma^n)
\leq 
z'_1(\gamma)-\epsilon/2
.
\end{equation}
The last equation implies that the intersection of the strip $\{(u,v)\in\R^2 : v\in z'_1(\gamma) + [-\epsilon/2,0]\}$ with the droplet $\Gamma^n$ associated with $\gamma^n$ has vanishing volume when $n$ is large. 
On the other hand, the volume of this strip intersected with $\Gamma$ (the droplet associated with $\gamma$) is strictly positive. 
This is absurd, since $\lim_nd_{L^1}(\gamma^n,\gamma)=0$ by assumption. 
Thus $\bar z'_1> z'_1(\gamma)-\epsilon$ for arbitrary $\epsilon>0$ and the lower semi-continuity.

We now prove that $\limsup_nw'_1(\gamma^n)\leq w'_1(\gamma)$ where $w'_1:=\tilde \gamma\mapsto w_1(\tilde \Gamma')$. 
We may suppose without loss of generality that $\gamma,\gamma^n$ are simple, 
so that $w_1(\gamma)=w_1(\Gamma')$ and idem for $\gamma^n$ ($n\in\N$). 
Suppose first that the first region of $\gamma$ is reduced to a point. 
Then the structure of elements of $\Omega$ implies that $|\Gamma\cap H|=0$, 
with $H$ the half plane of points to the right of $L_1(\gamma)$. 
Thus $\lim_n|\Gamma^n\cap H|=0$ which, $\Gamma^n$ being simple, implies that $\limsup_nw_1(\gamma^n)\leq w_1(\gamma)$.

Assume now that the first region of $\Gamma$ is not reduced to a point and suppose by contradiction that there is $\epsilon>0$ such that $w_1(\gamma^n)\geq w_1(\gamma)+\epsilon$ along some subsequence still denoted by $n$. 
We can assume that $\epsilon$ is small enough that there is a point $x_\epsilon$ in region $1$ of $\gamma$ with abscissa $w_1(\gamma)+\epsilon$. 
Let $\epsilon'(\gamma)>0$ be the drop in height between the pole $P_1(\gamma)$ and the point $x_\epsilon$:
\begin{equation}
x_\epsilon\cdot {\bf b}_{\pi/2}
=
z_1(\gamma) - \epsilon'(\gamma)
.
\end{equation}
The lower semi-continuity of $z'_1$ and the fact that $\gamma,\gamma^n$ are simple implies that $z_1(\gamma^n)\geq z_1(\gamma)-\epsilon'(\gamma)/2$ for all large enough $n$. 
There are now two cases to consider. 
Define the set (see Figure~\ref{fig_lsc_wprime}):
\begin{equation}
A 
:=
\Big\{(u,v)\in\R^2 : u\in\big[x_\epsilon\cdot{\bf b}_0, w_1(\Gamma^n)\big], v\in\big[x_\epsilon\cdot{\bf b}_{\pi/2},z_1(\Gamma^n)\big]\Big\}
.
\end{equation}
\begin{figure}
\begin{center}
\includegraphics[width=11cm]{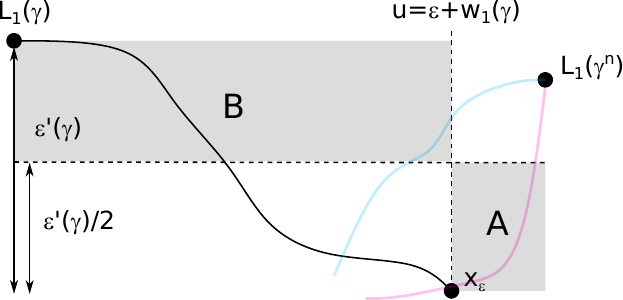} 
\caption{Representation of the sets $A,B$ (shaded areas) for a curve $\gamma$ (solid black line), with two possible $\gamma^n$ (cyan and magenta curves). The cyan curve corresponds to case 1 ($\liminf_n|\Gamma^n\cap A|>0$), 
the magenta curve to case 2
($\liminf_n|\Gamma^n\cap B|=0$, $|\Gamma\cap B|>0$). 
\label{fig_lsc_wprime}}
\end{center}
\end{figure}
Then either $\liminf_n|\Gamma^n\cap A|>0$, or $\liminf_n|\Gamma^n\cap A|=0$. 
Due to the the structure of curves in $\Omega$ (see Definition~\ref{def_state_space}), it always holds that $|\Gamma\cap A|=0$. 
In the first case (cyan line on Figure~\ref{fig_lsc_wprime}), 
this fact and the convergence $\lim_nd_{L^1}(\gamma,\gamma^n)=0$ give a contradiction. 
In the second case (magenta line on Figure~\ref{fig_lsc_wprime}), 
the structure of elements in $\Omega$ implies that, 
up to taking a subsequence, 
the height of a point in $\Gamma^n$ with abscissa $x_\epsilon\cdot{\bf b}_0$ (if such a point exists) is smaller than $x_\epsilon \cdot{\bf b}_{\pi/2}+\epsilon'(\gamma)/2$ for large enough $n$. 
Thus $|\Gamma^n\cap B|=0$ for large enough $n$, 
with:
\begin{equation}
B 
:=
\Big\{(u,v)\in\R^2 : u\in\big[w_1(\gamma), x_\epsilon\cdot{\bf b}_0, \big], v\in\big[z_1(\gamma)-\epsilon'(\gamma)/2, z_1(\gamma)\big]\Big\}
.
\end{equation}
However, the fact that $\Gamma$ is simple with first region not reduced to a point implies $|\Gamma\cap B|>0$. 
This gives a contradiction in this case as well and concludes the proof. 
\end{proof}
Convergence in volume is weaker than convergence in Hausdorff distance in general. 
For curves in $\Omega$, however, we now show that Hausdorff convergence is implied by convergence in volume, 
convergence of the $z_k$ and, depending on the curves, the $w_k$ as well ($1\leq k\leq 4$). 
We stress that convergence of the $w_k$ is not always implied by Hausdorff convergence, 
hence the complicated formula~\eqref{eq_def_d_tilde_H} below. 
We refer to Figure~\ref{fig_z_k_w_k} and to Remark~\ref{rmk_distance_tilde_hausdorff} for heuristics.
\begin{lemm}\label{lemm_lien_hausdorff_et_volume_et_poles}
Let $\iota : \R\rightarrow [0,1]$ be a strictly increasing continuous function such that $\iota(0) =0$. 
Consider the distance $\tilde d_{\mathcal H}$ on $\Omega$, defined for $\gamma,\tilde\gamma\in\Omega$ and associated droplets $\Gamma,\tilde\Gamma$ by:
\begin{align}
\tilde d_{\mathcal H}(\gamma,\tilde\gamma) &= d_{L^1}(\Gamma,\tilde\Gamma) + \sum_{k=1}^4\big|z_k(\Gamma)-z_k(\tilde\Gamma)\big| \nonumber\\
&\quad + \sum_{k=1}^4 \iota \Big(\max\big\{|z_k(\Gamma)-z_k(\Gamma')|,|z_k(\tilde\Gamma) - z_k(\tilde\Gamma')|\big\}\Big)\big|w_k(\Gamma)-w_k(\tilde\Gamma)\big|.\label{eq_def_d_tilde_H}
\end{align}
Then $\tilde d_{\mathcal H}$ and $d_{\mathcal H}$ are topologically equivalent on $\Omega$. 
\end{lemm}
\begin{figure}[H]
\begin{center}
\includegraphics[width=10cm]{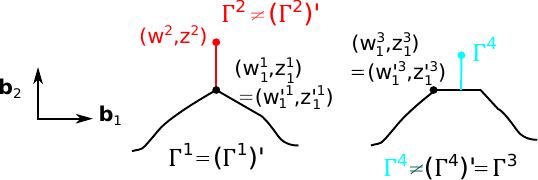} 
\caption{Neighbourhood of the north pole of four droplets $\Gamma^i$, $1\leq i \leq 4$ assumed to be identical except in the pictured portion. 
On the figure, $w'^i_1,z'^i_1$ is short for $w_1((\Gamma^i)'),z_1((\Gamma^i)')$ respectively and $w^i_1,z^i_1$ is short for $w_1(\Gamma^i),z_1(\Gamma^i)$. Recall that e.g. $(\Gamma^2)'$ is the largest droplet with simple boundary contained in $\Gamma^2$ (here $(\Gamma^2)'=\Gamma^1$, i.e. $(\Gamma^2)'$ is $\Gamma^2$ without the vertical red line), see also Figure~\ref{fig_gamma_prime_z_k_w_k}.\\
Convergence in volume to the $\Gamma^i$ ($1\leq i \leq 4$) ensures that the limit differs from $(\Gamma^i)'$ only by a vertical segment above the north pole of $(\Gamma^i)'$. 
To prove convergence in Hausdorff distance, in addition to the volume, one must therefore control convergence of the height $z_1$ of this segment and, depending on the limit curve, convergence of its lateral position $w_1$. \\
For $i\in\{1,3\}$, $\Gamma^i = (\Gamma^i)'$. 
Knowing convergence of $z_1$ is then enough, 
because it shows that this vertical segment is reduced to a point. \\
For $\Gamma^2$ (equal to $\Gamma_1$ except for the red part), $\Gamma^2\neq (\Gamma^2)'$ (or equivalently $z_1(\Gamma^2)> z_1((\Gamma^2)')$), 
so we in principle need to know convergence of $w_1$. 
However, since $(\Gamma^2)'$ has point-like pole, convergence of $w_1$ is already controlled by the volume convergence. 
Thus knowing convergence in volume and of $z_1$ is still enough.\\
In contrast, $\Gamma^4$ has point-like north pole while $(\Gamma^4)'=\Gamma^3\neq \Gamma^4$ satisfies $|P_1((\Gamma^4)')|>0$.    
Convergence in Hausdorff distance then requires convergence of the lateral position $w_1$ of the pole 
(in general only lower semi-continuous by Lemma~\ref{lemm_w_k}) 
in addition to convergence of $z_1$ and in volume. \label{fig_z_k_w_k}}
\end{center}
\end{figure}
\begin{rmk}\label{rmk_distance_tilde_hausdorff}
The statement of Lemma~\ref{lemm_lien_hausdorff_et_volume_et_poles} is important. Though already discussed at length in Figure~\ref{fig_z_k_w_k} from another point of view, let us therefore take the time to explain the definition~\eqref{eq_def_d_tilde_H} of $\tilde d_{\mathcal H}$. 

Note first that all microscopic curves, i.e. elements of $\Omega^N_{\text{mic}}$, are simple. 
For these curves, $\Gamma=\Gamma'$ and the second line of~\eqref{eq_def_d_tilde_H} vanishes. 
This second line is only useful to control singularities at the poles of elements of $\Omega$.

The volume distance does not control possible singularities at the poles of elements of $\Omega$, whereas the Hausdorff distance does. 
The question is then what kind of information must be added to the volume distance so as to guarantee convergence in Hausdorff distance.

If a curve $\gamma$ with associated droplet $\Gamma$ is simple in the neighbourhood of its pole $k$ ($1\leq k\leq 4$), 
then an additional control of the $z_k$ is sufficient to guarantee convergence in Hausdorff distance. 
This corresponds to $\Gamma^1$ and $\Gamma^3$ in Figure~\ref{fig_z_k_w_k}.  
Convergence of the $w_k$ is \emph{not} implied by convergence in Hausdorff distance in that case, as illustrated in the case of $\Gamma^3,\Gamma^4$. 
Note that in this case the function $\iota$ in the second line of~\eqref{eq_def_d_tilde_H} vanishes (this claim is established in the proof of Lemma~\ref{lemm_lien_hausdorff_et_volume_et_poles}), 
thus convergence for $\tilde d_{\mathcal H}$ only amounts to convergence in volume and of the $z_k$ but does not require convergence of the $w_k$ either.

Similarly, if $\gamma$ is not simple around pole $k$ but $\Gamma'$ has pole $k$ reduced to a point, 
then controlling the volume and $z_k$ is enough (as illustrated by $\Gamma^2$ in Figure~\ref{fig_z_k_w_k}). 
Again in this case the second line of~\eqref{eq_def_d_tilde_H} vanishes, 
because convergence in volume implies convergence of the $w_k$ to the only possible value $w_k(\Gamma)$. 

The last situation to consider is when a curve $\gamma$ has pole $k$ reduced to a point, but pole $k$ of $\Gamma'$ is not reduced to a point (this is the case for $\Gamma^4$ in Figure~\ref{fig_z_k_w_k}). 
In this case convergence in Hausdorff distance \emph{does} require convergence of $w_k$, 
which is why $w_k$ appears in the second line of~\eqref{eq_def_d_tilde_H} (and $\iota$ does not vanish in that situation).\demo
\end{rmk}
\begin{proof}[Proof of Lemma~\ref{lemm_lien_hausdorff_et_volume_et_poles}]
Let $(\gamma^n)_n\subset \Omega$ converge to $\gamma\in\Omega$ for $d_{\mathcal H}$, 
with as usual $(\Gamma^n)_n,\Gamma$ the associated droplets. 
Then $\lim_nd_{L^1}(\Gamma^n,\Gamma)=0$, 
and $(z_k(\gamma^n))_n$ converges to $z_k(\gamma) = z_k(\Gamma)$ for $1\leq k\leq 4$, 
as these objects are continuous in Hausdorff distance.

Suppose first $z_k(\Gamma)=z_k(\Gamma')$ for each $k$ ($1\leq k\leq 4$). 
Recall that $z_k(\Gamma^n)\geq z_k((\Gamma^n)')$ if $k\in\{1,2\}$, $z_k(\Gamma^n)\leq z_k((\Gamma^n)')$ if $k\in\{3,4\}$ by definition. 
Moreover, $\gamma\in \Omega\mapsto z_k(\Gamma')$ is lower semi-continuous for $k\in\{1,2\}$ and upper semi-continuous for $k\in\{3,4\}$ by Lemma~\ref{lemm_continuite_z_k'}. 
The convergence $\lim_nz_k(\Gamma^n) = z_k(\Gamma)$ then implies $\lim_n |z_k((\Gamma^n)')-z_k(\Gamma^n)| = 0$, 
thus $\iota$ vanishes in the second line of $\tilde d_{\mathcal H}(\gamma^n,\gamma)$, 
and $\lim_n\tilde d_{\mathcal H}(\gamma^n,\gamma)=0$.

If instead $z_k(\Gamma)>z_k(\Gamma')$ for some $1\leq k \leq 4$, 
then $\gamma$ has pole $k$ reduced to a point and $\iota$ is bounded from below at pole $k$ (recall that $\iota$ takes values in $[0,1]$). 
However, a point-like pole $k$ means that $w_k$ is a Hausdorff-continuous functional at $\gamma$ by~\eqref{eq_convergence_pole_at_fixed_time}, 
thus $\lim_nw_k(\gamma^n) = w_k(\gamma)$. 
It follows that $\lim_n\tilde d_{\mathcal H}(\gamma^n,\gamma)=0$.\\

Conversely, assume $(\gamma^n)_n\subset \Omega$ converges to $\gamma\in\Omega$ for $\tilde d_{\mathcal H}$. 
Convergence of the $(z_k(\gamma^n))_n$ ($1\leq k \leq 4$) implies convergence of the length $(|\gamma^n|)_n$, 
which is in particular bounded by some $C>0$. 
The set:
\begin{equation}
\Omega_C := \{\bar\gamma\in\Omega:|\bar\gamma|\leq C\}\label{eq_def_e_A}
\end{equation}
is compact for the Hausdorff distance by Proposition~\ref{prop_e_closed}. 
Let $\gamma^\infty$ denote a limit point of $(\gamma^n)_n$ for $d_{\mathcal H}$ and write $\Gamma,\Gamma^\infty$ for the droplets associated with $\gamma,\gamma^\infty$ respectively.  
By continuity of the $z_k$ and volume for both $d_{\mathcal H}$ and $\tilde d_{\mathcal H}$, 
one has:
\begin{equation}
d_{L^1}(\gamma ,\gamma^\infty) = 0\quad\text{(thus }\Gamma'=(\Gamma^\infty)'\text{)},\qquad 
z_k(\gamma)= z_k(\gamma^\infty),\quad 1\leq k \leq 4.\label{eq_lien_Gamma_Gamma_infty}
\end{equation}
If $z_k(\Gamma') = z_k(\Gamma)$ for each $1\leq k \leq 4$, 
then by~\eqref{eq_lien_Gamma_Gamma_infty} the same is true for each pole of $\Gamma^\infty$ 
and we can conclude $\Gamma = \Gamma^\infty$. 
The only possibility for $\Gamma,\Gamma^\infty$ to differ is thus when the situation represented by $\Gamma^3$ and $\Gamma^4$ in Figure~\ref{fig_z_k_w_k} occurs, 
i.e. when $z_k(\Gamma)>z_k(\Gamma')$ for some $1\leq k\leq 4$ and pole $k$ of $\Gamma'$ is not reduced to a point, 
so that $w_k(\gamma)$ may differ from $w_k(\gamma^\infty)$.

However, as soon as $z_k(\Gamma)=z_k(\Gamma^\infty)>z_k(\Gamma')=z_k((\Gamma^\infty)')$ for some $1\leq k\leq 4$, 
then $\gamma,\gamma^\infty$ have point-like pole $k$, 
thus $w_k$ is in particular continuous in Hausdorff distance at $\gamma^\infty$ by~\eqref{eq_convergence_pole_at_fixed_time}: 
$\lim_nw_k(\gamma^n)=w_k(\gamma^\infty)$ up to a subsequence. 
Since $z_k(\Gamma)>z_k(\Gamma')$ implies that $\iota$ is bounded from below at pole $k$, 
we also know that $\lim_nw_k(\gamma^n)=w_k(\gamma)$ by Definition~\eqref{eq_def_d_tilde_H} of $\tilde d_{\mathcal H}$. 
Thus $w_k(\gamma) = w_k(\gamma^\infty)$, 
and $\gamma=\gamma^\infty$: 
$\gamma$ is the limit of $(\gamma^n)$ in Hausdorff distance.
\end{proof}
We conclude the section by proving Lemma~\ref{lemm_local_jump_rates}, 
stated again as Lemma~\ref{lemm_stability_initial_condition_appendix}. 
Let $\gamma\in\Omega$, 
and define $q(\gamma)$ as the distance of $0$ to $\gamma$ in the following sense:
\begin{equation}
q(\gamma)
:=
\sup\Big\{ q>0 : B_1(0,q)\subset\Gamma\Big\} 
,
\label{eq_def_q_appB}
\end{equation}
with the subscript $1$ denoting the usual $\|\cdot\|_1$-norm on $\R^2$. 
Let also $r(\gamma)$ denote the distance between consecutive poles of the droplet $\Gamma'$ with $\partial\Gamma'\in\Omega'$ (see Lemma~\ref{lemm_continuite_z_k'}) associated with $\gamma$ (recall that $k+1:=1$ if $k=4$):
\begin{equation}
r(\gamma) 
:=
\inf\Big\{ r>0 : \forall 1\leq k \leq 4,
\quad \big|[R_k(\Gamma')-L_{k+1}(\Gamma')]\cdot{\bf b}_\theta\big|\geq r\ \text{ for } \theta\in\{0,\pi/2\}\Big\}
.
\label{eq_condition_sur_CI}
\end{equation}
Finally, let $r'(\gamma)$ denote the distance between opposite regions of $\partial\Gamma'$ in the following sense. 
Let $[L_k(\Gamma'),R_{k+1}(\Gamma')]_{\partial\Gamma'}$ denote region $k$ of $\partial\Gamma'$. 
Then, writing $d_1(A,B)=\inf_{x\in A,y\in B}\|x-y\|_1$ for the distance between bounded sets $A,B\subset\R^2$ in $1$-norm:
\begin{align}
r'(\gamma) 
= 
\sup\Big\{\epsilon>0: 
\ &d_1\Big([L_1(\Gamma'),R_{2}(\Gamma')]_{\partial\Gamma'}, [L_3(\Gamma'),R_{4}(\Gamma')]_{\partial\Gamma'}\Big)\geq \epsilon
,\quad 
\nonumber\\
&\ d_1\Big([L_2(\Gamma'),R_{3}(\Gamma')]_{\partial\Gamma'}, [L_4(\Gamma'),R_{1}(\Gamma')]_{\partial\Gamma'}\Big)\geq \epsilon
\Big\}
\label{eq_def_r'_appB}
.
\end{align}
\begin{lemm}\label{lemm_stability_initial_condition_appendix}
Let $\gamma\in\Omega$ be such that $q(\gamma)>0$, $r(\gamma)>0$ and $r'(\gamma)>0$ (i.e. $\gamma$ satisfies Property~\ref{prop_IC}). 
There is then $\epsilon>0$ such that:
\begin{itemize}
	\item all droplets $\tilde \Gamma$ associated with a curve $\tilde\gamma\in\Omega$ and such that $d_{L^1}(\gamma,\tilde \gamma)\leq \epsilon^2$ satisfy $q(\tilde\gamma)>q(\gamma)/2$, 
$r(\tilde\gamma)>r(\gamma)/2$ and $r'(\tilde\gamma)>r'(\gamma)/2$. 
	\item The jump rates $c(\tilde \gamma^N,\cdot)$ of any $\tilde \gamma^N\in\Omega^N_{\text{mic}}$ with $d_{L^1}(\gamma,\tilde\gamma^N)\leq \epsilon^2$ are local:
	there is $N(\gamma)\in\N_{\geq 1}$ such that, for any $N\geq N(\gamma)$ and any $x\in V(\tilde\gamma^N)$, 
	the value of $c(\tilde\gamma^N,(\tilde\gamma^N)^x)$ can be determined through the knowledge of points at $1$-distance at most $3$ from $x$. 
\end{itemize} 
\end{lemm}
\begin{proof}
The second item is a consequence of the first one 
since any $\gamma^N\in\Omega^N_{\text{mic}}$ satisfying Property~\ref{prop_IC} has local jump rates as discussed in Section~\ref{section_IC}. 
We therefore focus on the first item.

Notice that $q(\gamma),r(\gamma),r'(\gamma)$ are defined only in terms of volumes or the $\Gamma'$. 
This implies that $q,r,r'$ take the same value for $\gamma$ and for $\partial(\Gamma')$ which is a simple curve. 
It is therefore enough to work with simple curves exclusively, 
i.e. to assume that $\gamma$ is simple and prove the existence of $\epsilon>0$ such that any simple curve $\tilde\gamma\in \Omega$ with $d_{L^1}(\gamma,\tilde\gamma)\leq \epsilon^2$ satisfies $q(\tilde\gamma),r(\tilde\gamma),r'(\tilde\gamma)>0$. 
 
Consider first $r(\cdot)$. 
The continuity results of Lemma~\ref{lemm_continuite_z_k'} 
give that $\tilde\gamma\mapsto [L_2(\tilde\Gamma')-R_1(\tilde \Gamma')]\cdot{\bf b}_0$ and $\tilde \gamma\mapsto[R_1(\tilde\Gamma')-L_2(\tilde\Gamma')]\cdot{\bf b}_{\pi/2}$ are lower semi-continuous on the set of curves with $|\tilde\Gamma|>0$, 
thus in particular for simple $\tilde \gamma$. 
A similar statement is valid in other regions.  
This yields the claim for $r(\cdot)$.

Consider now $q(\cdot)$. 
Recall the definition $\Gamma^{(-\delta)} := \Gamma\setminus \bigcup_{x\notin\Gamma}B_1(x,\delta)$, 
with $B_1(x,\delta)$ the open ball in $1$-norm ($\delta\geq 0$).  
Let us show that, for any $\delta>0$, 
there is $\epsilon(\delta)>0$ such that $d_{L^1}(\tilde\gamma,\gamma)\leq \epsilon(\delta)$ for any simple curve $\tilde\gamma\in\Omega$ implies $\tilde\Gamma\supset\Gamma^{(-\delta)}$. 
This will imply $q(\tilde\gamma)\geq q(\gamma)/2$ for each $\tilde\gamma\in\Omega$ in a sufficiently small ball around $\gamma$ in volume distance. 
Let $\delta>0$ be such that $\Gamma^{(-\delta)}\neq\emptyset$; otherwise there is nothing to prove. 
For future reference, note that the ball $B_1(x,\delta)$ around a point $x$ can be split into four pieces corresponding to the closed sets with boundaries given by the triangles $(x,x+\delta{\bf b}_{(k-1)\pi/2},x+\delta{\bf b}_{k\pi/2})$ ($1\leq k\leq 4$).  
In the following we talk of an "open triangle in $B_1(x,\delta)$" to refer to the interior of one such piece.

Assume by contradiction that there is a sequence of simple curves $\tilde\gamma^n\in \Omega$ converging to $\gamma$ in volume distance and points $x_n\in \Gamma^{(-\delta)}\setminus\tilde \Gamma^n$ ($n\in\N$). 
As $\Gamma^{(-\delta)}$ is compact, we can assume up to taking a subsequence that $(x_n)$ converges to $x\in\Gamma^{(-\delta)}$. 
This implies:
\begin{equation}
|B_1(x,\delta)\cap\Gamma|
=
|B_1(x,\delta)|
.
\label{eq_volume_ball_proof_local_jump_rates}
\end{equation}
On the other hand, the structure of curves in $\Omega$ 
(recall Definition~\ref{def_state_space}) and the fact that $x_n\notin\tilde\Gamma^n$ imply that at least one open triangle in the ball $B_1(x_n,\delta)$ does not intersect $\tilde\Gamma^n$, thus:
\begin{equation}
|B_1(x_n,\delta)\cap\tilde \Gamma^n|
\leq 
\frac{3}{4}|B_1(x_n,\delta)|
=
\frac{3}{4}|B_1(x,\delta)|
.
\label{eq_volume_ball_proof_local_jump_rates2}
\end{equation}
Equations~\eqref{eq_volume_ball_proof_local_jump_rates}--\eqref{eq_volume_ball_proof_local_jump_rates2} and 
the fact that $\lim_n x_n=x$ contradict $\lim_nd_{L^1}(\gamma,\tilde\gamma^n)=0$. 

Consider finally $r'(\cdot)$. 
Again we proceed by contradiction. 
Let $\delta\in(0,r'(\gamma))$ and assume that there are simple curves $\tilde\gamma^n\in\Omega$ and points $x_n,y_n\in\partial\tilde\Gamma^n$ ($n\in\N$), 
say in regions $1$ and $3$ respectively for definiteness, 
such that $\lim_nd_{L^1}(\gamma,\tilde\gamma^n)=0$ and 
$\|x_n-y_n\|_1\leq \delta$. 
Let us prove that $(x_n)$, $(y_n)$ are bounded using the bound on the volume of $\tilde\Gamma^n$. 
By definition of the poles one always has:
\begin{equation}
x_n\cdot{\bf b}_0 \in[ w_1(\tilde \gamma^n), z_2(\tilde\gamma^n)]
,\quad
x_n\cdot{\bf b}_{\pi/2} \in[ R_2(\tilde \gamma^n)\cdot{\bf b}_{\pi/2}, z_1(\tilde\gamma^n)]
.
\end{equation}
Similar bounds hold for $y_n$. 
The semi-continuity results of Lemma~\ref{lemm_continuite_z_k'} for $w_1$, $R_2\cdot{\bf b}_{\pi/2}$ then imply that, for each $\epsilon>0$, 
there is $n(\epsilon)\geq 0$ such that, if $n\geq n(\epsilon)$:
\begin{align}
x_n\cdot {\bf b}_0
&\geq 
w_1(\gamma)-\epsilon
,\qquad
x_n\cdot {\bf b}_{\pi/2}
\geq 
R_2(\gamma)\cdot {\bf b}_{\pi/2}-\epsilon,
\nonumber\\
y_n\cdot {\bf b}_0
&\leq 
w_3(\gamma)+\epsilon
,\qquad
y_n\cdot {\bf b}_{\pi/2}
\leq 
R_4(\gamma)\cdot {\bf b}_{\pi/2}+\epsilon
.
\label{eq_lsc_region}
\end{align}
The fact that $\|x_n-y_n\|_1\leq \delta$ thus implies that both $(x_n),(y_n)$ are bounded. 
Up to taking subsequences, let us assume they converge to $x$, $y$ respectively with $\|x-y\|_1\leq \delta$. 
Since $x_n$, $y_n$ are in $\tilde\gamma^n$ and due to the structure of curves in $\Omega$, at most three open triangles out of four of any ball around $x_n$, $y_n$ can be included in $\tilde\Gamma^n$:
\begin{equation}
|B_1(x_n,\zeta)\cap\tilde\Gamma^n|
\leq 
\frac{3}{4}|B_1(x_n,\zeta)|
,\qquad
|B_1(y_n,\zeta)\cap\tilde\Gamma^n|
\leq 
\frac{3}{4}|B_1(y_n,\zeta)|
,\qquad
\zeta>0
.
\end{equation}
Taking limits, this implies that $x$, $y$ are not in the interior of $\Gamma$. 
We can nonetheless locate $x,y$ relative to $\Gamma$: by the proof for $q(\cdot)$, $\Gamma^{(-\epsilon)}\subset\tilde\Gamma^n$ for each $\epsilon>0$ and each $n$ larger than some $n'(\epsilon)\in\N$. 
This and~\eqref{eq_lsc_region} tell us that $x$ is above and to the right of the first region of $\gamma$:
\begin{equation}
\forall z\in[L_1(\gamma),R_2(\gamma)]_\gamma
,\qquad
[x-z]\cdot{\bf b}_{\pi/4}
\geq 
0
.
\end{equation}
Similarly $y$ is below and to the left of the third region of $\gamma$, i.e. $[y-z]\cdot{\bf b}_{5\pi/4}\geq 0$ for any $z$ in the third region of $\gamma$.  
This yields a contradiction, concluding the proof:
\begin{equation}
r'(\gamma)
>
\delta
\geq 
\|x-y\|_1
\geq 
d_1\Big([L_1(\gamma),R_2(\gamma)]_\gamma,[L_3(\gamma),R_4(\gamma)]_\gamma\Big)
\geq
r'(\gamma)
.
\end{equation}
\end{proof}
\subsection{The set $E([0,T],\e)$}\label{sec_the_set_E(0,T0)}
For $T>0$, the set $E([0,T],\e)$ was defined in~\eqref{eq_def_E_0_T_Omega} as follows:
\begin{equation}
E([0,T],\e) := D_{L^1}([0,T],\e)\cap\Big\{\gamma_\cdot : \int_{0}^T|\gamma_t|\, dt<\infty\Big\}.
\end{equation}
This set is equipped with the distance $d_E$, defined  by: 
\begin{equation}
d_E := d^S_{L_1} + \int_0^{T}d_{\mathcal H}\, dt ,\label{eq_def_d_E_appen}
\end{equation}
with $d^S_{L^1}$ the Skorokhod distance associated with convergence in the topology induced by the volume distance $d_{L^1}$ (see~\eqref{eq_convention_volume_appendice}). For properties of the Skorokhod topology, we refer the reader to Chapter 3 of \cite{Ethier2009}.

The main purpose of this section is to characterise relatively compact subsets of $E([0,T],\e)$, 
in Appendix~\ref{app_some_compact_sets}, 
to pave the way for the proof of tightness in Section~\ref{appen_tightness}. 
We also prove Proposition~\ref{prop_preuve_continuite_integrale_avec_v_moins_un} on the continuity of the functionals appearing in the proof of large deviations, in Section~\ref{sec_continuity_J_H_eps_ell_H_eps}.

To characterise compactness, we first exhibit a distance topologically equivalent to $d_E$ on $E([0,T],\e)$, but with explicit dependence on the distance between poles. 
This result is the analogue on trajectories of Lemma~\ref{lemm_lien_hausdorff_et_volume_et_poles}.
\begin{lemm}\label{lemm_lien_hausdorff_et_volume_et_poles_integral}
Let $\iota:\R\rightarrow[0,1]$ be a strictly increasing continuous function with $\iota(0)=0$. Let $\tilde d_E$ be the distance on $E([0,T],\e)$ defined for two trajectories $\gamma_\cdot,\tilde\gamma_\cdot$ with associated droplets $\Gamma_\cdot,\tilde\Gamma_\cdot$ by:
\begin{align}
\tilde d_E(\gamma_\cdot,\tilde\gamma_\cdot) 
&= 
d^S_{L_1}(\gamma_\cdot,\tilde\gamma_\cdot) + \sum_{k=1}^4\int_0^{T}\big|z_k(\Gamma_t)-z_k(\tilde\Gamma_t)\big|\, dt\nonumber\\
&\quad+ 
\sum_{k=1}^4 \int_0^T \iota\Big(\max\big\{|z_k(\Gamma_t)-z_k(\Gamma'_t)|,|z_k(\tilde \Gamma_t)-z_k(\tilde \Gamma'_t)|\big\}\Big)\big|w_k(\Gamma_t)-w_k(\tilde\Gamma_t)\big|\, dt
.
\label{eq_def_tilde_d_E}
\end{align}
Then $\tilde d_E$ and $d_E$ are topologically equivalent on $E([0,T],\e)$.
\end{lemm}
The proof of Lemma~\ref{lemm_lien_hausdorff_et_volume_et_poles_integral} is obtained as a consequence of the following lemma.
\begin{lemm}\label{lemm_continuite_implies_continuity_integral}
Let $F:\e\rightarrow\R$ be a continuous functional in Hausdorff distance (or, equivalently, for $\tilde d_H$) and assume:
\begin{equation}
\exists C,C'>0,\qquad |F(\gamma)|\leq C|\gamma|+C', \qquad\gamma\in\e.\label{eq_assumption_F}
\end{equation}
 Then $\gamma_\cdot \mapsto \int_0^T F(\gamma_t)\, dt$ is a continuous functional on $E([0,T],\e)$ for both the distances $\int_0^Td_{\mathcal{H}}\, dt$ and $\int_0^T \tilde d_{\mathcal H}\, dt$. 
 The conclusion of the lemma remains valid if $F$ is replaced by a function $G:[0,T]\times \e\rightarrow\R$ such that $G(t,\cdot)$ is continuous, $|G(t,\gamma)|\leq C|\gamma|+C'_t$ with $C,C'_t>0$ and $\int_0^TC'_t\, dt<\infty$.
\end{lemm}
\begin{proof}
Let $F:\e\rightarrow\R$ be a continuous functional and let $(\gamma^n_\cdot)_n\subset E([0,T],\e)$ converge to $\gamma_\cdot\in E([0,T],\e)$ for $\int_0^T d_{\mathcal H}\, dt$ (the proof is identical for $\tilde d_{\mathcal H}$). 
The key argument consists in proving that one may work with $F(\gamma^n_\cdot),F(\gamma_\cdot)$ bounded. One can then use the continuity of $F$ on $\e$ and a compactness argument to conclude.\\
We will use the following elementary identity: for any $c,X\geq0$,
\begin{equation}
X {\bf 1}_{X>c}
= c{\bf 1}_{X>c} + \int_c^\infty {\bf 1}\Big\{X>\lambda\Big\}\, d\lambda
. 
\label{eq_point_as_integral}
\end{equation}
Let $A>0$ and let us prove that it is enough to consider curves with length bounded in terms of $A$. The last equation applied with $c=0$ and $X = |F(\gamma^n_t) - F(\gamma_t)|$ for each $t\leq T$ yields:
\begin{align}
\int_0^T\big|F(\gamma^n_t)-F(\gamma_t)\big|\, dt &= \int_0^T\int_0^A{\bf 1}\Big\{\big|F(\gamma^n_t)-F(\gamma_t)\big|>\lambda\Big\}\, d\lambda \, dt \nonumber \\
&\quad +\int_0^T\int_A^\infty{\bf 1}\Big\{\big|F(\gamma^n_t)-F(\gamma_t)\big|>\lambda\Big\}\, d\lambda \, dt.\label{eq_decomp_F_appendix}
\end{align}
Assume $A\geq 4C'$. By Assumption~\eqref{eq_assumption_F} on $F$ and using~\eqref{eq_point_as_integral} in the last line below, the integral on the second line of~\eqref{eq_decomp_F_appendix} can be bounded as follows:
\begin{align}
\int_0^T\int_{A}^\infty {\bf 1}\Big\{\big|F(\gamma^n_t)-F(\gamma_t)&\big|>\lambda\Big\}\, d\lambda \, dt 
\leq 
\int_0^T\int_A^\infty {\bf 1}\Big\{ 2C\max\{|\gamma_t|,|\gamma^n_t|\}+2C'>\lambda\Big\}\, d\lambda \, dt\nonumber\\
&\leq 
\int_0^T\int_A^\infty {\bf 1}\Big\{ \max\{|\gamma_t|,|\gamma^n_t|\}>\frac{\lambda}{4C}\Big\}{\bf 1}\Big\{\big||\gamma^n_t| - |\gamma_t|\big|\leq \frac{\lambda}{8C}\Big\}\, d\lambda \, dt \nonumber\\ 
&\quad + 
\int_0^T\int_A^\infty {\bf 1}\Big\{\big||\gamma^n_t| - |\gamma_t|\big|>\frac{\lambda}{8C}\Big\}\, d\lambda \, dt\nonumber \\
&\leq \int_0^T\int_A^\infty {\bf 1}\Big\{ |\gamma_t|>\frac{\lambda}{8C}\Big\}\, d\lambda \, dt
+ 
8C\int_0^T\big||\gamma^n_t| - |\gamma_t|\big|\, dt.\label{eq_decomp_F_appendix_bis}
\end{align}
Since the length is Lipschitz in Hausdorff distance, one has $\lim_n\int_0^T\big||\gamma^n_t|-|\gamma_t|\big|\, dt =0$. In addition,~\eqref{eq_point_as_integral} implies that the first integral in~\eqref{eq_decomp_F_appendix_bis} is bounded as follows:
\begin{equation}
\int_0^T\int_A^\infty {\bf 1}\Big\{ |\gamma_t|>\frac{\lambda}{8C}\Big\}\, d\lambda \, dt \leq  8C\int_0^T|\gamma_t|{\bf 1}\Big\{|\gamma_t|>\frac{A}{8C}\Big\}\, dt \underset{A\rightarrow\infty}{\longrightarrow} 0.
\end{equation}
To prove that $\lim_n\int_0^T|F(\gamma^n_t)-F(\gamma_t)|\, dt=0$, it is therefore enough to prove:
\begin{equation}
\forall A,\epsilon>0,\qquad \lim_{n\rightarrow\infty}\int_0^T {\bf 1}\Big\{|F(\gamma^n_t)-F(\gamma_t)|>\epsilon \Big\}{\bf 1}\Big\{\max\{|\gamma_t|,|\gamma^n_t|\}\leq A\Big\}\, dt 
=
0
.\label{eq_continuity_F_interm}
\end{equation}
For any time $t\leq T$ for which the integrand above does not vanish, both $\gamma_t$ and $\gamma^n_t$ belong to the set $\e_A := \e \cap \big\{|\gamma|\leq A\big\}$. This set is compact for the Hausdorff distance by Proposition~\ref{prop_e_closed}. As $F$ is continuous on $\e_A$, thus uniformly continuous, there is a modulus of uniform continuity $m_A>0$ such that:
\begin{equation}
\forall \gamma,\tilde\gamma\in \e_A,\qquad |F(\gamma)-F(\tilde\gamma)|> \epsilon \quad \Rightarrow\quad d_{\mathcal H}(\gamma,\tilde \gamma)> m_A(\epsilon).
\end{equation}
The fact that $\lim_n\int_0^Td_{\mathcal H}(\gamma^n_t,\gamma_t)\, dt = 0$ concludes the proof of~\eqref{eq_continuity_F_interm}, thus of the first claim of the lemma.\\

Consider now $G:[0,T]\times \e\rightarrow\R$ bounded as in the lemma. Then the argument is the same, except that the decomposition~\eqref{eq_decomp_F_appendix} now reads:
\begin{align}
\int_0^T \big|G_t(\gamma^n_t)-G_t(\gamma_t)\big|\, dt 
&\leq 
\int_0^T\int_0^A{\bf 1}\Big\{\big| G_t(\gamma^n_t)- G_t(\gamma_t)\big|>\lambda\Big\}\, d\lambda \, dt \nonumber \\
&\quad +\int_0^T\int_A^\infty{\bf 1}\Big\{\big|G_t(\gamma^n_t)-G_t(\gamma_t)\big|>\lambda\Big\}{\bf 1}\Big\{ C'_t\leq\frac{\lambda}{4}\Big\}\, d\lambda \, dt\nonumber\\
&\quad +\int_0^T\int_A^\infty{\bf 1}\Big\{ C'_t>\frac{\lambda}{4}\Big\}\, d\lambda \, dt
\end{align}
The last integral vanishes when $A$ is large and is independent of $n$. The middle integral is treated as for $F$ in~\eqref{eq_decomp_F_appendix_bis}.
\end{proof}
\begin{proof}[Proof of Lemma~\ref{lemm_lien_hausdorff_et_volume_et_poles_integral}]
Notice that $\tilde d_E$ is topologically equivalent to $d_{L^1}^S + \int_0^T\tilde d_{\mathcal H}\, dt$. To prove Lemma~\ref{lemm_lien_hausdorff_et_volume_et_poles_integral}, it is thus enough to prove that $\int_0^T\tilde d_{\mathcal H}\, dt$ and $\int_0^Td_{\mathcal H}\, dt$ are topologically equivalent on $E([0,T],\e)$. By Lemma~\ref{lemm_lien_hausdorff_et_volume_et_poles}, 
$d_{\mathcal H}$ is a continuous functional for $\tilde d_{\mathcal H}$ and vice-versa. 
It is thus enough to check that both distances satisfy the hypotheses of Lemma~\ref{lemm_continuite_implies_continuity_integral}. Let $\gamma_\cdot \in E([0,T],\e)$. Then:
\begin{equation}
\forall \gamma\in\e,\qquad F_t(\gamma):= d_{\mathcal H}(\gamma,\gamma_t)\leq |\gamma|+|\gamma_t|,
\end{equation}
where the last bound comes from the fact that both $\gamma,\gamma_t$ surround the point $0$ and the definition~\eqref{eq_def_hausdorff_distance} of $d_{\mathcal H}$. As $\int_0^T|\gamma_t|\, dt<\infty$ and each $F_t$ is continuous on $\e$, the lemma applies: $\int_0^TF_t\,  dt$ is continuous on $E([0,T],\e)$, in particular at $\gamma_\cdot$. Thus convergence for $\int_0^T\tilde d_{\mathcal H}\, dt$ implies convergence for $\int_0^T d_{\mathcal H}\, dt$.\\

To conversely prove that convergence for $\int_0^T d_{\mathcal H}\, dt$ implies convergence for $\int_0^T \tilde d_{\mathcal H}\, dt$, let $t\in[0,T]$ and define, for each $\gamma\in\e$:
\begin{equation}
G_t(\gamma):=  d_{L^1}(\gamma,\gamma_t)+ \sum_{k=1}^4\|L_k(\gamma)-L_k(\gamma_t)\|_1\leq C(|\gamma|+|\gamma_t|)+C',
\end{equation}
where the $C'$ corresponds to the bound $\sup_{\gamma,\tilde \gamma\in \e}d_{L^1}(\gamma,\tilde\gamma)\leq C'$ and the $C$ comes from (see below~\eqref{eq_def_z_k_appB} and Lemma~\ref{lemm_w_k}): 
\begin{equation}
\forall 1\leq k\leq 4,\qquad |z_k|\leq |\gamma|,\qquad w_k\in[z_{k-1},z_{k+1}]\quad\text{with}\quad z_{k-1}\leq 0\leq z_{k+1},
\end{equation}
with the convention $z_{k-1} = z_4$ if $k=1$, $z_{k+1} = z_1$ if $k=4$. 
Then $G_t$ satisfies the hypothesis of Lemma~\ref{lemm_continuite_implies_continuity_integral} and $G_t(\gamma)\geq \tilde d_{\mathcal H}(\gamma,\gamma_t)$, hence the claim.
\end{proof}
\subsubsection{Compact sets in $E([0,T],\e)$}\label{app_some_compact_sets}
Thanks to Lemma~\ref{lemm_lien_hausdorff_et_volume_et_poles_integral}, we now have a distance topologically equivalent to $d_E$ that involves only the volume distance and the distance between the poles. This is the key to the following proposition, giving sufficient condition for compactness for $d_E$.
\begin{prop}[Compact sets for $d_E$]\label{prop_compact_sets_for_d_E}
Suppose that $\mathcal K\subset E([0,T],\e)$ satisfies the following.
\begin{itemize}
	\item One has:
\begin{equation}
\sup_{\gamma_\cdot\in \mathcal K}\sup_{t\leq T}|\gamma_t|<\infty.\label{eq_length_bound_theo_compactness}
\end{equation}
 \item If $m^{L^1}_\cdot(\Gamma_\cdot)$ is the Skorokhod modulus of continuity associated with volume convergence for trajectories in $E([0,T],\e)$ (see \cite[Equation (12.6)]{Billingsley1999}), then:
\begin{equation}
\lim_{\eta\rightarrow 0}\sup_{h\leq\eta}\sup_{\gamma_\cdot\in \mathcal K}\Big\{m_h^{L^1}(\gamma_\cdot) + \sum_{k=1}^4\int_0^{T-h}\Big[|z_k(\gamma^n_{t+h})-z_k(\gamma^n_t)| + |w_k(\gamma^n_{t+h})-w_k(\gamma^n_t)|\Big]\, dt \Big\} = 0.\label{eq_condition_compacite_pour_d_E}
\end{equation}
\end{itemize}
Then $\mathcal K$ is relatively compact for the topology induced by $d_E$.
\end{prop}
\begin{proof}
Recall the definition~\eqref{eq_def_d_E_appen} of $d_E$ and let $(\gamma^n_\cdot)_n\subset E([0,T],\e)$ be a sequence in $\mathcal K$. 
By Lemma~\ref{lemm_lien_hausdorff_et_volume_et_poles_integral}, it is enough to prove that a subsequence of $(\gamma^N_\cdot)$ converges in $\tilde d_E$ distance. 
This means that it is enough to control the convergence in volume and of the time integral of the distance between the poles. 

According to the characterisation of relatively compact sets in the Skorokhod topology in \cite[Theorem 6.3]{Ethier2009}, $\mathcal K$ is relatively compact in the Skorokhod space $D_{L^1}([0,T],\Omega'\cap\e)$ equipped with $d_{L^1}^S$, where $\Omega'$ is the subset of simple curves in $\Omega$ given in~\eqref{eq_def_eprime}. 
It follows that, up to a subsequence, $(\Gamma^n_\cdot)_n$ converges in $d^S_{L^1}$-distance to a trajectory $\tilde \Gamma^\infty_\cdot\in D_{L^1}([0,T],\Omega'\cap\e)$.

Let us now control the convergence of the $z_k,w_k$ for each $k$ with $1\leq k \leq 4$. Recall that the length of a curve $\gamma\in\e$ in $1$-distance satisfies:
\begin{equation}
|\gamma| = 2\big[z_1(\gamma)-z_3(\gamma)\big]+2\big[z_2(\gamma)-z_4(\gamma)\big].\label{eq_length_as_diff_z_k_appendix}
\end{equation}
Recall also that the droplet associated with $\gamma\in\e$ must contain the point $0$, thus $z_1(\gamma),z_2(\gamma)\geq 0$ and $z_3(\gamma),z_4(\gamma)\leq 0$. This observation,~\eqref{eq_length_as_diff_z_k_appendix} and the bound~\eqref{eq_length_bound_theo_compactness} thus translate into:
\begin{equation}
\max_{1\leq k \leq 4}\sup_{n}\sup_{t\leq T}|z_k(\gamma^n_t)|<\infty.\label{eq_uniform_bound_z_k}
\end{equation}
Similarly, $\sup_{n,t}|w_k(\gamma^n_t)|<\infty$ as $w_k \in[z_{k-1},z_{k+1}]$ by Lemma~\ref{lemm_w_k} 
(with $k+1 :=1$ if $k=4$, $k-1:= 4$ if $k=1$). 
Equation~\eqref{eq_condition_compacite_pour_d_E} and the Kolmogorov-Riesz theorem~\cite[Theorem 4.26]{Brezis2010} imply that the sets $\{(z_k(\gamma^n_\cdot))_n\}_n$, $\{(w_k(\gamma^n_\cdot))_n\}_n$ are relatively compact subsets of $L^1([0,T],\R)$ for each $k$. 
Up to a subsequence, they thus converge to $z_k^\infty,w_k^\infty\in L^1([0,T],\R)$ respectively.

It remains to build a limit point $\gamma^\infty_\cdot$ of $(\gamma^n_\cdot)_n$ for $d_E$. 
Define, for each $t\leq T$, the curve $\gamma^\infty_t$ as the boundary of the droplet $\Gamma^\infty_t$, with:
\begin{equation}
(\Gamma^\infty_t)' := \tilde\Gamma^\infty_t,\quad z_k(\Gamma^\infty_t) := z_k^\infty(t)
\quad\text{ for }1\leq k\leq 4,
\end{equation}
and:
\begin{equation}
w_k(\Gamma^\infty_t) := w^\infty_k(t){\bf 1}\Big\{z_k^\infty(t)>z_k(\tilde\Gamma^\infty_t)\Big\} + w_k(\tilde \Gamma^\infty_t){\bf 1}\Big\{z_k^\infty(t)=z_k(\tilde\Gamma^\infty_t)\Big\}
\quad\text{ for }1\leq k\leq 4.
\end{equation}
Then $\gamma^\infty_\cdot\in E([0,T],\e)$ and $(\gamma^n_\cdot)$ converges to $\gamma^\infty_\cdot$ for $\tilde d_E$ up to a subsequence by construction. 
The fact that $d_E$ and $\tilde d_E$ have the same converging sequences by Lemma~\ref{lemm_lien_hausdorff_et_volume_et_poles_integral} concludes the proof.
\end{proof}
\subsubsection{Continuity properties of the functionals $J_{H,\epsilon}^\beta$}\label{sec_continuity_J_H_eps_ell_H_eps}
Let $\beta>\log 2$ and $H\in\C$ be fixed. 
In this section, we prove Proposition~\ref{prop_preuve_continuite_integrale_avec_v_moins_un} on the regularity of the functionals $J_{H,\epsilon}^\beta$ for $\epsilon>0$. 
Let $\epsilon>0$. In view of the expression~\eqref{eq_def_J_H_epsilon_zeta}--\eqref{eq_def_ell_H_epsilon_zeta} of $J_{H,\epsilon}^\beta$, we need to prove two things. 

The first is that elements of the set $E_{pp}([0,T],\e)$ of trajectories with almost always point-like poles are point of continuity for the distance $d_E$ of the following functionals, defined for $\gamma_\cdot\in E([0,T],\e)$ by:
\begin{align}
&\bigg(\frac{1}{4}-\frac{e^{-\beta}}{2}\bigg)\sum_{k=1}^4\int_0^{T} \big[H(t,R_k(\gamma_t))+H(t,L_k(\gamma_t))\big]\, dt,
\quad \text{and:}
\label{eq_terme_exp_minus_beta_pour_lsc}\\
&-\int_0^{T}\int_{\gamma_t(\epsilon)} \frac{(\mathsf{v}^\epsilon)^2}{4\mathsf{v}} \big[{\bf T}^\epsilon \cdot {\bf m}\big] {\bf T}^\epsilon\cdot \nabla H_t\, ds\, dt
-\frac{1}{2}\int_0^{T}\int_{\gamma_t(\epsilon)} \frac{(\mathsf{v}^\epsilon)^2}{\mathsf{v}}|{\bf T}_1^\epsilon {\bf T}_2^\epsilon| H_t^2 \,ds \,dt
.
\label{eq_terme_IPP_pour_lsc}
\end{align}
Recall that ${\bf m}$ is given by Definition~\ref{def_m(gamma)_gamma(eta)}, 
$\gamma_t(\epsilon)$ is the subset of $\gamma_t$ of points at $1$-distance at least $\epsilon$ from the poles and $\mathsf{v}^\epsilon,{\bf T}^\epsilon$ are defined in~\eqref{eq_def_cont_versions_T_v}.

The second thing is the convergence $\lim_{\epsilon\rightarrow 0}J_{H,\epsilon}^\beta(\gamma_\cdot) = J^\beta_{H}(\gamma_\cdot)$ for each $\gamma_\cdot\in E([0,T],\e)$, which amounts to the convergence of the functional~\eqref{eq_terme_IPP_pour_lsc} when $\epsilon\rightarrow0$. \\

Let us start by proving regularity of \eqref{eq_terme_exp_minus_beta_pour_lsc}--\eqref{eq_terme_IPP_pour_lsc}. The regularity of~\eqref{eq_terme_exp_minus_beta_pour_lsc} is the object of the following lemma.
\begin{lemm}[Convergence of the poles]\label{lemm_elementary_cv_cont_fction_hausdorff}
For $n\in\N$, let $\gamma^n_\cdot\in E([0,T],\e)$ and assume that $(\gamma^n_\cdot)$ converges to $\gamma_\cdot \in E_{pp}([0,T],\e)$ for $d_E$. Then:
\begin{equation}
\forall k\in\{1,...,4\},\qquad \lim_{n\rightarrow\infty}\int_0^{T}\Big[\|L_k(\gamma^n_t)-L_k(\gamma_t)\|_1 + \|R_k(\gamma^n_t)-L_k(\gamma_t)\|_1\Big]dt = 0.
\end{equation}
\end{lemm}
\begin{proof}
By Lemma~\ref{lemm_lien_hausdorff_et_volume_et_poles}, convergence for $d_E$ implies convergence for $\tilde d_E$, defined in~\eqref{eq_def_tilde_d_E}. As the limiting trajectory has almost always point-like poles, convergence for $\tilde d_E$ implies convergence of both coordinates of $L_k,R_k$ in the topology of $L^1([0,T],\R)$, hence the result.
\end{proof}
Let us now show that the functional in~\eqref{eq_terme_IPP_pour_lsc} is continuous at each point of the set $\e_{pp}([0,T],\e)$ of trajectories with almost point-like poles. 
The proof is quite technical, but the idea is simple: 
first, control the convergence at the poles using Lemma~\ref{lemm_elementary_cv_cont_fction_hausdorff}. Then, express line integrals in each regions as integrals on the corresponding SSEP by the correspondence presented in Section~\ref{sec_heuristics}. At this point the desired regularity properties are proven as for the SSEP, see \cite[Chapter 10]{Kipnis1999}. 
 
To prove the continuity of~\eqref{eq_terme_exp_minus_beta_pour_lsc}--\eqref{eq_terme_IPP_pour_lsc}, 
we will use Lemma~\ref{lemm_lien_hausdorff_et_volume_et_poles} relating convergence in Hausdorff distance and convergence of the poles and volume, 
or more precisely the following consequence of Lemma~\ref{lemm_lien_hausdorff_et_volume_et_poles}.
\begin{lemm}\label{lemm_volume_and_height_pole_implies_hausdorff}
Let $(\gamma^n)_n\subset \e$ and $\gamma\in\e$. Let $x\in\gamma$ be away from the poles in the sense that, for some $\zeta>0$, the set $B' := B_1(x,\zeta)\cap\gamma$ only contains points in the same region of $\gamma$ as $x$. Then convergence in volume implies convergence in Hausdorff distance:
\begin{equation}
\lim_{n\rightarrow\infty}d_{L^1}\big(\Gamma^n\cap B',\Gamma\cap B'\big)=0\quad \Rightarrow\quad \lim_{n\rightarrow\infty}d_{\mathcal H}\big(\Gamma^n\cap B',\Gamma\cap B'\big)=0.
\end{equation}
\end{lemm}
Fix $\gamma_\cdot = (\gamma_t)_{t\leq T}\in E_{pp}([0,T],\e)$ and let $(\gamma^n_\cdot)_n\subset E([0,T],\e)$ be a sequence converging to $\gamma_\cdot$ for $d_E$. 
Introduce the functionals $F_{H_t,\epsilon},\tilde F_{H_t,\epsilon}$ on $\e$ as follows: for $\gamma\in \e$,
\begin{align}
F_{H_t,\epsilon}(\gamma) 
= 
\int_{\gamma(\epsilon)} \frac{(\mathsf{v}^\epsilon)^2}{4\mathsf{v}} \big[{\bf T}^\epsilon \cdot {\bf m}\big] {\bf T}^\epsilon\cdot \nabla H_t \, ds,
\quad 
\tilde F_{H_t,\epsilon}(\gamma) 
= 
\int_{\gamma(\epsilon)} \frac{(\mathsf{v}^\epsilon)^2}{2\mathsf{v}}|{\bf T}_1^\epsilon {\bf T}_2^\epsilon| H_t^2 \, ds
.
\label{eq_def_F_H_epsilon}
\end{align}
To prove the continuity of the functional in~\eqref{eq_terme_IPP_pour_lsc} at $\gamma_\cdot$, we need to show:
\begin{equation}
\lim_{n\rightarrow\infty}\int_0^T F_{H_t,\epsilon}(\gamma_t^n)\, dt 
= 
\int_0^T F_{H_t,\epsilon}(\gamma_t)\, dt,
\qquad 
\lim_{n\rightarrow\infty}\int_0^T \tilde F_{H_t,\epsilon}(\gamma_t^n)\, dt 
= 
\int_0^T \tilde F_{H_t,\epsilon}(\gamma_t)\, dt
.
\label{eq_conv_F_H_t}
\end{equation}
We only deal with $F_{H_\cdot,\epsilon}$, $\tilde F_{H_\cdot,\epsilon}$ being similar. 
Note that there is $C(H)>0$ such that:
\begin{equation}
\forall \gamma\in\e,\forall t\leq T,\qquad \big|F_{H_t,\epsilon}(\gamma)\big|\leq C(H)|\gamma|.
\end{equation}
The functional $F_{H_t,\epsilon}$ is not continuous on $\e$ at each $t\in[0,T]$ because of the lack of continuity at the poles, 
so we cannot directly conclude through Lemma~\ref{lemm_continuite_implies_continuity_integral}. 
However, each curve with point-like pole is a point of continuity of $F_{H_t,\epsilon}$ for each $t$. 
The idea will be to use this fact and the above control of the length to restrict to times where $\gamma_t$ has point-like pole, then mimic the argument of Lemma~\ref{lemm_continuite_implies_continuity_integral} for those times.

By similar arguments as in the proof of Lemma~\ref{lemm_continuite_implies_continuity_integral}, 
the fact that $\lim_n\int_0^T\big||\gamma^n_t|-|\gamma_t|\big|\, dt=0$ implies that it is enough to prove:
\begin{equation}
\forall A>0,\qquad 
\lim_{n\rightarrow\infty}\int_0^T {\bf 1}\Big\{\big|F_{H_t,\epsilon}(\gamma_t^n)- F_{H_t,\epsilon}(\gamma_t)\big|\leq A\Big\}
\big|F_{H_t,\epsilon}(\gamma_t^n)- F_{H_t,\epsilon}(\gamma_t)\big|\, dt 
= 0
.
\end{equation}
This allows us to only prove convergence on a subset of times with length $T(1-o_n(1))$, 
which we chose as those times where poles of $(\gamma^n_\cdot)_n$ converge to those of $\gamma_\cdot$. 
In view of Lemma~\ref{lemm_elementary_cv_cont_fction_hausdorff}, it is enough to prove:
\begin{align}
\lim_{\zeta\rightarrow 0}\limsup_{n\rightarrow\infty}\int_0^T {\bf 1}\Big\{\max_{1\leq k\leq 4}\big\{\|L_k(\gamma^n_t)-L_k(\gamma_t)\|_1 + \|&R_k(\gamma^n_t)-L_k(\gamma_t)\|_1\big\}\leq \zeta\Big\}\nonumber\\
&\times \big|F_{H_t,\epsilon}(\gamma_t^n)- F_{H_t,\epsilon}(\gamma_t)\big|\, dt = 0.\label{eq_to_prove_convergence_F_H_t_interm}
\end{align}
Now that we have restricted to times where poles are well-behaved, 
the line integral in the definition~\eqref{eq_def_F_H_epsilon} will converge at each time. 
To prove it, we split the line integral between different regions and map the integrand the the SSEP as in Section~\ref{sec_heuristics}.

Let $\gamma\in \e$. Recall that, by definition of $\Omega\supset \e$, region $k$ ($1\leq k \leq 4$) of $\gamma$ is the graph of a $1$-Lipschitz function $f^k$ in a suitable reference frame (see Figure~\ref{fig_graph_function}):
\begin{figure}
\begin{center}
\includegraphics[width=13cm]{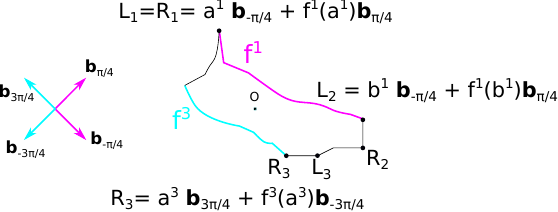} 
\caption{A curve in $\e$ and associated splitting in terms of the poles and the four regions $[R_k,L_{k+1}]$ (with $L_{k+1} := L_1$ if $k=4$) where the curve corresponds to the graph of a function $f^k$ in the reference frame $({\bf b}_{\pi/4-k\pi/2},{\bf b}_{\pi/4-(k-1)\pi/2})$ ($1\leq k\leq 4$). Region 1 is in magenta, regions 2 and 4 in black and region 3 in cyan. The coordinates of $L_1=R_1$ and $L_2$ are written in the the first reference frame $({\bf b}_{-\pi/4},{\bf b}_{\pi/4})$ and those of $R_3$ in the third reference frame $({\bf b}_{3\pi/4},{\bf b}_{-3\pi/4})$.\label{fig_graph_function}}
\end{center}
\end{figure}
\begin{align}
\gamma = \bigcup_{k=1}^4P_k(\gamma)\cup \bigcup_{k=1}^4\Big\{ u{\bf b}_{\pi/4-k\pi/2} + f^k(u){\bf b}_{\pi/4-(k-1)\pi/2}: u\in [a^k,b^k] \Big\},
\end{align}
where the extremities $a^k,b^k$ are chosen here to correspond to coordinates of $R_k(\gamma),L_{k+1}(\gamma)$ (with $L_{k+1}(\gamma):= L_1(\gamma)$ if $k=4$):
\begin{equation}
a^k 
= 
a^k(\gamma) 
:= 
R_k(\gamma)\cdot {\bf b}_{\pi/4-k\pi/2},
\quad 
b^k 
= 
b^k(\gamma)
:= 
L_{k+1}(\gamma)\cdot {\bf b}_{\pi/4-k\pi/2}
.
\label{eq_def_a_k_b_k}
\end{equation}
For $1\leq k\leq 4$, let $u\in(a^k,b^k)$ and write $\gamma(u)$ for the corresponding point of $\gamma$:
\begin{equation}
\gamma(u) 
:= 
u{\bf b}_{\pi/4-k\pi/2} + f^k(u){\bf b}_{\pi/4-(k-1)\pi/2}
.
\end{equation}
The derivative of $w\mapsto \gamma(w)$ at $u$, if it exists, is given by:
\begin{equation}
\frac{\sqrt{2}}{2}\Big({\bf b}_{\pi/4-k\pi/2} + \partial_u f^k(u){\bf b}_{\pi/4-(k-1)\pi/2}\Big) 
= 
{\bf t}(\gamma(u))
,
\end{equation}
where ${\bf t}$ is the $1$-normed tangent vector, defined in~\eqref{eq_def_t_first_region} in the first region. Recall that ${\bf T} := {\bf t}/\mathsf{v}$, where $\mathsf{v} := \|{\bf t}\|_2$ and the arclength coordinate $s(u)$ satisfy:
\begin{equation}
ds(u) 
= 
\sqrt{1+ (\partial_u f^k(u))^2}\, du 
= 
\sqrt{2}\mathsf{v}(\gamma(u))\, dx
. 
\end{equation}
For each $t\in[0,T]$ and each curve $\gamma\in\e$, one can then write for $F_{H_t,\epsilon}$:
\begin{equation}
F_{H_t,\epsilon}(\gamma) 
= 
\frac{1}{2\sqrt{2}}\sum_{k=1}^4 \int_{a^k +\epsilon/\sqrt{2}}^{b^k-\epsilon/\sqrt{2}} q^k_{H_t,\epsilon}(u,\gamma)\, du
,
\label{eq_decomp_F_H_epsilon_regions}
\end{equation}
with, for $u\in(a^k+\epsilon/\sqrt{2},b^k-\epsilon/\sqrt{2})$:
\begin{align}
q^k_{H_t,\epsilon}(u,\gamma) := \Big((\mathsf{v}^{\epsilon})^2\big[{\bf T}^\epsilon\cdot {\bf m}\big]&{\bf T}^\epsilon\Big)\big(\gamma(u)\big)\cdot \nabla H_t\big(\gamma(u)\big).
\end{align}
Recall from Definition~\ref{def_m(gamma)_gamma(eta)} that ${\bf m} = (\pm 1,\pm1)$ has a fixed value inside each region. Recall also the definitions of $v^\epsilon,{\bf T}^\epsilon$ and their relationship with ${\bf t}^\epsilon$ from~\eqref{eq_def_cont_versions_T_v}--\eqref{eq_exp_t_epsilon_sans_N_as_integral_sur_un_cadrant}: if $u\in(a^1+\epsilon/\sqrt{2},b^1-\epsilon/\sqrt{2})$ is in the first region for definiteness,
\begin{align}
{\bf t}^{\epsilon}(\gamma(u)) 
= 
\mathsf{v}^\epsilon(\gamma(u)){\bf T}^\epsilon(\gamma(u)) 
&= 
\frac{1}{\sqrt{2}\epsilon}\int_{u-\epsilon/\sqrt{2}}^{u+\epsilon/\sqrt{2}} {\bf t}(\gamma(w))\, dw
\nonumber\\
&= 
\frac{\sqrt{2}}{2}\Big({\bf b}_{-\pi/4} + \frac{f^1\big(u+\epsilon/\sqrt{2}\big)- f^1\big(u-\epsilon/\sqrt{2}\big)}{\epsilon \sqrt{2}}{\bf b}_{\pi/4}\Big) 
.
\label{eq_expression_developpee_t_epsilon}
\end{align}
For future reference, note that~\eqref{eq_expression_developpee_t_epsilon} implies the continuity of $q^k_{H_t,\epsilon}(u,\cdot)$ on $\e$ in Hausdorff distance for each $u\in(a^k+\epsilon/\sqrt{2},b^k-\epsilon/\sqrt{2})$, in the following sense: if $(\gamma^n)_n\subset\e,\gamma\in\e$ satisfy $\lim_nd_{\mathcal H}(\gamma^n,\gamma)=0$ and if $u\in \big(a^k(\gamma)+\epsilon/\sqrt{2},b^k(\gamma)-\epsilon/\sqrt{2}\big)$, then $\lim_nq^k_{H_t,\epsilon}(u,\gamma^n) = q^k_{H_t,\epsilon}(u,\gamma)$.  
Indeed, convergence in Hausdorff distance implies uniform convergence of $f^k$ around each such point $u$ and the expression~\eqref{eq_expression_developpee_t_epsilon} (or a similar one in region $k\neq 1$) yields the continuity.\\ 

We now apply the decomposition~\eqref{eq_decomp_F_H_epsilon_regions} to each $\gamma^n_t$  and use the control of the pole in~\eqref{eq_to_prove_convergence_F_H_t_interm} to express the integral on each region independently of $a^k(\gamma^n_t),b^k(\gamma^n_t)$. 
Thanks to the indicator function in~\eqref{eq_decomp_F_H_epsilon_regions}, one has for each $\zeta<\epsilon/2$, each time $t$ and each $k$ with $1\leq k \leq 4$:
\begin{equation}
a^k(\gamma_t) \leq a^k(\gamma^n_t)+\frac{\zeta}{2}\leq a^k(\gamma^n_t)+\frac{\epsilon}{2\sqrt{2}},\qquad
b^k(\gamma_t) -\frac{\epsilon}{2\sqrt{2}}\leq b^k(\gamma_t) -\frac{\zeta}{2}\leq b^k(\gamma^n_t).
\end{equation}
Since $q_{H_t,\epsilon}$ is bounded on $\e$:
\begin{equation}
\max_{1\leq k \leq 4}\sup_{t\leq T}\sup_{\gamma\in\e}\sup_{u\in(a^k+\epsilon/\sqrt{2},b^k-\epsilon/\sqrt{2})}|q_{H_t,\epsilon}(\gamma,u)|<\infty,\label{eq_bound_q}
\end{equation} 
it is enough to prove:
\begin{align}
&\lim_{\zeta\rightarrow\infty}\limsup_{n\rightarrow\infty}\int_0^T {\bf 1}\Big\{\max_{1\leq k\leq 4}\big\{\|L_k(\gamma^n_t)-L_k(\gamma_t)\|_1 + \|R_k(\gamma^n_t)-L_k(\gamma_t)\|_1\big\}\leq \zeta\Big\}
\nonumber\\
&\hspace{4cm}\times
\sum_{k=1}^4 \int_{a^k(\gamma_t) +\epsilon/\sqrt{2}}^{b^k(\gamma_t)-\epsilon/\sqrt{2}} \big|q^k_{H_t,\epsilon}(u,\gamma_t^n)-q^k_{H_t,\epsilon}(u,\gamma_t)\big|\, du\, dt
=
0
.\label{eq_continuite_line_integral_final}
\end{align}
At this point the dependence on the poles has been completely taken care of and it will be enough to study $q^k_{H_t,\epsilon}$. 
Let $I_\zeta$ be the set of times in the first line above. To prove~\eqref{eq_continuite_line_integral_final}, we prove: 
\begin{equation}
\lim_{n\rightarrow\infty}q^k_{H_t,\epsilon}(u,\gamma_t^n)= q^k_{H_t,\epsilon}(u,\gamma_t)\quad \text{for each }(t,u)\in I_\zeta\times(a^k(\gamma_t)+\epsilon/\sqrt{2},b^k(\gamma_t)-\epsilon/\sqrt{2}).\label{eq_convergence_q_H}
\end{equation}
Since $q^k_{H_t,\epsilon}$ is bounded on $\e$ for each $k$ by~\eqref{eq_bound_q},~\eqref{eq_convergence_q_H} and the dominated convergence theorem yield~\eqref{eq_continuite_line_integral_final}.\\
To prove~\eqref{eq_convergence_q_H}, notice that $\lim_nd_E(\gamma^n_\cdot,\gamma_\cdot)=0$ implies convergence in volume at almost every time: $\lim_nd_{L^1}(\gamma^n_t,\gamma_t)=0$ for almost every $t$. Now, by Lemma~\ref{lemm_volume_and_height_pole_implies_hausdorff}, for each point $x$ at 1-distance at least $\lambda>0$ from the poles, convergence in volume of $\gamma^n_t\cap B_1(x,\lambda)$ implies convergence in Hausdorff distance. But we already saw below~\eqref{eq_expression_developpee_t_epsilon} that $q^k_{H_t,\epsilon}(u,\cdot)$ is continuous in Hausdorff distance for any $u$ corresponding to a point at 1-distance at least $\epsilon$ from the poles. This implies~\eqref{eq_convergence_q_H} and concludes the proof of the continuity of~\eqref{eq_terme_IPP_pour_lsc}.\\

To conclude the proof of Proposition~\ref{prop_preuve_continuite_integrale_avec_v_moins_un}, it remains to establish:
\begin{equation}
\forall \gamma_\cdot \in E([0,T],\e),\qquad \lim_{\epsilon\rightarrow0}J_{H,\epsilon}^\beta(\gamma_\cdot) = J_{H}^\beta(\gamma_\cdot).
\end{equation}
Recalling Definition~\eqref{eq_def_J_H} of $J_H^\beta$, the above statement boils down to proving convergence of the terms in~\eqref{eq_terme_IPP_pour_lsc}. As $\gamma_\cdot\in E([0,T],\e)$ implies that $\int_0^T|\gamma_t|\, dt<\infty$, this is an immediate consequence of the expression~\eqref{eq_expression_developpee_t_epsilon} of the tangent vector, of the bound~\eqref{eq_bound_q} and of the dominated convergence theorem. This concludes the proof of Proposition~\ref{prop_preuve_continuite_integrale_avec_v_moins_un}.
\subsection{Exponential tightness}\label{appen_tightness}
In this section, we use the characterisation of compact sets of Proposition~\ref{prop_compact_sets_for_d_E} to prove exponential tightness of $\{\Prob^N_{\beta}:N\in \N_{\geq 1}\}$ for each $T>0,\beta>\log 2$ for trajectories in $E([0,T],\e)$. 
We first give a sufficient condition for exponential tightness, in Corollary~\ref{coro_sufficient_condition_for_tightness}, 
then prove that it is satisfied in the rest of the section. 
The main difficulty lies, once again, in the control of the poles. 
To start with, the following characterisation of convergence in the volume distance $d_{L^1}$ will be useful.

\begin{lemm}
Let $(G_\ell)_{\ell\geq 1}$ be a family of functions of $C^2_c(\R^2,\R)$, dense for the uniform norm $\sup_{\R^2}|\cdot|$ in the separable set $C_c(\R^2,R)$. Then $d_{L^1}$ is topologically equivalent to the distance $\tilde d_{L^1}$ defined as follows: if $\gamma^1,\gamma^2\in \Omega$ have associated droplet $\Gamma^1,\Gamma^2$,
\begin{equation}
\tilde d_{L^1}(\Gamma^1,\Gamma^2) = \sum_{\ell\geq 1}\frac{1}{2^\ell}\frac{\big|\big<\Gamma^1,G_{\ell}\big> - \big<\Gamma^2,G_\ell\big>\big|}{1+\big|\big<\Gamma^1,G_{\ell}\big> - \big<\Gamma^2,G_\ell\big>\big|}.
\end{equation}
In the sequel, $\tilde d_{L^1}$ and $d_{L^1}$ are identified.
\end{lemm}
To prove exponential tightness, we replace the condition on the Hausdorff distance, in Proposition~\ref{prop_compact_sets_for_d_E}, by a condition on the positions of the extremities $L_k,R_k,1\leq k \leq 4$ of the poles. This condition, stated next, is more convenient to check at the microscopic level.
\begin{coro}[Sufficient condition for tightness]
\label{coro_sufficient_condition_for_tightness}
Let $T>0$. Assume that, for each $G\in C^2_c(\R^2)$ and each $\epsilon>0$,
\begin{align}
&\lim_{\eta\rightarrow 0}\limsup_{N\rightarrow\infty}
\frac{1}{N}\log \Prob^N_{\beta}\bigg(\gamma^N_\cdot \in  E([0,T],\e)\cap 
\bigg\{\gamma_\cdot :\sup_{|s-t|\leq \eta}\big|\big<\Gamma_t,G\big>-\big<\Gamma_s,G\big>\big|
\label{eq_exp_tightness_pour_une_fction_test}
\\
&\hspace{3.2cm} + \sup_{h\leq \eta}\sum_{k=1}^4\int_0^{T-h}\|L_k(\gamma_t)-L_k(\gamma_{t+h})\|_1 \, dt 
\geq 
\epsilon\bigg\}\bigg)
=
-\infty
.
\nonumber
\end{align}
Then for each $H\in\C$ and $q\in\N_{\geq 1}$, there are compact sets $K_q = K_q(H)\subset E([0,T],\Omega)$ such that:
\begin{equation}
\sup_{N}\frac{1}{N}\log \mathbb{P}^N_{\beta,H}\big(\gamma^N_\cdot\in E([0,T],\e)\cap (K_q)^c\big) \leq -q.
\end{equation}
\end{coro}
\begin{proof}
As~\eqref{eq_exp_tightness_pour_une_fction_test} also holds under $\Prob^N_{\beta,H}$ for any $H\in\C$, we prove the corollary only for $H\equiv 0$.\\
Consider a sequence $G_\ell\in C^ 2_c(\R^ 2),\ell\geq 1$, dense for the uniform norm. According to~\eqref{eq_exp_tightness_pour_une_fction_test}, for each $q,n,\ell\in\N^ *$, there is $\eta = \eta(q,\ell,n)$ and $N_0 = N_0(\eta)$ such that:
\begin{align}
\sup_{N\geq N_0}&\frac{1}{N}\log\Prob^N_{\beta}\bigg(\gamma^N_\cdot\in E([0,T],\e)\cap
\bigg\{\sup_{|s-t|\leq \eta}\big|\big<\Gamma_t,G_\ell\big>-\big<\Gamma_s,G_\ell\big>\big|\label{eq_exp_tightness_quantitative_pour_une_fction_test} \\
&\hspace{1.2cm}+ \sup_{h\leq \eta}\sum_{k=1}^4\int_0^{T-h}\|L_k(\gamma_t)-L_k(\gamma_{t+h})\|_1\, dt \geq \frac{1}{n}\bigg\}\cap \Big\{\sup_{t\leq T}|\gamma_t|\leq \frac{q\ell n}{C(\beta)}\Big\}\bigg)\leq -qn\ell.\nonumber
\end{align}
For $N\leq N_0$, 
each of $\big<\Gamma^N_\cdot,G\big>$ $L_k(\gamma^N_\cdot)$ and $R_k(\gamma^N_\cdot)$ for $k$ with $1\leq k\leq 4$ is a càdlàg function when $\gamma^N_\cdot$ is a trajectory in the space of $\Omega^N_{\text{mic}}\cap \e\cap\{|\gamma|\leq qn\ell/C(\beta)\}$-valued trajectories that are càdlàg in Hausdorff distance. 
This set is complete and separable, as $\Omega^N_{\text{mic}}\cap \e\cap\{|\gamma|\leq qn\ell/C(\beta)\}$ is compact by Proposition~\ref{prop_e_closed}. 
As a result,~\eqref{eq_exp_tightness_quantitative_pour_une_fction_test} holds for $N\leq N_0$ as well up to choosing $\eta' = \eta'(q,\ell,n) \leq \eta$, hence for all $N$ in $\N_{\geq 1}$. 
For $G\in C^2_c(\R^2)$, 
let thus $m_\cdot^{L^1}\big(\big<\Gamma_\cdot,G\big>\big)$ be the Skorokhod modulus of continuity associated with the trajectory $\big(\big<\Gamma_t,G\big>\big)_t$. It satisfies:
\begin{equation}
\forall \theta>0,\qquad m_\theta^{L^1}\big(\big<\Gamma_\cdot,G\big>\big)\leq \sup_{|s-t|\leq \theta}\big|\big<\Gamma_t,G\big> - \big<\Gamma_s,G\big>\big|.
\end{equation}
Recall the control on the length obtained in Lemma~\ref{lemm_tightness_sup_length}, in particular the definition of $C(\beta)>0$. Define then $K_q = \bar U_q$, with $U_q$ as follows:
\begin{align}
U_q &:= \Big\{\sup_{t\leq T}|\gamma_t|\leq \frac{q}{C(\beta)}\Big\}\nonumber\\
&\qquad\cap \bigcap_{\ell,n\in\N^ *} \bigg\{m_{\eta'}^{L^1}\big(\big<\Gamma_\cdot,G_\ell\big>\big) + \sup_{h\leq \eta'}&\sum_{k=1}^4\int_0^{T-h}\|L_k(\gamma_t)-L_k(\gamma_{t+h})\|_1\, dt 
\leq 
\frac{1}{n}\bigg\}
.
\end{align}
By Proposition~\ref{prop_compact_sets_for_d_E} and Lemma~\ref{lemm_lien_hausdorff_et_volume_et_poles}, $K_q$ is compact. Moreover, it satisfies by construction:
\begin{equation}
\sup_{N\in\N_{\geq 1}}\frac{1}{N}\log \Prob^ N_{\beta}\big(\gamma^N_\cdot\in E([0,T],\e)\cap(K_q)^ c\big)\leq -q.
\end{equation}
This concludes the proof of exponential tightness inside $E([0,T],\e)$.
\end{proof}
We conclude the section by a proof of relative compactness of the laws of the dynamics for short time.
\begin{coro}\label{coro_tightness_small_time}
Let $\beta>\log 2$, $H\in\C$ and $(\mu^N)_N$ be a sequence of probability measures on $\big(\e,d_{L^1}\big)$ converging weakly to $\delta_{\gamma^{\mathrm{ref}}}$. 
Assume further that there is $t_0>0$ such that: 
\begin{equation}
\lim_{N\rightarrow\infty}\Prob^{\mu^N}_{\beta,H}(\gamma^N_\cdot\in E([0,t_0],\e)) = 1.\label{eq_tightness_pas_exponentielle_corollary}
\end{equation}
Still write $d_E$ for the distance~\eqref{eq_def_d_E_appen} defined on a time interval $[0,t_0]$. Then the set $\{\Prob^{\mu^N}_{\beta,H}: N\in\N_{\geq 1}\}$ is relatively compact in the set of probability measures on $\big(E([0,t_0],\Omega),d_E\big)$ and its limit points are supported on trajectories in $E([0,t_0],\e)$ that are continuous in $d_{L^1}$ distance.
\end{coro}
\begin{proof}
By the direct half of  Prokhorov theorem (Theorem 5.1 in \cite{Billingsley1999}), relative compactness is implied by tightness. Let us therefore prove that $\{\Prob^N_{\beta,H}\}_N$ is tight in $E([0,t_0],\Omega)$. The proof is a bit indirect because Corollary~\ref{coro_sufficient_condition_for_tightness} only gives a good control of trajectories in $E([0,t_0],\e)$, not in $E([0,t_0],\Omega)$. 
For each measurable set $B\subset E([0,t_0],\Omega)$, write:
\begin{align}
\Prob^{\mu^N}_{\beta,H}(\gamma^N_\cdot\in B) &= \Prob^{\mu^N}_{\beta,H}\big(\gamma^N_\cdot\in B\cap E([0,t_0],\e)\big) + \Prob^{\mu^N}_{\beta,H}\big(\gamma^N_\cdot\in B\cap E([0,t_0],\e)^c\big). 
\end{align}
Fix $\eta>0$. By Assumption~\eqref{eq_tightness_pas_exponentielle_corollary}, there is $N_0(\eta)\in\N_{\geq 1}$ such that:
\begin{equation}
\forall N\geq N_0(\eta),\qquad 
\Prob^{\mu^N}_{\beta,H}\big(\gamma^N_\cdot\in E([0,t_0],\e)^c\big) 
\leq 
\eta
.
\end{equation}
On the other hand, the initial conditions $(\mu^N)_N$ are probability measures on $\e$, which is separable and complete for $d_{L^1}$ (seeing $\e$ as a closed subset of $L^1(\R^2)$ by identifying curves with the indicator functions of their associated droplets). It follows from the converse half of Prokhorov's theorem (Theorem 5.2 in \cite{Billingsley1999}) that $(\mu^N)_N$ is tight. For each $\eta>0$, let thus $K^0_\eta\subset \e$ be a compact set for the distance $d_{L^1}$, such that:
\begin{equation}
\forall N\in\N_{\geq 1},\qquad \Prob^{\mu^N}_{\beta,H}\Big(\gamma^N_0\in\big(K^0_\eta\big)^c\Big) 
= 
\mu^N \Big(\gamma^N\in\big(K^0_\eta\big)^c\Big)
\leq 
\eta
.
\end{equation}
Then, for each $q\in\N_{\geq 1}$ with $e^{-q}\leq \eta$, recalling the definition  of $K_q$ from Corollary~\ref{coro_sufficient_condition_for_tightness}:
\begin{equation}
\forall N\in\N_{\geq 1},\qquad 
\Prob^{\mu^N}_{\beta,H}\Big(\gamma^N_\cdot\in E([0,t_0],\e)\cap \big(K^0_\eta\cap K_q\big)^c\Big)
\leq 
2\eta
.
\end{equation}
As a result, for $N\geq N_0(\eta)$, we have:
\begin{equation}
\Prob^{\mu^N}_{\beta,H}\Big(\gamma^N_\cdot \in \big(K^0_\eta\cap K_q\big)^c\Big)\leq 3\eta.
\end{equation}
Now, each $\Prob^{\mu^N}_{\beta,H}$ for $N<N_0(\eta)$ is a probability measure on the complete, separable set $D_{\mathcal H}([0,t_0],\Omega)$ of càdlàg trajectories in Hausdorff distance with values in $\Omega$. In particular, for each $N<N_0(\eta)$, $\Prob^{\mu^N}_{\beta,H}$ is tight: there is a compact set $K^N_\eta\subset D_{\mathcal H}([0,t_0],\Omega)$ such that:
\begin{equation}
\Prob^{\mu^N}_{\beta,H}\Big(\gamma^N_{\cdot} \in \big(K^N_\eta\big)^c\Big)\leq \eta,\qquad N<N_0(\eta).
\end{equation}
Since convergence in $D_{\mathcal H}([0,t_0],\Omega)$ implies convergence for $d_E$, each $K^N_\eta$ is also a compact set for $d_E$, whence the proof of tightness in $E([0,t_0],\Omega)$:
\begin{equation}
\forall N\in\N_{\geq 1},\qquad \Prob^{\mu^N}_{\beta,H}\Big(\gamma^N_\cdot\in  \big(K^0_\eta\cap K_q\big)^c\cap\bigcap_{M<N_0(\eta)}\big(K^M_\eta\big)^c\Big)\leq 3\eta.
\end{equation}
It remains to check that $\{\Prob^{\mu^N}_{\beta,H}:N\in\N_{\geq 1}\}$ concentrates on trajectories that are continuous in volume, $d_{L^1}$ distance. This is a standard consequence of the estimate~\eqref{eq_exp_tightness_pour_une_fction_test}, so we conclude the proof here.
\end{proof}
\subsubsection{Estimate in $L^1(\R^2)$ topology}\label{exp_tightness_volume}
In this section, we prove exponential tightness in volume, i.e. in $L^1(\R^2)$. 
\begin{lemm}\label{lemm_exp_tightness}
Let $T>0$ and $G\in C^2_c(\R^2)$. Then, for each $\epsilon>0$:
\begin{align}
\lim_{\eta\rightarrow 0}\limsup_{N\rightarrow\infty}\frac{1}{N}\log \Prob^N_{\beta}\Big(\gamma^N_\cdot\in E([0,T],\e)\cap\Big\{\sup_{|t-s|\leq \eta} \big|\big<\Gamma^N_t,G\big>-\big<\Gamma^N_s,G\big>\big|>\epsilon\Big\}\Big) 
=
-\infty
.
\end{align}
The result also holds under $\Prob^N_{\beta,H}$ for $H\in\C$ by Corollary~\ref{coro_change_measure_sous_exp}.
\end{lemm}
\begin{proof}
Compared to Chapter 10 in \cite{Kipnis1999}, the only subtleties to prove Lemma~\ref{lemm_exp_tightness} are in the introduction of the condition $E([0,T],\e)$ and in the control of the change in volume induced by the motion of the poles. As these do not present any particular difficulty, the proof is omitted.
\end{proof}
\subsubsection{Precise control of the slope and volume around the poles}\label{sec_precise_control_slope_volume}
In this section and the next, 
we prove the estimate on the poles appearing in~\eqref{eq_exp_tightness_pour_une_fction_test}. 
As preliminary, we prove in this section that the volume beneath each pole is fixed by the reservoir-like behaviour induced by the dynamics. 
This will be used in Section~\ref{sec_tightness_poles} to argue that a displacement of the poles must result in a change in volume, which is unlikely for short time by Lemma~\ref{lemm_exp_tightness}.

The estimate of the volume beneath a pole relies on the microscopic estimate of the slope at the pole, obtained in Corollary~\ref{coro_1_2_blocks_deviations_slope}. All results are stated for $\Prob^N_\beta$ but apply to $\Prob^N_{\beta,H}$ for $H\in\C$ by Corollary~\ref{coro_change_measure_sous_exp}.
\begin{lemm}[Control of the deviations of the width at distance $\alpha>0$ below the pole]\label{lemm_deviations_g_pm_alpha}
Let $\beta>\log 2$. 
For $\alpha>0$ and $\gamma\in\e$, let $g^+(\alpha) = g^+(\alpha) (\gamma)$ be the width of the horizontal segment of $\gamma$ at height $z_1(\gamma)-\alpha$ to the right of $L_1(\gamma)$ (see Figure~\ref{fig_link_slope_width}). Define similarly $g^-(\alpha)$ to the left of $L_1(\gamma)$. For each $\delta,\eta>0$:
\begin{align}
\lim_{\alpha\rightarrow0}\limsup_{N\rightarrow\infty}\frac{1}{N}\log \Prob^ {N}_{\beta}\bigg(\gamma^N_\cdot\in E([0,T],\e)\cap\Big\{\frac{1}{T}\int_0^ {T} {\bf 1}_{|\alpha^{-1}g^{\pm}(\alpha)-(e^ {\beta}-1)|\geq \delta}\, dt >\eta\Big\}\bigg)
=
-\infty
.
\end{align}
\end{lemm}
\begin{proof}
The proof is a formalisation of Figure~\ref{fig_link_slope_width}: 
since the slope on both sides of the pole is fixed by Corollary~\ref{coro_1_2_blocks_deviations_slope}, 
we can obtain upper and lower bounds on $g^\pm$ in terms of $\beta$. \\  
Take $\zeta^1,\zeta^2>0$ to be determined later and $\theta>0$ which will be small. The proof of the result is similar for $g^+$ and $g^-$, so we focus on $g^+$. It is sufficient to prove:
\begin{align}
&\limsup_{\zeta^1,\zeta^2\rightarrow 0}\limsup_{\alpha\rightarrow0}\limsup_{N\rightarrow\infty}\frac{1}{N}\log \Prob^ {N}_{\beta}\bigg(\gamma_\cdot^N \in E([0,T],\e);\\
&\hspace{3cm}\frac{1}{T}\int_0^ {T} {\bf 1}_{|\alpha^{-1}g^{+}(\alpha)-(e^ {\beta}-1)|\geq \delta}{\bf 1}_{|\xi^{+,\zeta^1 N}_{L_1}-e^{-\beta}|\leq \theta}{\bf 1}_{|\xi^{+,\zeta^2 N}_{L_1}-e^{-\beta}|\leq \theta}\,  dt 
>\eta/3\bigg)
=
-\infty
.
\nonumber
\end{align}
Consider the event bearing on $\xi^{+,\zeta^1 N}_{L_1}$. It enforces:
\begin{equation}
\xi^{+,\zeta^1 N}_{L_1} \in[e^{-\beta}-\theta,e^{-\beta}+\theta].
\end{equation}
\begin{figure}
\begin{center}
\includegraphics[width=13cm]{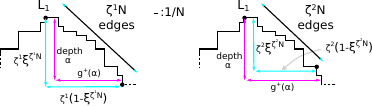} 
\caption{Neighbourhood of the north pole of a microscopic curve for which $g^+(\alpha)$ is drawn (horizontal magenta arrow). On the left figure, black dots delimit the $\zeta^1 N$ edges to the right of $L_1$, with $\zeta^1$ chosen to ensure $\zeta^1\xi^{+,\zeta^1 N}\geq \alpha$ (this quantity represented by cyan arrows). On the right figure, $\zeta^2$ is similarly chosen so that $\zeta^2\xi^{+,\zeta^2 N}\leq \alpha$ (in cyan arrows again). Since bounds on $\zeta^1\xi^{+,\zeta^1 N},\zeta^2\xi^{+,\zeta^2 N}$ are available in terms of $\beta$, $g^+(\alpha)$ can be bounded. \label{fig_link_slope_width}}
\end{center}
\end{figure}
Choose $\zeta^1$ such that $(e^{-\beta}-\theta)\zeta^1=\alpha$. Then $\zeta^1\xi^{+,\zeta^1 N}_{L_1}\geq \alpha$ (see Figure~\ref{fig_link_slope_width}). By definition, $g^+(\alpha)$ must thus be smaller than $\zeta^1(1-\xi^{+,\zeta^1 N}_{L_1})$:
\begin{align}
\xi^{+,\zeta^1 N}_{L_1} &\in[e^{-\beta}-\theta,e^{-\beta}+\theta]\text{ and }(e^{-\beta}-\theta)\zeta^1=\alpha\nonumber\\
&\quad \Rightarrow\quad \alpha^{-1}g^+(\alpha) \leq \frac{1-e^{-\beta}+\theta}{e^{-\beta}-\theta}=e^{\beta}-1+O(\theta),
\end{align}
where $O(\theta)$ is a positive function. Similarly, choose $\zeta^2$ such that $(e^{-\beta}+\theta)\zeta^2=\alpha$. Then $g^+(\alpha)\geq\zeta^2(1-\xi^{+,\zeta^2 N}_{L_1})$, thus:
\begin{align}
\xi^{+,\zeta^2 N}_{L_1} &\in[e^{-\beta}-\theta,e^{-\beta}+\theta]\text{ and }(e^{-\beta}+\theta)\zeta^2=\alpha\nonumber\\
&\quad\Rightarrow\quad \alpha^{-1}g^+(\alpha) \geq \frac{1-e^{-\beta}-\theta}{e^{-\beta}+\theta}=e^{\beta}-1-O(\theta).
\end{align}
$O(\theta)$ is again a positive function. Taking $\theta$ small enough to contradict $|\alpha^{-1}g ^{+}(\alpha) - (e^{\beta}-1)|\geq \delta$ concludes the proof.
\end{proof}
\begin{lemm}[Control of the deviations of the volume at distance $\alpha>0$ below the pole]\label{lemm_controle_deviations_volume}
For $\gamma\in\e$ with associated droplet $\Gamma$, let $V^ \alpha = V^ \alpha(\gamma)$ be defined as:
\begin{equation}
V^ \alpha(\gamma) = \alpha^{-2}\big|\{ x\in\Gamma: x\cdot {\bf b}_{\pi/2}\geq z_1(\gamma)-\alpha\}\big|.
\end{equation}
Then for each $\beta>\log 2$ and each $\delta,\eta>0$:
\begin{align}
\lim_{\alpha\rightarrow0}\limsup_{N\rightarrow\infty}\frac{1}{N}\log \Prob^ {N}_{\beta}\bigg(\gamma^N_\cdot\in E([0,T],\e)\cap\Big\{\frac{1}{T}\int_0^ {T} {\bf 1}_{|V^{\alpha}-(e^ {\beta}-1)|>\delta}\, dt >\eta\Big\}\bigg)
=
-\infty
.
\end{align}
\end{lemm}
\begin{proof}
The idea is to use Lemma~\ref{lemm_deviations_g_pm_alpha} at multiple depths to prove that a droplet beneath the pole must be approximately triangular. Fix $\ell\in\N_{\geq 1}$ and $\theta>0$ to be chosen later. By Lemma~\ref{lemm_deviations_g_pm_alpha}, it is sufficient to prove:
\begin{align}
\lim_{\alpha\rightarrow0}\limsup_{N\rightarrow\infty}\frac{1}{N}\log \Prob^ {N}_{\beta}\bigg(&\gamma^N_\cdot\in E([0,T],\e);
\label{eq_to_prove_lemm_controle_devs_volume}\\
&\frac{1}{T}\int_0^ {T} {\bf 1}_{|V^{\alpha}-(e^ {\beta}-1)|>\delta}{\bf 1}_{\forall j\in\{1,...,\ell\}, \big|\frac{\ell}{j\alpha}g^{\pm}(j\alpha/\ell)-(e^{\beta}-1)\big|\leq \theta}\, dt >\eta/2\bigg)
=
-\infty
.
\nonumber
\end{align}
By definition of $g^{\pm}(\alpha) = g^\pm(\alpha)(\gamma)$ for $\alpha>0$ and $\gamma\in \e$ (see Lemma~\ref{lemm_deviations_g_pm_alpha}), the quantity $V^\alpha(\gamma)$ satisfies:
\begin{equation}
V^\alpha(\gamma) 
= 
\alpha^{-2}\int_{0}^\alpha (g^+(u)+g^-(u))\, du
.
\end{equation}
As elements of $\Omega$ have $1$-Lipschitz boundaries, if a curve $\gamma\in\Omega$ is such that $(\ell/j\alpha)g^\pm(j\alpha/\ell) \in [e^{\beta}+1-\theta,e^{\beta}-1-\theta]$ for each $1\leq j\leq \ell$, then:
\begin{equation}
\alpha^{2}V^{\alpha}(\gamma) 
= 
|\{x\in\Gamma : x\cdot {\bf b}_{\pi/2}\geq z_1-\alpha\}|\geq 2\sum_{j=1}^{\ell-1} \frac{j}{\ell\alpha}(e^{\beta}+1-\theta)\times\frac{\alpha}{\ell} 
=
\frac{\ell-1}{\ell}(e ^{\beta}-1-\theta)\alpha^2
.
\end{equation}
Similarly,
\begin{equation}
\alpha^{2}V^{\alpha}(\gamma)\leq 2\sum_{j=1}^{\ell} \frac{j}{\ell\alpha}(e^{\beta}+1-\theta)\times\frac{\alpha}{\ell} =\frac{\ell+1}{\ell}(e ^{\beta}-1+\theta)\alpha^{2}.
\end{equation}
To conclude the proof, it remains to take $\ell,\theta$ such that the indicator functions appearing in~\eqref{eq_to_prove_lemm_controle_devs_volume} bear on incompatible events. This is achieved provided:
\begin{align}
\frac{\ell-1}{\ell}(e^{\beta}-1-\theta) \geq  e^{\beta}-1-\delta\quad \text{and}\quad \frac{\ell+1}{\ell}(e^{\beta}-1+\theta)
\leq  
e^{\beta}-1+\delta
.
\end{align}
\end{proof}
\subsubsection{Tightness in $L^1([0,T])$ for the trajectory of the poles}\label{sec_tightness_poles}
In this section, we conclude the proof of~\eqref{eq_exp_tightness_pour_une_fction_test} by providing the estimate on the motion of the poles, more precisely on the components $w_k,z_k$ of the $L_k$, $1\leq k \leq 4$. We prove the result for $z_1$, the other seven coordinates being similar (see Remark~\ref{rmk_tightness_w_1} below).
\begin{lemm}[Tightness in $L^{1}$ distance for $z_1$]\label{lemm_tightness_y_max}
Let $\beta>\log 2$ and $A,\epsilon>0$. Then:
\begin{align}
\lim_{\eta\rightarrow 0}\limsup_{N\rightarrow\infty}\frac{1}{N}\log \Prob^ {N}_{\beta}\bigg(\gamma^N_\cdot&\in E([0,T],\e)\cap\Big\{\ \sup_{t\leq T}|\gamma^N_t|\leq A\Big\}\label{eq_tightness_y_max}\\
&\quad 
\cap\Big\{\sup_{h\leq \eta}\frac{1}{T}\int_0^ {T-h} |z_1(\gamma^N_{t+h}) - z_1(\gamma^N_t)|\, dt>\epsilon\Big\}\bigg)
=
-\infty
.
\nonumber
\end{align}
\end{lemm}
\begin{figure}
\begin{center}
\includegraphics[width=11cm]{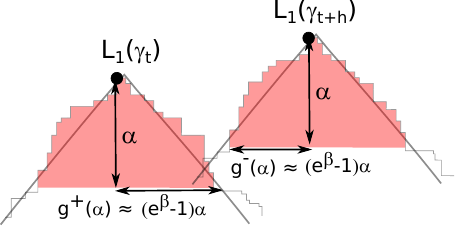} 
\caption{Neighbourhood of the north pole of a microscopic curve at times $t$, $t+h$. Under the contour dynamics, sufficiently close to the poles (corresponding to the parameter $\alpha\ll 1$ on the figure), the poles stand atop a triangular shape (red shaded area) with a slope fixed in terms of $\beta$. Equivalently, the width $g^{\pm}(\alpha)$ is approximately given by $\alpha(e^{\beta}-1)$. 
A displacement of the poles leads to a shift of the triangle and thus implies a change in volume, i.e. for the distance $d_{L^1}$.\label{fig_kinks}}
\end{center}
\end{figure}
\begin{proof}
Let $\gamma_\cdot \in E([0,T],\e)$. The idea is the following. Results of Section~\ref{sec_precise_control_slope_volume} imply that the pole dynamics creates triangular shapes in the curves, with a slope fixed in terms of $\beta$ (see Figure~\ref{fig_kinks}). Moving a pole thus means moving one of these triangles, which has a volume. In that way motion of the poles is linked with a displacement of the volume, which we know cannot happen instantaneously by Lemma~\ref{lemm_exp_tightness}. \\

For each $h\leq \eta$ and each $t\in[0,T-h]$, write $\Delta_h(t)$ for the difference $|z_1(\gamma_{t+h}) - z_1(\gamma_t)|$ for brevity. By Definition~\ref{def_CI} of the initial condition and owing to the bound on the length, $z_1(\gamma_\cdot)$ is bounded by $A+C_0$ for some $C_0=C_0(\gamma^{\mathrm{ref}})>0$ on the event in~\eqref{eq_tightness_y_max}. Equation~\eqref{eq_tightness_y_max} is thus proven as soon as:
\begin{align}
&\lim_{\eta\rightarrow 0}\limsup_{N\rightarrow\infty}\\
&\quad\frac{1}{N}\log \Prob^ {N}_{\beta}\bigg(\gamma^N_\cdot\in E([0,T],\e)
\cap\Big\{ \sup_{h\leq \eta}\frac{1}{T}\int_0^ {T-h}{\bf 1}_{\Delta_h(t)\geq \epsilon/2}\, dt
>\frac{\epsilon}{2(A+C_0)}\Big\}\bigg)
=
-\infty
.
\nonumber
\end{align}
Fix $\delta>0$ that will be chosen small enough in the following. Recall that, for $\alpha>0$, $V^\alpha$ is the volume below the pole times $\alpha^{-2}$, see Lemma~\ref{lemm_controle_deviations_volume}. For $\alpha>0$ and $t\in[0,T]$, define then $\Delta V^ {\alpha}(t)$ as follows :
\begin{equation}
\Delta V^ \alpha(t) = |V^ \alpha(\gamma_t) - (e^ \beta-1)|.
\end{equation}
Lemma~\ref{lemm_controle_deviations_volume} gives:
\begin{align}
\lim_{\alpha\rightarrow 0}\limsup_{N\rightarrow\infty}\frac{1}{N}\log \Prob^ {N}_{\beta}\bigg(\gamma^N_\cdot\in E([0,T],\e)\cap \Big\{\frac{1}{T}\int_0^{T}{\bf 1}_{\Delta V^ \alpha(t)>\delta} \, dt>\frac{\epsilon}{6(A+C_0)}\Big\}\bigg)=-\infty.
\end{align}
Notice in addition that:
\begin{equation}
\Big\{\sup_{h\leq \eta}\frac{1}{T}\int_0^{T-h}{\bf 1}_{\Delta V^ \alpha(t+h)>\delta}\, dt>\frac{\epsilon}{6(A+C_0)}\Big\} \subset \Big\{ \frac{1}{T}\int_0^{T}{\bf 1}_{\Delta V^ \alpha(t)>\delta} \,dt>\frac{\epsilon}{6(A+C_0)}\Big\}.
\end{equation}
As a result, if $\lambda$ denotes $T^{-1}$ times the Lebesgue measure on $[0,T]$,~\eqref{eq_tightness_y_max} holds as soon as:
\begin{align}
&\lim_{\alpha\rightarrow 0}\limsup_{\eta\rightarrow 0}\limsup_{N\rightarrow\infty}\label{eq_tightness_y_max_transit_0}\frac{1}{N}\log \Prob^ {N}_{\beta}\bigg(\gamma^N_\cdot\in E([0,T],\e);\\
&\hspace{3cm}
\sup_{h\leq \eta}\lambda\Big[\Delta_h(t)\geq \epsilon/2, |\Delta V^ \alpha(t)|\leq \delta,|\Delta V^ \alpha(t+h)|\leq \delta\Big]>\frac{\epsilon}{6(A+C_0)}\bigg)
=
-\infty
.
\nonumber
\end{align}
By Lemma~\ref{lemm_exp_tightness} on exponential tightness in $d^S_{L^1}$ topology,~\eqref{eq_tightness_y_max_transit_0} is proven as soon as the following holds:
\begin{align}
&\lim_{\alpha\rightarrow 0}\limsup_{\eta\rightarrow 0}\limsup_{N\rightarrow\infty}\frac{1}{N}\log \Prob^ {N}_{\beta}\bigg(\gamma^N_\cdot\in E([0,T],\e);
\ \sup_{t\leq T}|\gamma^N_t|\leq A;
\nonumber\\
&\hspace{3cm}
\sup_{\substack{(s,t)\in[0,T]^2\\ |s-t|\leq \eta}}d_{L^ 1}(\gamma^N_s,\gamma^N_t)<\frac{\alpha^2 (e^{\beta}-1)}{2};
\label{eq_tightness_y_max_transit_1}\\
&\hspace{3cm} 
\sup_{h\leq \eta}\lambda\Big[\Delta_h(t)\geq \epsilon/2, |\Delta V^ \alpha(t)|\leq \delta,|\Delta V^ \alpha(t+h)|\leq \delta\Big]>\frac{\epsilon}{6(A+C_0)} \bigg)
=
-\infty
.
\nonumber
\end{align}
Take $\delta< (e^{\beta}-1)/2$ and an arbitrary $\alpha\in(0,\epsilon/2]$. Let us prove that the event in the probability in~\eqref{eq_tightness_y_max_transit_1} is empty. For any trajectory $(\gamma^N_t)_{t\in[0,T]}$ in this event, there must be $t\in[0,T]$ and $h<\eta$ such that, simultaneously:
\begin{itemize}
	\item The north poles of $\Gamma^N_t,\Gamma^N_{t+h}$ are at vertical distance at least $\epsilon/2$, so that either $\{x\in \Gamma^N_t:x\cdot {\bf b}_{\pi/2}\geq z_1(\gamma^N_t)-\alpha\} \cap \Gamma^N_{t+h} = \emptyset$ or $\{x\in \Gamma^N_{t+h}:x\cdot {\bf b}_{\pi/2}\geq z_1(\gamma^N_{t+h})-\alpha\} \cap \Gamma^N_{t} = \emptyset$.
	\item Recall that $V^\alpha(t) = \alpha^{-2}|\{x\in\Gamma^N_t:x\cdot {\bf b}_{\pi/2}\geq z_1(\gamma^N_t)-\alpha\}|$. $V^\alpha(t)$ and $V^{\alpha}(t+h)$ are both bounded below by $e^{\beta}-1 -\delta>(e^{\beta}-1)/2$ so that, by the first point, the difference in volume between $\Gamma^N_t$ and $\Gamma_{t+h}$ is at least $\alpha^2(e^{\beta}-1)/2$.
	\item Yet, $d_{L^1}(\gamma^N_t,\gamma^N_{t+h})< \alpha^2(e^{\beta}-1)/2$, which is incompatible with point 2. This concludes the proof.
\end{itemize}
\end{proof}
\begin{rmk}\label{rmk_tightness_w_1}
The proof for $z_k$ ($2\leq k\leq 4$) is identical to the above. For the $w_k$, i.e. $L_1\cdot {\bf b}_0, L_2\cdot {\bf b}_{\pi/2},L_3\cdot {\bf b}_0$ and $L_4\cdot {\bf b}_{\pi/2}$, slight modifications are required: in addition to the indicator functions on the volumes $\Delta V^\alpha(t)<\delta$, $\Delta V^\alpha(t+h)<\delta$, one has to introduce the events $\{|\alpha^{-1}g^\pm_\alpha(t) - (e^\beta-1)|<\delta\}, \{|\alpha^{-1}g^\pm_\alpha(t+h) - (e^\beta-1)|<\delta\}$, where $g^\pm_\alpha$, the width of the level at distance $\alpha$ beneath the pole, is defined in Lemma~\ref{lemm_deviations_g_pm_alpha}.

The idea is then that, e.g. for $w_1$, if $\alpha$ is taken small enough as a function of $\epsilon$ and $\beta$ ($(e^{\beta}-1)^{-1}\epsilon/6$ works), then $|w_1(\gamma^N_t) - w_1(\gamma^N_{t+h})|\geq \epsilon/2$ implies that this difference must be larger than $\min\{g_\alpha^+(t+h)+g_\alpha^-(t), g_\alpha^+(t)+g_\alpha^-(t+h)\}$ (see Figure~\ref{fig_kinks}). 

But then this means that the set of points above $z_1(\gamma^N_t)-\alpha$ in $\Gamma^N_t$ and the set of points above $z_1(\gamma^N_{t+h})-\alpha$ in $\Gamma^N_{t+h}$ are disjoint. Thanks to the indicator functions on the volumes $\Delta V^\alpha$, this implies a difference in volume, which is again impossible for $\eta$ small enough.\demo
\end{rmk}
\end{appendices}
\section*{Acknowledgements}
The author would like to thank his Ph.D. advisor Thierry Bodineau for continuous help and discussions on the content of this article, as well as an anonymous referee for his extremely useful feedback and an improved proof of Lemma A.1.
Part of this work was done while the author was supported by the European Research Council under the European Union's Horizon 2020 research and innovation programme
(grant agreement No.~851682 SPINRG).

\bibliographystyle{alpha}
\bibliography{Contour_dynamics_final}

\end{document}